\documentclass[11pt]{amsart}
\usepackage[T1]{fontenc}
\usepackage{amssymb,amsmath, amsthm, amsfonts, tikz-cd}
\usepackage{graphicx}
\usepackage{listings}
\usepackage{lstautogobble}
\usepackage{enumerate}
\usepackage[shortlabels]{enumitem}
\usepackage{thmtools}
\usepackage{thm-restate}
\usepackage{amsthm}
\usepackage{verbatim}

\usepackage{mathtools}
\usepackage{physics}
\usepackage{color}
\usepackage{hyperref}
\usepackage[capitalise]{cleveref}
\usepackage[bottom]{footmisc}
\crefformat{equation}{(#2#1#3)}

\usepackage{float}
\restylefloat{table}

\usepackage{mathrsfs}
\setlist{  
  listparindent=\parindent,
  parsep=0pt,
}
\theoremstyle{plain}
\newtheorem{thm}{Theorem}[section]
\newtheorem{prop}[thm]{Proposition}
\newtheorem{lemma}[thm]{Lemma}
\newtheorem{cor}[thm]{Corollary}

\theoremstyle{definition}
\newtheorem{mydef}[thm]{Definition}
\newtheorem{ex}[thm]{Example}

\newtheorem{remark}[thm]{Remark}

\newtheorem{ques}{Question}[section]

\newtheorem*{sum-invol}{Summary of result 1 - $\mathcal{H}_{n}$ is involution}
\newtheorem*{sum-GP_ham}{Summary of result 2 - GP Hamiltonian flows}

\numberwithin{equation}{section} 

\DeclarePairedDelimiter\ipp{\langle}{\rangle}
\DeclarePairedDelimiter{\brak}{\lbrack}{\rbrack}
\DeclarePairedDelimiter{\paren}{\lparen}{\rparen}

\DeclarePairedDelimiter{\jp}{\langle}{\rangle}

\DeclareMathOperator{\supp}{supp}
\DeclareMathOperator{\ssupp}{sing\,supp}

\DeclareMathOperator{\sgn}{sgn}

\renewcommand{\det}{\mathrm{det}}

\newcommand{\V}{{\mathcal{V}}}

\newcommand{\p}{{\partial}}

\renewcommand{\d}{\delta}

\newcommand{\tens}{\overset{\otimes}{,}}

\newcommand{\R}{{\mathbb{R}}}
\newcommand{\C}{{\mathbb{C}}}
\newcommand{\N}{{\mathbb{N}}}

\newcommand{\Z}{{\mathbb{Z}}}
\newcommand{\K}{{\mathbb{K}}}
\newcommand{\Ss}{{\mathbb{S}}}
\renewcommand{\H}{{\mathcal{H}}}
\newcommand{\T}{{\mathbb{T}}}
\newcommand{\g}{{\mathfrak{g}}}
\newcommand{\G}{{\mathfrak{G}}}

\renewcommand{\L}{{\mathcal{L}}}
\newcommand{\Sc}{{\mathcal{S}}}
\newcommand{\A}{{\mathcal{A}}}

\newcommand{\D}{\mathcal{D}}
\newcommand{\m}{\mathfrak{m}}

\newcommand{\wt}{\widetilde}
\newcommand{\tl}{\tilde}

\newcommand{\ol}{\overline}
\newcommand{\ul}{\underline}

\newcommand{\ux}{\underline{x}}

\newcommand{\W}{{\mathbf{W}}}
\newcommand{\E}{{\mathbf{E}}}

\DeclareMathOperator{\Sym}{Sym}

\newcommand{\wh}{\widehat}
\newcommand{\ueta}{\underline{\eta}}

\newcommand{\uxi}{\underline{\xi}}
\renewcommand{\ueta}{\underline{\eta}}
\newcommand{\F}{{\mathcal{F}}}

\DeclareMathOperator{\WF}{WF}

\def\XXint#1#2#3{{\setbox0=\hbox{$#1{#2#3}{\int}$ }
\vcenter{\hbox{$#2#3$ }}\kern-.6\wd0}}

\let\oldtocsection=\tocsection
 
\let\oldtocsubsection=\tocsubsection
 
\let\oldtocsubsubsection=\tocsubsubsection
 
\renewcommand{\tocsection}[2]{\hspace{0em}\oldtocsection{#1}{#2}}
\renewcommand{\tocsubsection}[2]{\hspace{1em}\oldtocsubsection{#1}{#2}}
\renewcommand{\tocsubsubsection}[2]{\hspace{2em}\oldtocsubsubsection{#1}{#2}}

\title[Poisson commuting energies for a system of infinitely many bosons]{Poisson commuting energies for a system of infinitely \\ many bosons}

\author[D. Mendelson]{Dana Mendelson$^1$} 
\address{$^1$  
Department of Mathematics \\ 
University of Chicago\\  
5734 S. University Avenue \\ 
Chicago, IL  60637}
\email{dana@math.uchicago.edu}
\thanks{\hspace{-4.6mm} $^1$ D.M. is funded in part by NSF DMS-1800697.}

\author[A. Nahmod]{Andrea R. Nahmod$^2$}
\address{$^2$ 
Department of Mathematics \\ University of Massachusetts\\  710 N. Pleasant Street, Amherst MA 01003}
\email{nahmod@math.umass.edu}
\thanks{$^2$ A.N. is funded in part by NSF DMS-1463714, DMS 1800852, and the Simons Foundation.}

\author[N. Pavlovi\'{c}]{Nata\v{s}a Pavlovi\'c$^3$}
\address{$^3$  
Department of Mathematics\\ 
University of Texas at Austin\\ 
2515 Speedway, Stop C1200\\
Austin, TX 78712}
\email{natasa@math.utexas.edu}
\thanks{$^3$ N.P. is funded in part by NSF DMS-1516228 and NSF DMS-1840314.}

\author[M. Rosenzweig]{Matthew Rosenzweig$^4$}
\address{$^4$  
Department of Mathematics\\ 
University of Texas at Austin\\ 
2515 Speedway, Stop C1200\\
Austin, TX 78712}
\email{rosenzweig.matthew@math.utexas.edu}
\thanks{$^4$ M.R. is funded in part by NSF DMS-1516228 and a Provost Excellence Graduate Fellowship from the University of Texas at Austin.}

\author[G. Staffilani]{Gigliola Staffilani$^5$}
\address{$^5$ Department of Mathematics\\
Massachusetts Institute of Technology\\ 
77 Massachusetts Avenue,  Cambridge, MA 02139}
\email{gigliola@math.mit.edu}
\thanks{$^5$ G.S. is funded in part by NSF DMS-1462401, DMS-1764403, and the Simons Foundation.}

\begin{document}
\maketitle{}

\begin{abstract}
We consider the cubic Gross-Pitaevskii (GP) hierarchy in one spatial dimension. We establish the existence of an infinite sequence of observables such that the corresponding trace functionals, which we call ``energies,'' commute with respect to the weak Lie-Poisson structure defined by the authors in \cite{MNPRS1_2019}. The Hamiltonian equation associated to the third energy functional is precisely the GP hierarchy. The equations of motion corresponding to the remaining energies generalize the well-known nonlinear Schr\"odinger hierarchy, the third element of which is the one-dimensional cubic nonlinear Schr\"odinger equation. This work provides substantial evidence for the GP hierarchy as a new integrable system.
\end{abstract}

\tableofcontents

\section{Introduction} 
\subsection{Main discussion}
Integrable partial differential equations (PDE) are a special class of equations which, broadly speaking, can be solved explicitly,\footnote{Originally, the typical method employed to solve such systems was by method of ``quadratures,'' or, in other words, integration.} for instance by the inverse scattering transform (IST) discovered by Gardner, Greene, Kruskal and Miura \cite{GGKM} and its subsequent reformulation by Lax \cite{Lax68}. In the years since these (and many other) landmark works, there has been much activity on determining which equations, and more generally, systems, are or should be integrable and the mathematical consequences of being integrable. The reader may acquire a sense for the scope of this activity in the very nice survey \cite{deift_survey} of Deift. Despite the lively, ongoing debate \cite{Zhakharov91} over the defining features of integrability, consensus holds that certain equations, such as the \emph{Korteweg-de Vries (KdV)} or \emph{one-dimensional cubic nonlinear Schr\"odinger equation (NLS)}, should be integrable under any reasonable definition of the term.

Even with the vast research on the implications of an equation's integrability, such as conserved quantities, solitons, or hidden symmetries, it remains unclear \emph{why} equations which are so physically relevant also happen to be integrable. Mathematical insight into this line of inquiry would certainly deepen our understanding of the important models that comprise the extensive catalog of known integrable systems. In an article \cite{Calogero1991} on this very question, Calogero advances his thesis that equations are integrable because they are scaling limits of integrable (or conjecturally integrable) systems, which we refer to as \emph{progenitor models} in this discussion.

Inspired by Calogero's suggestion, the present article considers the familiar one-dimensional cubic NLS
\begin{equation}\label{nls}
i\p_{t}\phi + \Delta \phi = 2\kappa|\phi|^{2}\phi, \qquad \phi : \R\times \mathbb{R}\to \mathbb{C}, \quad \kappa \in \{\pm 1\},
\end{equation}
which was shown by Zakharov and Shabat \cite{ZS72} to be exactly solvable by the IST (see also \cite{AKNS74, ZS79, FT07}). We consider equation \eqref{nls} from the viewpoint that it arises as a \emph{mean field} scaling limit from the progenitor \emph{Lieb-Liniger (LL) model} \cite{LL1963_I}, which is a well-known exactly solvable model describing finitely many bosons with $\delta$-potential interactions. Keeping with Calogero's thesis, we conjecture that integrability of the NLS is a consequence of the exact solvability of the underlying LL model, leading us to the expectation of some manifestation of integrability intrinsically at the level of the system describing infinitely many bosons with $\delta$-potential interactions called the \emph{Gross-Pitaevskii (GP) hierarchy}, for which the NLS is a special case.  Accordingly, this work focuses on providing evidence for the GP hierarchy as a new integrable system.

\medskip

Given the aforementioned debate over the precise definition of an integrable PDE, this work focuses on a particular type of integrability known as \emph{Liouville integrability}. The notion of a Liouville integrable Hamiltonian system was originally introduced in the 19th century and refers to a finite-dimensional Hamiltonian system where there is a maximal (in the sense of degrees of freedom) independent set of Poisson commuting integrals. In the finite-dimensional setting, a Liouville completely integrable system, which satisfies some technical conditions, can be solved by so-called action angle variables, which allow for explicit integration of the system.

The exact solvability of the one-dimensional cubic NLS by the IST was formally shown in the aforementioned work \cite{ZS72} and was mathematically revisited by Beals and Coifman \cite{BC1984, BC1985, BC1987, BC1989}, Terng and Uhlenbeck \cite{Terng1997, TU1998}, Deift and Zhou \cite{Zhou1998, Zhou1989, DZ2003}, among others. Liouville integrability is a particular consequence of this exact solvability, which asserts that the Hamiltonian is one element of a countable sequence of functionals in nontrivial\footnote{By nontrivial, we mean that these functionals are not all Casimirs for the Poisson structure (i.e. they Poisson commute with any functional).} mutual involution. More precisely, one recursively defines (see \cref{ssec:nls_pc_im}) a sequence of operators
\begin{equation}
\label{eq:w_rec}
\begin{split}
w_{n}:\Sc(\R) \rightarrow \Sc(\R), \qquad \begin{cases} \displaystyle w_{1}[\phi] &\coloneqq \phi \\
\displaystyle w_{n+1}[\phi] &\coloneqq -i\p_{x}w_{n}[\phi] + \kappa\bar{\phi}\sum_{k=1}^{n-1}w_{k}[\phi]w_{n-k}[\phi].\end{cases}
\end{split}
\end{equation}
Each $w_n$ generates a functional $I_{n}:\Sc(\R) \rightarrow \C$ by
\begin{equation}
\label{eq:In_def}
I_{n}(\phi) \coloneqq \int_{\R}dx\ol{\phi(x)}w_{n}[\phi](x), \qquad \forall \phi\in\Sc(\R),
\end{equation}
which is, in fact, real-valued (see \cref{lem:I_invol}). Endowing the Schwartz space $\Sc(\R)$ with the standard weak symplectic structure\footnote{See \cref{ssec:pre_WP} for background material on weak symplectic and weak Poisson structures.} given by
\begin{equation}\label{equ:poisson1}
\omega_{L^2}(\phi,\psi) = 2\Im{\int_{\R}dx\ol{\phi(x)}\psi(x)},
\end{equation}
we obtain a canonical weak Poisson structure on $\Sc(\R)$ as follows: consider the real unital\footnote{I.e. the algebra has a multiplicative identity.} algebra with respect to point-wise multiplication
\begin{align}\label{equ:A_sc}
\A_{\Sc} = \{ H \in C^\infty(\Sc(\R);\R) \,:\, \grad_{s}H \in C^{\infty}(\Sc(\R);\Sc(\R)) \}.
\end{align}
Here, $\grad_{s}$ is the symplectic gradient associated to the form $\omega_{L^2}$ (see \cref{def:re_grad} and \cref{schwartz_deriv} for definitions). For $F,G\in \A_{\Sc}$, define their Poisson bracket by
\begin{equation}\label{equ:poisson2}
\pb{F}{G}_{L^{2}}(\phi) \coloneqq \omega_{L^2}(\grad_{s}F(\phi), \grad_{s}G(\phi)), \qquad \forall \phi\in\Sc(\R).
\end{equation}
Then one can verify (see \cref{ssec:nls_pc}) that 
\begin{equation}
\label{eq:I_pb_com}
\pb{I_{n}}{I_{m}}_{L^2}(\phi)= 0, \qquad \forall \phi\in\Sc(\R), \enspace\forall n,m\in\N.
\end{equation}

Furthermore, the solution to the NLS \eqref{nls} is the integral curve to the Hamiltonian equation of motion associated to the third functional $I_3$. That is,
\begin{equation}
\paren*{\frac{d}{dt}\phi}(t) = \grad_{s}I_{3}(\phi(t)).
\end{equation}
In particular, if $\phi\in C^{\infty}([t_0,t_1];\Sc(\R))$ is a classical solution to \eqref{nls}, then $I_{n}(\phi)$ is conserved on the lifespan $[t_0,t_1]$ of $\phi$ for every $n \in \N$. Furthermore, each of the functionals $I_{n}$ has an associated equation of motion
\begin{equation}\label{equ:nnls}
\paren*{\frac{d}{dt}\phi}(t) = \grad_{s}I_{n}(\phi(t)).
\end{equation}
Following the terminology of Faddeev and Takhtajan \cite{FT07}, we call \eqref{equ:nnls} the \emph{$n$-th nonlinear Schr\"{o}dinger equation (nNLS)}. The $n=1,2$ equations are trivial, the $n=3$ equation is the NLS \eqref{nls}, and the $n=4$ equation is the complex mKdV equation
\begin{equation}
\label{eq:mkdv}
\p_{t}\phi = \p_{x}^{3}\phi -6\kappa|\phi|^{2}\p_{x}\phi, \qquad \kappa \in\{\pm 1\}.
\end{equation}
To our knowledge, the $n$-th nonlinear Schr\"odinger equations do not have specific names for $n\geq 5$. Together, the family of $n$-th nonlinear Schr\"odinger equations constitutes the \emph{nonlinear Schr\"odinger hierarchy}, as termed by Palais \cite{Palais1997}.

To set the stage for our work in this paper, we begin by briefly discussing the progenitor LL model and its relation to the NLS. The LL model is the many-body problem
\begin{equation}
\label{eq:LL_model}
i\p_t\Phi_N = H_N\Phi_N, \qquad H_N = \sum_{j=1}^N -\Delta_{x_j} + \frac{2\kappa}{(N-1)} \sum_{1 \leq j < k \leq N} \delta(X_j - X_k),
\end{equation}
where $\Phi_N\in L_{sym}^2(\R^N)$, the coupling constant has been taken to be proportional to $1/N$ so that we are in the mean field scaling regime. The value of $\kappa\in\{\pm 1\}$ determines whether the system is repulsive ($\kappa=1$) or attractive ($\kappa=-1$). Here, $H_N$ may be realized as a self-adjoint operator on $L_{sym}^2(\R^N)$, the space of bosonic wave functions, by means of the KLMN theorem (see Theorem X.17 of \cite{RSII}). The many-body problem \eqref{eq:LL_model} is a toy model for an interacting Bose gas in one dimension, and both mathematical and physical interest in \eqref{eq:LL_model} stems from its remarkable property of being \emph{exactly solvable}. More precisely, Lieb and Liniger used the Bethe ansatz\footnote{Bethe ansatz refers to a technique in the study of exactly solvable models introduced by Hans Bethe to find exact eigenvalues and eigenvectors of the antiferromagnetic Heisenberg spin chain \cite{Bethe1931}. For more on this technique, we refer the reader to the monographs 
\cite{Gaudin2014} and \cite{KBI1993}.} in their seminal paper \cite{LL1963_I} to obtain explicit formulae for the eigenfunctions and spectrum of the Hamiltonian $H_N$. Analogous to the free Schr\"{o}dinger equation, one has an explicit distorted Fourier transform associated to $H_N$, which by solving an ordinary differential equation in the distorted Fourier domain yields a formula for the solution to \eqref{eq:LL_model}.

The connection between the LL model and the NLS is via an infinite particle limit. To understand the dynamics of \eqref{eq:LL_model} in the limit as the particle number $N \to \infty$, one needs to address the fact that the wave functions $\{\Phi_N\}_{N \in \N}$ do not live in a common topological space. One way to do so is to consider sequences of hierarchies of reduced density matrices
\begin{equation}
\gamma_{N}^{(k)} \coloneqq \Tr_{k+1, \ldots, N}\paren*{\ket*{\Phi_N}\bra*{\Phi_N}}, \qquad k\in\N
\end{equation}
where $\Tr_{k+1,\ldots,N}$ denotes the partial trace over the $k+1,\ldots,N$ coordinates. The $\{\gamma_N^{(k)}\}_{k=1}^N$ then solve the \emph{BBGKY hierarchy},\footnote{Bogoliubov-Born-Green-Kirkwood-Yvon hierarchy} which is a coupled system of linear equations describing the evolution of finitely many interacting bosons. In the limit as $N \to\infty$, the sequence $\{\gamma_N^{(k)}\}_{k\in\N}$ (where by convention, $\gamma_N^{(k)} = 0$ for $k > N$) formally converges to a solution $\{\gamma^{(k)}\}_{k\in\N}$ of the cubic Gross-Pitaevskii (GP) hierarchy
\begin{equation}
\label{eq:GP}
i\p_t\gamma^{(k)} = \comm{-\Delta_{\ux_k}}{\gamma^{(k)}} + 2\kappa\sum_{j=1}^k \Tr_{k+1}\paren*{\comm{\delta(X_j-X_{k+1})}{\gamma^{(k+1)}}}, \qquad k\in\N,
\end{equation}
where we have introduced the notation $\Delta_{\ux_k}\coloneqq \sum_{j=1}^k \Delta_{x_j}$. While \eqref{eq:GP} is a linear system, it is \emph{coupled}, rendering its mathematical study nontrivial. The connection with the NLS \eqref{nls} is then as follows:
\begin{equation}\label{equ:fac}
(\gamma^{(k)})_{k\in\N}, \enspace \gamma^{(k)} \coloneqq \ket*{\phi^{\otimes k}}\bra*{\phi^{\otimes k}} \text{ solves the GP \eqref{eq:GP}} \Longleftrightarrow \phi:[0,T]\times\R\rightarrow\C \text{ sovles the NLS \eqref{nls}}.
\end{equation}
The above formal discussion has been made rigorous in works by Adami, Bardos, Golse, and Teta \cite{ABGT2004} and Ammari and Breteaux \cite{AmBre2012}.

\medskip
In light of our previous discussion on Liouville integrability of the NLS, we turn to our search for evidence of integrability at the infinite-particle level. We note that this search necessitates a Hamiltonian formulation of the GP hierarchy, for which we rely on the recent work of the authors \cite[Theorem 2.10]{MNPRS1_2019} that shows that the GP hierarchy is the equation of motion on a weak Poisson manifold for a Hamiltonian $\H_{GP}$. We formulate the following question:

\begin{ques}\label{q:GP_int}
Does the one-dimensional cubic GP hierarchy possess an infinite sequence of functionals $\{\H_{n}\}_{n \in \mathbb{N}}$ containing the Hamiltonian $\H_{GP}$ for the GP hierarchy, which are in nontrivial involution?
\end{ques}

We provide an affirmative answer to \cref{q:GP_int} with our \cref{thm:GP_invol}, evidencing Liouville integrability of the GP hierarchy. Note that an immediate consequence of the affirmative answer to \cref{q:GP_int} is that the functionals $\H_n$ are conserved along the flow of the GP hierarchy. 

\medskip
The functionals $\H_n$ which we construct are trace functionals associated to the family of observable $\infty$-hierarchies $\{-i\W_n\}_{n \in \N}$ which belong to the Lie algebra $\G_{\infty}$ defined in \cite{MNPRS1_2019}, the definition of which we review in \cref{prop:G_inf_br} below. Heuristically speaking, our definition of these observable hierarchies proceeds by a quantization of the recursive formula \eqref{eq:w_rec} for the one-particle nonlinear operators $\{w_n\}_{n \in \N}$. More precisely, we observe that the functionals $I_n$ defined in \eqref{eq:In_def} are finite sums of multilinear forms whose arguments are restricted to a single function $\phi\in\Sc(\R)$ and its complex conjugate $\ol{\phi}\in\Sc(\R)$: 
\begin{equation}
I_n(\phi) = \sum_{k=1}^{N(n)} I_n^{(k)}[\underbrace{\phi,\ldots,\phi}_{k};\underbrace{\ol{\phi},\ldots,\ol{\phi}}_{k}], \qquad N(n)\in\N.
\end{equation}
A posteriori of our construction, we show that the $k$-particle component $\W_n^{(k)}$ of $\W_n=~(\W_n^{(j)})_{j\in\N}$ is the Schwartz kernel of each $I_n^{(k)}$. 

To prove the Poisson commutativity of the functionals $\H_n$ with respect to the Poisson structure underlying the GP hierarchy from \cite{MNPRS1_2019}, we simultaneously proceed at the level of the GP hierarchy and at the level of the NLS equation. We combine a good understanding of the multilinear structure of the $I_n$ with a knowledge of the structure of bosonic density matrices to show that Poisson commutativity of the $\H_n$ is equivalent to that of certain functionals $I_{b,n}$ defined in \eqref{eq:Ibn_intro_def}, which are associated to an integrable system generalizing the NLS.\footnote{The inspiration for considering this system comes from a remark of Faddeev and Takhtajan \cite[Remark 13, pg. 181]{FT07}.}  We rewrite the NLS \eqref{nls} as the system
\begin{equation}
\label{nls_sys}
\begin{cases}
i\p_t\phi = -\Delta\phi + 2\kappa\phi^2\ol{\phi} \\
i\p_t\ol{\phi} = \Delta\ol{\phi} -2\kappa\ol{\phi}^2\phi
\end{cases},
\end{equation}
and relax the requirement that $\ol{\phi}$ denotes the complex conjugate of $\phi$ (i.e. $\phi$ and $\ol{\phi}$ are independent coordinates on $\Sc(\R)$). We then show that the family $\{I_{b,n}\}_{n\in\N}$ is mutually involutive (see \cref{prop:Ib_invol}). By also showing that there is a Poisson morphism from the phase space of \eqref{nls_sys}\footnote{Strictly speaking, the domain of the morphism is a quotient space of the phase space of \eqref{nls_sys} with the property that the elements are ``self-adjoint''.}  to the phase space of the GP hierarchy, we obtain the desired conclusion. This equivalence we prove, recorded in \cref{ibn_inv} below, is quite interesting in its own right and was not expected by the authors at the onset of this project.

\begin{remark}
In \cite{MNPS2016}, four of the co-authors of the present article identified an infinite sequence of conserved quantities for the GP hierarchy, which agreed with the $I_n$ defined in \eqref{eq:In_def} when evaluated on factorized states. At the time of \cite{MNPS2016}, a Hamiltonian structure for the GP hierarchy had not been identified, so it was premature to ask if the conservation of these quantities was a consequence of their Poisson commuting with the GP Hamiltonian, let alone their being in mutual involution, as is the case with the functionals $I_n$. The current work also provides a substantial generalization of the previous work \cite{MNPS2016}, in that the definition of the functionals $\H_n$ in \cite{MNPS2016} used the quantum de Finetti theorems \cite{HM1976, Stormer1969, LNR2015}. Indeed, these functionals are initially defined on factorized states of the form in \eqref{equ:fac}, and then their domain of definition is extended to statistical averages of such factorized states by means of quantum de Finetti. In contrast, we now establish that these functionals are defined on the entire GP phase space. In particular, we construct $\H_n$ without any considerations of admissibility\footnote{An infinite sequence of trace-class density matrices $\{\gamma^{(k)}\}_{k\in\N}$ is said to be \emph{admissible} if $\gamma^{(k)} = \Tr_{k+1}(\gamma^{(k+1)})$.} and without any recourse to representation theorems, such as the quantum de Finetti theorems. In fact, admissibility plays no role in this paper.
\end{remark}

\medskip
Following our affirmative answer to \cref{q:GP_int}, one may wonder from a more dynamical perspective, if there is a natural connection between the flows generated by the Poisson commuting functionals $\H_n$ and other well-known one-particle equations. We are thus motivated to address the following question:

\begin{ques}\label{q:nGP}
Does each of the functionals $\H_{n}$ generate a Hamiltonian equation of motion related to the $n$-th nonlinear Schr\"{o}dinger equation \eqref{equ:nnls} via factorized solutions in the spirit of \eqref{equ:fac}? 
\end{ques}

Our \cref{thm:GP_fac} below provides an affirmative answer to \cref{q:nGP}, proving that factorized solutions of the equation of motion with Hamiltonian $\H_n$ are of the form \eqref{equ:fac}, where now each factor solves the $n$-th NLS equation. In this sense, we establish that the family comprised of the $n$-th GP hierarchies is the appropriate infinite-particle generalization of the nonlinear Schr\"{o}dinger hierarchy. As with the proof of our involution result, our proof of this factorization connection relies on a good understanding of the multilinear structure underlying the $I_n$. We then use this understanding to find a formula for the symplectic gradients $\grad_s I_n$, which together with a general formula for Hamiltonian vector fields on the GP phase space allows us to arrive at the desired conclusion. The result \cref{thm:GP_fac} can be viewed as a one-dimensional extension of Theorem 2.10 from our companion work \cite{MNPRS1_2019}, which proves a Hamiltonian formulation for the GP hierarchy in all dimensions. We also include an explicit computation of the fourth GP hierarchy in \cref{ssec:gp_flows_ex}, which corresponds to the complex mKdV equation \eqref{eq:mkdv}.

\medskip
We close the introduction by returning to the aforementioned thesis of Calogero with an eye towards future work. As we previously commented, if Calogero's thesis is correct for the NLS, as we believe it is, then there should be some evidence of integrability at the level of the GP hierarchy. Our work provides such evidence by showing that there is a family of Poisson commuting functionals which encode the nonlinear Schr\"{o}dinger hierarchy. Given that the work of Adami et al. \cite{ABGT2004} and Ammari et al. \cite{AmBre2012} mathematically demonstrates that the NLS \eqref{nls_sys} is the mean field limit of the LL model \cref{eq:LL_model}, it is natural to ask if there exists a connection between our functionals $\H_n$ together with the family of $n$-th GP hierarchies--and by implication the functionals $I_n$ together with the nonlinear Schr\"odinger hierarchy--and the LL model. Establishing this connection in rigorous mathematical terms seems a difficult but worthwhile task. We believe that the core difficulty lies in understanding the connection between classical and quantum field theories via the processes of quantization and mean field limit. This connection figures prominently in the work of Fr\"ohlich, Tsai, and Yau \cite{FTY2000} and Fr\"olich, Knowles, and Pizzo \cite{FKP} and references therein. We also mention the work \cite{Thacker1978_pcl}, in which Thacker posits a conjecture related to this line of inquiry, and the work \cite{Davies1990}, in which Davies discusses the issues with naive quantization of classical approaches to integrability. We hope that the work of our paper together with the derivation of the Hamiltonian formulation of the NLS in our companion paper \cite{MNPRS1_2019} will inspire others to join us in elucidating these fascinating connections.

\subsection{Acknowledgments}
The authors thank Karen Uhlenbeck for helpful discussion on the role of geometry in the study of integrable systems during the course of this project. The authors also thank J\"urg Fr\"ohlich for helpful comments regarding references which have enhanced the presentation of the article.

\section{Statement of main results and blueprint of proofs}
\label{sec:mr_bp}
In this section, we provide an outline and discussion of the main results of this article and their proofs. We begin by recalling in \cref{ssec:mr_geo_rev} several of the main geometric results from \cite{MNPRS1_2019} which are needed in the current work.
  
\subsection{Review of \cite{MNPRS1_2019}}
\label{ssec:mr_geo_rev} 
A major soure of difficulty in \cite{MNPRS1_2019} is the construction of an infinite-dimensional Lie algebra of observable $\infty$-hierarchies and its dual weak Lie-Poisson manifold of density matrix $\infty$-hierarchies, which together form the geometric foundation of the Hamiltonian formulation of the GP hierarchy. The analytic difficulties in this definition stem primarily from the fact that the GP Hamiltonian $\H_{GP}=\H_3$ is the trace functional associated to a \emph{distribution-valued operator (DVO)}.\footnote{Not to be confused with operator-valued distributions in quantum field theory.} The natural Lie bracket for such operators requires composition of two operators in a given particle coordinate. Such a definition is not possible in general since the composition of two DVOs may be ill-defined. Overcoming these difficulties necessitated the identification of a property for DVOs which we termed the \emph{good mapping property}, whose definition we recall here.

\begin{restatable}[Good mapping property]{mydef}{gmp}
\label{def:gmp}
Let $\ell\in\N$. We say that an operator $A^{(\ell)}\in \L(\Sc(\R^{\ell}),\Sc'(\R^{\ell}))$ has the \emph{good mapping property} if for any $\alpha\in\N_{\leq \ell}$, the continuous bilinear map
\begin{equation*}
\begin{split}
&\Sc(\R^{\ell}) \times\Sc(\R^{\ell}) \rightarrow \Sc_{x_\alpha'}(\R;\Sc_{x_\alpha}'(\R)) \\
&(f^{(\ell)},g^{(\ell)}) \mapsto \int_{\R^{\ell-1}} dx_{1} \ldots dx_{\alpha-1}dx_{\alpha+1} \ldots dx_{\ell} A^{(\ell)}(f^{(\ell)})(x_1, \ldots, x_\ell) g^{(\ell)}(x_1, \ldots, x_{\alpha-1},x_{\alpha}',x_{\alpha+1}, \ldots,x_{\ell}),
\end{split}
\end{equation*}
may be identified with a continuous bilinear map $\Sc(\R^{\ell})\times \Sc(\R^{\ell}) \rightarrow \Sc(\R^{2})$.\footnote{Here and throughout this paper, an integral involving a distribution should be understood as a distributional pairing unless specified otherwise.} 
\end{restatable}

The good mapping property has the following important consequence: let $(\alpha,\beta)\in\N_{\leq\ell}\times \N_{\leq j}$, and let $A^{(\ell)}\in \L(\Sc(\R^\ell),\Sc'(\R^\ell))$ and $B^{(j)}\in \L(\Sc(\R^j),\Sc'(\R^j))$ have the good mapping property. If $k\coloneqq \ell+j-1$, then the bilinear map
\begin{equation}
\begin{split}
&\Sc(\R^k)^2 \rightarrow \Sc_{(\ux_{\alpha-1},\ux_{\alpha+1;\ell},\ux_\ell')}(\R^{\alpha-1}\times\R^{\ell-\alpha}\times\R^\ell;\Sc_{x_\alpha}'(\R))\\
&(f^{(k)},g^{(k)}) \mapsto \begin{cases} \ipp*{B_{(1,\ldots,j)}^{(j)}(f^{(k)}(\ux_{\alpha-1},\cdot,\ux_{\alpha+1;\ell},\cdot)),(\cdot)\otimes g^{(k)}(\ux_\ell',\cdot)}_{\Sc'(\R^j)-\Sc(\R^j)}, &{\beta=1} \\ \ipp*{B_{(2,\ldots,\beta,1,\beta+1,\ldots,j)}^{(j)}(f^{(k)}(\ux_{\alpha-1},\cdot,\ux_{\alpha+1;\ell},\cdot)),(\cdot)\otimes g^{(k)}(\ux_\ell',\cdot)}_{\Sc'(\R^j)-\Sc(\R^j)}, & {\beta \neq 1} \end{cases}
\end{split}
\end{equation}
may be identified with a unique smooth bilinear map
\begin{equation}
\Phi_{B^{(j)},\alpha,\beta}:\Sc(\R^k)\times\Sc(\R^k)\rightarrow \Sc_{(\ux_\ell,\ux_\ell')}(\R^{2\ell})
\end{equation}
via
\begin{equation}
\begin{split}
&\int_{\R}dx_\alpha \Phi_{B^{(j)},\alpha,\beta}(f^{(k)},g^{(k)})(\ux_\ell;\ux_\ell')\phi(x_\alpha) \\
&=\begin{cases}\displaystyle \ipp*{B_{(1,\ldots,j)}^{(j)}(f^{(k)}(\ux_{\alpha-1},\cdot,\ux_{\alpha+1;\ell},\cdot)),\phi\otimes g^{(k)}(\ux_\ell',\cdot)}_{\Sc'(\R^j)-\Sc(\R^j)}, & {\beta=1} \\
\displaystyle \ipp*{B_{(2,\ldots,\beta,1,\beta+1,\ldots,j)}^{(j)}(f^{(k)}(\ux_{\alpha-1},\cdot,\ux_{\alpha+1;\ell},\cdot)),\phi\otimes g^{(k)}(\ux_\ell',\cdot)}_{\Sc'(\R^j)-\Sc(\R^j)}, &{\beta\neq 1}, \end{cases}
\end{split}
\end{equation}
for any $\phi\in\Sc(\R)$ and $(\ux_{1;\alpha-1},\ux_{\alpha+1;\ell},\ux_\ell')\in \R^{2\ell-1}$. Here, the subscript $(2,\ldots,\beta,1,\beta+1,\ldots,j)$ is to be interpreted in the sense of the subscript notation in \eqref{eq:sym_repeat} (see also \cref{prop:ext_k}).\footnote{So as to avoid a cumbersome consideration of cases in the sequel, we will not distinguish between the $\beta=1$ and $\beta\neq 1$ cases going forward.} Hence, by the Schwartz kernel theorem isomorphism
\begin{equation}
\L(\Sc(\R^k),\Sc'(\R^k)) \cong \Sc(\R^{2k}),
\end{equation}
we can define the following composition as an element
\begin{equation}
(A^{(\ell)}\circ_\alpha^\beta B^{(j)})\in \L(\Sc(\R^k),\Sc'(\R^k))
\end{equation}
by
\begin{equation}
\ipp*{(A^{(\ell)}\circ_\alpha^\beta B^{(j)})f^{(k)},g^{(k)}}_{\Sc'(\R^k)-\Sc(\R^k)} \coloneqq \ipp*{K_{A^{(\ell)}},\Phi_{B^{(j)},\alpha,\beta}^t(f^{(k)},g^{(k)})}_{\Sc'(\R^{2k})-\Sc(\R^{2k})},
\end{equation}
where $K_{A^{(\ell)}}$ denotes the Schwartz kernel of $A^{(\ell)}$ and $\Phi_{B^{(j)},\alpha,\beta}^t(f^{(k)},g^{(k)})$ denotes the transpose of $\Phi_{B^{(j)},\alpha,\beta}(f^{(k)},g^{(k)})$ defined by
\begin{equation}
\Phi_{B^{(j)},\alpha,\beta}^t(f^{(k)},g^{(k)})(\ux_j;\ux_j') \coloneqq \Phi_{B^{(j)},\alpha,\beta}(f^{(k)},g^{(k)})(\ux_j';\ux_j), \qquad \forall (\ux_j,\ux_j')\in\R^{2j}.
\end{equation}
Note that $A^{(\ell)}\circ_\alpha^\beta B^{(j)}$ coincides with the composition
\begin{equation}
A_{(1,\ldots,\ell)}^{(\ell)}B_{(\ell+1,\ldots,\ell+\beta-1,\alpha,\ell+\beta,\ldots,k)}^{(j)}
\end{equation}
when the latter is defined. We let $\L_{gmp}(\Sc(\R^{\ell}),\Sc'(\R^{\ell}))$ denote the subset of $\L(\Sc(\R^{\ell}),\Sc'(\R^{\ell}))$ of elements with the good mapping property, and $\L_{gmp,*}(\Sc(\R^\ell),\Sc'(\R^\ell))$ denote the further subset of elements which are skew-adjoint (see \cref{lem:dvo_adj} and \cref{def:dvo_sa} for the definitions of adjoint and skew-adjoint for a DVO). We established in \cite[Lemma 6.1, Remark 6.3]{MNPRS1_2019} that the composition
\begin{equation}
(\cdot)\circ_\alpha^\beta(\cdot): \L_{gmp,*}(\Sc(\R^\ell),\Sc'(\R^\ell)) \times \L_{gmp,*}(\Sc(\R^j),\Sc'(\R^j)) \rightarrow \L_{gmp,*}(\Sc(\R^k),\Sc'(\R^k))
\end{equation}
is a separately continuous, bilinear map.

With the composition map $(\cdot)\circ_\alpha^\beta (\cdot)$ in hand, we proceed to reviewing the main geometric actors from \cite{MNPRS1_2019}. We recall that
\begin{align}\label{gk_def}
\g_{k,gmp} \coloneqq \{A^{(k)} \in  \L_{gmp}(\Sc_{s}(\R^{k}),\Sc_{s}'(\R^{k})) : (A^{(k)})^{*} = -A^{(k)} \},
\end{align}
where $\Sc_s(\R^k)$ is the subspace of $\Sc(\R^k)$ consisting of functions invariant under permutation of coordinates (see \cref{sym_schwartz}), and 
\begin{equation}
\G_{\infty} \coloneqq \bigoplus_{k=1}^{\infty} \g_{k,gmp}
\end{equation}
endowed with the locally convex topology. We equip $\G_{\infty}$ with a Lie bracket given by
\begin{equation}
\label{eq:LB_def}
\begin{split}
&\comm{A}{B}_{\G_{\infty}} = C=(C^{(k)})_{k\in\N}\\
&C^{(k)} \coloneqq \Sym_k\paren*{\sum_{\ell,j\geq 1; \ell+j-1=k} \sum_{\alpha=1}^{\ell}\sum_{\beta=1}^{j} \paren*{(A^{(\ell)} \circ_\alpha^\beta B^{(j)}) - (B^{(j)}\circ_\beta^\alpha A^{(\ell)})}},\footnotemark
\end{split}
\end{equation}
\footnotetext{Strictly speaking, a priori it is not the operators $A^{(\ell)}$ and $B^{(j)}$ that appear in the right-hand side, but instead extensions $\tl{A}^{(\ell)}\in \L_{gmp}(\Sc(\R^\ell),\Sc'(\R^\ell))$ and $\tl{B}^{(j)}\in\L_{gmp}(\Sc(\R^j),\Sc'(\R^j))$. The right-hand side is independent of the choice of extension, as shown in \cite[Remark 6.5]{MNPRS1_2019}, and therefore we will not comment on this technical point in the sequel.}
where $\Sym_k$ denotes the bosonic symmetrization operator given by
\begin{equation}\label{eq:sym_repeat}
\Sym_k(A^{(k)}) \coloneqq \frac{1}{k!} \sum_{\pi\in\Ss_{k}} A^{(k)}_{(\pi(1), \ldots, \pi(k))}, \quad A^{(k)}_{(\pi(1), \ldots, \pi(k))} = \pi\circ A^{(k)} \circ \pi^{-1}.
\end{equation}

\begin{prop}[\protect{\cite[Proposition 2.7]{MNPRS1_2019}}]
\label{prop:G_inf_br}
$(\G_{\infty},\comm{\cdot}{\cdot}_{\G_{\infty}})$ is a Lie algebra.
\end{prop}

Next, we recall the definition of the weak Lie-Poisson manifold $(\G_{\infty}^{*},\A_{\infty},\pb{\cdot}{\cdot}_{\G_{\infty}^{*}})$, which is the phase space underlying the GP hierarchy. We  define the real topological vector space
\begin{equation}
\label{eq:gk*_def}
\g_k^* \coloneqq \left\{\gamma^{(k)} \in \L(\Sc_{s}'(\R^{k}),\Sc_{s}(\R^{k})) : \gamma^{(k)} = (\gamma^{(k)})^{*}\right\}
\end{equation}
and define the topological direct product
\begin{equation}
\label{eq:Ginf*_def}
\G_{\infty}^{*} \coloneqq \prod_{k=1}^{\infty} \g_k^*.
\end{equation}
Attached to $\G_\infty^*$ is the admissible algebra of functionals $\A_{\infty}$ defined to be the real algebra with respect to point-wise product generated by functionals in the set
\begin{equation}
\label{eq:Ainf_gen}
\begin{split}
&\{F\in C^{\infty}(\G_{\infty}^{*};\R) : F(\cdot) = i\Tr(\W\cdot),\,\, \W\in\G_{\infty}\}  \cup \{F\in C^{\infty}(\G_{\infty}^{*};\R) : F(\cdot) \equiv C\in\R\}.
\end{split}
\end{equation}
Most importantly, our choice of $\A_{\infty}$ contains the trace functionals associated to the observable $\infty$-hierarchies $\{-i\W_{n}\}_{n=1}^{\infty}$. We can then define the Poisson bracket of functionals $F,G\in\A_\infty$ by
\begin{equation}\label{equ:poisson_def}
\pb{F}{G}_{\G_{\infty}^{*}}(\Gamma) = i\Tr\paren*{\comm{dF[\Gamma]}{dG[\Gamma]}_{\G_{\infty}}\cdot\Gamma}, \qquad \forall\Gamma\in\G_\infty^*.
\end{equation}
In the right-hand side of \eqref{equ:poisson_def}, we identify the G\^ateaux derivatives $dF[\Gamma]$ and $dG[\Gamma]$, which are a priori continuous linear functionals, as elements of $\G_{\infty}$. This identification is possible thanks to the definition of $\A_\infty$ and the next lemma, which characterizes the dual of $\G_\infty^*$. 

\begin{lemma}[{\cite[Lemma 6.8]{MNPRS1_2019}}]\label{dual_dual}
The topological dual of $\G_{\infty}^{*}$, denoted by $(\G_{\infty}^{*})^{*}$ and endowed with the strong dual topology, is isomorphic to
\begin{equation}
\widetilde{\G}_{\infty}\coloneqq \{A\in \bigoplus_{k=1}^{\infty}\L(\Sc_{s}(\R^{k}), \Sc_{s}'(\R^{k})) : (A^{(k)})^{*} = -A^{(k)}\},
\end{equation}
equipped with the subspace topology induced by $\bigoplus_{k=1}^{\infty}\L(\Sc_{s}(\R^{k}),\Sc_{s}'(\R^{k}))$, via the canonical bilinear form
\begin{equation}
i\Tr(A\cdot \Gamma) = i\sum_{k=1}^{\infty}\Tr_{1,\ldots,k}(A^{(k)}\gamma^{(k)}), \qquad \forall \Gamma=(\gamma^{(k)})_{k\in\N}\in \G_{\infty}^{*}, \enspace A=(A^{(k)})_{k\in\N}\in\wt{\G}_\infty.
\end{equation}
\end{lemma}

In \cite{MNPRS1_2019}, classical results on the existence of a Lie-Poisson manifold associated to a Lie algebra were unavailable to us due to functional analytic difficulties, such as the fact that $\G_{\infty}\subsetneq \wt{\G}_\infty$. Nevertheless, we verified directly that our choices for $\G_\infty^*$, $\A_{\infty}$, and $\pb{\cdot}{\cdot}_{\G_{\infty}^{*}}$ satisfy the weak Poisson axioms of \cref{def:WP}, thereby establishing the following result. 

\begin{prop}[\protect{\cite[Proposition 2.8, Lemma 6.15]{MNPRS1_2019}}]
\label{prop:LP}
$(\G_{\infty}^{*},\A_{\infty},\pb{\cdot}{\cdot}_{\G_{\infty}^{*}})$ is a weak Poisson manifold. Furthermore, for any $F\in \A_\infty$, the Hamiltonian vector field $X_F$ is given by the formula
\begin{equation}
X_F(\Gamma)^{(\ell)} = \sum_{j=1}^\infty j\Tr_{\ell+1,\ldots,\ell+j-1}\paren*{\comm{\sum_{\alpha=1}^\ell dH[\Gamma]_{(\alpha,\ell+1,\ldots,\ell+j-1)}^{(j)}}{\gamma^{(\ell+j-1)}}}, \qquad \ell\in\N, \enspace \Gamma\in\G_\infty^*,
\end{equation}
where the extension $dH[\Gamma]_{(\alpha,\ell+1,\ldots,\ell+j-1)}^{(j)}$ is defined via \cref{prop:ext_k}.
\end{prop}

\subsection{Statement of main results}
Having reviewed the results from \cite{MNPRS1_2019} presently germane, we are now prepared to state the main results of the current work.  We previously introduced the GP hierarchy in \eqref{eq:GP}, which we recall now.  We say that a sequence of time-dependent kernels $(\gamma^{(k)})_{k\in\N}$ of $k$-particle density matrices is a solution to the GP hierarchy if
\begin{equation}\label{eq:GP'}
i\p_{t}\gamma^{(k)} = -\comm{\Delta_{\ux_{k}}}{\gamma^{(k)}}+ 2\kappa B_{k+1}(\gamma^{(k+1)}), \qquad k\in\N,
\end{equation}
with $\kappa\in\{\pm 1\}$, and
\begin{equation}
\label{eq:B_con}
B_{k+1}(\gamma^{(k+1)}) = \sum_{j=1}^{k} \paren*{B_{j;k+1}^{+}-B_{j;k+1}^{-}}(\gamma^{(k+1)}),
\end{equation}
where for every $(\ux_k,\ux_k')\in\R^{2k}$,
\begin{equation}
\label{eq:B_pm}
\begin{split}
B_{j;k+1}^{+}(\gamma^{(k+1)})(t,\ux_{k};\ux_{k}') &\coloneqq \gamma^{(k+1)}(t,\ux_{k},x_j;\ux_{k}',x_j),\\
B_{j;k+1}^{-}(\gamma^{(k+1)})(t,\ux_k;\ux_k') &\coloneqq \gamma^{(k+1)}(t,\ux_k, x_j';\ux_k',x_j').
\end{split}
\end{equation}
When $\kappa=1$, we say that the hierarchy is \emph{defocusing} and for $\kappa=-1$, we say that the hierarchy is \emph{focusing} (in analogy with the defocusing and focusing NLS, respectively). 

\medskip
To address \cref{q:GP_int}, we must first establish the existence of an infinite sequence of observable $\infty$-hierarchies $\{-i\W_{n}\}_{n\in\N} \in \G_{\infty}$ by a recursion argument inspired by that for the operators $w_{n}$ in \eqref{eq:w_rec}. Due to analytic difficulties, once again stemming primarily from the need to consider the composition of DVOs, we proceed in three steps.

The first step consists of constructing an element
\[
\wt{\W}_{n}\in \bigoplus_{k=1}^{\infty} \L(\Sc(\R^{k}),\Sc'(\R^{k}))
\]
by the recursive formula
\begin{equation}\label{eq:Wn_rec}
\begin{split}
\wt{\W}_{1} &\coloneqq \E_{1} = (Id_1, 0, \ldots) \\
\wt{\W}_{n+1}^{(k)} &\coloneqq (-i\p_{x_{1}})\wt{\W}_{n}^{(k)}  +\kappa\sum_{m=1}^{n-1}\sum_{\ell,j\geq 1;\ell+j=k}\delta(X_{1}-X_{\ell+1}) \paren*{\wt{\W}_{m}^{(\ell)} \otimes \wt{\W}_{n-m}^{(j)}}, \qquad \forall k\in \N,
\end{split}
\end{equation}
Note the structural similarity between this recursion and the one for the operators $w_{n}$ stated in \cref{eq:w_rec}. While the DVO $\wt{\W}_{m}^{(\ell)} \otimes \wt{\W}_{n-m}^{(j)}$ is well-defined by the universal property of the tensor product, the composition
\begin{equation}\label{eq:intro_comp_ex}
\delta(X_{1}-X_{\ell+1}) \paren*{\wt{\W}_{m}^{(\ell)} \otimes \wt{\W}_{n-m}^{(j)}}
\end{equation}
is a priori purely formal, since evaluation on a Schwartz function leads to products of distributions, in particular products of $\delta$ functions and their higher-order derivatives. Thus, the challenge is to give meaning to this composition. The key property which allow us to make sense of the composition is that if we formally expand the recursion, we will only find products such as $\delta(x_{1}-x_{2})\delta(x_{2}-x_{3})$, which is well-defined as the Lebesgue measure on the hyperplane $\{\ux_{k}\in\R^{k}: x_{1}=x_{2}=x_{3}\}$. To systematically handle the products of distributions, we use the wave front set and a useful criterion of H\"{o}rmander for the multiplication of distributions (see \cref{prop:H_crit} and more generally, \cref{app:WF}).

A priori, H\"{o}rmander's criterion only yields that the product of two tempered distributions is a distribution, not necessarily tempered, which is problematic since we work exclusively with tempered distributions. Moreover, we wish any definition of the composition \eqref{eq:intro_comp_ex} to satisfy the property
\begin{equation}
\label{eq:des_def}
\begin{split}
&\ipp*{\delta(X_{1}-X_{\ell+1})\paren*{\wt{\W}_{m}^{(\ell)}\otimes\wt{\W}_{n-m}^{(j)}}(f^{(\ell)}\otimes f^{(j)}), g^{(\ell)}\otimes g^{(j)}}_{\Sc'(\R^{k})-\Sc(\R^{k})} \\
&\phantom{=}=\int_{\R}dx \: \Phi_{\wt{\W}_m^{(\ell)}}(f^{(\ell)},g^{(\ell)})(x,x) \Phi_{\wt{\W}_{n-m}^{(j)}}(f^{(j)},g^{(j)})(x,x),
\end{split}
\end{equation}
where
\begin{equation}
\Phi_{\wt{\W}_m^{(\ell)}}: \Sc(\R^\ell)^2\rightarrow \Sc(\R^2),\quad  \Phi_{\wt{\W}_{n-m}^{(j)}}:\Sc(\R^j)^2\rightarrow\Sc(\R^2)
\end{equation}
are the necessarily unique maps identifiable with
\begin{equation}
\begin{split}
&\Sc(\R^\ell)^2 \rightarrow \Sc_{x'}(\R;\Sc_x'(\R)) \quad (f^{(\ell)},g^{(\ell)}) \mapsto \ipp*{\wt{\W}_m^{(\ell)}f^{(\ell)}, (\cdot)\otimes g^{(\ell)}(x',\cdot)}_{\Sc'(\R^\ell)-\Sc(\R^\ell)},\\
&\Sc(\R^j)^2\rightarrow \Sc_{x'}(\R;\Sc_x'(\R)) \quad (f^{(j)},g^{(j)}) \mapsto \ipp*{\wt{\W}_{n-m}^{(j)}f^{(j)}, (\cdot)\otimes g^{(j)}(x',\cdot)}_{\Sc'(\R^j)-\Sc(\R^j)}
\end{split}
\end{equation}
via
\begin{equation}
\begin{split}
\int_{\R}dx \Phi_{\wt{\W}_m^{(\ell)}}(f^{(\ell)},g^{(\ell)})(x;x')\phi(x) &= \ipp*{\wt{\W}_m^{(\ell)}f^{(\ell)}, \phi\otimes g^{(\ell)}(x',\cdot)}_{\Sc'(\R^\ell)-\Sc(\R^\ell)},\\
\int_{\R}dx \Phi_{\wt{\W}_{n-m}^{(j)}}(f^{(j)},g^{(j)})(x;x')\phi(x) &= \ipp*{\wt{\W}_{n-m}^{(j)}f^{(j)}, \phi\otimes g^{(j)}(x',\cdot)}_{\Sc'(\R^j)-\Sc(\R^j)},
\end{split}
\end{equation}
for any $\phi\in\Sc(\R)$.

We ensure that this is achieved thanks once more to the good mapping property of \cref{def:gmp}. Indeed, proceeding inductively and exploiting the recursion formula and the induction hypothesis that 
\[
\wt{\W}_{1},\ldots,\wt{\W}_{n}\in \bigoplus_{k=1}^{\infty} \L_{gmp}(\Sc(\R^{k}),\Sc'(\R^{k}))
\]
 together with some Fourier analysis,  we show that the composition \eqref{eq:intro_comp_ex} is exactly what we think it should be, namely, the unique distribution in $\D'(\R^{k})$ satisfying \eqref{eq:des_def}, which can then be shown to be tempered. Moreover, by further appealing to the good mapping property and the universal property of the tensor product, we can show that the composition \eqref{eq:intro_comp_ex} indeed belongs to $\L_{gmp}(\Sc(\R^{k}),\Sc'(\R^{k}))$. The preceding discussion is summarized by the following proposition.

\begin{restatable}{prop}{WS}
\label{prop:W_S1}
For each $n\in\N$, there exists an element
\[
\wt{\W}_{n}\in\bigoplus_{k=1}^{\infty}\L_{gmp}(\Sc(\R^{k}),\Sc'(\R^{k}))
\]
defined according to the recursive formula \cref{eq:Wn_rec}, where the composition \eqref{eq:intro_comp_ex} is well-defined in the sense of \cref{prop:H_crit}.
\end{restatable}

Since we are interested in the action of the elements $\wt{\W}_n$ on density matrices, which are self-adjoint, the second step in the construction is to make each $\wt{\W}_{n}$ self-adjoint in the sense of \cref{def:dvo_sa}. By the involution property of the adjoint operation (see \cref{lem:dvo_adj}), the DVO
\begin{equation}
\label{eq:Wn_self_ad}
\W_{n,sa} \coloneqq \frac{1}{2}\paren*{\wt{\W}_{n} + \wt{\W}_{n}^{*}}
\end{equation}
is a self-adjoint element of $\L(\Sc(\R^{k}),\Sc'(\R^{k}))$. Since we want to preserve the good mapping property throughout each step of the construction, the challenge is to show that $\wt{\W}_{n}^{*}$ also has the good mapping property. Naively taking the adjoint of the recursive formula \cref{eq:Wn_rec}, we should formally have that
\begin{equation}\label{eq:Wn_rec_adj}
\wt{\W}_{n+1}^{(k),*} \,``=\textup{''} \ \,  \wt{\W}_{n}^{(k),*}(-i\p_{x_{1}}) + \kappa\sum_{m=1}^{n-1}\sum_{\ell,j\geq 1;\ell+j=k}\paren*{\wt{\W}_{m}^{(\ell),*} \otimes \wt{\W}_{n-m}^{(j),*}}\delta(X_{1}-X_{\ell+1}).
\end{equation}
While the expression on the right-hand side is, a priori, meaningless,\footnote{Among other issues, we note that for $f^{(k)}\in\Sc(\R^{k})$, the tempered distribution $\delta(x_{1}-x_{\ell+1})f^{(k)}$ does not belong to the domain of $\wt{\W}_{m}^{(\ell),*} \otimes \wt{\W}_{n-m}^{(j),*}$,} by inducting on the statement that $\wt{\W}_1^*,\ldots,\wt{\W}_{n-1}^*$ having the good mapping property and exploiting duality, the recursion for $\wt{\W}_n$, and the good mapping property for $\wt{\W}_n$, we are able to prove that the $\wt{\W}_{n}^{*}$ have the good mapping property, as desired.

The third, final, and easiest step of the construction is to symmetrize the $\W_{n,sa}$, so that we obtain an $\infty$-hierarchy which belongs to $\G_{\infty}$. The motivation is that we always restrict to permutation-invariant test functions, reflecting the bosonic nature of the underlying physics. To obtain a formula for $\W_{n}$ from $\W_{n,sa}$ is straightforward. We record this definition in the following proposition:

\begin{prop}\label{prop:Wn_con}
For each $n\in\N$,
\begin{equation}
\label{eq:Wn_fin_def}
-i\W_n\coloneqq -i\Sym\paren*{\W_{n,sa}} = -\frac{i}{2}\paren*{\Sym\paren*{\wt{\W}_{n}} + \Sym\paren*{\wt{\W}_{n}^{*}}} \in \G_{\infty},
\end{equation}
where $\Sym$ is a bosonic symmetrization operator, the definition of which is given in \cref{def:sym_A}.
\end{prop}

\medskip
Having constructed the $\infty$-hierarchies $\{-i\W_{n}\}_{n=1}^{\infty}$, we define trace functionals $\H_n\in \A_\infty$ by
\begin{equation}
\label{Hn_trace}
\H_n(\Gamma) \coloneqq \Tr\paren*{\W_n\cdot\Gamma}, \qquad \Gamma\in\G_\infty^*.
\end{equation}
Since the functionals $I_{n}$ are generated by the operators $w_{n}$, much in the same manner as the trace functionals $\H_{n}$ are generated by the $\W_{n}$, our next task is to relate $\W_n$ to the one-particle nonlinear operators $w_n$ defined in \eqref{eq:w_rec}. Doing so necessitates understanding the action of the $k$-particle components $\wt{\W}_n^{(k)}$ and $\wt{\W}_n^{(k),*}$ on pure tensors of the form
\begin{equation}
\ket*{\phi_1\otimes\cdots\otimes\phi_k}\bra*{\psi_1\otimes\cdots\otimes\psi_k}, \qquad \phi_1,\ldots,\phi_k,\psi_1,\ldots,\psi_k\in\Sc(\R).
\end{equation}
To make this connection precise for the arguments in \cref{sec:gp_flows}, our strategy is to replace the nonlinear operator $w_{n}$ with a multilinear operator by generalizing the recursion \eqref{eq:w_rec}. See \cref{ssec:cor_multi} for more details. As most of the results in \cref{sec:cor} are of a technical nature, and perhaps not so enlightening at this stage, we mention only the following result, which connects $\H_n$ to the functionals $I_{n}$ and can be obtained as an easy corollary of \cref{prop:Hn_Ibn}:
\begin{equation}
\H_n(\Gamma) = I_n(\phi), \qquad \forall\Gamma = (\ket*{\phi^{\otimes k}}\bra*{\phi^{\otimes k}})_{k\in\N}, \enspace \phi\in\Sc(\R).
\end{equation}

Next, we turn to establishing the involution statement of \cref{q:GP_int}, which we record in the following theorem:

\begin{restatable}[Involution theorem]{thm}{GPinvol}
\label{thm:GP_invol}
Let $n,m\in\N$. Then
\begin{equation}
\pb{\H_{n}}{\H_{m}}_{\G_{\infty}^{*}} \equiv 0 \text{ on $\G_\infty^*$.}
\end{equation}
\end{restatable}

To prove \cref{thm:GP_invol}, we proceed on both the one-particle and infinite-particle fronts. We prove that there is an equivalence between the involution of the functionals $\H_n$ and the involution of certain real-valued functionals $I_{b,n}$, defined in \eqref{eq:Ibn_intro_def} below, on a weak Poisson manifold of mixed states. We find this equivalence, explicitly stated in \cref{thm:Equiv} below, quite interesting its own right. We now provide some details of the proof of this equivalence.

On the one-particle front, we relax \eqref{nls} to a system
\begin{equation}
\label{eq:nls_sys_intro}
\begin{cases}
i\p_t\phi_1 = -\Delta\phi_1 + 2\kappa\phi_1^2\phi_2, \\
i\p_t\phi_2 = \Delta\phi_2 -2\kappa\phi_2^2\phi_1
\end{cases},
\end{equation}
where $\phi_1,\phi_2:\R\times\R\rightarrow\C$. We study \eqref{eq:nls_sys_intro} as an integrable system on a \emph{complex} weak Poisson manifold $(\Sc(\R^2),\A_{\Sc,\C},\pb{\cdot}{\cdot}_{L^2,\C})$, see \cref{schwartz2_wpoiss} for the precise definition of this manifold, by revisiting in detail the treatment of the NLS \eqref{nls} in \cite{FT07}. Specifically, we show that there are functionals
\begin{equation}\label{tilde_in_intro}
\tl{I}_n(\phi_1,\phi_2) \coloneqq \int_{\R}dx\phi_2(x) w_{n,(\phi_1,\phi_2)}(x), \qquad \forall (\phi_1,\phi_2)\in\Sc(\R)^2, \enspace n\in\N,
\end{equation}
where $w_{n,(\phi_1,\phi_2)}(x)$ satisfies a similar recursion formula to the $w_n$, see \cref{eq:wn_b_rec}, such that $\tl{I}_3$ is the Hamiltonian for NLS system \eqref{eq:nls_sys_intro}, and such that the $\tl{I}_n$ commute on $(\Sc(\R^2),\A_{\Sc,\C},\pb{\cdot}{\cdot}_{L^2,\C})$.

Since we are ultimately interested in \emph{real}, not complex, weak Poisson manifolds, we pass to another weak Poisson manifold of \emph{mixed states}, $(\Sc(\R;\mathcal{V}),\A_{\Sc,\mathcal{V}},\pb{\cdot}{\cdot}_{L^2,\mathcal{V}})$, where the space $\Sc(\R;\mathcal{V})$ consists of Schwartz functions $\gamma$ taking values in the space $\mathcal{V}$ of self-adjoint, off-diagonal $4\times 4$ complex matrices:
\begin{equation}
\label{eq:Schw_ms}
\gamma = \frac{1}{2}\mathrm{odiag}(\phi_1,\ol{\phi_2},\phi_2,\ol{\phi_1}) = \frac{1}{2}\begin{pmatrix} 0&0&0&\phi_1\\0&0&\ol{\phi_2}&0\\0&\phi_2&0&0\\ \ol{\phi_1}&0&0&0\end{pmatrix}, \qquad \phi_1,\phi_2\in\Sc(\R).
\end{equation}
We refer to \eqref{symp_v}, \eqref{alg_v}, and \cref{prop:Schw_WP_V} for the precise definition and properties of this weak Poisson manifold.

We use the $\tl{I}_n$ to define real-valued functionals $I_{b,n}\in\A_{\Sc,\mathcal{V}}$ on the manifold $(\Sc(\R;\mathcal{V}),\A_{\Sc,\mathcal{V}},\pb{\cdot}{\cdot}_{L^2,\mathcal{V}})$ via the formula
\begin{equation}
\label{eq:Ibn_intro_def}
I_{b,n}(\gamma) \coloneqq \frac{1}{2}\paren*{\tl{I}_n(\phi_1,\ol{\phi_2}) + \tl{I}_n(\phi_2,\ol{\phi_1})},
\end{equation}
and we show in \cref{prop:Ib_invol} that the family $\{I_{b,n}\}_{n\in\N}$ is in mutual involution with respect to the Poisson bracket $\pb{\cdot}{\cdot}_{L^2,\mathcal{V}}$. As we do not feel the results described in this paragraph are the primary contribution of this work, but nevertheless believe they may be of independent interest to the community, we have placed them in \cref{app:nls_pc} and not the main body of the paper.

On the infinite-particle front, we first demonstrate that there is a Poisson morphism
\begin{equation}
\begin{split}
&\iota_{\m}: (\Sc(\R;\mathcal{V}), \A_{\Sc,\mathcal{V}}, \pb{\cdot}{\cdot}_{L^2,\mathcal{V}}) \rightarrow (\G_\infty^*,\A_{\infty},\pb{\cdot}{\cdot}_{\G_\infty^*})\\
&\iota_{\m}(\gamma) \coloneqq \frac{1}{2}\paren*{\ket*{\phi_1^{\otimes k}}\bra*{\phi_2^{\otimes k}} + \ket*{\phi_2^{\otimes k}}\bra*{\phi_1^{\otimes k}}}_{k\in\N}, \qquad \gamma = \frac{1}{2}\mathrm{odiag}(\phi_1,\ol{\phi_2},\phi_2,\ol{\phi_1}).
\end{split}
\end{equation}
The subscript $\m$ signifies that $\iota_\m$ produces a mixed state element of $\G_\infty^*$.

\begin{restatable}{thm}{Pomo}
\label{thm:GP_pomo}
The map $\iota_\m$ is a Poisson morphism of $(\Sc(\R;\mathcal{V}), \A_{\Sc,\mathcal{V}}, \pb{\cdot}{\cdot}_{L^2,\mathcal{V}})$ into $(\G_\infty^*,\A_{\infty},\pb{\cdot}{\cdot}_{\G_\infty^*})$; i.e., it is a smooth map with the property that
\begin{equation}
\iota_{\m}^*\pb{\cdot}{\cdot}_{\G_\infty^*} = \pb{\iota_{\m}^* \cdot }{\iota_{\m}^* \cdot }_{L^2,\mathcal{V}},
\end{equation}
where $\iota_\m^*$ denotes the pullback of $\iota_\m$.
\end{restatable}

\cref{thm:GP_pomo} is a generalization of \cite[Theorem 2.12]{MNPRS1_2019} in our companion paper and, in fact, recovers this previous theorem since \cref{prop:Schw_WP_V} demonstrates that there is also a Poisson morphism
\begin{equation}
\iota_{\mathfrak{pm}}: (\Sc(\R),\A_{\Sc},\pb{\cdot}{\cdot}_{L^2}) \rightarrow (\Sc(\R;\mathcal{V}),\A_{\Sc,\mathcal{V}},\pb{\cdot}{\cdot}_{L^2,\mathcal{V}}), \quad \phi \mapsto \frac{1}{2}\mathrm{odiag}(\phi,\ol{\phi},\phi,\ol{\phi}),
\end{equation}
and the composition of Poisson morphisms is again a Poisson morphism.

The motivation for \cref{thm:GP_pomo} is the following. Since
\begin{equation}
I_{b,n}(\gamma) = \H_n(\iota_{\m}(\gamma)), \qquad \forall\gamma\in\Sc(\R;\mathcal{V})
\end{equation}
by \cref{prop:Hn_Ibn}, and since $\pb{I_{b,n}}{I_{b,m}}_{L^2,\mathcal{V}}\equiv 0$ on $\Sc(\R;\V)$, for any $n,m\in\N$, by \cref{prop:Ib_invol}, \cref{thm:GP_pomo} implies that
\begin{equation}\label{com_span}
0=\pb{\H_n}{\H_m}_{\G_\infty^*}(\iota_{\m}(\gamma)) = \frac{1}{2}\sum_{k=1}^\infty i\Tr_{1,\ldots,k}\paren*{\comm{-i\W_n}{-i\W_m}_{\G_\infty}^{(k)}\paren*{\ket*{\phi_1^{\otimes k}}\bra*{\phi_2^{\otimes k}}+\ket*{\phi_2^{\otimes k}}\bra*{\phi_1^{\otimes k}}}}.
\end{equation}
Note that only finitely many terms in the above summation are nonzero. Next, we use a scaling argument to show that \eqref{com_span} implies that each of the summands in the right-hand side of \eqref{com_span} are identically zero:
\begin{equation}
\label{eq:intro_sum_van}
\frac{i}{2}\Tr_{1,\ldots,k}\paren*{\comm{-i\W_n}{-i\W_m}_{\G_\infty}^{(k)}\paren*{\ket*{\phi_1^{\otimes k}}\bra*{\phi_2^{\otimes k}}+\ket*{\phi_2^{\otimes k}}\bra*{\phi_1^{\otimes k}}}} = 0, \qquad \forall \phi_1,\phi_2\in\Sc(\R), \enspace k\in\N.
\end{equation}
The intuition is that if a polynomial is identically zero then all of its coefficients are zero. By unpacking the definition of the Poisson bracket $\pb{\H_n}{\H_m}_{\G_\infty^*}$, \eqref{eq:intro_sum_van} yields
\begin{equation}
\label{eq:pb_r1_simp}
\pb{\H_n}{\H_m}_{\G_\infty^*}(\Gamma) = 0, \qquad \forall\Gamma=\frac{1}{2}\paren*{\ket*{\phi_{k,1}^{\otimes k}}\bra*{\phi_{k,2}^{\otimes k}} + \ket*{\phi_{k,2}^{\otimes k}}\bra*{\phi_{k,1}^{\otimes k}}}_{k\in\N},
\end{equation}
where $\phi_{k,1},\phi_{k,2}\in\Sc(\R)$ for every $k\in\N$. By then using an approximation argument from \cref{app:multi_alg} involving symmetric-rank-1 approximations (see \cref{cor:Sch_ST_DM_d}) together with the continuity of $\pb{\H_n}{\H_m}_{\G_\infty^*}$, we obtain from \eqref{eq:pb_r1_simp} that Poisson commutativity of the $I_{b,n}$ implies the Poisson commutativity of $\H_n$. The reverse implication is a straightforward consequence of \cref{thm:GP_pomo}. Summarizing the preceding discussion, we have the following equivalence result:

\begin{restatable}[Poisson commutativity equivalence]{thm}{Equiv}
\label{thm:Equiv}
For any $n,m\in\N$,
\begin{equation}\label{ibn_inv}
\pb{I_{b,n}}{I_{b,m}}_{L^2,\mathcal{V}}(\gamma) =0, \qquad \forall \gamma \in\Sc(\R;\mathcal{V}),
\end{equation}
if and only if
\begin{equation}
\pb{\H_n}{\H_m}_{\G_\infty^*}(\Gamma) = 0, \qquad \forall \Gamma\in\G_\infty^*.
\end{equation}
\end{restatable}

In light of \cref{prop:Ib_invol}, which asserts the validity of \eqref{ibn_inv}, we then obtain \cref{thm:GP_invol} from \cref{thm:Equiv}, thus answering \cref{q:GP_int}.

\medskip
Having resolved  \cref{q:GP_int}, we turn to answering \cref{q:nGP}. For each $n\in\N$, we define the \emph{$n$-th~GP hierarchy (nGP)} to be the Hamiltonian equation of motion generated by the functional $\H_{n}$ with respect to the Poisson structure on $\G_{\infty}^{*}$:
\begin{equation}\label{eq:nGP}
\paren*{\frac{d}{dt}\Gamma} = X_{\H_{n}}(\Gamma),
\end{equation}
where $X_{\H_{n}}$ is the unique Hamiltonian vector field defined by $\H_{n}$. See \ref{item:wp_P3} of \cref{def:WP} for the definition of the Hamiltonian vector field. We generalize the fact that solutions to the NLS generate a special class of factorized solutions to the GP hierarchy by proving that the same correspondence is true for the (nNLS) and (nGP). Thus, we are led to our final main theorem, providing an affirmative answer to \cref{q:nGP}.

\begin{restatable}[Connection between (nGP) and (nNLS)]{thm}{GPfac}
\label{thm:GP_fac}
Let $n\in\N$. Let $I\subset\R$ be a compact interval and let $\phi \in C^{\infty}(I;\mathcal{S}(\R))$ be a solution to the (nNLS) with lifespan $I$. If we define
\begin{equation}
\Gamma\in C^{\infty}(I;\G_{\infty}^{*}), \qquad  \Gamma \coloneqq \paren*{\ket*{\phi^{\otimes k}}\bra*{\phi^{\otimes k}}}_{k \in \N},
\end{equation}
then $\Gamma$ is a solution to the (nGP).
\end{restatable}

\begin{remark}\label{gp_ham}
In \cite{MNPRS1_2019}, we defined the \emph{Gross-Pitaevskii Hamiltonian functional} $\H_{GP}$ by
\begin{equation}
\H_{GP}(\Gamma) \coloneqq  \Tr_{1}\paren*{-\Delta_{x_1}\gamma^{(1)}} + \kappa\Tr_{1,2}\paren*{\delta(X_{1}-X_{2})\gamma^{(2)}}, \qquad \forall \Gamma=(\gamma^{(k)})_{k\in\N}\in\G_\infty^*.
\end{equation}
In particular, $\H_{GP} = \H_3$, and in the one-dimensional case, we recover Theorem 2.10 from \cite{MNPRS1_2019}, which asserts that the GP hierarchy \eqref{eq:GP'} is the Hamiltonian equation of motion on $(\G_\infty^*,\A_{\infty},\pb{\cdot}{\cdot}_{\G_\infty^*})$ induced by $\H_{GP}$.
\end{remark}

\begin{remark}
\cref{thm:GP_fac} does not assert that the factorized solution $(\ket*{\phi^{\otimes k}}\bra*{\phi^{\otimes k}})_{k\in\N}$ is the unique solution to the $n$-th GP hierarchy starting from factorized initial data, only that it is \emph{a} particular solution. More generally, \cref{thm:GP_fac} makes no assertion about the uniqueness of solutions to the (nGP) in the class $C^{\infty}(I;\G_{\infty}^{*})$. While the (nNLS) are known to be globally well-posed in the Schwartz class by the work of Beals and Coifman \cite{BC1984} and Zhou \cite{Zhou1989}, unconditional uniqueness of the $n$-th GP hierarchy in the class $C^{\infty}(I;\G_\infty^*)$, for some compact interval $I$, is an open problem, the resolution of which we do not address in this work.
\end{remark}

To prove \cref{thm:GP_fac}, we need to show that the $n$-th GP Hamiltonian vector field $X_{\H_n}$ can be written as
\begin{equation}
\label{eq:GPham_wts}
X_{\H_n}(\Gamma)^{(k)} = \sum_{\alpha=1}^{k} \paren*{\ket*{\phi^{\otimes (\alpha-1)}\otimes \grad_s I_n(\phi)\otimes \phi^{(k-\alpha)}}\bra*{\phi^{\otimes k}} + \ket*{\phi^{\otimes k}}\bra*{\phi^{\otimes (\alpha-1)}\otimes \grad_s I_n(\phi)\otimes \phi^{(k-\alpha)}}},
\end{equation}
for $\Gamma$ as in the statement of \cref{thm:GP_fac}. We remind the reader that $\grad_s I_n$ denotes the symplectic gradient of $I_n$ with respect to the form $\omega_{L^2}$, see \cref{def:re_grad}. To establish the identity \eqref{eq:GPham_wts}, we use a formula from \cref{ssec:cor_symgrad} for $\grad_s I_n$, which is in terms of the G\^ateaux derivatives of the nonlinear operators $w_n$. Combining this formula with the computation of $X_{\H_n}(\Gamma)$ for factorized $\Gamma$ (see \cref{lem:XHn_Sch_ker}), which extensively uses the good mapping property of the generators of the $\H_n$ (i.e. $-i\W_n$), we obtain \eqref{eq:GPham_wts} and hence the desired conclusion.

\subsection{Organization of the paper}
We close \cref{sec:mr_bp} by commenting on the organization of the paper. In \cref{sec:pre}, we review the notation and background material used throughout this paper. \cref{ssec:pre_calc} briefly reviews the G\^ateaux derivative and calculus in the setting of locally convex spaces. \cref{ssec:pre_bos} introduces the relevant spaces of bosonic test functions and distributions, symmetrization and contraction operators, and tensor products. \cref{ssec:pre_WP} is a crash course on weak symplectic and Poisson manifolds in addition to discussing several important examples of such objects which appear frequently in this work. Lastly, \cref{ssec:pre_LA_facts} quickly reviews the definition of a Lie algebra as well as the classical Lie-Poisson construction. As this subject is thoroughly treated in Section 4.2 of our companion paper \cite{MNPRS1_2019}, we have omitted proofs and instead refer the reader to that work for more details.

In \cref{sec:con_GP_W}, we construct our observable $\infty$-hierarchies $-i\W_{n}$, thereby proving \cref{prop:Wn_con}. The section is divided into three subsections corresponding to each stage of the construction: the preliminary version, followed by the self-adjoint version, followed by the final bosonic, self-adjoint version.

\cref{sec:cor} is devoted to analyzing the correspondence between the $w_n$ and the $\W_n$ and the consequences of this correspondence. \cref{ssec:cor_multi} contains the ``multilinearization'' of the $w_n$. \cref{ssec:cor_symgrad} contains the proof of a formula for the symplectic gradients of the $I_n$. \cref{ssec:cor_ptr} connects the multilinearizations of the $w_n$ from \cref{ssec:cor_multi} with the partial traces of the $\W_n$.

In \cref{sec:invol}, we prove our involution result, \cref{thm:GP_invol}, in addition to the main auxiliary results involved in the proof of this theorem, which might be of independent interest. This section is broken down into four subsections in order to make the presentation more modular. \cref{ssec:invol_pomo} contains the proof of the Poisson morphism result, \cref{thm:GP_pomo}. \cref{ssec:cor_Ibn_Hn} connects the infinite-particle functionals $\H_n$ to the one-particle functions $I_{b,n}$ via the Poisson morphism of \cref{thm:GP_pomo} and the correspondence results of \cref{ssec:cor_ptr}. \cref{ssec:invol_pb_an} contains the proofs of the Poisson commutativity equivalence result, \cref{thm:Equiv}, and the involution result, \cref{thm:GP_invol}. Lastly, \cref{ssec:invol_ntriv} contains the proof of \cref{prop:invol_ntriv}, which asserts that there is at least one functional which does not Poisson commute with a given $\H_n$. 

In the last section, \cref{sec:gp_flows}, of the main body of the paper, we prove our $n$-th GP/$n$-th NLS  correspodence result,  \cref{thm:GP_fac}. \cref{ssec:gp_flows_vf} is devoted to the computation of the Hamiltonian vector fields of the $\H_n$ evaluated on factorized states, and \cref{ssec:gp_flows_prf} is devoted to the proof of \cref{thm:GP_fac}. To close the section,  we compute in \cref{ssec:gp_flows_ex} the fourth GP hierarchy, which corresponds to the complex mKdV equation.

We have also included several appendices to make this work as self-contained as possible. \cref{app:nls_pc} revisits the treatment in Faddeev and Takhtajan's monograph \cite{FT07} of the involution of the functionals $I_n$ in the more general setting of the system \eqref{eq:nls_sys_intro}. We were unable to find a reference covering this generalization. Therefore, we provide a fairly thorough presentation at the expense of a lengthy appendix. \cref{app:multi_alg} contains a quick review of some facts from multilinear algebra on symmetric tensors, which we use to establish approximation results for bosonic Schwartz functions and density matrices. \cref{app:DVO} is devoted to technical facts about distribution-valued operators and topological tensor products, which justify the manipulations used extensively in this paper.  Furthermore, this appendix includes an elaboration on the good mapping property, in particular, some technical consequences of it which are used in the body of the paper. \cref{app:DVO} is also included in \cite{MNPRS1_2019}; however, we include it here, with most of the proofs omitted, for convenient referencing. \cref{app:WF} contains technical material on products of distributions, specifically on when the product of two distributions can be rigorously defined.

\section{Notation}\label{sec:not}
\subsection{Index of notation}

We include \cref{tab:notation}, located at the end of the manuscript, as a guide for the frequently used symbols in this work. In this table, we either provide a definition of the notation or a reference for where the symbol is defined.

\section{Preliminaries}\label{sec:pre}
\subsection{Calculus on locally convex spaces}\label{ssec:pre_calc}
We begin by recalling some definitions related to calculus on locally convex spaces, which we make use of in the sequel. For further background material, we refer the reader to the lecture notes of Milnor \cite{Milnor1984}.

\begin{mydef}[Locally convex space]
A topological vector space (tvs) $X$ over a scalar field $\K$ is said to be \emph{locally convex} if every neighborhood $U\ni 0$ contains a neighborhood $U'\ni 0$ which is convex.
\end{mydef}

A particularly nice consequence of local convexity is the following Hahn-Banach type result.
\begin{prop}[Hahn-Banach]\label{prop:HB}
If $X$ is locally convex, then given two distinct vectors $x,y\in X$, there exists a continuous $\K$-linear map $\ell: X\rightarrow \K$ with $\ell(x)\neq \ell(y)$.
\end{prop}

\begin{mydef}[G\^{a}teaux derivative]\label{gateaux_deriv}
Let $X$ and $Y$ be locally convex $\K$-tvs, let $X_{0}\subset X$ and $Y_{0}\subset Y$ be open sets, and let $f: X_{0}\rightarrow Y_{0}$ be a continuous map. Given a point $x\in X_{0}$ and a direction $v\in X$, we define the \emph{directional derivative} or \emph{G\^{a}teaux derivative} of $f$ at $x$ in the direction $v$ to be the vector
\begin{equation}
f'(x;v) \coloneqq f_{x}'(v) \coloneqq \lim_{t\rightarrow 0} \frac{f(x+tv) - f(x)}{t},
\end{equation}
if this limit exists. We call the map $f_{x}': X\rightarrow Y$ the \emph{derivative of $f$ at the point $x_{0}$}. We use the notation $df[x](v) \coloneqq f'(x;v)$.

The map $f:X_{0}\rightarrow Y_{0}$ is \emph{$C^{2}$ G\^{a}teaux} if $f$ is a $C^{1}$ G\^{a}teaux map and for each $v_{1}\in X$ fixed, the map
\begin{equation}
X_{0} \rightarrow Y, \qquad x\mapsto f'(x;v_{1})
\end{equation}
is $C^1$ with G\^{a}teaux derivative
\begin{equation}
\lim_{t\rightarrow 0} \frac{f'(x+tv_{2}; v_{1}) - f'(x;v_{1})}{t}
\end{equation}
depending continuously on $(x;v_{1},v_{2})\in X_{0}\times X\times X$ equipped with the product topology. If this limit exists, we call it the \emph{second G\^{a}teaux derivative} of $f$ at $x$ in the directions $v_{1},v_{2}$ and denote it by $f''(x;v_{1},v_{2})$. We inductively define \emph{$C^{r}$ maps} $X_{0}\rightarrow Y_{0}$. If a map is $C^{r}$ for every $r\in \N$, then we say that $f$ is a \emph{$C^{\infty}$ map} or alternatively, a \emph{smooth map}.
\end{mydef}

\begin{prop}[Symmetry and $r$-linearity of $f_{x_{0}}^{(r)}$]
If for $r\in\N$, the map $f$ is $C^{r}$, then for each fixed $x_{0}\in X_{0}$, the map
\begin{equation}
\underbrace{X\times\cdots \times X}_{r} \rightarrow Y, \qquad (v_{1},\ldots,v_{r}) \mapsto f^{(r)}(x_{0};v_{1},\ldots,v_{r})
\end{equation}
is $r$-linear and symmetric, i.e. for any permutation $\pi\in\Ss_{r}$,
\begin{equation}
f^{(r)}(x_{0};v_{\pi(1)},\ldots,v_{\pi(r)}) = f^{(r)}(x_{0};v_{1},\ldots,v_{r}).
\end{equation}
\end{prop}

\begin{prop}[Composition]
If $f:X_{0}\rightarrow Y_{0}$ and $g: Y_{0}\rightarrow Z_{0}$ are $C^{r}$ maps, then $g\circ f: X_{0}\rightarrow Z_{0}$ is $C^{r}$ and the derivative of $(g\circ f)$ at the point $x\in X_{0}$ is the map $g_{f(x)}' \circ f_{x}': X\rightarrow Z$.
\end{prop}

\subsection{Bosonic functions, operators, and tensor products}\label{ssec:pre_bos}
We now review the main spaces of test functions and distributions and some basic facts about tensor products used extensively in the body of the paper.

We denote the pairing of a tempered distribution $u\in\Sc'(\R^{k})$ with a Schwartz function $f\in\Sc(\R^{k})$ by
\begin{equation}
\ipp{u,f}_{\Sc'(\R^{k})-\Sc(\R^{k})}.
\end{equation}
For $1\leq p\leq \infty$, we use the notation $L^{p}(\R^{k})$ to denote Banach space of $p$-integrable functions with norm $\|\cdot\|_{L^{p}(\R^{k})}$. In particular, when $p=2$, we denote the $L^{2}$ inner product by
\begin{equation}
\ip{f}{g} \coloneqq \int_{\R^{k}}d\ux_{k}\ol{f(\ux_{k})} g(\ux_{k}).
\end{equation}
Note that we use the physicist's convention that the inner product is complex linear in the second entry. Similarly, for $u\in\Sc'(\R^{k})$ and $f\in\Sc(\R^{k})$, we use the notation $\ip{u}{f}$ and $\ip{f}{u}$ to denote
\begin{equation}
\ip{u}{f} \coloneqq \ol{\ipp{u,\bar{f}}_{\Sc'(\R^{k})-\Sc(\R^{k})}} \qquad \text{ and }\qquad \ip{f}{u} \coloneqq \ol{\ip{u}{f}}.
\end{equation}
Alternatively, the right-hand side of the first definition may be taken as the definition of the tempered distribution $\bar{u}$. Throughout the paper, we will use an integral to represent the pairing of a distribution and a test function. 

We denote the symmetric group on $k$ letters by $\Ss_{k}$.  For a permutation $\pi\in\Ss_{k}$, we define the map $\pi: \R^{k}\rightarrow \R^{k}$ by
\begin{equation}
\label{eq:pi_vec_def}
\pi(\ul{x}_{k}) \coloneqq (x_{\pi(1)},\ldots,x_{\pi(k)}).
\end{equation}
For a complex-valued, measurable function $f : \R^k \to \C$, we define the permuted function
\begin{equation}\label{eq:pi_func_def}
(\pi f)(\ul{x}_{k}) \coloneqq (f\circ\pi)(\ul{x}_{k}) = f(x_{\pi(1)},\ldots,x_{\pi(k)}), \qquad \forall\ux_k\in\R^k.
\end{equation}

\begin{mydef}
We say that a measurable function $f:\R^{k}\rightarrow \C$ is \emph{symmetric} or \emph{bosonic} if
\begin{equation}
\pi(f) = f, \qquad \forall \pi\in\Ss_{k}.
\end{equation}
\end{mydef}

\begin{mydef}\label{def:sym_f}
We define the \emph{symmetrization operator} $\Sym_k$ on the space of measurable complex-valued functions by
\begin{equation}
\Sym_k(f)(\ul{x}_{k}) \coloneqq \frac{1}{k!}\sum_{\pi\in\Ss_{k}} \pi(f)(\ul{x}_{k}), \qquad \forall\ux_k\in\R^k.
\end{equation}
By duality and continuity of the symmetrizing operation, we can extend the symmetrization operator to $\mathcal{S}'(\R^{k})$.
\end{mydef}

\begin{mydef}[Bosonic test functions/distributions]\label{sym_schwartz}
For $k\in\N$, let $\mathcal{S}_{s}(\R^{k})$ denote the subspace of $\mathcal{S}(\R^{k})$ consisting of Schwartz functions which are bosonic. We say that a tempered distribution $u\in\mathcal{S}'(\R^{k})$ is \emph{symmetric} or \emph{bosonic} if for every permutation $\pi\in\Ss_{k}$,
\begin{equation}
\ipp{u,g\circ\pi^{-1}}_{\Sc'(\R^k)-\Sc(\R^k)} = \ipp{u,g}_{\Sc'(\R^k)-\Sc(\R^k)},
\end{equation}
for all $g\in\mathcal{S}(\R^{k})$. We denote the subspace of such tempered distributions by $\mathcal{S}_{s}'(\R^{k})$.
\end{mydef}

\begin{remark}\label{rem:sym}
It is straightforward to check that $\Sym_k$ is a continuous operator $\mathcal{S}(\R^{k}) \rightarrow \mathcal{S}_{s}(\R^{k})$ and $\mathcal{S}'(\R^{k}) \rightarrow \mathcal{S}_{s}'(\R^{k})$. Furthermore, a tempered distribution $u$ is bosonic if and only if $u=\Sym_k(u)$.
\end{remark}

\begin{lemma}
\label{lem:b_td}
We have the identification
\begin{equation}
\mathcal{S}_{s}'(\R^{k}) \cong (\mathcal{S}_{s}(\R^{k}))^*,
\end{equation}
where $(\Sc_s(\R^k))^*$ denotes the topological dual of $\Sc_s(\R^k)$.
\end{lemma}

Given two locally convex spaces $E$ and $F$, we denote the space of continuous linear maps $E\rightarrow F$ by $\L(E,F)$. We topologize $\L(E,F)$ with the topology of bounded convergence.  For our purposes, we will typically have $E,F\in\{\mathcal{S}(\R^{k}), \mathcal{S}_{s}(\R^{k}), \mathcal{S}'(\R^{k}), \mathcal{S}_{s}'(\R^{k})\}$. In the case that $E=\Sc(\R^k)$ and $F=\Sc'(\R^k)$, the bounded topology is generated by the seminorms
\begin{equation}
\|A\|_{\mathfrak{R}} \coloneqq \sup_{f,g\in\mathfrak{R}} |\ipp{Af,g}_{\Sc'(\R^k)-\Sc(\R^k)}|, \qquad \forall A\in \L(\Sc(\R^k), \Sc'(\R^k)),
\end{equation}
where $\mathfrak{R}$ ranges over the bounded subsets of $\Sc(\R^k)$. An identical statement holds with all spaces replaced by their symmetric counterparts.  We topologize $\mathcal{S}'(\R^k)$ with the \emph{strong dual topology}, which is the locally convex topology generated by the seminorms of the form
\begin{equation}
\|f\|_{B} \coloneqq \sup_{\varphi\in B} \left|\ipp{f,\varphi}_{\Sc'(\R^k)-\Sc(\R^k)}\right|,
\end{equation}
where $B$ ranges over the family of all bounded subsets of $\mathcal{S}(\R^k)$. Note that since $\mathcal{S}(\R^k)$ is a Montel space, bounded subsets are precompact. An identical statement holds with all spaces replaced by their symmetric counterparts.

Given two locally convex spaces $E$ and $F$ over a field $\K$, we denote an\footnote{The reader will recall that the algebraic tensor product is only defined up to unique isomorphism.} algebraic tensor product of $E$ and $F$ consisting of finite linear combinations
\begin{equation}
\sum_{j=1}^{n} \lambda _{j} e_{j} \otimes f_{j}, \qquad e_{j}\in E, \enspace f_{j}\in F, \enspace \lambda_j\in\K
\end{equation}
by $E\otimes F$. We note that since the spaces we deal with in this paper are nuclear, the topologies of the injective and projective tensor products coincide. Hence, we can unambiguously write $E\hat{\otimes} F$ to denote the completion of $E\otimes F$ under either of the aforementioned topologies. Given locally convex spaces $E_{j}$ and $F_{j}$ for $j=1,2$ and linear maps $T:E_{1}\rightarrow E_{2}$ and $S:F_{1}\rightarrow F_{2}$, and a tensor product
\begin{equation}
B:E_{1}\times E_{2} \rightarrow E_{1}\otimes E_{2},
\end{equation}
the notation $T\otimes S$ denotes the unique linear map $T\otimes S:E_{1}\otimes F_{1} \rightarrow E_{2}\times F_{2}$ such that
\begin{equation}
(T\otimes S) \circ B = T\times S.
\end{equation}
Note that the existence of such a unique map is guaranteed by the universal property of the tensor product.

When $E$ and $F$ are subspaces of measurable functions on $\R^{m}$ and $\R^{n}$ respectively, and $e \in E$ and $f \in F$, we let $e\otimes f$ denote the realization of the tensor product given by
\begin{equation}
e\otimes f : \R^{m}\times\R^{n} \rightarrow \C, \qquad (e\otimes f)(\ux_{m};\ux_{n}') \coloneqq e(\ux_{m})f(\ux_{n}'), \qquad \forall (\ux_m,\ux_n')\in\R^m\times\R^n.
\end{equation}
which induces a bilinear map $E\times F\rightarrow E\otimes F$.  Similarly, if $E'$ and $F'$ are the duals of spaces of test functions $E$ and $F$ (e.g. $E'=\D'(\R^{m})$ and $F'=\D'(\R^{n})$ ), we let $u\otimes v$ denote the unique distribution satisfying 
\begin{equation}
(u\otimes v)(e\otimes f) = u(e) \cdot v(f).
\end{equation}
Finally, if $\phi:\R^{m}\rightarrow\C$ is a measurable function, we use the notation $\phi^{\otimes k}$, for $k\in\N$, to denote the measurable function $\phi^{\otimes k}: \R^{mk}\rightarrow\C$ defined by
\begin{equation}\label{tensor_def}
\phi^{\otimes k}(\ux_{m,1},\ldots,\ux_{m,k}) \coloneqq \prod_{\ell=1}^{k} \phi(\ux_{m;\ell}),
\end{equation}
and we use the notation $\phi^{\times k}$ to denote the measurable function $\phi^{\times k}:\R^{m}\rightarrow\C^k$
\begin{equation}\label{prod_coords}
\phi^{\times k}(\ux_m) \coloneqq (\phi(\ux_m),\ldots,\phi(\ux_m)).
\end{equation}

For $A^{(k)} \in \L(\Sc_{s}(\R^{k}),\Sc_{s}'(\R^{k}))$ and $\pi \in \Ss_k$, we define
\begin{equation}\label{eq:op_coord}
A_{(\pi(1),\ldots, \pi(k))}^{(k)} \coloneqq \pi \circ A^{(k)} \circ \pi^{-1}.
\end{equation}
In particular, $A_{(1,\ldots,k)}^{(k)}=A^{(k)}$.

\begin{mydef}\label{def:sym_A} Given $A^{(k)} \in \L(\Sc(\R^{k}),\Sc'(\R^{k}))$, we define its \emph{bosonic symmetrization} $\Sym_k(A^{(k)})$ by
\begin{equation}\label{eq:sym_repeat}
\Sym_k(A^{(k)})  \coloneqq \frac{1}{k!} \sum_{\pi\in\Ss_{k}} A_{(\pi(1), \ldots, \pi(k))}^{(k)}.
\end{equation}
For $A=(A^{(k)})_{k\in\N}\in\bigoplus_{k=1}^\infty \L(\Sc(\R^{k}),\Sc'(\R^{k}))$, we define
\begin{equation}
\Sym(A) \coloneqq \paren*{\Sym_k(A^{(k)})}_{k\in\N}.
\end{equation}
\end{mydef}

\begin{mydef}[Bosonic operators]
Let $k\in\N$. We say that an operator $A^{(k)}: \mathcal{S}(\R^k) \rightarrow \mathcal{S}'(\R^k)$ is \emph{bosonic} or \emph{permutation invariant} if $A^{(k)}$ maps $\mathcal{S}_{s}(\R^k)$ into $\mathcal{S}_{s}'(\R^k)$.
\end{mydef}

The analogue of \cref{rem:sym} holds for the symmetrization of operators: bosonically symmetrized operators are indeed maps from bosonic Schwartz functions to bosonic tempered distributions.

\begin{lemma}\label{lem:sym_op_space}
Let $k \in \N$. If $A^{(k)} \in \mathcal{L}(\mathcal{S}(\R^{k}), \mathcal{S}'(\R^{k}))$, then
\begin{equation}
\Sym_k(A^{(k)}) \in \mathcal{L}(\mathcal{S}_{s}(\R^{k}), \mathcal{S}_{s}'(\R^{k})).
\end{equation} 
Furthermore, if $A^{(k)}\in \L_{gmp}(\Sc(\R^{k}),\Sc'(\R^{k}))$, then
\begin{equation}
\Sym_k(A^{(k)}) \in \L_{gmp}(\Sc_{s}(\R^{k}),\Sc_{s}'(\R^{k})).
\end{equation}
\end{lemma}

The following technical lemma is frequently used implicitly in the sequel. For definitions and discussion of the generalized trace, see \cref{def:gen_trace}.

\begin{lemma}\label{lem:tr_bos}
Let $k\in\N$, and let $\gamma^{(k)}\in\L(\Sc_{s}'(\R^{k}),\Sc_{s}(\R^{k}))$ and $A^{(k)}\in\L(\Sc(\R^{k}),\Sc'(\R^{k}))$. Then for any permutation $\tau\in\Ss_{k}$, we have that
\begin{equation}
\Tr_{1,\ldots,k}\paren*{A_{(\tau(1),\ldots,\tau(k))}^{(k)}\gamma^{(k)}} = \Tr_{1,\ldots,k}\paren*{A^{(k)}\gamma^{(k)}}.
\end{equation}
\end{lemma}

\subsection{Weak Poisson structures and Hamiltonian systems}
\label{ssec:pre_WP}
We recall the definition of a \emph{weak Poisson structure} due to Neeb et al. \cite{NST2014} generalized to allow for complex-valued functionals. The presentation closely follows that of Section 4.1 in our companion paper \cite{MNPRS1_2019}.

\begin{mydef}[Weak Poisson manifold]\label{def:WP}
A \emph{weak Poisson structure} on $M$ is a unital sub-algebra $\mathcal{A}\subset C^{\infty}(M;\C)$ and a bilinear map $\pb{\cdot}{\cdot}:\mathcal{A}\times\mathcal{A}\rightarrow\mathcal{A}$ satisfying the following properties:
\begin{enumerate}[(P1)]
\item\label{item:wp_P1}
The bilinear map $\pb{\cdot}{\cdot}$, is a Lie bracket\footnote{See \cref{def:la} for details.} and satisfies the Leibnitz rule
\begin{equation}
\pb{F}{GH} = \pb{F}{G}H+G\pb{F}{H}, \qquad \forall F,G,H\in\mathcal{A}.
\end{equation}
We call $\pb{\cdot}{\cdot}$ a \emph{Poisson bracket}.
\item\label{item:wp_P2}
For all $m\in M$ and $v\in T_{m}M$ satisfying $dF[m](v)=0$ for all $F\in\A$, we have that $v=0$.
\item\label{item:wp_P3}
For every $H\in\A$, there exists a smooth vector field $X_{H}$ on $M$ satisfying
\begin{equation}
X_{H}F=\pb{F}{H}, \qquad \forall F\in\A,
\end{equation} 
where in the left-hand side of the identity, we regard $X_H$ as a derivation. We call $X_{H}$ the \emph{Hamiltonian vector field} associated to $H$.
\end{enumerate}
If properties \ref{item:wp_P1} - \ref{item:wp_P3} are satisfied, then we call the triple $(M,\A,\pb{\cdot}{\cdot})$ a \emph{weak Poisson manifold}.
\end{mydef}

We now record some observations from \cite{NST2014} about the definition of a weak Poisson structure.

\begin{remark}\label{rem:hvf_u}
\ref{item:wp_P2} implies that the Hamiltonian vector field $X_{H}$ associated to some $H\in\A$ is uniquely determined by the relation
\begin{equation}
\pb{F}{H}(m) = (X_{H}F)(m)=dF[m](X_{H}(m)), \qquad \forall F\in\A.
\end{equation}
Indeed, if $X_{H,1}$ and $X_{H,2}$ are two smooth vector fields satisfying the preceding relation, then the smooth vector $\widetilde{X}_{H}\coloneqq X_{H,1}-X_{H,2}$ satisfies
\begin{equation}
dF[m](\widetilde{X}_{H}(m))=0, \qquad \forall F\in\A,
\end{equation}
for all $m\in M$, which by \ref{item:wp_P2} implies that $\widetilde{X}_{H}\equiv 0$.
\end{remark}

\begin{remark}
For all $F,G,H\in\A$, we have that
\begin{align}
\comm{X_{F}}{X_{G}}H &= \pb{\pb{H}{G}}{F}-\pb{\pb{H}{F}}{G} \nonumber\\
&= \pb{H}{\pb{G}{F}} \nonumber\\
&= X_{\pb{G}{F}}H.
\end{align}
Hence, by \cref{rem:hvf_u}, $\comm{X_{F}}{X_{G}}=X_{\pb{G}{F}}$ for $F,G\in\A$. Additionally, the Leibnitz rule for $\pb{\cdot}{\cdot}$ implies the identity
\begin{equation}\label{eq:vf_L}
X_{FG}=FX_{G}+GX_{F}, \qquad \forall F,G\in\A.
\end{equation}
\end{remark}

There is also a notion of a weak symplectic manifold, which we have generalized to allow for complex-valued symplectic forms.

\begin{mydef}[Weak symplectic manifold]\label{weak_sym}
Let $M$ be a smooth locally convex manifold, and let $\mathcal{X}(M)$ denote the space of smooth vector fields on $M$. A \emph{weak symplectic manifold} is a pair $(M,\omega)$ consisting of a smooth manifold $M$ and a closed non-degenerate 2-form $\omega: TM\times TM\rightarrow\C$.
\end{mydef}

Given a weak symplectic manifold, we denote the Lie algebra of Hamiltonian vector fields on $M$ by
\begin{equation}
\mathrm{ham}(M,\omega) \coloneqq \{X\in\mathcal{X}(M) : \exists H\in C^{\infty}(M;\C) \enspace \text{s.t.} \enspace \omega(X, \cdot) = dH\}.
\end{equation}

With this definition  in hand, we see that weak symplectic manifolds canonically lead to weak Poisson manifolds. 
\begin{remark}[Weak symplectic $\Rightarrow$ weak Poisson]
\label{rem:sym_poi}
Let $(M,\omega)$ be a weak symplectic manifold. Let
\begin{equation}
\label{eq:sym_alg}
\A \coloneqq \{H\in C^{\infty}(M;\C) : \exists X_{H}\in\mathcal{X}(M) \enspace \text{s.t.} \enspace \omega( X_{H}, \cdot) =dH\},
\end{equation}
and let
\begin{equation}
\label{eq:ws_pb}
\pb{\cdot}{\cdot}:\A\times\A\rightarrow\A, \qquad \pb{F}{G} \coloneqq \omega(X_{F},X_{G}) = dF[X_{G}]=X_{G}F.
\end{equation}
The formula \eqref{eq:ws_pb} defines a Poisson bracket satisfying properties \ref{item:wp_P1} and \ref{item:wp_P3}. If we additionally have that for each $v\in T_{m}M$, 
\begin{equation}
\paren*{\omega(X(m),v)=0, \enspace \forall X\in\mathrm{ham}(M,\omega)} \Longrightarrow v=0,
\end{equation}
then property \ref{item:wp_P2} is also satisfied. Consequently, the triple $(M,\A,\pb{\cdot}{\cdot})$ is a weak Poisson manifold.
\end{remark}

We now turn to mappings between weak Poisson manifolds which preserve the Poisson structures. This leads to the notion of a Poisson map, alternatively Poisson morphism.

\begin{mydef}[Poisson map]\label{def:po_map}
Let $(M_{j},\A_{j},\pb{\cdot}{\cdot}_{j})$, for $j=1,2$, be weak Poisson manifolds. We say that a smooth map $\varphi:M_{1}\rightarrow M_{2}$ is a \emph{Poisson map}, or \emph{morphism of Poisson manifolds}, if $\varphi^{*}\A_{2}\subset \A_{1}$ and
\begin{equation}
\varphi^{*}\pb{F}{G}_{2}=\pb{\varphi^{*}F}{\varphi^{*}G}_{1}, \qquad \forall F,G\in\A_{2},
\end{equation}
where $\varphi^*$ denotes the pullback of $\varphi$.
\end{mydef}

\begin{remark}
In \cite{NST2014}, the authors define a Poisson morphism 
\[
\varphi: (M_{1},\A_{1},\pb{\cdot}{\cdot}_{1}) \rightarrow (M_{2},\A_{2},\pb{\cdot}{\cdot}_{2})
\]
with the requirement that $\varphi^{*}\A_{2}=\A_{1}$. We relax this requirement in our \cref{def:po_map}.
\end{remark}

We need several examples of weak Poisson/symplectic manifolds in this work. An example we discussed at length in \cite{MNPRS1_2019} is the Schwartz space $\Sc(\R^{k})$, as well as its bosonic counterpart $\Sc_{s}(\R^{k})$. We collect the main conclusions and refer the reader to \cite{MNPRS1_2019} for proofs.

We equip the space $\Sc(\R^{k})$ with a real pre-Hilbert inner product by defining
\begin{equation}
\ip{f}{g}_{\Re} \coloneqq 2\Re{ \int_{\R^{k}}d\ux_{k}\ol{f(\ux_{k})} g(\ux_{k})}, \qquad \forall f,g\in\Sc(\R^k).
\end{equation}
The operator $J: \Sc(\R^{k})\rightarrow\Sc(\R^{k})$ defined by $J(f) \coloneqq if$ defines an almost complex structure on $(\Sc(\R^{k}),\ip{\cdot}{\cdot}_{\Re})$, leading to the \emph{standard $L^{2}$ symplectic form}
\begin{equation}\label{l2_symp}
\omega_{L^{2}}(f,g) \coloneqq \ip{Jf}{g}_{\Re} = 2\Im{\int_{\R^{k}}d\ux_{k}\ol{f(\ux_{k})}g(\ux_{k})}, \qquad \forall f,g\in\Sc(\R^{k}).
\end{equation}

With these definitions in hand, we record the following well-known fact. 

\begin{prop}\label{prop:Sch_wsym}
$(\Sc(\R^{k}),\omega_{L^2})$ is a weak symplectic manifold.
\end{prop}

Now given a functional $F\in C^{\infty}(\Sc(\R^{k});\R)$, the G\^ateaux derivative of $F$ at the point $f\in\Sc(\R^{k})$ , denoted by $dF[f]$, defines an element of $\Sc'(\R^k)$. We consider the case when $dF[f]$ can be identified with a Schwartz function via the inner product $\ip{\cdot}{\cdot}_{\Re}$. 

\begin{lemma}[Uniqueness]\label{lem:grad_unq}
Let $F\in C^{\infty}(\Sc(\R^{k});\R)$ and $f\in\Sc(\R^{k})$. Suppose that there exist $g_{1},g_{2}\in\Sc(\R^{k})$ such that
\begin{equation}
\ip{g_{1}}{\delta f}_{\Re} = dF[f](\delta f) = \ip{g_{2}}{\delta f}_{\Re}, \qquad \forall \delta f\in\Sc(\R^{k}).
\end{equation}
Then $g_{1}=g_{2}$.
\end{lemma}

Letting $\A$ be the algebra defined in \eqref{eq:sym_alg} and $F\in\A$, we see from \cref{rem:hvf_u} and \cref{rem:sym_poi} that $X_F(f)$ is the unique element, hereafter denoted by $\grad_s F(f)$, satisfying
\begin{equation*}
dF[f](\delta f) = \omega_{L^2}(\grad_s F(f), \delta f), \qquad \forall \delta f\in \Sc(\R^k).
\end{equation*}
Consequently, we can define the real and symplectic gradients of functionals.

\begin{mydef}[Real/Symplectic $L^{2}$ gradient]\label{def:re_grad}
We define the \emph{real and symplectic $L^{2}$ gradient} of $F\in C^{\infty}(\Sc(\R^{k});\R)$ at the point $f\in\Sc(\R^{k})$, denoted by $\grad F(f)$ and $\grad_s F(f)$, respectively, to be the unique elements of $\Sc(\R^{k})$ (if they exist) such that
\begin{equation}
dF[f](\delta f) = \ip{\grad F(f)}{\delta f}_{\Re} = \omega_{L^2}(\grad_s F(f),\delta f), \qquad \forall \delta f\in\Sc(\R^{k}).
\end{equation}
We say that $F$ has a real, or respectively symplectic, $L^{2}$ gradient if $\grad F:\Sc(\R^{k}) \rightarrow\Sc(\R^{k})$, respectively $\grad_s F: \Sc(\R^k)\rightarrow \Sc(\R^k)$, is a smooth map.
\end{mydef}

\begin{remark}\label{schwartz_deriv}
Recalling that
\begin{equation*}
\omega_{L^2}(f,g) = \ip{Jf}{g}_{\Re},
\end{equation*}
we see that $\grad_s F(f) = -i\grad F(f)$.
\end{remark}

\cref{rem:sym_poi} implies that the symplectic form $\omega_{L^2}$ canonically induces a Poisson structure on $\Sc(\R^{k})$, a fact we record in the next proposition.

\begin{prop}
\label{schwartz_wpoiss}
Define a subset $A_{\Sc}\subset C^{\infty}(\Sc(\R^{k});\R)$ by
\begin{equation}\label{equ:Asc}
\begin{split}
\A_{\Sc} &\coloneqq \{H\in C^\infty(\Sc(\R^k);\R) :  \grad_s H\in C^\infty(\Sc(\R^k);\Sc(\R^k))\}.
\end{split}
\end{equation}
and define a bracket $\pb{\cdot}{\cdot}_{L^{2}}$ by
\begin{equation}\label{l2_bracket}
\pb{F}{G}_{L^{2}} \coloneqq \omega_{L^{2}}(\grad_{s}F,\grad_{s}G), \qquad \forall F,G\in\A_{\Sc}.
\end{equation}
Then $(\Sc(\R^{k}),\A_{\Sc}, \pb{\cdot}{\cdot}_{L^{2}})$ is a weak Poisson manifold.
\end{prop}

\begin{remark}[Variational derivatives]
\label{rem:pb_vd}
For functionals $F,G\in C^\infty(\Sc(\R^k);\R)$ having a special form discussed below, there is a computationally more convenient way to express their symplectic gradients and Poisson bracket in terms of \emph{variational derivatives}. Given a smooth functional $\tl{F}:\Sc(\R^k)^2\rightarrow\C$, we define the variational derivatives $\grad_1 \tl{F}$ and $\grad_{\bar{2}}\tl{F}$ by the property\footnote{Our notation for variational derivatives is nonstandard. In the calculus of variations literature, one typically finds $\frac{\d f}{\d\phi_1}$ and $\frac{\d f}{\d\ol{\phi_2}}$ instead of $\grad_1 f(\phi_1,\ol{\phi_2})$ and $\grad_{\bar{2}}(\phi_1,\ol{\phi_2})$, respectively. We prefer our notation as it emphasizes the nature of the variational derivatives as vector fields. The motivations for the seemingly odd use of the subscript $\bar{2}$, as opposed to just $2$, will become clear later in this subsection.}
\begin{equation}
\label{eq:vd_prop}
d\tl{F}[\phi_1,\ol{\phi_2}](\d\phi_1,\d\ol{\phi_2}) =\int_{\R^k}d\ux_k \paren*{\grad_1 \tl{F}(\phi_1,\ol{\phi_2})\d\phi_1 + \grad_{\bar{2}}\tl{F}(\phi_1,\ol{\phi_2})\d\ol{\phi_2}}(\ux_k), \ \forall(\phi_1,\ol{\phi_2}),(\d\phi_1,\d\ol{\phi_2})\in\Sc(\R^k)^2.
\end{equation}
The reader can verify that the variational derivatives, if they exist, are unique.

Let $F,G\in C^\infty(\Sc(\R^k);\R)$. Suppose that
\begin{equation}
F(\phi) = \tl{F}(\phi,\ol{\phi}), \qquad \tl{F}\in C^\infty(\Sc(\R^k)^2;\C),
\end{equation}
where $\tl{F}$ satisfies the conditions
\begin{equation}
\label{eq:tlF_cons}
\ol{\tl{F}(\phi_1,\ol{\phi_2})} = \tl{F}(\phi_2,\ol{\phi_1}), \qquad \grad_1 \tl{F}, \ \grad_{\bar{2}}\tl{F} \in C^\infty(\Sc(\R^k)^2; \Sc(\R^k)),
\end{equation}
and similarly for $G$ and $\tl{G}$. Then we claim that $F,G\in A_{\Sc}$ and their Poisson bracket $\pb{F}{G}_{L^2}$ may be rewritten as
\begin{equation}
\label{eq:L2_pb_vd}
\pb{F}{G}_{L^2}(\phi) = -i\int_{\R}dx\paren*{\grad_1 \tl{F}(\phi,\ol{\phi})\grad_{\bar{2}}\tl{G}(\phi,\ol{\phi}) - \grad_{\bar{2}}\tl{F}(\phi,\ol{\phi})\grad_1 \tl{G}(\phi,\ol{\phi})}(x).
\end{equation}
Indeed, observe that
\begin{align}
d\tl{F}[\phi_1,\ol{\phi_2}](\delta\phi_1,\delta\ol{\phi_2}) &=\lim_{\varepsilon\rightarrow 0} \frac{\tl{F}(\phi_1+\varepsilon\delta\phi_1,\ol{\phi_2}+\varepsilon\delta\ol{\phi_2}) - \tl{F}(\phi_1,\ol{\phi_2})}{\varepsilon} \nonumber\\
&=\lim_{\varepsilon\rightarrow 0} \frac{\ol{\tl{F}(\phi_2+\varepsilon\ol{\delta\ol{\phi_2}},\ol{\phi_1}+\varepsilon\ol{\delta\phi_1})-\tl{F}(\phi_2,\ol{\phi_1})}}{\varepsilon} \nonumber\\
&=\ol{d\tl{F}[\phi_2,\ol{\phi_1}](\ol{\delta\ol{\phi_2}},\ol{\delta\phi_1})} \nonumber\\
&=\int_{\R^k}d\ux_k \paren*{\ol{\grad_1 \tl{F}(\phi_2,\ol{\phi_1})}\delta\ol{\phi_2} + \ol{\grad_{\bar{2}}\tl{F}(\phi_2,\ol{\phi_1})}\delta\phi_1}(\ux_k),
\end{align}
where the ultimate equality follows by definition of the variational derivatives. Since
\begin{equation}
d\tl{F}[\phi_1,\ol{\phi_2}](\delta\phi_1,\delta\ol{\phi_2}) = \int_{\R^k}d\ux_k\paren*{\grad_1 \tl{F}(\phi_1,\ol{\phi_2})\delta\phi_1+\grad_{\bar{2}}\tl{F}(\phi_1,\ol{\phi_2})\delta\ol{\phi_2}}(\ux_k),
\end{equation}
we conclude by uniqueness of variational derivatives that
\begin{equation}
\label{eq:vd_invol}
\grad_1 \tl{F}(\phi_1,\ol{\phi_2}) = \ol{\grad_{\bar{2}}\tl{F}(\phi_2,\ol{\phi_1})}, \qquad \grad_{\bar{2}}\tl{F}(\phi_1,\ol{\phi_2}) = \ol{\grad_1 \tl{F}(\phi_2,\ol{\phi_1})}.
\end{equation}
Now recalling the definition of the symplectic gradient, we have that
\begin{align}
\omega_{L^2}(\grad_s F(\phi),\psi) &= dF[\phi](\psi) \nonumber\\
&=d\tl{F}[\phi,\ol{\phi}](\psi, \ol{\psi}) \nonumber\\
&=\int_{\R^k}d\ux_k \paren*{\grad_1 \tl{F}(\phi,\ol{\phi})\psi + \grad_{\bar{2}}\tl{F}(\phi,\ol{\phi})\ol{\psi}}(\ux_k) \nonumber\\
&=2\Re {\int_{\R^k}d\ux_k \grad_1 \tl{F}(\phi,\ol{\phi})(\ux_k) \psi(\ux_k)},
\end{align}
where the ultimate equality follows from the relations \eqref{eq:vd_invol}. By uniqueness of the symplectic gradient, we conclude that
\begin{equation}
\label{eq:sgrad_F}
\grad_s F(\phi) = -i\ol{\grad_1 \tl{F}(\phi,\ol{\phi})} = -i\grad_{\bar{2}} \tl{F}(\phi,\ol{\phi}) = \frac{1}{2}\paren*{\ol{i\grad_1 \tl{F}(\phi,\ol{\phi})} - i\grad_{\bar{2}}\tl{F}(\phi,\ol{\phi})}.
\end{equation}
Since the right-hand side of the preceding identity defines an element of $C^\infty(\Sc(\R^k);\Sc(\R^k))$, we obtain that $F\in\A_{\Sc}$. Now we can rewrite the Poisson bracket as
\begin{align}
\omega_{L^2}(\grad_s F(\phi),\grad_s G(\phi)) &= 2\Im{\int_{\R^k}d\ux_k \paren*{i\grad_1 \tl{F}(\phi,\ol{\phi}) \ol{i\grad_1 \tl{G}(\phi,\ol{\phi})}}(\ux_k)} \nonumber\\
&=-i\int_{\R^k}d\ux_k \paren*{\grad_1 \tl{F}(\phi,\ol{\phi})\ol{\grad_1 \tl{G}(\phi,\ol{\phi})} - \ol{\grad_1 \tl{F}(\phi,\ol{\phi})}\grad_1 \tl{G}(\phi,\ol{\phi})}(\ux_k) \nonumber\\
&=-i\int_{\R^k}d\ux_k \paren*{\grad_1 \tl{F}(\phi,\ol{\phi})\grad_{\bar{2}} \tl{G}(\phi,\ol{\phi}) - \grad_{\bar{2}} \tl{F}(\phi,\ol{\phi})\grad_1 \tl{G}(\phi,\ol{\phi})}(\ux_k),
\end{align}
where the ultimate equality follows from the relations \eqref{eq:vd_invol}.

In the sequel, all of the functionals we consider will satisfy the requirements \eqref{eq:tlF_cons}. Consequently, we will pass between the variational derivative formulation \eqref{eq:L2_pb_vd} and the symplectic gradient formulation of the Poisson bracket without comment.
\end{remark}

To motivate our next extension of the weak Poisson manifold $(\Sc(\R^k),\A_{\Sc},\pb{\cdot}{\cdot}_{L^2})$, we observe that we can identify a one-particle wave function $\phi$ with the pure state
\begin{equation*}
\ket*{\phi}\bra*{\phi}.
\end{equation*}
We can define a real topological vector space of pure states by considering the space of Schwartz functions taking values in the space of self-adjoint, off-diagonal $2\times 2$ complex matrices:
\begin{equation}
\begin{pmatrix} 0 & \phi \\ \ol{\phi} & 0\end{pmatrix}.
\end{equation}
The natural generalization of a pure state is a mixed state,
\begin{equation*}
\frac{1}{2}\paren*{\ket*{\phi_1}\bra*{\phi_2}+\ket*{\phi_2}\bra*{\phi_1}},
\end{equation*}
and we can define a real topological vector space of mixed states as follows: let $\mathcal{V}$ denote the real vector space of self-adjoint, off-diagonal $4\times 4$ matrices of the form 
\begin{equation} \label{equ:mixed}
\frac{1}{2}\mathrm{odiag}(a, \ol{b}, b, \ol{a}), \qquad a,b\in\C.
\end{equation}
We let $\Sc(\R^k;\mathcal{V})$ denote the space of Schwartz functions taking values in the space $\mathcal{V}$. Elements of $\Sc(\R^k;\mathcal{V})$ have the form
\begin{equation}
\gamma(\ux_k) = \frac{1}{2}\mathrm{odiag}(\phi_1(\ux_k),\ol{\phi_2}(\ux_k),\phi_2(\ux_k),\ol{\phi_1}(\ux_k)), \qquad \forall \ux_k\in\R^k, \enspace \phi_1,\phi_2\in\Sc(\R^k).
\end{equation}

We can define a real pre-Hilbert inner product on $\Sc(\R^k;\mathcal{V})$ by
\begin{equation}
\ip{\gamma_1}{\gamma_2}_{\Re,\mathcal{V}} \coloneqq 2\int_{\R^k}d\ux_k \tr_{\C^2\otimes \C^2}(\gamma_1(\ux_k)\gamma_{2,swap}(\ux_k)), \qquad \forall\gamma_1,\gamma_2\in\Sc(\R^k;\V),
\end{equation}
where $\tr_{\C^2\otimes\C^2}$ denotes the $4\times 4$ matrix trace and
\begin{equation}
\gamma_{2,swap} = \frac{1}{2}\mathrm{odiag}(\phi_2,\ol{\phi_1},\phi_1,\ol{\phi_2}), \qquad \gamma_2 = \frac{1}{2}\mathrm{odiag}(\phi_1,\ol{\phi_2},\phi_2,\ol{\phi_1}).
\end{equation}
The matrix left-multiplication operator
\begin{equation}
J:\Sc(\R^k;\mathcal{V}) \rightarrow \Sc(\R^k;\mathcal{V}), \qquad  J=\mathrm{diag}(i,-i,i,-i)
\end{equation}
defines an almost complex structure. We can then define a symplectic form $\omega_{L^2,\mathcal{V}}$ by
\begin{equation}\label{symp_v}
\omega_{L^2,\mathcal{V}}(\gamma_1,\gamma_2)\coloneqq \ip{J\gamma_1}{\gamma_{2,swap}}_{\Re,\mathcal{V}}.
\end{equation}
Analogous to \cref{schwartz_wpoiss}, we have that $(\Sc(\R^k; \mathcal{V}),\omega_{L^2,\mathcal{V}})$ is a weak symplectic manifold. Moreover, the obvious map
\begin{equation}
\iota_{\mathfrak{pm}}: \Sc(\R^k) \rightarrow \Sc(\R^k;\mathcal{V}), \qquad \phi \mapsto \frac{1}{2}\mathrm{odiag}(\phi,\ol{\phi},\phi,\ol{\phi})
\end{equation}
is a symplectomorphism. Additionally, if we denote the symplectic gradient with respect to the form $\omega_{L^2,\mathcal{V}}$ by $\grad_{s,\mathcal{V}}$, then one can show that if we define
\begin{equation}\label{alg_v}
\A_{\Sc,\mathcal{V}} \coloneqq \{F\in C^\infty(\Sc(\R^k;\mathcal{V});\R) : \grad_{s,\mathcal{V}} F\in C^\infty(\Sc(\R^k;\mathcal{V}),\Sc(\R^k;\mathcal{V}))\},
\end{equation}
and let $\pb{\cdot}{\cdot}_{L^2,\mathcal{V}}$ be the Poisson bracket canonically induced by the form $\omega_{L^2,\mathcal{V}}$, then the triple
\begin{equation}
(\Sc(\R^k;\mathcal{V}),\A_{\Sc,\mathcal{V}},\pb{\cdot}{\cdot}_{L^2,\mathcal{V}})
\end{equation}
is a weak Poisson manifold. We summarize the preceding discussion with the following proposition.

\begin{prop}
\label{prop:Schw_WP_V}
$(\Sc(\R^k;\mathcal{V}),\omega_{L^2,\mathcal{V}})$ is a weak symplectic manifold, and $(\Sc(\R^k;\mathcal{V}),\mathcal{A}_{\Sc,\mathcal{V}}, \pb{\cdot}{\cdot}_{L^2,\mathcal{V}})$ is a weak Poisson manifold, where
\begin{equation}
\label{eq:pb_L2_V}
\pb{F}{G}_{L^2,\V}(\gamma) \coloneqq \omega_{L^2,\V}(\grad_{s,\V} F(\gamma), \grad_{s,\V} G(\gamma)).
\end{equation}
Furthermore, the map $\iota_{\mathfrak{pm}}$ is a symplectomorphism; i.e., it is a smooth map such that
\begin{equation}
\iota_{\mathfrak{pm}}^*\omega_{L^2,\mathcal{V}} = \omega_{L^2},
\end{equation}
where $\iota_{\mathfrak{pm}}^*$ denotes the pullback of $\iota_{\mathfrak{pm}}$, so that
\begin{equation}
\iota_{\mathfrak{pm}}: (\Sc(\R^k),\mathcal{A}_{\Sc},\pb{\cdot}{\cdot}_{L^2}) \rightarrow (\Sc(\R^k;\mathcal{V}),\A_{\Sc,\mathcal{V}},\pb{\cdot}{\cdot}_{L^2,\mathcal{V}})
\end{equation}
is a Poisson morphism.
\end{prop}

\begin{remark}
\label{rem:pb_V_vd}
\cref{rem:pb_vd} carries over to the setting of $\Sc(\R^k;\mathcal{V})$. More precisely, suppose $F\in C^\infty(\Sc(\R^k;\mathcal{V});\R)$ is such that
\begin{equation}
F(\gamma) = \tl{F}(\phi_1,\ol{\phi_2},\phi_2,\ol{\phi_1}), \qquad \gamma =\frac{1}{2}\mathrm{odiag}(\phi_1,\ol{\phi_2},\phi_2,\ol{\phi_1})\in\Sc(\R^k;\mathcal{V}),
\end{equation}
where $\tl{F}\in C^\infty(\Sc(\R^k)^4;\C)$, is such that
\begin{equation}
\grad_1 \tl{F}, \,\,\,  \grad_{\bar{2}}\tl{F}, \,\,\,\grad_{2}\tl{F}, \,\,\,\grad_{\bar{1}}\tl{F} \in C^\infty(\Sc(\R^k)^4;\Sc(\R^k)),
\end{equation}
where the four variational derivatives are uniquely defined by
\begin{equation}
\label{eq:vd_prop_V}
\begin{split}
&d\tl{F}[\phi_1,\phi_{\bar{2}},\phi_2,\phi_{\bar{1}}](\delta\phi_1,\delta\phi_{\bar{2}},\d\phi_2,\d\phi_{\bar{1}}) \\
&=\int_{\R^k}d\ux_k \paren*{\paren*{\grad_1 \tl{F}\d\phi_1 + \grad_{\bar{2}}\tl{F}\d\phi_{\bar{2}} + \grad_2 \tl{F} \d\phi_2 + \grad_{\bar{1}}\tl{F} \d\phi_{\bar{1}}}(\phi_1,\phi_{\bar{2}},\phi_2,\phi_{\bar{1}})}(\ux_k),
\end{split}
\end{equation}
and $\tl{F}$ has the involution property
\begin{equation}
\tl{F}(\phi_1,\phi_{\bar{2}},\phi_2,\phi_{\bar{1}}) = \ol{\tl{F}(\ol{\phi_{\bar{1}}}, \ol{\phi_2}, \ol{\phi_{\bar{2}}},\ol{\phi_1})}.
\end{equation}
Then $F\in\A_{\Sc,\mathcal{V}}$. Additionally, if $F,G$ are two such functionals, then their Poisson bracket may be rewritten as
\begin{equation}
\label{eq:L2V_pb_vd}
\begin{split}
\pb{F}{G}_{L^2,\mathcal{V}}(\gamma) &= -i\int_{\R^k}d\ux_k \paren*{\grad_1 \tl{F}(\gamma)\grad_{\bar{2}} \tl{G}(\gamma)- \grad_{\bar{2}}\tl{F}(\gamma)\grad_1 \tl{G}(\gamma}(\ux_k) \\
&\phantom{=} -i\int_{\R^k}d\ux_k \paren*{\grad_{2} \tl{F}(\gamma)\grad_{\bar{1}} \tl{G}(\gamma) - \grad_{\bar{1}}\tl{F}(\gamma)\grad_2 \tl{G}(\gamma)}(\ux_k),
\end{split}
\end{equation}
where we identify $\gamma$ with the 4-tuple $(\phi_1,\ol{\phi_2},\phi_2,\ol{\phi_1})$ for the sake of more compact notation.

In the sequel, all the functionals on $\Sc(\R^k;\mathcal{V})$ we consider satisfy the conditions of the remark. Consequently, we will pass between the variational derivative and symplectic gradient formulations for the Poisson bracket without comment.
\end{remark}

Lastly, we make heavy use of a ``complexified'' version of the weak symplectic manifold $(\Sc(\R^k),\omega_{L^2})$. More precisely, consider the cartesian product $\Sc(\R^k)^2$ and define a complex-valued map
\begin{equation}
\label{eq:om_L2_sq}
\omega_{L^2,\C}(\ul{f}_2,\ul{g}_2) \coloneqq \int_{\R^k}d\ux_k \tr_{\C^2}(J_{\C}\ul{f}_2\ul{g}_2) (\ux_k),
\end{equation}
where
\begin{equation}
\ul{f}_2 = \begin{pmatrix} 0 & f_1\\f_2&0\end{pmatrix}, \ \ul{g}_2=\begin{pmatrix}0&g_1\\g_2&0\end{pmatrix} \in \Sc(\R^k)^2,
\end{equation}
$\tr_{\C^2}$ denotes the $2\times 2$ matrix trace, and $J_\C$ is the left-matrix multiplication operator $\mathrm{diag}(i,-i)$. Here, we identify a Schwartz function taking values in the space of off-diagonal $2\times 2$ matrices with an element of $\Sc(\R^k)^2$ in the obvious manner.

\begin{remark}
\label{rem:sym_form_tr}
Note that if $\ul{f}_2 = \mathrm{odiag}(f,\ol{f})$ and $\ul{g}_2=\mathrm{odiag}(g,\ol{g})$, for $f,g\in\Sc(\R^k)$, then
\begin{equation}
\omega_{L^2,\C}(\ul{f}_2,\ul{g}_2) = i\int_{\R^k}d\ux_k\paren*{f\ol{g} - \ol{f}g}(\ux_k) = 2\Im{\int_{\R^k}d\ux_k\ol{f(\ux_k)}g(\ux_k)} = \omega_{L^2}(f,g).
\end{equation}
\end{remark}

\begin{prop}\label{schwartz2_wpoiss}
Define a subset $A_{\Sc,\C}\subset C^{\infty}(\Sc(\R^{k})^2;\C)$ by
\begin{equation}\label{eq:Asc}
\A_{\Sc,\C} \coloneqq \bigl \{ H \in C^\infty(\Sc(\R^k);\C) \,:\, \grad_{s,\C}H \in C^{\infty}(\Sc(\R)^2;\Sc(\R)^2) \bigr\} ,
\end{equation}
and define a bracket $\pb{\cdot}{\cdot}_{L^2,\C}$ by
\begin{equation}
\label{eq:pb_L2_C}
\pb{F}{G}_{L^2,\C} \coloneqq \omega_{L^2,\C}(\grad_{s,\C}F,\grad_{s,\C}G).
\end{equation}
Then $(\Sc(\R^{k})^2,\A_{\Sc,\C}, \pb{\cdot}{\cdot}_{L^2,\C})$ is a weak Poisson manifold.
\end{prop}

\begin{remark}
\label{rem:pb_vd_C}
As before, if $F,G\in C^\infty(\Sc(\R^k)^2;\C)$ satisfy the condition \eqref{eq:tlF_cons}, then $F,G\in\A_{\Sc,\C}$ and
\begin{equation}
\pb{F}{G}_{L^2,\C}(\phi_1,\ol{\phi_2}) = -i\int_{\R^k}d\ux_k \paren*{\grad_1 F(\phi_1,\ol{\phi_2})\grad_{\bar{2}}G(\phi_1,\ol{\phi_2}) - \grad_{\bar{2}}F(\phi_1,\ol{\phi_2})\grad_1 G(\phi_1,\ol{\phi_2})}(\ux_k).
\end{equation}
\end{remark}

\begin{remark}
All the Schwartz space examples given in this subsection have their $2L$-periodic analogues, where $\Sc(\R^k)$ is replaced by $C^\infty(\T_L^k)$. We will need the periodic examples in \cref{app:nls_pc}.
\end{remark}

\subsection{Some Lie algebra facts}\label{ssec:pre_LA_facts}
We close \cref{sec:pre} by collecting some facts about Lie algebras for easy referencing. Following our presentation in \cite[Section 4.2]{MNPRS1_2019}, we outline a canonical construction of a Poisson structure on the dual of a Lie algebra, which is known as a \emph{Lie-Poisson structure} following the terminology of Marsden and Weinstein \cite{MW1983}. We refer the reader to \cite{MR2013, MMW1984} for a more thorough discussion. 

\begin{mydef}[Lie algebra]\label{def:la}
A \emph{Lie algebra} is a locally convex space $\g$ over the field $\mathbb{F} \in \{\R,\C\}$ together with a separately continuous binary operation $\brak{\cdot,\cdot}:\g\times \g \rightarrow \g$ called the \emph{Lie bracket}, which satisfies the following properties:
\begin{enumerate}[(L1)]
\item\label{item:LA_1}
$[\cdot,\cdot]$ is bilinear.
\item\label{item:LA_2}
$[x,x]=0$ for all $x\in \g$.
\item\label{item:LA_3}
$\brak{\cdot,\cdot}$ satisfies the \emph{Jacobi identity}
\begin{equation}
\brak*{x,\brak*{y,z}} + \brak*{z,\brak*{x,y}} + \brak*{y,\brak*{z,x}}=0
\end{equation}
for all $x,y,z\in \g$.
\end{enumerate}
\end{mydef}

\begin{remark}
Usually (see, for instance, \cite{Omori1979}), a Lie bracket is required to be continuous, as opposed to separately continuous. We drop this requirement in this work, due to functional analytic difficulties stemming from the separate continuity of the distributional pairing.
\end{remark}

\begin{mydef}[Nondegenerate pairings]
Let $V$ and $W$ be topological vector spaces over the field $\mathbb{F}$, and let 
\[
\ip{\cdot}: V\times W\rightarrow \mathbb{F}
\]
be a bilinear pairing between $V$ and $W$. We say that the pairing is \emph{$V$-nondegenerate} (resp., \emph{$W$-nondegenerate}) if the map $V\rightarrow W^{*}, x\mapsto \ip{x}{\cdot}$ (resp., $W\rightarrow V^{*}, y\mapsto \ip{\cdot}{y}$) is an isomorphism. If the pairing is both $V$- and $W$-nondegenerate, then we say that the pairing is \emph{nondegenerate}.
\end{mydef}

\begin{mydef}[dual space $\g^{*}$]
Let $(\g,\brak{\cdot,\cdot})$ be a Lie algebra. We say that a topological vector $\g^{*}$ is a \emph{dual space} to $\g$ if there exists a pairing $\ip{\cdot}_{\g-\g^*}: \g \times \g^{*}\rightarrow \mathbb{F}$ which is nondegenerate.
\end{mydef}

\begin{ex}
If $\g$ is a reflexive Fr\'{e}chet space, for instance the Schwartz space $\mathcal{S}(\mathbb{R}^{d})$, then taking $\g^{*}$ to be the topological dual of $\g$ equipped with the strong dual topology, the standard duality pairing
\begin{equation*}
\g\times \g^{*} \rightarrow \mathbb{F}: \ip{x}{\varphi}_{\g-\g^*} = \varphi(x)
\end{equation*}
is nondegenerate.
\end{ex}

\begin{lemma}[Existence of functional derivatives]
Let $\g$ be a Lie algebra, and let $\g^{*}$ be dual to $\g$ with respect to the nondegenerate pairing $\ip{\cdot}{\cdot}_{\g-\g^*}$. For any functional $F\in C^{1}(\g^{*};\mathbb{F})$, there exists a unique element $\frac{\delta F}{\delta \mu}\in \g$ such that
\begin{equation}
\ip{\frac{\delta F}{\delta \mu}}{\delta\mu}_{\g-\g^*} = dF[\mu](\delta \mu), \qquad \mu, \delta\mu \in \g^{*}.
\end{equation}
\end{lemma}

\begin{prop}[Lie-Poisson structure]\label{prop:LP_rev}
Let $(\g,\comm{\cdot}{\cdot}_{\g})$ be a Lie algebra, such that the Lie bracket is continuous, and let $\g^{*}$ be dual to $\g$ with respect to the nondegenerate pairing $\ip{\cdot}_{\g-\g^*}$. Define the \emph{Lie-Poisson bracket} 
\begin{equation}
\pb{\cdot}{\cdot} : C^{\infty}(\g^{*};\mathbb{F})\times C^{\infty}(\g^{*};\mathbb{F}) \rightarrow C^{\infty}(\g^{*};\mathbb{F})
\end{equation}
by
\begin{equation}
\pb{F}{G}(\mu) \coloneqq \ip{\comm{\frac{\delta F}{\delta\mu}}{\frac{\delta G}{\delta\mu}}_{\g}}{\mu}_{\g-\g^*}, \qquad \mu \in \g^{*}.
\end{equation}
Then $(C^{\infty}(\g^{*};\mathbb{F}), \pb{\cdot}{\cdot})$ is a Lie algebra.
\end{prop}

\section{The construction: defining the $\W_{n}$}\label{sec:con_GP_W}
We now define the operators $\W_n$ giving rise to the Hamiltonian functionals $\H_n$. As detailed in \cref{sec:mr_bp}, in order to construct the operators $\W_n$, we proceed incrementally.

\subsection{Step 1: Preliminary definition of operators}\label{ssec:con_GP_W_S1}
Let
\begin{equation}\label{eq:Wn_base}
\widetilde{\W}_{1} = (\widetilde{\W}_{1}^{(k)})_{k\in\N} \in  \bigoplus_{k=1}^{\infty}\mathcal{L}_{gmp}(\mathcal{S}(\R^{k}), \mathcal{S}'(\R^{k})), \qquad \widetilde{\W}_{1} \coloneqq \E_{1},
\end{equation}
where we recall that
\begin{equation}\label{defn_ei}
\E_{j}=(\E_{j}^{(k)})_{k \in\N} \in \bigoplus_{k=1}^{\infty}\mathcal{L}_{gmp}(\mathcal{S}(\R^{k}), \mathcal{S}'(\R^{k})), \qquad \E_{j}^{(k)} \coloneqq Id_{k}\,\delta_{jk},
\end{equation}
where $Id_{k}$ is the identity operator in $\L(\Sc(\R^{k}),\Sc'(\R^{k}))$ and $\delta_{jk}$ is the Kronecker delta function. We regard $\E_{j}$ as the $j^{th}$ coordinate element of $\bigoplus_{k=1}^{\infty}\mathcal{L}(\mathcal{S}(\R^{k}), \mathcal{S}'(\R^{k}))$. It is clear that these operators satisfy the good mapping property.

We would like to recursively define
\begin{equation}
\widetilde{\W}_{n+1}=(\widetilde{\W}_{n+1}^{(k)})_{k\in\N} \in \bigoplus_{k=1}^{\infty}\mathcal{L}_{gmp}(\mathcal{S}(\R^{k}), \mathcal{S}'(\R^{k}))
\end{equation}
by the formula
\begin{equation}\label{eq:Wn_recur}
\widetilde{\W}_{n+1}^{(k)} \coloneqq -i\p_{x_{1}}\widetilde{\W}_{n}^{(k)}+ \kappa\sum_{m=1}^{n-1}\sum_{\ell,j\geq 1;\ell+j=k}\delta(X_{1}-X_{\ell+1}) \paren*{\widetilde{\W}_{m}^{(\ell)} \otimes \widetilde{\W}_{n-m}^{(j)}}, \qquad k\in\N,
\end{equation}
where we regard the multiplier operator $-i\p_{x_{1}}$ as a $k$-particle operator by tensoring with the identity in the $X_{2},\ldots,X_{k}$ coordinates. Similarly, we regard the multiplication $\delta(X_{1}-X_{\ell+1})$ as $k$-particle operator simply by tensoring with the identity in the $X_{2},\ldots,X_{\ell},X_{\ell+2},\ldots,X_{k}$ coordinates. 

Our aim is then two-fold. First, we need to make sense of the definition \cref{eq:Wn_recur}. At first glance, the right-hand side of \cref{eq:Wn_recur} is purely formal, since for $n\geq 4$, the sum will contain products of $\delta$ functions. However, as we will prove in the next lemma, the operators in \cref{eq:Wn_recur} are well-defined elements of $\mathcal{L}_{gmp}(\mathcal{S}(\R^{k}), \mathcal{S}'(\R^{k}))$. Intuitively, this is because the products in \eqref{eq:Wn_recur} never contain delta functions with identical arguments, such as $\delta^2(X_1-X_2)$. Subsequently, we will show that all but finitely many terms in the recursion are non-zero, which justifies our use of the direct sum notation. Thus, we are led to \cref{prop:W_S1}, the statement of which we recall below.

\WS*

We begin the proof of \cref{prop:W_S1} with establishing the recursion \eqref{eq:Wn_recur}.

\begin{lemma}[Rigorous recursion]\label{lem:Wn_wd}
For every $k, n \in \N$, the distribution-valued operator $\wt{\W}_{n}^{(k)}$ is an element of $\L_{gmp}(\Sc(\R^{k}),\Sc'(\R^{k}))$ and satisfies the following:
\begin{enumerate}[(R1)]
\item\label{item:Wn_wd_1}
There exists a finite subset $\mathsf{A}_{n}^{(k)}\subset \N_0^{k}$ of multi-indices such that
\begin{equation}\label{w_rep}
\widetilde{\W}_{n}^{(k)}f^{(k)} = \sum_{\ul{\alpha}_{k}\in \mathsf{A}_{n}^{(k)}} u_{\ul{\alpha}_{k},n}\p_{\ux_{k}}^{\ul{\alpha}_{k}}f^{(k)}, \qquad \forall f^{(k)}\in\Sc(\R^{k}),
\end{equation}
where $u_{\ul{\alpha}_{k},n}\in\Sc'(\R^{k})$.  
\item\label{item:Wn_wd_2} For every $\ul{\alpha}_{k}\in \mathsf{A}_{n}^{(k)}$, either
\begin{description}
\item[Case 1]\label{item:wf_c1}
$\WF(u_{\ul{\alpha}_{k},n}) = \emptyset$, or
\item[Case 2]\label{item:wf_c2}
$\WF(u_{\ul{\alpha}_{k},n})\neq \emptyset$ and satisfies the \emph{non-vanishing pair property}:
\begin{equation}
(\ux_{k},\uxi_{k}) \in \WF(u_{\ul{\alpha}_{k},n}) \Longrightarrow \exists \ell,j \in\N_{\leq k} \enspace \text{s.t.} \enspace \text{$\ell<j$ and  both $\xi_{\ell}\neq 0$ and $\xi_{j}\neq 0$}.
\end{equation}
\end{description}
\end{enumerate}
\end{lemma}

\begin{remark}
In other words, \ref{item:Wn_wd_1} means that $\widetilde{\W}_{n}^{(k)}$ can be written as a linear combination of terms, where each term consists of a differential operator left-composed with a distributional multiplication operator. The motivation for the non-vanishing pair property is to exploit the fact that the products of delta functions in \cref{eq:Wn_recur} do not have the same arguments.
\end{remark}

\begin{proof}[Proof of \cref{lem:Wn_wd}]
We prove the assertion by strong induction on $n \geq 1$.  The base case, namely that the claims hold for $n=1$, is clear. Next, let $n \geq 1$ and suppose that for every $k \in \N$, we have that
\begin{equation}
\widetilde{\W}_{1}^{(k)}, \ldots,\widetilde{\W}_{n}^{(k)}\in \L_{gmp}(\Sc(\R^{k}),\Sc'(\R^{k}))
\end{equation}
are defined according to \cref{eq:Wn_base} and \cref{eq:Wn_recur} and satisfy the properties \ref{item:Wn_wd_1} and \ref{item:Wn_wd_2}. We will show that for any $k \in \N$, the observable $\widetilde{\W}_{n+1}^{(k)}$ is a well-defined element of $\L_{gmp}(\Sc(\R^{k}),\Sc'(\R^{k}))$ and satisfies the properties \ref{item:Wn_wd_1} and \ref{item:Wn_wd_2}.  We organize our argument into several steps:

\medskip
\noindent \textbf{Step I:} We first prove \ref{item:Wn_wd_1}. If $\mathsf{A}_{n}^{(k)}\subset \N_{0}^{k}$ is a finite subset of multi-indices such that
\begin{equation}
\widetilde{\W}_{n}^{(k)}f^{(k)} = \sum_{\ul{\alpha}_{k}\in\mathsf{A}_{n}^{(k)}} u_{\ul{\alpha}_{k},n}\p_{\ux_{k}}^{\ul{\alpha}_{k}}f^{(k)}, \qquad \forall f^{(k)}\in\Sc(\R^{k}),
\end{equation}
where $u_{\ul{\alpha}_{k},n}\in\Sc'(\R^{k})$, then by the product rule,
\begin{equation}
(-i\p_{x_{1}})\widetilde{\W}_{n}^{(k)}f^{(k)} = \sum_{\ul{\alpha}_{k}\in\mathsf{A}_{n}^{(k)}} \paren*{(-i\p_{x_{1}}u_{\ul{\alpha}_{k},n})\p_{\ux_{k}}^{\ul{\alpha}_{k}}f^{(k)} -i u_{\ul{\alpha}_{k},n}\p_{x_{1}}\p_{\ux_{k}}^{\ul{\alpha}_{k}}f^{(k)}}, \qquad \forall f^{(k)}\in\Sc(\R^{k}).
\end{equation}
Let $\mathsf{A}_{m}^{(\ell)}$ and $\mathsf{A}_{n-m}^{(j)}$ be finite subsets of $\N_{0}^{\ell}$ and $\N_{0}^{j}$, respectively, such that
\begin{align}
\widetilde{\W}_{m}^{(\ell)}f^{(\ell)} &= \sum_{\ul{\alpha}_{\ell}\in\mathsf{A}_{m}^{(\ell)}} u_{\ul{\alpha}_{\ell},m}\p_{\ux_{\ell}}^{\ul{\alpha}_{\ell}}f^{(\ell)}, \qquad \forall f^{(\ell)}\in\Sc(\R^{\ell}) \\
\widetilde{\W}_{n-m}^{(j)}f^{(j)} &= \sum_{\ul{\alpha}_{j}\in\mathsf{A}_{n-m}^{(j)}} u_{\ul{\alpha}_{j},n-m}\p_{\ux_{j}}^{\ul{\alpha}_{j}}f^{(j)}, \qquad \forall f^{(j)}\in\Sc(\R^{j}),
\end{align}
where $u_{\ul{\alpha}_{\ell},m}\in\Sc'(\R^{\ell})$ and $u_{\ul{\alpha}_{j},n-m}\in\Sc'(\R^{j})$. Define the set
\begin{equation}\label{anm}
\mathsf{A}_{n,m}^{(k)} \coloneqq \mathsf{A}_{m}^{(\ell)}\times\mathsf{A}_{n-m}^{(j)} \subseteq \N_{0}^{\ell}\times \N_{0}^{j}
\end{equation}
so that
\begin{equation}
\label{eq:Wmn_ot_rep}
\paren*{\widetilde{\W}_{m}^{(\ell)} \otimes \widetilde{\W}_{n-m}^{(j)}}f^{(k)} = \sum_{(\ul{\alpha}_{\ell},\ul{\alpha}_{j})\in\mathsf{A}_{n,m}^{(k)}} \paren*{u_{\ul{\alpha}_{\ell},m}\otimes u_{\ul{\alpha}_{j},n-m}} \paren*{\p_{\ux_{\ell}}^{\ul{\alpha}_{\ell}}\otimes \p_{\ux_{j}}^{\ul{\alpha}_{j}}}f^{(k)}, \qquad \forall f^{(k)} \in\Sc(\R^{k}).
\end{equation}
Hence, to prove the claim, it suffices to show that
\begin{equation}\label{delta_def}
\delta(X_{1}-X_{\ell+1})\paren*{\widetilde{\W}_{m}^{(\ell)}\otimes \widetilde{\W}_{n-m}^{(j)}}
\end{equation}
is well-defined in $\L(\Sc(\R^{k}), \Sc'(\R^{k}))$, and that for all $f^{(k)}\in\Sc(\R^{k})$, \eqref{delta_def} admits the representation
\begin{equation}\label{delta_rep}
\paren*{\delta(X_{1}-X_{\ell+1})\paren*{\widetilde{\W}_{m}^{(\ell)} \otimes \widetilde{\W}_{n-m}^{(j)}}}f^{(k)} = \hspace{-5mm} \sum_{(\ul{\alpha}_{\ell},\ul{\alpha}_{j}) \in\mathsf{A}_{n,m}^{(k)}} \delta(x_{1}-x_{\ell+1})\paren*{u_{\ul{\alpha}_{\ell},m}\otimes u_{\ul{\alpha}_{j},n-m}}\paren*{\p_{\ul{x}_{\ell}}^{\ul{\alpha}_{\ell}} \otimes \p_{\ul{x}_{j}}^{\ul{\alpha}_{j}}}f^{(k)},
\end{equation}
where $\delta(x_{1}-x_{\ell+1})(u_{\ul{\alpha}_{\ell},m}\otimes u_{\ul{\alpha}_{j},n-m})$ is well-defined in $\Sc'(\R^{k})$. We will do this in two steps:

\begin{itemize}
\item First, we will show that \eqref{delta_def} admits the representation \eqref{delta_rep} for all $f^{(k)}\in\Sc(\R^{k})$, and that $\delta(x_{1}-x_{\ell+1})(u_{\ul{\alpha}_{\ell},m}\otimes u_{\ul{\alpha}_{j},n-m})\in\mathcal{D}'(\R^{k})$ in the H\"ormander product sense of \cref{prop:H_crit}.
\item Second, we will show that the products are, in fact, tempered distributions.
\end{itemize}

To show that the product of distributions
\begin{equation}
\delta(x_{1}-x_{\ell+1})\paren*{\widetilde{\W}_{m}^{(\ell)}\otimes \widetilde{\W}_{n-m}^{(j)}}(f^{(k)})
\end{equation}
is well-defined in $\D'(\R^{k})$ for every $f^{(k)}\in \Sc(\R^{k})$, it suffices by H\"{o}rmander's criterion (\cref{prop:H_crit}) to show that
\begin{equation}
(\ux_{k},\uxi_{k})\in \WF(\delta(x_{1}-x_{\ell+1})) \Longrightarrow (\ux_{k},-\uxi_{k})\notin \WF\paren*{\paren*{\widetilde{\W}_{m}^{(\ell)}\otimes \widetilde{\W}_{n-m}^{(j)}}f^{(k)}}.
\end{equation}
By \cref{lem:wf_del}, which computes the wave front set of $\delta(x_{1}-x_{\ell+1})$, we need to show that if $\xi_{1}\neq 0$, then
\begin{equation}\label{eq:wf_goal}
((x_{1},\ux_{2;\ell},x_{1},\ux_{\ell+2;k}), (\xi_{1},\ul{0}_{2;\ell},-\xi_{1},\ul{0}_{\ell+2;k})) \notin \WF\paren*{\paren*{\widetilde{\W}_{m}^{(\ell)}\otimes \widetilde{\W}_{n-m}^{(j)}}f^{(k)}}.
\end{equation}
Since for any $(\ul{\alpha}_{\ell},\ul{\alpha}_{j})\in\mathsf{A}_{n,m}^{(k)}$ and for any $g^{(k)}\in\Sc(\R^{k})$, we have the inclusion
\begin{equation}
\WF\paren*{\paren*{u_{\ul{\alpha}_{\ell},m}\otimes u_{\ul{\alpha}_{j},n-m}}g^{(k)}} \subset \WF\paren*{u_{\ul{\alpha}_{\ell},m}\otimes u_{\ul{\alpha}_{j},n-m}},
\end{equation}
by \cref{prop:wf_prop}\ref{inclusion}, it follows from \cref{prop:wf_prop}\ref{item:wf_prop_sum} and \eqref{eq:Wmn_ot_rep} that
\begin{equation}
\WF\paren*{\paren*{\widetilde{\W}_{m}^{(\ell)} \otimes \widetilde{\W}_{n-m}^{(j)}}f^{(k)}} \subset \bigcup_{(\ul{\alpha}_{\ell},\ul{\alpha}_{j})\in \mathsf{A}_{n,m}^{(k)}} \WF\paren*{u_{\ul{\alpha}_{\ell},m}\otimes u_{\ul{\alpha}_{j},n-m}}, \qquad \forall f^{(k)}\in\Sc(\R^{k}).
\end{equation}
Now by \cref{prop:wf_prop}\ref{item:wf_prop_tp}, we have that
\begin{equation}\label{eq:wf_set_ub}
\begin{split}
\WF\paren*{u_{\ul{\alpha}_{\ell},m}\otimes u_{\ul{\alpha}_{j},n-m}} &\subset \paren*{\WF\paren*{u_{\ul{\alpha}_{\ell},m}} \times \WF\paren*{u_{\ul{\alpha}_{j},n-m}}} \\
&\phantom{=} \cup \paren*{\supp\paren*{u_{\ul{\alpha}_{\ell},m}}\times\{\ul{0}_\ell\}} \times \WF\paren*{u_{\ul{\alpha}_{j},n-m}} \\
&\phantom{=} \cup \WF\paren*{u_{\ul{\alpha}_{\ell},m}} \times \paren*{\supp\paren*{u_{\ul{\alpha}_{j},n-m}}\times\{\ul{0}_j\}}.
\end{split}
\end{equation}
Note that we abuse notation with the cartesian products on the right-hand side of the preceding inclusion in the following sense: we denote an element of $\WF(u_{\ul{\alpha}_{\ell},m}) \times \WF(u_{\ul{\alpha}_{j},n-m})$ by
\begin{equation}
(\ux_{\ell},\ux_{\ell+1;k},\uxi_{\ell},\uxi_{\ell+1;k}), 
\end{equation}
where
\[
(\ux_{\ell},\uxi_{\ell}) \in \WF(u_{\ul{\alpha}_{\ell},m}), \quad (\ux_{\ell+1;k},\uxi_{\ell+1;k}) \in \WF(u_{\ul{\alpha}_{j},n-m})
\]
and similarly for elements of $(\supp(u_{\ul{\alpha}_{\ell},m})\times \{\ul{0}_\ell\})\times\WF(u_{\ul\alpha_{j},n-m})$ and $\WF(u_{\ul\alpha_{\ell},m})\times (\supp(u_{\ul\alpha_{j},n-m})\times\{\ul{0}_j\})$. We now consider three cases based on the values of the sets $\WF(u_{\ul{\alpha}_{\ell},m})$ and $\WF(u_{\ul{\alpha}_{j},n-m})$.
\begin{enumerate}[(i)]
\item 
Suppose that $\WF\paren*{u_{\ul{\alpha}_{\ell},m}}$ and $\WF\paren*{u_{\ul{\alpha}_{j},n-m}}$ are both empty. Then it follows readily from \cref{eq:wf_set_ub} that
\begin{equation}
\WF\paren*{u_{\ul{\alpha}_{\ell},m}\otimes u_{\ul{\alpha}_{j},n-m}} = \emptyset,
\end{equation}
and so \cref{eq:wf_goal} is satisfied.

\item 
Without loss of generality, suppose that $\WF\paren*{u_{\ul{\alpha}_{j},n-m}}=\emptyset$ and that $\WF\paren*{u_{\ul{\alpha}_{\ell},m}}\neq\emptyset$ and satisfies the non-vanishing pair property. Then by \cref{eq:wf_set_ub}, we have
\begin{equation}
\WF\paren*{u_{\ul{\alpha}_{\ell},m} \otimes u_{\ul{\alpha}_{j},n-m}} \subset \WF\paren*{u_{\ul{\alpha}_{\ell},m}} \times \paren*{\supp\paren*{u_{\ul{\alpha}_{j},n-m}}\times\{\ul{0}_j\}}.
\end{equation}
Observe that the set on the right-hand side does not contain an element of the form 
\begin{equation}
((x_{1},\ux_{2;\ell},x_{1},\ux_{\ell+2;k}), (\xi_{1},\ul{0}_{2;\ell},-\xi_{1},\ul{0}_{\ell+2;k})), \qquad \xi_{1}\neq 0.
\end{equation}
since $\WF( u_{\ul{\alpha}_{\ell},m} )$ is nonempty and satisfies the non-vanishing pair property.
\item 
Suppose that both $\WF\paren*{u_{\ul{\alpha}_{\ell},m}}$ and $\WF\paren*{u_{\ul{\alpha}_{j},n-m}}$ are both nonempty and satisfy the non-vanishing pair property. Then if $(\ux_{k},\uxi_{k}) \in \WF\paren*{u_{\ul{\alpha}_{\ell},m} \otimes u_{\ul{\alpha}_{j},n-m}}$, one of three sub-cases must occur:
\begin{enumerate}[1.]
\item
$\uxi_{\ell}=0$ and there exists $l_{1},l_{2}\in\{\ell+1,\ldots,\ell+j\}$ such that $\xi_{l_1} \neq 0$ and $\xi_{l_2} \neq 0$.
\item
$\uxi_{\ell+1;k}=0$ and there exists $l_{1},l_{2}\in\{1,\ldots,\ell\}$ such that $\xi_{l_1} \neq 0$ and $\xi_{l_2} \neq 0$.
\item
$\uxi_{\ell}\neq 0$, $\uxi_{\ell+1;k}\neq 0$, and there exist $l_{1},l_{2}\in\{1,\ldots,\ell\}$ and $l_{3},l_{4}\in\{\ell+1,\ldots,k\}$ such that $\xi_{l_{1}} \neq 0$, $\xi_{l_{2}} \neq 0$, $\xi_{l_{3}} \neq 0$ and $\xi_{l_{4}} \neq 0$.
\end{enumerate}
Any of these three sub-cases guarantees \eqref{eq:wf_goal}.
\end{enumerate}
To summarize, we have shown that
\begin{equation}
((x_{1},\ux_{2;\ell},x_{1},\ux_{\ell+2;k}), (\xi_{1},\ul{0}_{2;\ell},-\xi_{1},\ul{0}_{\ell+2;k})) \notin \bigcup_{(\ul{\alpha}_{\ell},\ul{\alpha}_{j}) \in \mathsf{A}_{n,m}^{(k)}} \WF\paren*{u_{\ul{\alpha}_{\ell},m} \otimes u_{\ul{\alpha}_{j},n-m}},
\end{equation}
and therefore
\begin{equation}
\delta(x_{1}-x_{2})\paren*{\widetilde{\W}_{m}^{(\ell)} \otimes \widetilde{\W}_{n-m}^{(j)}}(f^{(k)})
\end{equation}
is defined in $\D'(\R^{k})$ according to \cref{prop:H_crit}, proving the first claim.

We now show that this H\"ormander product is tempered:
\begin{equation}
\delta(x_{1}-x_{\ell+1})\paren*{\widetilde{\W}_{m}^{(\ell)} \otimes \widetilde{\W}_{n-m}^{(j)}}(f^{(k)}) \in \Sc'(\R^{k}), \qquad \forall f^{(k)}\in\Sc'(\R^{k}).
\end{equation}
Since by the inductive hypothesis, $\widetilde{\W}_{m}^{(\ell)}$ and $\widetilde{\W}_{n-m}^{(j)}$ satisfy the good mapping property of \cref{def:gmp} (and we refer to  \cref{ssec:GMP} for more details on the good mapping property), there exist unique continuous bilinear maps
\begin{equation}
\label{eq:Phi_W_def}
\Phi_{\wt{\W}_m^{(\ell)},\alpha}: \Sc(\R^\ell)^2\rightarrow \Sc_{(x_\alpha,x_\alpha')}(\R^2), \enspace \Phi_{\wt{\W}_{n-m}^{(j)},\beta}: \Sc(\R^j)^2\rightarrow\Sc_{(x_\beta,x_\beta')}(\R^2), \qquad \alpha\in\N_{\leq\ell}, \enspace \beta\in\N_{\leq j}
\end{equation}
identifiable with the maps
\begin{equation}
\begin{split}
\Sc(\R^\ell)^2\rightarrow \Sc_{x_\alpha'}(\R;\Sc_{x_\alpha}'(\R)), \quad (f^{(\ell)}, g^{(\ell)}) &\mapsto \ipp*{\wt{\W}_m^{(\ell)} f^{(\ell)},(\cdot)\otimes_\alpha g^{(\ell)}(\cdot,x_\alpha',\cdot)}_{\Sc'(\R^\ell)-\Sc(\R^\ell)}, \\
\Sc(\R^j)^2\rightarrow \Sc_{x_\beta'}(\R;\Sc_{x_\beta}'(\R)), \quad (f^{(j)}, g^{(j)}) &\mapsto \ipp*{\wt{\W}_m^{(j)} f^{(j)},(\cdot)\otimes_\beta g^{(j)}(\cdot,x_\beta',\cdot)}_{\Sc'(\R^j)-\Sc(\R^j)},
\end{split}
\end{equation}
via
\begin{equation}
\label{eq:Phi_W_id}
\begin{split}
&\int_{\R}dx_\alpha\Phi_{\wt{\W}_m^{(\ell)},\alpha}(f^{(\ell)},g^{(\ell)})(x_\alpha;x_\alpha')\phi(x_\alpha) = \ipp*{\wt{\W}_n^{(\ell)} f^{(\ell)}, \phi\otimes_\alpha g^{(\ell)}(\cdot,x_\alpha',\cdot)}_{\Sc'(\R^\ell)-\Sc(\R^\ell)}, \\
&\int_{\R}dx_\beta \Phi_{\wt{\W}_{n-m}^{(j)},\beta}(f^{(j)},g^{(j)})(x_\beta;x_\beta')\phi(x_\beta) = \ipp*{\wt{\W}_{n-m}^{(j)} f^{(j)}, \phi\otimes_\beta g^{(j)}(\cdot,x_\beta',\cdot)}_{\Sc'(\R^j)-\Sc(\R^j)}, 
\end{split}
\end{equation}
for $\phi\in\Sc(\R)$, respectively. Above, the notation $(\cdot) \otimes_\alpha g^{(\ell)}(\cdot,x_\alpha',\cdot)$ and $(\cdot)\otimes_{\beta} g^{(j)}(\cdot,x_\beta',\cdot)$ is defined by
\begin{equation}
\label{eq:otimes_sub}
\begin{split}
\paren*{\phi \otimes_\alpha g^{(\ell)}(\cdot,x_\alpha',\cdot)}(\ul{y}_\alpha) &\coloneqq \phi(y_\alpha) g^{(\ell)}(\ul{y}_{1;\alpha-1},x_\alpha',\ul{y}_{\alpha+1;\ell}), \qquad \forall \ul{y}_\ell\in\R^\ell \\
\paren*{\phi\otimes_\beta g^{(j)}(\cdot,x_\beta',\cdot)}(\ul{y}_\beta) &\coloneqq \phi(y_\beta)g^{(j)}(\ul{y}_{1;\beta-1},x_\beta',\ul{y}_{\beta+1;j}), \qquad \forall \ul{y}_j\in\R^j
\end{split}, \qquad \forall \phi\in\Sc(\R).
\end{equation}
Now given $f^{(k)},g^{(k)}\in\Sc(\R^k)$, we see that
\begin{equation}
(\ux_\ell,\ux_\ell') \mapsto \Phi_{\wt{\W}_{n-m}^{(j)},1}(f^{(k)}(\ux_\ell,\cdot), g^{(k)}(\ux_\ell',\cdot)) \in \Sc_{(\ux_\ell,\ux_\ell')}(\R^{2\ell}; \Sc_{(y_1,y_1')}(\R^2)).
\end{equation}
Thus, we can define a map $\Psi_{\wt{\W}_{n-m}^{(j)},1}:\Sc(\R^k)^2 \rightarrow \Sc(\R^{2(\ell+1)})$
\begin{equation}
\begin{split}
&\Psi_{\wt{\W}_{n-m}^{(j)},1}(f^{(k)},g^{(k)})(\ux_{\ell+1};\ux_{\ell+1}') \\
&\coloneqq \Phi_{\wt{\W}_{n-m}^{(j)},1}(f^{(k)}(\ux_\ell,\cdot),g^{(k)}(\ux_\ell',\cdot))(x_{\ell+1};x_{\ell+1}'), \qquad \forall (\ux_{\ell+1},\ux_{\ell+1}')\in\R^{2(\ell+1)},
\end{split}
\end{equation}
which is bilinear and continuous. Now since $\Phi_{\wt{\W}_m^{(\ell)},1}:\Sc(\R^{\ell})^2\rightarrow\Sc(\R^2)$ is bilinear and continuous, the universal property of the tensor product and the identification of $\Sc(\R^{2\ell})\cong \Sc(\R^\ell) \hat{\otimes} \Sc(\R^\ell)$ implies that there exists a unique continuous linear map
\begin{equation}
\bar{\Phi}_{\wt{\W}_m^{(\ell)},1}: \Sc(\R^{2\ell})\rightarrow \Sc(\R^2),
\end{equation}
with the property that
\begin{equation}
\label{eq:Phibar_tp}
\Phi_{\wt{\W}_m^{(\ell)},1}(f^{(\ell)},g^{(\ell)}) = \bar{\Phi}_{\wt{\W}_m^{(\ell)},1}(f^{(\ell)}\otimes g^{(\ell)}), \qquad \forall f^{(\ell)},g^{(\ell)}\in\Sc(\R^\ell).
\end{equation}
Hence, the function
\begin{equation*}
\bar{\Phi}_{\wt{\W}_{m,1}^{(\ell)},1}\paren*{\Psi_{\wt{\W}_{n-m}^{(j)},1}(f^{(k)},g^{(k)})(\cdot,x_{\ell+1};\cdot,x_{\ell+1}')}(x_1;x_1'), \qquad \forall (x_1,x_{\ell+1},x_1',x_{\ell+1}')\in\R^{4}
\end{equation*}
defines an element of $\Sc(\R^4)$, and moreover,
\begin{equation}
\begin{split}
&\Sc(\R^k)^2 \rightarrow \Sc(\R^{4}), \\
&(f^{(k)},g^{(k)})\mapsto \bar{\Phi}_{\wt{\W}_{m,1}^{(\ell)},1}\paren*{\Psi_{\wt{\W}_{n-m}^{(j)},1}(f^{(k)},g^{(k)})(\cdot,x_{\ell+1};\cdot,x_{\ell+1}')}(x_1;x_1'), \qquad \forall (x_1,x_{\ell+1},x_1', x_{\ell+1}')\in\R^{4}
\end{split}
\end{equation}
is a continuous bilinear map. Thus, we may define a functional $u_{f^{(k)}}$ on $\Sc(\R^k)$ by
\begin{equation}\label{eq:map_func}
\begin{split}
&\ipp{u_{f^{(k)}},g^{(k)}}_{\Sc'(\R^{k})-\Sc(\R^{k})} \\
&\coloneqq \int_{\R^2}dx_1 dx_{\ell+1} \delta(x_1-x_{\ell+1}) \bar{\Phi}_{\wt{\W}_m^{(\ell)},1}\paren*{\Psi_{\wt{\W}_{n-m}^{(j)},1}(f^{(k)},g^{(k)})(\cdot,x_{\ell+1};\cdot,x_{\ell+1})}(x_1;x_1), \qquad \forall g^{(k)}\in\Sc(\R^k).
\end{split}
\end{equation}
This functional $u_{f^{(k)}}$ is evidently linear, and it follows from the continuity of $\bar{\Phi}_{\wt{\W}_m^{(\ell)},1}$ and $\Psi_{\wt{\W}_{n-m}^{(j)},1}$ that it is continuous $\Sc(\R^{k}) \rightarrow\C$, hence a tempered distribution. Furthermore, we claim that the map
\begin{equation}
\Sc(\R^k)\rightarrow \Sc'(\R^k), \qquad f^{(k)} \mapsto \ipp{u_{f^{(k)}}, \cdot}_{\Sc'(\R^k)-\Sc(\R^k)}
\end{equation}
satisfies the good mapping property. Indeed, replacing $f^{(k)}, g^{(k)}$ with $\pi f^{(k)}, \pi g^{(k)}$, for any $\pi\in\Ss_k$, it suffices to verify this assertion for the case $\alpha=1$ in \cref{def:gmp}. Additionally, it suffices by the universal property of the tensor product and the Schwartz kernel theorem isomorphism $\Sc(\R^k)\cong \Sc(\R^\ell)\hat{\otimes}\Sc(\R^j)$ to show that there is a (necessarily unique) continuous, multilinear map
\begin{equation*}
\Phi_{u}: \paren*{\Sc(\R^\ell)\times\Sc(\R^j)}^2 \rightarrow\Sc(\R^2),
\end{equation*}
such that for $f^{(\ell)},g^{(\ell)}\in\Sc(\R^\ell)$ and $f^{(j)},g^{(j)}\in\Sc(\R^j)$, 
\begin{equation}
\begin{split}
&\int_{\R}dx \Phi_u(f^{(\ell)},f^{(j)}, g^{(\ell)}, g^{(j)})(x;x')\phi(x) \\
&= \ipp{u_{f^{(\ell)}\otimes f^{(j)}}, \phi\otimes (g^{(\ell)}\otimes g^{(j)})(x',\cdot)}_{\Sc'(\R^k)-\Sc(\R^k)}, \qquad \forall \phi\in\Sc(\R), \ x'\in\R.
\end{split}
\end{equation}
Now for any $\phi\in\Sc(\R)$, the bilinearity of $\Phi_{\wt{\W}_{n-m}^{(j)},1}$ implies
\begin{equation}
\begin{split}
&\Phi_{\wt{\W}_{n-m}^{(j)},1}\paren*{(f^{(\ell)}\otimes f^{(j)})(\ux_\ell,\cdot), (\phi \otimes (g^{(\ell)}\otimes g^{(j)})(x',\cdot))(\ux_\ell',\cdot)}(x_{\ell+1};x_{\ell+1}')\\
&=f^{(\ell)}(\ux_\ell) \phi(x_1')g^{(\ell)}(x',\ux_{2;\ell}') \Phi_{\wt{\W}_{n-m}^{(j)},1}\paren*{f^{(j)},g^{(j)}}(x_{\ell+1};x_{\ell+1}'), \qquad \forall (\ux_{\ell+1},\ux_{\ell+1}',x')\in\R^{2\ell+3}.
\end{split}
\end{equation}
Hence,
\begin{equation}
\begin{split}
&\Psi_{\wt{\W}_{n-m}^{(j)},1}\paren*{f^{(\ell)}\otimes f^{(j)}, \phi \otimes (g^{(\ell)}\otimes g^{(j)})(x',\cdot)}(\ux_{\ell+1};\ux_{\ell+1}')\\
&=f^{(\ell)}(\ux_\ell)\phi(x_1')g^{(\ell)}(x',\ux_{2;\ell}')\Phi_{\wt{\W}_{n-m}^{(j)},1}(f^{(j)},g^{(j)})(x_{\ell+1};x_{\ell+1}'), \qquad \forall (\ux_{\ell+1},\ux_{\ell+1}')\in\R^{2(\ell+1)}.
\end{split}
\end{equation}
For $x'\in\R$ and $\phi\in\Sc(\R)$, define the function $\tl{g}_{x',\phi}^{(\ell)}\in \Sc(\R^\ell)$ by
\begin{equation}
\label{eq:tlf_tlg}
\begin{split}
\tl{g}_{x',\phi}^{(\ell)}(\ux_\ell') &\coloneqq \phi(x_1')g^{(\ell)}(x',\ux_{2;\ell}'), \qquad \forall \ux_\ell'\in\R^{\ell},
\end{split}
\end{equation}
so that we can write
\begin{equation}
\label{eq:Psi_prod}
\begin{split}
&\Psi_{\wt{\W}_{n-m}^{(j)},1}\paren*{f^{(\ell)}\otimes f^{(j)}, \phi \otimes (g^{(\ell)}\otimes g^{(j)})(x',\cdot)}(\ux_{\ell+1};\ux_{\ell+1}') \\
&= (f^{(\ell)}\otimes \tl{g}_{x',\phi}^{(\ell)})(\ux_\ell;\ux_\ell')\Phi_{\wt{\W}_{n-m}^{(j)},1}(f^{(j)},g^{(j)})(x_{\ell+1};x_{\ell+1}'), \qquad \forall (\ux_{\ell+1},\ux_{\ell+1}')\in\R^{2(\ell+1)}.
\end{split}
\end{equation}
Therefore, using identity \eqref{eq:Psi_prod} and the linearity of the map $\bar{\Phi}_{\wt{\W}_m^{(\ell)},1}$, we see that
\begin{align}
&\bar{\Phi}_{\wt{\W}_m^{(\ell)},1}\paren*{\Psi_{\wt{\W}_{n-m}^{(j)},1}\paren*{f^{(\ell)}\otimes f^{(j)}, \phi \otimes (g^{(\ell)}\otimes g^{(j)})(x',\cdot)}(\cdot,x_{\ell+1};\cdot,x_{\ell+1}')}(x_1;x_1') \nonumber  \\
&= \Phi_{\wt{\W}_{n-m}^{(j)},1}(f^{(j)},g^{(j)})(x_{\ell+1};x_{\ell+1}')\bar{\Phi}_{\wt{\W}_m^{(\ell)},1}\paren*{f^{(\ell)}\otimes \tl{g}_{x',\phi}^{(\ell)}}(x_1;x_1')  \nonumber\\
&=\Phi_{\wt{\W}_{n-m}^{(j)},1}(f^{(j)},g^{(j)})(x_{\ell+1};x_{\ell+1}')\Phi_{\wt{\W}_m^{(\ell)},1}(f^{(\ell)}, \tl{g}_{x',\phi}^{(\ell)})(x_1;x_1'),
\end{align}
where the ultimate equality follows from the property \eqref{eq:Phibar_tp}. Recalling the definition \eqref{eq:map_func} for $u_{f^{(k)}}$, we obtain that
\begin{align}
&\ipp*{u_{f^{(\ell)}\otimes f^{(j)}}, \phi \otimes (g^{(\ell)}\otimes g^{(j)})(x',\cdot)}_{\Sc'(\R^k)-\Sc(\R^k)} \nonumber\\
&= \int_{\R^2}dx_1dx_{\ell+1}\delta(x_1-x_{\ell+1}) \bar{\Phi}_{\wt{\W}_m^{(\ell)},1}\paren*{\Psi_{\wt{\W}_{n-m}^{(j)},1}\paren*{f^{(\ell)}\otimes f^{(j)}, \phi\otimes (g^{(\ell)}\otimes g^{(j)})(x',\cdot)}(\cdot,x_{\ell+1};\cdot,x_{\ell+1})}(x_1;x_1) \nonumber\\
&=\int_{\R^2}dx_1dx_{\ell+1}\delta(x_1-x_{\ell+1}) \Phi_{\wt{\W}_{n-m}^{(j)},1}(f^{(j)},g^{(j)})(x_{\ell+1};x_{\ell+1})\Phi_{\wt{\W}_m^{(\ell)},1}(f^{(\ell)}, \tl{g}_{x',\phi}^{(\ell)})(x_1;x_1) \nonumber\\
&=\int_{\R}dx\Phi_{\wt{\W}_{n-m}^{(j)},1}(f^{(j)},g^{(j)})(x;x)\Phi_{\wt{\W}_m^{(\ell)},1}(f^{(\ell)}, \tl{g}_{x',\phi}^{(\ell)})(x;x) \nonumber \\
&=\ipp*{\wt{\W}_m^{(\ell)} f^{(\ell)}, \Phi_{\wt{\W}_{n-m}^{(j)},1}(f^{(j)},g^{(j)})|_{y=y'} \tl{g}_{x',\phi}^{(\ell)}}_{\Sc'(\R^\ell)-\Sc(\R^\ell)} \nonumber,
\end{align}
where $\Phi_{\wt{\W}_{n-m}^{(j)},1}(f^{(j)},g^{(j)})|_{y=y'}$ denotes the restriction to the hyperplane $\{(y,y'):y=y'\}\subset \R^2$ and the ultimate equality follows from the definition of $\Phi_{\wt{\W}_m^{(\ell)},1}$ in \eqref{eq:Phi_W_def}. Unpacking the definition of $\tl{g}_{x',\phi}^{(\ell)}$ from \eqref{eq:tlf_tlg} and applying the definition of $\Phi_{\wt{\W}_m^{(\ell)},1}$ once more, we conclude that
\begin{align}
&\ipp*{\wt{\W}_m^{(\ell)}f^{(\ell)}, \Phi_{\wt{\W}_{n-m}^{(j)},1}(f^{(j)},g^{(j)})|_{y=y'}\tl{g}_{x',\phi}^{(\ell)}}_{\Sc'(\R^\ell)-\Sc(\R^\ell)} \nonumber\\
&= \ipp*{\wt{\W}_m^{(\ell)}f^{(\ell)}, (\phi\Phi_{\wt{\W}_{n-m}^{(j)},1}(f^{(j)},g^{(j)})|_{y=y'})\otimes g^{(\ell)}(x',\cdot)}_{\Sc'(\R^\ell)-\Sc(\R^\ell)} \nonumber\\
&=\int_{\R}dx \Phi_{\wt{\W}_m^{(\ell)},1}(f^{(\ell)},g^{(\ell)})(x;x')\phi(x)\Phi_{\wt{\W}_{n-m}^{(j)},1}(f^{(j)},g^{(j)})(x;x).
\end{align}
Therefore, the desired map $\Phi_u$ is given by
\begin{equation}
\Phi_u(f^{(\ell)},f^{(j)},g^{(\ell)},g^{(j)})(x;x') \coloneqq \Phi_{\wt{\W}_m^{(\ell)},1}(f^{(\ell)},g^{(\ell)})(x;x')\Phi_{\wt{\W}_{n-m}^{(j)},1}(f^{(j)},g^{(j)})(x;x),
\end{equation}
which is evidently multilinear and continuous $(\Sc(\R^\ell)\times\Sc(\R^j))^2\rightarrow\Sc(\R^2)$ being the composition maps. Thus, the proof that $f^{(k)}\mapsto u_{f^{(k)}}$ has the good mapping property is complete.

Lastly, we claim that $u_{f^{(k)}}$ coincides with the H\"ormander product 
\[
\delta(x_{1}-x_{\ell+1})\paren*{\widetilde{\W}_{m}^{(\ell)}\otimes\widetilde{\W}_{n-m}^{(j)}}(f^{(k)})
\]
defined above via \cref{prop:H_crit}. To prove the claim, we rely on the uniqueness criterion for the product. We set
\begin{equation}
g^{(k)}\coloneqq g^{(1)}\otimes g^{(\ell-1)}\otimes \tilde{g}^{(1)}\otimes g^{(j-1)}, \quad \phi^{(k)}\coloneqq \phi^{(1)}\otimes \phi^{(\ell-1)}\otimes\tl{\phi}^{(1)}\otimes \phi^{(j-1)}
\end{equation}
for $g^{(1)},\tilde{g}^{(1)},\phi^{(1)},\tilde{\phi}^{(1)}\in\Sc(\R)$, $g^{(\ell-1)},\phi^{(\ell-1)}\in\Sc(\R^{i-1})$, and $g^{(j-1)},\phi^{(j-1)}\in\Sc(\R^{j-1})$. By density of linear combinations of tensor products, it suffices to show that
\begin{equation}
\ipp{\F({g^{(k)}}^{2}u_{f^{(k)}}),\phi^{(k)}}_{\Sc'(\R^{k})-\Sc(\R^{k})}=\ipp{\F(g^{(k)}\delta(x_{1}-x_{\ell+1})) \ast \F(g^{(k)}(\widetilde{\W}_{m}^{(\ell)}\otimes \widetilde{\W}_{n-m}^{(j)})(f^{(k)})), \phi^{(k)}}_{\Sc'(\R^{k})-\Sc(\R^{k})},
\end{equation}
since pointwise equality then follows from the localization lemma (see Chapter 2, \S{2} of \cite{HormPDOI}) together with the continuity of the Fourier transforms involved. This is then an exercise, the details of which we leave to the reader, relying on the good mapping property and the distributional Plancherel theorem.

\medskip
\noindent \textbf{Step II:} The property \ref{item:Wn_wd_2} is readily established by the arguments in the previous step and the fact that $\mathsf{A}_{n,m}^{(k)}$ defined in \eqref{anm} has finite cardinality, it then follows from another application of \cref{prop:wf_prop}\ref{item:wf_prop_sum}  that either 
\[
\WF\paren*{\widetilde{\W}_{n+1}^{(k)}f^{(k)}}=\emptyset
\]
or 
\[
\WF\paren*{\widetilde{\W}_{n+1}^{(k)}f^{(k)}}\neq \emptyset \textup{ and satisfies the non-vanishing pair property}.
\]

\medskip
\noindent \textbf{Step III:} Next, we show that the map $f^{(k)}\mapsto \widetilde{\W}_{n+1}^{(k)}f^{(k)}$ satisfies the good mapping property for every $k\in\N$. Since differentiation is a continuous endomorphism of $\Sc'(\R^{k})$, it is immediate from the induction hypothesis that
\begin{equation}
-i\p_{x_{1}}\widetilde{\W}_{n}^{(k)}\in \L_{gmp}(\Sc(\R^{k}),\Sc'(\R^{k})).
\end{equation}
Since $\L_{gmp}(\Sc(\R^k),\Sc'(\R^k))$ is a vector space, it remains to show that
\begin{equation}
\label{eq:hp_rhs}
f^{(k)}\mapsto \delta(x_{1}-x_{\ell+1})\paren*{\widetilde{\W}_{m}^{(\ell)}\otimes \widetilde{\W}_{n-m}^{(j)}}(f^{(k)})
\end{equation}
satisfies the good mapping property for every $\ell,j\in\N$ with $\ell+j=k$ and $m\in\N_{\leq n-1}$. But this follows from Step II, where we showed that $u_{f^{(k)}}$ defined in \eqref{eq:map_func} coincides with the H\"ormander product in the right-hand side of \eqref{eq:hp_rhs} and that the DVO $f^{(k)}\mapsto u_{f^{(k)}}$ defined in \eqref{eq:map_func} has the good mapping property.

\medskip
\noindent \textbf{Step IV:} Finally, we show that 
\[
\widetilde{\W}_{n}^{(k)}:\Sc(\R^{k})\rightarrow \Sc'(\R^{k})
\]
is a continuous map. As argued before, it suffices to show that the map
\begin{equation}
(f^{(\ell)},f^{(j)}) \mapsto \delta(x_{1}-x_{\ell+1})\paren*{\widetilde{\W}_{m}^{(\ell)}\otimes \widetilde{\W}_{n-m}^{(j)}}(f^{(\ell)}\otimes f^{(j)})
\end{equation}
is a continuous bilinear map $\Sc(\R^{\ell})\times\Sc(\R^{j}) \rightarrow \Sc'(\R^{k})$. Bilinearity is obvious. For continuity, suppose that $(f_{r}^{(\ell)},f_{r}^{(j)}) \rightarrow 0\in\Sc(\R^{\ell})\times\Sc(\R^{j})$ as $r\rightarrow\infty$. We need to show that for any bounded subset $\mathfrak{R}$ of $\Sc(\R^{k})$,
\begin{equation}\label{null_conv}
\lim_{r\rightarrow\infty} \sup_{g^{(k)}\in\mathfrak{R}} \left|\ipp{\delta(x_{1}-x_{\ell+1})\paren*{\widetilde{\W}_{m}^{(\ell)}\otimes \widetilde{\W}_{n-m}^{(j)}}(f_{r}^{(\ell)}\otimes f_{r}^{(j)}), g^{(k)}}_{\Sc'(\R^{k})-\Sc(\R^{k})}\right|=0.
\end{equation}
But this follows from our analysis proving the good mapping property of the map $f^{(k)} \mapsto u_{f^{(k)}}$ in Step II.
\end{proof}

We now turn to showing that only finitely many components of $\wt{\W}_{n}$ are nonzero for a given $n\in\N$. This property justifies our use of the direct sum notation.

\begin{lemma}\label{lem:aa_zero}
For all $n\in\N$, we have
\begin{equation}
\wt{\W}_{2n}^{(k)} = 0 \in \L(\mathcal{S}(\R^{k}), \mathcal{S}'(\R^{k})) \qquad k\in\N_{\geq n+1},
\end{equation}
and
\begin{equation}
\wt{\W}_{2n+1}^{(k)}=0\in \L(\Sc(\R^{k}),\Sc'(\R^{k})), \qquad k\in\N_{\geq n+2}.
\end{equation}
\end{lemma}
\begin{proof}
We prove the lemma by strong induction on $n$. We first establish the base case $n=1$. It follows from the recursion \eqref{eq:Wn_recur} that
\begin{equation}
\widetilde{\W}_{2} = -i\p_{x_{1}}\E_{1}.
\end{equation}
Since $\E_{1}^{(k)}=0$ for $k\geq 2$, it follows that $\widetilde{\W}_{2}^{(k)}=0$ for $k\geq 2$. To see that $\wt{\W}_{3}^{(k)}=0$ for $k\geq 3$, observe that
\begin{equation}
(-i\p_{x_{1}})\wt{\W}_{2}^{(k)} = 0 \in\L(\Sc(\R^{k}),\Sc'(\R^{k})),
\end{equation}
since $\wt{\W}_{2}^{(k)}=0$. If $k\geq 3$ and $\ell,j\in\N$ satisfy $\ell+j=k$, then $\max\{\ell,j\} \geq 2$. Since $\wt{\W}_{1}^{(m)}=0$ for $m\geq 2$, we obtain that
\begin{equation}
\wt{\W}_{1}^{(\ell)}\otimes \wt{\W}_{1}^{(j)} = 0 \in\L(\Sc(\R^{k}),\Sc'(\R^{k})),
\end{equation}
which implies that $\delta(X_{1}-X_{\ell+1})\paren*{\wt{\W}_{1}^{(\ell)}\otimes\wt{\W}_{1}^{(j)}}=0$.

We now proceed to the inductive step. Let $n\in\N_{\geq 2}$ and suppose that for all integers $m\in\N_{\leq n}$,
\begin{align}
\wt{\W}_{2m}^{(k)}&=0\in\L(\Sc(\R^{k}),\Sc'(\R^{k})), \qquad \forall k\in\N_{\geq m+1} \\
\wt{\W}_{2m+1}^{(k)}&=0\in\L(\Sc(\R^{k}),\Sc'(\R^{k})), \qquad \forall k\in\N_{\geq m+2}.
\end{align}
We now need to show that these identities hold with $m = n+1$. We first handle the case of even indices. Specifically, we show that 
\[
\wt{\W}_{2(n+1)}^{(k)}=0\in\L(\Sc(\R^{k}),\Sc'(\R^{k})), \quad k\in \N_{\geq n+2}.
\]
Observe that if $k\geq n+2$, then by the induction hypothesis, $\widetilde{\W}_{2(n+1) -1}^{(k)}=0$ and therefore
\begin{equation}
-i\p_{x_{1}}\widetilde{\W}_{2(n+1) -1}^{(k)} = 0 \in\L(\Sc(\R^{k}),\Sc'(\R^{k})).
\end{equation}
We now consider the H\"ormander product terms
\begin{equation}
\delta(X_{1}-X_{\ell+1})\paren*{\wt{\W}_{m}^{(\ell)}\otimes\wt{\W}_{2n+1-m}^{(j)}}, \qquad \ell+j=k
\end{equation}
arising in the recursion relation \eqref{eq:Wn_recur} for $\wt{\W}_{2(n+1)}^{(k)}$. By symmetry, it suffices to consider the following case: if $m$ is odd (i.e. $m=2r+1$ for some $r\in\N_0$) then $2n+1-m$ is even (i.e. $2n+1-m=2r'$ for some $r'\in\N$), and we can write $n=r+r'$. By the induction hypothesis
\begin{align}
\wt{\W}_{m}^{(\ell)} &= 0 , \quad \forall \ell \in \N_{\geq r + 2}\\
\wt{\W}_{2n+1-m}^{(j)} &= 0 , \quad \forall j \in \N_{\geq r'+1}.
\end{align}
If $k\geq n + 2 = r+r'+2$, then either $\ell\geq r+2$ or $j\geq r'+1$, since if both $\ell\leq r+1$ and $j\leq r' $, then
\begin{equation}
k=\ell+j\leq r+r' + 1.
\end{equation}
Thus,
\begin{equation}
\delta(X_{1}-X_{\ell+1})\paren*{\wt{\W}_{m}^{(\ell)}\otimes\wt{\W}_{2n+1-m}^{(j)}} = 0\in\L(\Sc(\R^{k}),\Sc'(\R^{k})),
\end{equation}
and so it follows from the recursion relation \eqref{eq:Wn_recur} that $\wt{\W}_{2(n+1)}^{(k)}=0\in\L(\Sc(\R^{k}),\Sc'(\R^{k}))$ for $k\geq n+2$.

We next handle the case of odd indices, namely we show that
\begin{equation}
\wt{\W}_{2(n+1)+1}^{(k)}=0\in\L(\Sc(\R^{k}),\Sc'(\R^{k})), \qquad k\geq n+3.
\end{equation}
As before, observe that if $k\geq n+3$, then
\begin{equation}
(-i\p_{x_{1}})\wt{\W}_{2(n+1)}^{(k)} = 0\in\L(\Sc(\R^{k}),\Sc'(\R^{k}))
\end{equation}
by the result of the preceding paragraph. Now consider the H\"ormander product terms
\begin{equation}
\delta(X_{1}-X_{\ell+1})\paren*{\wt{\W}_{m}^{(\ell)}\otimes \wt{\W}_{2n+2-m}^{(j)}}
\end{equation}
in the recursion relation \eqref{eq:Wn_recur} for $\wt{\W}_{2(n+1)+1}^{(k)}$. We consider two cases:
\begin{enumerate}[C1.]
\item
Suppose $m$ is odd (i.e. $m=2r+1$ for some $r\in\N_{0}$). Then $2n+2-m$ is odd (i.e. $2n+2-m=2r'+1$ for some $r'\in\N_0$), and we can write $2(n+1)+1=2(r+r'+1)+1$. If $k\geq (r+r'+1)+2$, then either $\ell\geq r+2$ or $j\geq r'+2$, since if both $\ell\leq r+1$ and $j\leq r'+1$, we have that
\begin{equation}
k=\ell+j \leq (r+r'+1)+1.
\end{equation}
Hence applying the induction hypothesis to obtain $\wt{\W}_{m}^{(\ell)}=0$ or $\wt{\W}_{2n+2-m}^{(j)}=0$, respectively, we conclude that
\begin{equation}
\delta(X_{1}-X_{\ell+1})\paren*{\wt{\W}_{m}^{(\ell)}\otimes\wt{\W}_{2n+2-m}^{(j)}} = 0 \in\L(\Sc(\R^{k}),\Sc'(\R^{k})).
\end{equation}
\item
Suppose $m$ is even (i.e. $m=2r$ for some $r\in\N$). Then $2n+2-m$ is even (i.e. $2n+2-m=2r'$ for some $r'\in\N$), and we can write $2n+2=2(r+r')$. Once again, if $k\geq r+r'+1$, then either $\ell\geq r+1$ or $j\geq r'+1$, since if $\ell\leq r$ and $j\leq r'$, then
\begin{equation}
k=\ell+j\leq r+r'.
\end{equation}
Hence, we obtain again that
\begin{equation}
\delta(X_{1}-X_{\ell+1})\paren*{\wt{\W}_{m}^{(\ell)}\otimes\wt{\W}_{2n+2-m}^{(j)}}=0\in\L(\Sc(\R^{k}),\Sc'(\R^{k})).
\end{equation}
by the induction hypothesis.
\end{enumerate}
In now follows from the recursion relation \eqref{eq:Wn_recur} that $\wt{\W}_{2(n+1)+1}^{(k)}=0\in\L(\Sc(\R^{k}),\Sc'(\R^{k}))$ for $k\geq n+3$, completing the proof of the inductive step.
\end{proof}

\subsection{Step 2: Defining self-adjoint operators}
Our goal is now to define the self-adjoint elements $\W_{n,sa}$, proving the following:

\begin{prop}\label{prop:Wn_sa}
For each $n\in\N$, there exists an element
\[
\W_{n,sa} \in \bigoplus_{k=1}^{\infty}\L_{gmp,*}(\Sc(\R^{k}),\Sc'(\R^{k})) ,
\]
given by
\begin{equation}\label{wn_sa_def}
\W_{n,sa} \coloneqq \frac{1}{2}\paren*{\widetilde{\W}_{n}+\widetilde{\W}_{n}^{*}}.
\end{equation}
\end{prop}
\begin{remark}
Recall that 
\[
(\widetilde{\W}_{n}^{*})^{(k)} \coloneqq  \widetilde{\W}_{n}^{(k),*}.
\]
is the adjoint operator defined in \cref{lem:dvo_adj}. 
\end{remark}

It follows readily from \cref{lem:dvo_adj} that 
\[
\W_{n,sa}\in \bigoplus_{k=1}^{\infty}\L(\Sc(\R^{k}),\Sc'(\R^{k}))
\]
and is self-adjoint. Thus, in order to prove \cref{prop:Wn_sa}, we only need to verify each $\W_{n,sa}$ satisfies the good mapping property, for which it suffices by linearity and the fact that each $\widetilde{\W}_{n}^{(k)}\in\L_{gmp}(\Sc(\R^{k}),\Sc'(\R^{k}))$ to prove that
\begin{align}\label{eq:Wn_a_gmp}
\widetilde{\W}_{n}^{(k),*}\in\L_{gmp}(\Sc(\R^{k}),\Sc'(\R^{k})), \qquad \forall k\in\N.
\end{align}
Using the recursion \eqref{eq:Wn_recur}, the linearity of the adjoint operation, and the fact that
\begin{equation}
\paren*{-i\p_{x_1}\wt{\W}_n^{(k)}}^* = \wt{\W}_n^{(k),*}(-i\p_{x_1}) \in \L_{gmp}(\Sc(\R^k),\Sc'(\R^k))
\end{equation}
by \cref{lem:adj_comp}, we just need to show that
\begin{equation}
\label{eq:HP_gmp}
\paren*{\delta(X_1-X_{\ell+1})\paren*{\wt{\W}_{m}^{(\ell)}\otimes\wt{\W}_{n-m}^{(j)}}}^* \in \L_{gmp}(\Sc(\R^k),\Sc'(\R^k))
\end{equation}
for any $m\in\N_{\leq n-1}$ and $\ell,j\in\N$ satisfying $\ell+j=k$. We prove this assertion by another induction argument.

\begin{lemma}
\label{lem:Wn_sa_recur}
Let $n\in\N_{\geq 2}$, and suppose that $\widetilde{\W}_{1}^{*},\ldots,\widetilde{\W}_{n-1}^{*}\in \bigoplus_{k=1}^{\infty}\L_{gmp}(\Sc(\R^{k}),\Sc'(\R^{k}))$. Then \eqref{eq:HP_gmp} holds.
\end{lemma}
\begin{proof} 
Let $k\in\N$. Given $f^{(k)}\in\Sc(\R^{k})$, we define the tempered distribution $v_{f^{(k)}}$ by
\begin{equation}
g^{(k)}\mapsto \ip{f^{(k)}}{\delta(X_{1}-X_{\ell+1})\paren*{\widetilde{\W}_{m}^{(\ell)} \otimes \widetilde{\W}_{n-m}^{(j)}}g^{(k)}},
\end{equation}
where the composition $\delta(X_{1}-X_{\ell+1})(\widetilde{\W}_{m}^{(\ell)} \otimes \widetilde{\W}_{n-m}^{(j)})$ is well-defined by \cref{lem:Wn_wd}. It is easy to check that the map
\begin{equation}
\Sc(\R^{k}) \rightarrow \Sc'(\R^{k}), \qquad f^{(k)}\mapsto v_{f^{(k)}}
\end{equation}
is a continuous linear map, so it remains for us to verify the good mapping property. As in the proof of \cref{lem:Wn_wd}, it suffices to show that for any $\alpha\in\N_{\leq k}$, the map
\begin{equation}\label{eq:gmp_goal}
\begin{split}
&(\Sc(\R^\ell)\times\Sc(\R^j))^2 \rightarrow \Sc_{x_\alpha'}(\R;\Sc_{x_\alpha}'(\R)) \\
&(f^{(\ell)},f^{(j)},g^{(\ell)},g^{(j)}) \mapsto \ip{v_{f^{(\ell)}\otimes f^{(j)}}}{(\cdot)\otimes_\alpha (g^{(\ell)}\otimes g^{(j)})(\cdot,x_\alpha',\cdot)}, \qquad x_\alpha'\in\R.
\end{split}
\end{equation}
may be identified with a (necessarily unique) continuous map $(\Sc(\R^{\ell})\times\Sc(\R^{j}))^{2} \rightarrow \Sc(\R^{2})$, which is antilinear in the $f^{(\ell)},f^{(j)}$ variables and linear in the $g^{(\ell)}, g^{(j)}$ variables. The reader will recall that the notation $\otimes_\alpha$ is defined in \eqref{eq:otimes_sub}. To simplify the presentation, we will assume $\alpha\leq \ell$. The case $\ell<\alpha\leq k$ follows mutatis mutandis. Moreover, by replacing $f^{(\ell)}, g^{(\ell)}$ with $\pi f^{(\ell)}, \pi g^{(\ell)}$, for $\pi\in\Ss_\ell$, we may assume that $\alpha=1$. For any $\phi\in\Sc(\R)$, we have by the distributional Fubini-Tonelli theorem that,
\begin{align}
&\ipp*{\ip{v_{f^{(\ell)}\otimes f^{(j)}}}{(\cdot)\otimes (g^{(\ell)}\otimes g^{(j)})(x_1',\cdot)}, \phi}_{\Sc'(\R)-\Sc(\R)} \nonumber\\
& =\ip{v_{f^{(\ell)}\otimes f^{(j)}}}{\phi\otimes (g^{(\ell)}\otimes g^{(j)})(x_1',\cdot)} \nonumber\\
&= \ip{f^{(\ell)}\otimes f^{(j)}}{\delta(x_{1}-x_{\ell+1})\paren*{\widetilde{\W}_{m}^{(\ell)}\otimes \widetilde{\W}_{n-m}^{(j)}}\paren*{\phi \otimes g^{(\ell)}(x_1',\cdot)\otimes g^{(j)}}} \nonumber\\
&= \ipp*{\delta(x_{1}-x_{\ell+1})\paren*{\widetilde{\W}_{m}^{(\ell)}\otimes \widetilde{\W}_{n-m}^{(j)}}\paren*{\phi \otimes g^{(\ell)}(x_1',\cdot)\otimes g^{(j)}}, \ol{f^{(\ell)}\otimes f^{(j)}}}_{\Sc'(\R^k)-\Sc(\R^k)}. \label{eq:tp_adj}
\end{align}
Using the identifications of \eqref{eq:Phi_W_id} and the action of the DVO $\delta(X_1-X_{\ell+1})(\wt{\W}_m^{(\ell)}\otimes\wt{\W}_{n-m}^{(j)})$ given by \eqref{eq:map_func} in Step II of the proof of \cref{lem:Wn_wd}, we find that
\begin{align}
\eqref{eq:tp_adj} &= \int_{\R}dx_1 \Phi_{\wt{\W}_{n-m}^{(j)},1}(g^{(j)},\ol{f^{(j)}})(x_1;x_1) \Phi_{\wt{\W}_m^{(\ell)},1}(\phi\otimes g^{(\ell)}(x_1',\cdot),\ol{f^{(\ell)}})(x_1;x_1) \nonumber\\
&=\ip{f^{(\ell)} \ol{\Phi_{\W_{n-m}^{(j)},1}(g^{(j)},\ol{f^{(j)}})|_{y=y'}}}{\wt{\W}_m^{(\ell)}\paren*{\phi\otimes g^{(\ell)}(x_1',\cdot)}} \nonumber\\
&=\ip{\wt{\W}_m^{(\ell),*} \paren*{f^{(\ell)} \ol{\Phi_{\W_{n-m}^{(j)},1}(g^{(j)},\ol{f^{(j)}})|_{y=y'}}}}{\phi\otimes g^{(\ell)}(x_1',\cdot)}, \label{eq:Wm_adj_gmp}
\end{align}
where the ultimate equality follows from the definition of the adjoint of a DVO, see \cref{lem:dvo_adj}. As before, the notation $|_{y=y'}$ denotes restriction to the hyperplane $\{(y,y'):y=y'\}\subset\R^2$. By the induction hypothesis, $\wt{\W}_m^{(\ell),*}$ possesses the good mapping property. Therefore, for any $\alpha\in\N_{\leq \ell}$, we can uniquely identify the map
\begin{equation}
\Sc(\R^\ell)^2\rightarrow\Sc_{x_\alpha'}(\R;\Sc_{x_\alpha}'(\R)), \qquad (\tl{f}^{(\ell)},\tl{g}^{(\ell)}) \mapsto \ipp*{\wt{\W}_m^{(\ell),*}\tl{f}^{(\ell)}, (\cdot)\otimes_\alpha \tl{g}^{(\ell)}(\cdot,x_\alpha',\cdot)}_{\Sc'(\R^\ell)-\Sc(\R^\ell)}
\end{equation}
with a continuous bilinear map
\begin{equation}
\begin{split}
&\Phi_{\wt{\W}_m^{(\ell),*},\alpha}: \Sc(\R^\ell)^2\rightarrow\Sc_{(x_\alpha,x_\alpha')}(\R^2) \\
&\int_{\R}dx_\alpha\Phi_{\wt{\W}_m^{(\ell),*},\alpha}(\tl{f}^{(\ell)},\tl{g}^{(\ell)})(x_\alpha;x_\alpha')\phi(x_\alpha) = \ipp*{\wt{\W}_m^{(\ell),*}\tl{f}^{(\ell)}, \phi\otimes_\alpha \tl{g}^{(\ell)}(\cdot,x_\alpha',\cdot)}_{\Sc'(\R^\ell)-\Sc(\R^\ell)}, \qquad \phi\in\Sc(\R).
\end{split}
\end{equation}
Hence,
\begin{align}
\eqref{eq:Wm_adj_gmp} &= \ol{\ipp*{\wt{\W}_m^{(\ell),*} \paren*{f^{(\ell)} \ol{\Phi_{\W_{n-m}^{(j)},1}(g^{(j)},\ol{f^{(j)}})|_{y=y'}}}, \ol{\phi\otimes g^{(\ell)}(x_1',\cdot)}}}_{\Sc'(\R^\ell)-\Sc(\R^\ell)} \nonumber\\
&=\ol{\int_{\R}dx_1 \Phi_{\wt{\W}_m^{(\ell),*},1}(f^{(\ell)} \ol{\Phi_{\W_{n-m}^{(j)},1}(g^{(j)},\ol{f^{(j)}})|_{y=y'}}, \ol{g^{(\ell)}})(x_1;x_1')\ol{\phi(x_1)}} \nonumber\\
&=\int_{\R}dx_1 \ol{\Phi_{\wt{\W}_m^{(\ell),*},1}(f^{(\ell)} \ol{\Phi_{\W_{n-m}^{(j)},1}(g^{(j)},\ol{f^{(j)}})|_{y=y'}}, \ol{g^{(\ell)}})}(x_1;x_1') \phi(x_1).
\end{align}
Defining the map
\begin{equation}
(f^{(\ell)},f^{(j)},g^{(\ell)},g^{(j)}) \mapsto \ol{\Phi_{\wt{\W}_m^{(\ell),*},1}(f^{(\ell)} \ol{\Phi_{\W_{n-m}^{(j)}}(g^{(j)},\ol{f^{(j)}})|_{y=y'}}, \ol{g^{(\ell)}})}
\end{equation}
yields the desired conclusion, being the composition of continuous maps, antilinear in the $f^{(\ell)},f^{(j)}$ variables, and linear in the $g^{(\ell)},g^{(j)}$ variables.
\end{proof}

Since the base case $\wt{\W}_1^{(k),*}\in \L_{gmp}(\Sc(\R^k),\Sc'(\R^k))$ for every $k\in\N$ is trivial, the lemma and the remarks preceding it imply the \cref{prop:Wn_sa}.

\subsection{Step 3: Bosonic symmetrization} \label{subsec:bosonic_sym}
We now modify the definition of the operators $\W_{n,sa}$ from the previous subsection in order to obtain a bosonic operator which generates the same trace functional as $\W_{n,sa}$ when evaluated on elements of $\G_\infty^*$. As an immediate consequence of \cref{lem:sym_op_space}, we obtain \cref{prop:Wn_con}, completing the main objective of \cref{sec:con_GP_W}. We conclude this subsection by explicitly computing $\W_{3}$  and $\W_{4}$.

\begin{ex}[Computation of $\W_{3}$]
From the recursion \eqref{eq:Wn_recur}, we have that
\begin{align}
\widetilde{\W}_{3}^{(k)} &= (-i\p_{x_{1}})\widetilde{\W}_{2}^{(k)} + \kappa\sum_{\ell+j=k}\delta(X_{1}-X_{\ell+1})\paren*{\widetilde{\W}_{1}^{(\ell)}\otimes \widetilde{\W}_{1}^{(j)}} \nonumber\\
&=
\begin{cases}
(-i\p_{x_{1}})^{2}, & {k=1} \\
\kappa\delta(X_{1}-X_{2}) Id_{2} = \kappa\delta(X_{1}-X_{2}), & {k=2} \\
0_{k}, & {k\geq 3}.
\end{cases}
\end{align}
Since the components $\widetilde{\W}_{3}^{(k)}$ are already self-adjoint and bosonic, it follows that
\begin{equation}
\W_{3}=\widetilde{\W}_{3} = \paren*{(-i\p_{x_{1}})^{2}, \kappa\delta(X_{1}-X_{2}), 0_3,\ldots}.
\end{equation}
\end{ex}

\begin{ex}[Computation of $\W_{4}$]\label{ex:W_4}
Similarly, from the recursion \eqref{eq:Wn_recur}, we have that
\begin{align}
\widetilde{\W}_{4}^{(k)} &= (-i\p_{x_{1}})\widetilde{\W}_{3}^{(k)} + \kappa\sum_{m=1}^{2}\sum_{\ell+j=k}\delta(X_{1}-X_{\ell+1})\paren*{\widetilde{\W}_{m}^{(\ell)}\otimes \widetilde{\W}_{3-m}^{(j)}}.
\end{align}
If $k=1$, then
\begin{equation}
\widetilde{\W}_{4}^{(1)} = (-i\p_{x_{1}})\widetilde{\W}_{3}^{(1)} = (-i\p_{x_{1}})^{3}=\W_{4}^{(1)},
\end{equation}
since $(-i\p_{x_{1}})^{3}$ is self-adjoint and bosonic. If $k=2$, then
\begin{align}
\widetilde{\W}_{4}^{(2)} &= (-i\p_{x_{1}})\widetilde{\W}_{3}^{(2)} + \kappa\delta(X_{1}-X_{2})\paren*{\widetilde{\W}_{1}^{(1)} \otimes \widetilde{\W}_{2}^{(1)}} + \kappa\delta(X_{1}-X_{2})\paren*{\widetilde{\W}_{2}^{(1)}\otimes \widetilde{\W}_{1}^{(1)}} \nonumber\\
&= \kappa\paren*{(-i\p_{x_{1}})\delta(X_{1}-X_{2}) + \delta(X_{1}-X_{2})\paren*{Id_{1}\otimes (-i\p_{x})} + \delta(X_{1}-X_{2})\paren*{(-i\p_{x})\otimes Id_{1}}} \nonumber\\
&=-i\kappa\paren*{\p_{x_{1}}\delta(X_{1}-X_{2})+\delta(X_{1}-X_{2})\paren*{\p_{x_{1}}+\p_{x_{2}}}}.
\end{align}
The term $-i\delta(X_1-X_2)(\p_{x_1}+\p_{x_2})$ is evidently bosonic, and it is self-adjoint since 
\[
\comm{\p_{x_{1}}+\p_{x_{2}}}{\delta(X_{1}-X_{2})}=0.
\]
For the term $-i\p_{x_1}\delta(X_1-X_2)$, \cref{lem:adj_comp} implies that the adjoint is given by $-i\delta(X_{1}-X_{2})\p_{x_{1}}$, and therefore
\begin{align}
&\frac{\kappa}{2}\Sym_2\paren*{(-i\p_{x_{1}})\delta(X_{1}-X_{2})+\delta(X_{1}-X_{2})(-i\p_{x_{1}})} \nonumber\\
 &= \frac{\kappa}{4}\paren*{(-i\p_{x_{1}}-i\p_{x_{2}})\delta(X_{1}-X_{2}) + \delta(X_{1}-X_{2})(-i\p_{x_{1}}-i\p_{x_{2}})} \nonumber\\
&= \frac{\kappa}{2}(-i\p_{x_{1}}-i\p_{x_{2}})\delta(X_{1}-X_{2}),
\end{align}
where we use that $\delta$ is an even distribution and again that $\comm{\p_{x_{1}}+\p_{x_{2}}}{\delta(X_{1}-X_{2})}=0$. We conclude that
\begin{equation}
\W_{4}^{(2)} = \frac{3\kappa}{2}(-i\p_{x_{1}}-i\p_{x_{2}})\delta(X_{1}-X_{2}).
\end{equation}
Finally, it is evident that $\W_{4}^{(k)} = 0_{k}$ for $k\geq 3$.
\end{ex}

\section{The correspondence: $\W_n$ and $w_n$}
\label{sec:cor}
\subsection{Multilinear forms $w_n$}\label{ssec:cor_multi}
In this subsection, we analyze the structure of the nonlinear operators $w_n$ as sums of restricted multilinear forms. For each $k\in\N$, we define a $(2k-1)$-$\C$-linear operator
\begin{equation}
 w_{n}^{(k)}: \Sc(\R)^{k} \times \Sc(\R)^{k-1}\rightarrow \Sc(\R), \quad (\phi_{1},\ldots,\phi_{k}; \psi_{2},\ldots,\psi_{k}) \mapsto w_{n}^{(k)}[\phi_{1},\ldots,\phi_{k};\psi_{2},\ldots,\psi_{k}],
\end{equation}
recursively by
\begin{equation}\label{eq:wn(k)_recur}
\begin{split}
&w_{1}^{(k)}[\phi_{1},\ldots,\phi_k;\psi_2,\ldots,\psi_k]  \coloneqq \phi_{1} \delta_{k1},\\
&w_{n+1}^{(k)}[\phi_{1},\ldots,\phi_{k};\psi_{2},\ldots,\psi_{k}] \\
&= (-i\p_{x})w_{n}^{(k)}[\phi_{1},\ldots,\phi_{k};\psi_{2},\ldots,\psi_{k}] \\
&\phantom{=} \qquad + \kappa\sum_{m=1}^{n-1}\sum_{\ell,j\geq 1;\ell+j=k}\psi_{\ell+1}w_{m}^{(\ell)}[\phi_{1},\ldots,\phi_{\ell};\psi_{2},\ldots,\psi_{\ell}] w_{n-m}^{(j)}[\phi_{\ell+1},\ldots,\phi_{k}; \psi_{\ell+2},\ldots,\psi_{k}],
\end{split}
\end{equation}
where $\delta_{k1}$ denotes the usual Kronecker delta. The next lemma establishes several important structural properties of the $w_n$, including that $w_{n}^{(k)}$ is identically zero for all but finitely many $k\in\N$.

\begin{lemma}[Properties of $w_n^{(k)}$]
\label{lem:wn(k)_prop}
The following properties hold:
\begin{itemize}
\item  For each odd $n\in \N$, $w_n^{(k)}\equiv 0$ for $k>\frac{n+1}{2}$ and for $k \leq \frac{n+1}{2}$ we have
\begin{equation}
\label{eq:w_odd}
w_n^{(k)}[\phi_1,\ldots,\phi_k;\psi_2,\ldots,\psi_k] = \sum_{{ (\ul{\alpha}_k,\ul{\alpha}_{k-1}')\in \N_0^{2k-1}}\atop {|\ul{\alpha}_k| + |\ul{\alpha}_{k-1}'| = n-1-2(k-1)}} a_{n, (\ul{\alpha}_k,\ul{\alpha}_{k-1}')} (\prod_{r=1}^k \p_x^{\alpha_r} \phi_r) (\prod_{r=2}^k \p_x^{\alpha_r'}\psi_r),
\end{equation}
where $a_{n,(\ul{\alpha}_k,\ul{\alpha}_{k-1}')}\in\R$. 
\item For each even $n\in\N$, $w_n^{(k)} \equiv 0$ for $k>\frac{n}{2}$ and for $k \leq \frac{n}{2}$ we have
\begin{equation}
\label{eq:w_even}
w_n^{(k)}[\phi_1,\ldots,\phi_k;\psi_2,\ldots,\psi_k] = i\sum_{{ (\ul{\alpha}_k,\ul{\alpha}_{k-1}')\in \N_0^{2k-1}}\atop {|\ul{\alpha}_k| + |\ul{\alpha}_{k-1}'| = n-1-2(k-1)}} a_{n, (\ul{\alpha}_k,\ul{\alpha}_{k-1}')} (\prod_{r=1}^k \p_x^{\alpha_r} \phi_r) (\prod_{r=2}^k \p_x^{\alpha_r'}\psi_r),
\end{equation}
where $a_{n,(\ul{\alpha}_k,\ul{\alpha}_{k-1}')}\in\R$.
\end{itemize}
\end{lemma}
\begin{proof}
We prove the lemma by strong induction on $n$. We begin with the base case $n=1$. That \eqref{eq:w_odd} holds for $n=1$ is tautological. For the induction step, suppose that there exists some $n\in\N$ such that either \eqref{eq:w_odd} or \eqref{eq:w_even} holds for every odd or even $j\in\N_{\leq n}$, respectively. We consider two cases based on whether $n$ is even or odd.

Consider the even index case. We first show that $w_n^{(k)} \equiv 0$ for $k>\frac{n}{2}$. Since $n-1$ is odd, the induction hypothesis implies that
\begin{equation}
(-i\p_x)w_{n-1}^{(k)} \equiv 0, \qquad k>\frac{n}{2}.
\end{equation}
Now suppose that $\ell, j\in\N$ are such that $\ell+j=k$ and
\begin{equation}
w_m^{(\ell)} \otimes w_{n-1-m}^{(j)} \not\equiv 0,
\end{equation}
where $1\leq m\leq n-2$. By symmetry, it suffices to consider when $m$ is odd and $n-1-m$ is even. By the induction hypothesis,
\begin{equation}
w_m^{(\ell)} \equiv 0, \enspace \ell >\frac{m+1}{2} \qquad \text{ and } \qquad w_{n-1-m}^{(j)} \equiv 0, \enspace j>\frac{n-1-m}{2}.
\end{equation}
Consequently, we must have that
\begin{equation}
k=\ell+j\leq \frac{m+1}{2}+\frac{n-1-m}{2} = \frac{n}{2}.
\end{equation}
It then follows from the recursion \eqref{eq:wn(k)_recur} that $w_n^{(k)} \equiv 0$ for $k>\frac{n}{2}$.

Next we establish the asserted expansion formula. By the induction hypothesis,
\begin{equation}
w_{n-1}^{(k)}[\phi_1,\ldots,\phi_k;\psi_2,\ldots,\psi_k] = \sum_{{(\ul{\alpha}_{k}, \ul{\alpha}_{k-1}')\in\N_0^{2k-1}}\atop {|\ul{\alpha}_{k}| + |\ul{\alpha}_{k-1}'|= n-2-2(k-1)}} a_{n-1,(\ul{\alpha}_{k}, \ul{\alpha}_{k-1}')} (\prod_{r=1}^k \p_x^{\alpha_r} \phi_r) (\prod_{r=2}^k\p_x^{\alpha_r'}\psi_r),
\end{equation}
where the coefficients $a_{n-1,(\ul{\alpha}_k,\ul{\alpha}_{k-1}')}$ are real. Hence by the Leibnitz rule, we can define real coefficients $b_{n,(\ul{\alpha}_k,\ul{\alpha}_{k-1})}$ such that
\begin{equation}
\label{eq:b_sum}
-i\p_x w_{n-1}^{(k)}[\phi_1,\ldots,\phi_k;\psi_2,\ldots,\psi_k] = i\sum_{ {(\ul{\alpha}_k,\ul{\alpha}_{k-1}')\in\N_0^{2k-1}} \atop {|\ul{\alpha}_k| + |\ul{\alpha}_{k-1}'| = n-1-2(k-1)}} b_{n,(\ul{\alpha}_k,\ul{\alpha}_{k-1})}(\prod_{r=1}^k \p_x^{\alpha_r} \phi_r) (\prod_{r=2}^k\p_x^{\alpha_r'}\psi_r).
\end{equation}
Similarly, for $m\in\N_{\leq n-2}$ and $\ell,j\in\N$, the induction hypothesis implies that
\begin{equation}
w_m^{(\ell)}[\phi_1,\ldots,\phi_\ell; \psi_2,\ldots,\psi_\ell] = \begin{cases} \displaystyle\sum_{{(\ul{\alpha}_\ell,\ul{\alpha}_{\ell-1}')\in\N_0^{2\ell-1}}\atop {|\ul{\alpha}_\ell| + |\ul{\alpha}_{\ell-1}'|= m-1-2(\ell-1)}} a_{m,(\ul{\alpha}_\ell,\ul{\alpha}_{\ell-1}')} (\prod_{r=1}^\ell\p_x^{\alpha_r}\phi_r)(\prod_{r=2}^{\ell}\p_x^{\alpha_r'}\psi_r), \enspace \text{$m$ odd},\\
\displaystyle i\sum_{{(\ul{\alpha}_\ell,\ul{\alpha}_{\ell-1}')\in\N_0^{2\ell-1}}\atop {|\ul{\alpha}_\ell| + |\ul{\alpha}_{\ell-1}'|= m-1-2(\ell-1)}} a_{m,(\ul{\alpha}_\ell,\ul{\alpha}_{\ell-1}')} (\prod_{r=1}^\ell\p_x^{\alpha_r}\phi_r)(\prod_{r=2}^{\ell}\p_x^{\alpha_r'}\psi_r), \enspace \text{$m$ even}\end{cases}
\end{equation}
and
\begin{equation}
\begin{split}
&w_{n-1-m}^{(j)}[\phi_{\ell+1,\ldots,\phi_k};\psi_{\ell+2},\ldots,\psi_{k}]\\
&=\begin{cases} \displaystyle i\sum_{{(\ul{\alpha}_j,\ul{\alpha}_{j-1}')\in\N_0^{2j-1}}\atop {|\ul{\alpha}_j| + |\ul{\alpha}_{j-1}'|= n-2-m-2(j-1)}} a_{n-1-m,(\ul{\alpha}_j,\ul{\alpha}_{j-1}')} (\prod_{r=\ell+1}^{k}\p_x^{\alpha_r}\phi_r)(\prod_{r=\ell+2}^{k}\p_x^{\alpha_r'}\psi_r), \enspace \text{$m$ odd}\\
\displaystyle\sum_{{(\ul{\alpha}_j,\ul{\alpha}_{j-1}')\in\N_0^{2j-1}}\atop {|\ul{\alpha}_j| + |\ul{\alpha}_{j-1}'|= n-2-m-2(j-1)}} a_{n-1-m,(\ul{\alpha}_j,\ul{\alpha}_{j-1}')} (\prod_{r=\ell+1}^{k}\p_x^{\alpha_r}\phi_r)(\prod_{r=\ell+2}^{k}\p_x^{\alpha_r'}\psi_r), \enspace \text{$m$ even}\end{cases},
\end{split}
\end{equation}
where $a_{n-1-m,(\ul{\alpha}_\ell,\ul{\alpha}_{\ell-1}')}, a_{n-1-m,(\ul{\alpha}_j,\ul{\alpha}_{j-1}')}\in\R$. For $\ell+j=k$ and $(\ul{\alpha}_\ell,\ul{\alpha}_{\ell-1}'), (\ul{\alpha}_j,\ul{\alpha}_{j-1}')$ as in the summations above, the multi-index
\begin{equation*}
(\ul{\alpha}_{\ell},\ul{\alpha}_j, \ul{\alpha}_{\ell-1}',\ul{\alpha}_{j-1}') \in \N_0^{2k-2}
\end{equation*}
satisfies
\begin{equation}
|(\ul{\alpha}_\ell,\ul{\alpha}_j)| + |(\ul{\alpha}_{\ell-1}',\ul{\alpha}_{j-1}')| = m-1-2(\ell-1) + n-2-m -2(j-1) = n-1-2(k-1).
\end{equation}
Consequently, we can define real coefficients $c_{n,(\ul{\alpha}_k,\ul{\alpha}_{k-1}')}$ such that
\begin{equation}
\label{eq:c_sum}
\begin{split}
&\sum_{m=1}^{n-1} \psi_{\ell+1} w_m^{(\ell)}[\phi_1,\ldots,\phi_\ell;\psi_2,\ldots,\psi_j] w_{n-1-m}^{(j)}[\phi_{\ell+1},\ldots,\phi_k;\psi_{\ell+2},\ldots,\psi_{k}] \\
&=i\sum_{{(\ul{\alpha}_k,\ul{\alpha}_{k-1}')\in\N_0^{2k-1}}\atop {|\ul{\alpha}_k| + |\ul{\alpha}_{k-1}'|= n-1-2(k-1)}} c_{n,(\ul{\alpha}_k,\ul{\alpha}_{k-1}')} (\prod_{r=1}^k\p_{x}^{\alpha_r}\phi_r)(\prod_{r=2}^k \p_{x}^{\alpha_r'}\psi_r).
\end{split}
\end{equation}
Defining
\begin{equation}
a_{n,(\ul{\alpha}_k,\ul{\alpha}_{k-1}')} \coloneqq b_{n,(\ul{\alpha}_k,\ul{\alpha}_{k-1}')} + c_{n,(\ul{\alpha}_k,\ul{\alpha}_{k-1}')},
\end{equation}
and summing \eqref{eq:b_sum} and \eqref{eq:c_sum} shows that \eqref{eq:w_even} holds.

Next, consider the odd index case. To establish that $w_n^{(k)}\equiv 0$ for $k>\frac{n+1}{2}$, we have by our previous discussion in the even case, that
\begin{equation}
-i\p_x w_{n-1}^{(k)} = 0, \qquad k>\frac{n-1}{2}.
\end{equation}
Suppose that $\ell,j\in\N$ are such that $\ell+j=k$ and
\begin{equation}
w_m^{(\ell)}\otimes w_{n-1-m}^{(j)}\not\equiv 0,
\end{equation}
where $1\leq m\leq n-2$. If $m$ is odd, then $n-1-m$ is odd, and so by the induction hypothesis,
\begin{equation}
w_m^{(\ell)}\equiv 0, \enspace \ell>\frac{m+1}{2} \qquad \text{ and } \qquad w_{n-1-m}^{(j)}\equiv 0, \enspace j>\frac{n-m}{2}.
\end{equation}
Consequently, we must have that
\begin{equation}
k=\ell+j\leq \frac{m+1}{2} + \frac{n-m}{2} = \frac{n+1}{2}.
\end{equation}
Similarly, if $m$ is even, then $n-1-m$ is even, and so by the induction hypothesis
\begin{equation}
w_m^{(\ell)}\equiv 0, \enspace \ell>\frac{m}{2} \qquad \text{ and } \qquad w_{n-1-m}^{(j)}\equiv 0, \enspace j>\frac{n-m-1}{2}.
\end{equation}
Consequently, we must have that
\begin{equation}
k=\ell+j\leq \frac{m}{2}+\frac{n-m-1}{2} = \frac{n-1}{2}.
\end{equation}
It now follows from the recursion \eqref{eq:wn(k)_recur} that $w_n^{(k)}\equiv 0$ for $k>\frac{n+1}{2}$. Repeating the proof mutatis mutandis from the $n$ even case, we see that $w_n^{(k)}$ has the representation \eqref{eq:w_odd}. Thus, the proof of the induction step is complete.
\end{proof}

We establish now some notation we will use here and in the sequel. For $k,n\in\N$, we define densities
\begin{equation}\label{eq:Pn(k)_def}
P_{n}^{(k)}[\phi_{1},\ldots,\phi_{k}; \psi_{1},\ldots,\psi_{k}] \coloneqq \psi_{1}w_{n}^{(k)}[\phi_{1},\ldots,\phi_{k}; \psi_{2},\ldots,\psi_{k}]\in\Sc(\R),
\end{equation}
and we define
\begin{equation}
\label{eq:In(k)_def}
I_n^{(k)}[\phi_1,\ldots,\phi_k;\psi_1,\ldots,\psi_k] \coloneqq \int_{\R}dx P_n^{(k)}[\phi_1,\ldots,\phi_k;\psi_1,\ldots,\psi_k](x).
\end{equation}
It is clear from \cref{lem:wn(k)_prop}, that $P_n^{(k)}:\Sc(\R)^{2k}\rightarrow\Sc(\R)$ is a $2k$-$\C$-linear, continuous map, and thus $I_n^{(k)}: \Sc(\R)^{2k}\rightarrow\C$ is a $2k$-$\C$-linear, continuous map. For $k \in \N$, we recall the notation $\phi^{\times k}$ from \eqref{prod_coords} to denote the measurable function $\phi^{\times k}:\R^{m}\rightarrow\C^k$
\begin{equation}
\phi^{\times k}(\ux_m) \coloneqq (\phi(\ux_m),\ldots,\phi(\ux_m)),
\end{equation}
and similarly for $\psi^{\times k}$. 
\begin{remark}
\label{rem:I_n/w_n}
It is clear from the recursion \eqref{eq:wn(k)_recur} that 
\begin{equation}
I_n(\phi) = \sum_{k=1}^\infty I_n^{(k)}[\phi^{\times k};\ol{\phi}^{\times k}], \qquad \forall \phi\in\Sc(\R),
\end{equation}
where $I_n$ is as defined in \eqref{eq:In_def}.
\end{remark}

\cref{rem:I_n/w_n} and the structure result \cref{lem:wn(k)_prop} allow us to give a proof of the seemingly obvious fact that the functionals $I_n$ are not constant on $\Sc(\R)$. We obtain this fact as a consequence of a more general lemma. Note that since $I_n(0)=0$, the nonconstancy of $I_n$ is equivalent to $I_n\not\equiv 0$.

\begin{lemma}
\label{lem:In_noncon}
Let $n\in\N$, and let $\ul{c} = \{c_k\}_{k\in\N}\subset\C$ such that $c_1\neq 0$. Define the map
\begin{equation}
I_{n,\ul{c}}:\Sc(\R)\rightarrow\C,\qquad I_{n,\ul{c}}(\phi)\coloneqq \sum_{k=1}^\infty c_k I_n^{(k)}[\phi^{\times k};\ol{\phi}^{\times k}], \qquad \forall\phi\in\Sc(\R).
\end{equation}
Then $I_{n,\ul{c}}\not\equiv 0$.
\end{lemma}
\begin{proof}
Assume the contrary. Then for any $\lambda\in\C$, we find from the $2k$-complex linearity of the functionals $I_n^{(k)}$ that
\begin{equation}
0 = I_{n,\ul{c}}(\lambda\phi) = \sum_{k=1}^\infty c_k I_n^{(k)}[(\lambda\phi)^{\times k}; \ol{(\lambda\phi)}^{\times k}] = \sum_{k=1}^\infty c_k|\lambda|^{2k} I_n^{(k)}[\phi^{\times k};\ol{\phi}^{\times k}], \qquad \forall \phi\in\Sc(\R).
\end{equation}
Now fix $\phi\in\Sc(\R)$ and define a function
\begin{equation}
\rho_{\phi,\ul{c}}:\C\rightarrow\C, \qquad \rho_{\phi,\ul{c}}(\lambda) \coloneqq \sum_{k=1}^\infty c_k|\lambda|^{2k} I_n^{(k)}[\phi^{\times k};\ol{\phi}^{\times k}],
\end{equation}
which is well-defined and smooth since $I_n^{(k)}\equiv 0$ for all but finitely many indices $k$. Now observe that
\begin{equation}
0 = (\p_\lambda\p_{\ol{\lambda}} \rho_{\phi,\ul{c}})(0) = c_1 I_n^{(1)}[\phi;\ol{\phi}] = c_1\int_{\R}dx \ \ol{\phi}(x) (-i\p_x)^{n-1}\phi(x).
\end{equation}
Choosing $\phi\in\Sc(\R)$ to be a function whose Fourier transform $\wh{\phi}$ satisfies $0\leq \wh{\phi}\leq 1$,
\begin{equation}
\wh{\phi}(\xi) =
\begin{cases}
1, & {2\leq \xi \leq 3}\\
0, & {\xi\leq 1, \ \xi\geq 4}
\end{cases},
\end{equation}
we obtain a contradiction from Plancherel's theorem, since $c_1\neq 0$ by assumption.
\end{proof}

\subsection{Variational derivatives}\label{ssec:cor_symgrad}
In this subsection, we show that the functionals $I_n$ satisfy the conditions of \cref{rem:pb_vd} and explicitly compute their symplectic gradients. To this end, we record here a recursive formula for the functions $w_{n, (\psi_1,\psi_2)}$, which generalizes the recursive formula \cref{eq:w_rec} for $w_n$, given by
\begin{equation}
\label{eq:wn_bv_rec}
\begin{split}
w_{1,(\psi_1,\psi_2)}(x) &= \psi_1(x) \\
w_{n+1,(\psi_1,\psi_2)}(x) &= -i\p_x w_{n,(\psi_1,\psi_2)}(x) + \kappa\psi_2(x)\sum_{m=1}^{n-1}w_{m,(\psi_1,\psi_2)}(x)w_{n-m,(\psi_1,\psi_2)}(x)
\end{split},
\end{equation}
and we refer to \eqref{eq:wn(,)_def} for more details. We define $\tl{I}_n:\Sc(\R)^2 \rightarrow\C$ by
\begin{equation}
\label{eq:tlIn_def}
\tl{I}_n(\psi_1,\psi_2) \coloneqq \int_{\R}dx \psi_2(x) w_{n,(\psi_1,\psi_2)}(x), \qquad \forall (\psi_1,\psi_2)\in\Sc(\R)^2.
\end{equation}

\begin{remark}
\label{rem:rec_comp}
By comparing the recursion \eqref{eq:wn_bv_rec} to the recursion \eqref{eq:wn(k)_recur}, we see that
\begin{equation}
w_{n,(\psi_1,\psi_2)} = \sum_{k=1}^\infty w_n^{(k)}[\psi_1^{\times k}; \psi_2^{\times (k-1)}]
\end{equation}
and consequently
\begin{equation}
\tl{I}_n(\psi_1,\psi_2) = \sum_{k=1}^\infty I_n^{(k)}[\psi_1^{\times k};\psi_2^{\times k}].
\end{equation}
\end{remark}

We now use the multilinear $w_n^{(k)}$ introduced in the previous subsection in order to compute the variational derivatives, defined in \eqref{eq:vd_prop}, of the functions $\tl{I}_n$. We first dispense with a technical lemma asserting the existence of a partial transpose for the $w_n^{(k)}$ in $C^\infty(\Sc(\R)^{2k-1};\Sc(\R))$. The proof follows from the structural formula of \cref{lem:wn(k)_prop} and integration by parts; we leave the details to the reader.

\begin{lemma}
\label{lem:wn(k)_tran}
Let $n,k\in\N$. Then for $1\leq j\leq k$, there exists a unique partial transpose $w_{n,j}^{(k),t} \in C^\infty(\Sc(\R)^{2k-1};\Sc(\R))$, such that for every $\delta\phi\in\Sc(\R)$ and $\phi_1, \ldots, \phi_k, \psi_2 , \ldots \psi_k \in \Sc(\R)$ we have
\begin{equation}
\begin{split}
&\int_{\R} dx \delta\phi(x) w_{n,j}^{(k),t}[\phi_1,\ldots,\phi_k; \psi_2,\ldots,\psi_k](x)\\
&=\int_{\R}dx \phi_j(x) w_n^{(k)}[\phi_1,\ldots,\phi_{j-1},\delta\phi,\phi_{j+1},\ldots,\phi_k; \psi_2,\ldots,\psi_k](x) ,
\end{split}
\end{equation}
Similarly, for $2\leq j\leq k$, there exists a unique partial transpose $w_{n,j'}^{(k),t}\in C^\infty(\Sc(\R)^{2k-1};\Sc(\R))$, such that for every $\delta\psi\in\Sc(\R)$ and $\phi_1, \ldots, \phi_k, \psi_2 , \ldots \psi_k \in \Sc(\R)$ we have
\begin{equation}
\begin{split}
&\int_{\R}dx \delta\psi(x) w_{n,j'}^{(k),t}[\phi_1,\ldots,\phi_k; \psi_2,\ldots,\psi_k](x)\\
&=\int_{\R}dx \psi_j(x) w_n^{(k)}[\phi_1,\ldots,\phi_k; \psi_2,\ldots,\psi_{j-1},\delta\psi,\psi_{j+1},\ldots,\psi_k](x) .
\end{split}
\end{equation}
For convenience of notation, we define $w_{n,1'}^{(k),t}\in C^\infty(\Sc(\R)^{2k-1};\Sc(\R))$ by
\begin{equation}
w_{n,1'}^{(k),t}[\phi_1,\ldots,\phi_k;\psi_2,\ldots,\psi_k]\coloneqq w_n^{(k)}[\phi_1,\ldots,\phi_k;\psi_2,\ldots,\psi_k].
\end{equation}
\end{lemma}

We may now proceed to establish formulae for the variational derivatives of the $\tl{I}_n$.

\begin{lemma}
\label{lem:tlI_grad_sym}
For $n\in\N$, we have that
\begin{align}
\grad_1\tl{I}_n(\phi,\psi) &= \sum_{k=1}^\infty\sum_{j=1}^k w_{n,j}^{(k),t}[\phi^{\times(j-1)},\psi,\phi^{\times(k-j)};\psi^{\times(k-1)}], \\
 \grad_{\bar{2}}\tl{I}_n(\phi,\psi) &= \sum_{k=1}^\infty \sum_{j=1}^k w_{n,j'}^{(k),t}[\phi^{\times k};\psi^{\times(k-1)}]
\end{align}
for every $(\phi,\psi)\in\Sc(\R)^2$. In particular,
\begin{equation}\label{symp_grad_tli}
\begin{split}
\grad_s I_n(\phi) &= -i \sum_{k=1}^\infty\sum_{j=1}^k \ol{w_{n,j}^{(k),t}[\phi^{\times(j-1)},\ol{\phi},\phi^{\times(k-j)};\ol{\phi}^{\times(k-1)}]} \\
&=-i\sum_{k=1}^\infty \sum_{j=1}^ k w_{n,j'}^{(k),t}[\phi^{\times k};\ol{\phi}^{\times(k-1)}] \\
&=-\frac{i}{2}\sum_{k=1}^\infty\sum_{j=1}^k \paren*{\ol{w_{n,j}^{(k),t}[\phi^{\times(j-1)},\ol{\phi},\phi^{\times(k-j)};\ol{\phi}^{\times(k-1)}]} + w_{n,j'}^{(k),t}[\phi^{\times k};\ol{\phi}^{\times(k-1)}]}.
\end{split}
\end{equation}

\end{lemma}
\begin{proof}
Fix a point $(\phi,\psi)\in\Sc(\R)^2$. Unpacking the definition of $\tl{I}_n$ and using the chain rule for the G\^ateaux derivative, we obtain that
\begin{equation}
\label{eq:tlI_gd}
\begin{split}
d\tl{I}_n[\phi,\psi](\delta\phi,\delta\psi) &= \sum_{k=1}^\infty \sum_{j=1}^k  \biggl( \int_{\R}dx P_n^{(k)}[\phi^{\times (j-1)},\delta\phi,\phi^{\times (k-j)};\psi^{\times k}](x) \\
&\phantom{=} \hspace{25mm} +\int_{\R}dx P_n^{(k)}[\phi^{\times k};\psi^{\times(j-1)},\delta\psi,\psi^{\times (k-j)}](x) \biggr).
\end{split}
\end{equation}
Since
\begin{equation}
P_n^{(k)}[\phi^{\times (j-1)},\delta\phi,\phi^{\times(k-j)};\psi^{\times k}] = \psi w_n^{(k)}[\phi^{\times(j-1)},\delta\phi,\phi^{\times(k-j)};\psi^{\times(k-1)}]
\end{equation}
and
\begin{equation}
P_n^{(k)}[\phi^{\times k};\psi^{\times(j-1)},\delta\psi,\psi^{\times(k-j)}] = \begin{cases} \delta\psi w_n^{(k)}[\phi^{\times k};\psi^{\times(k-1)}], & {j=1} \\
\psi w_n^{(k)}[\phi^{\times k};\psi^{\times(j-2)},\delta\psi,\psi^{\times(k-j)}], & {2\leq j\leq k}
\end{cases},
\end{equation}
upon application of \cref{lem:wn(k)_tran}, we see that
\begin{equation}
\label{eq:Pn(k)_phi_tran}
\begin{split}
&\int_{\R}dx P_n^{(k)}[\phi^{\times(j-1)},\delta\phi,\phi^{\times(k-j)};\psi^{\times k}](x)\\
&=\int_{\R}dx \delta\phi(x) w_{n,j}^{(k),^t}[\phi^{\times(j-1)},\psi,\phi^{\times(k-j)};\psi^{\times(k-1)}](x)
\end{split}
\end{equation}
and
\begin{equation}
\label{eq:Pn(k)_psi_tran}
\begin{split}
\int_{\R}dx P_n^{(k)}[\phi^{\times k};\psi^{\times(j-1)},\delta\psi,\psi^{\times(k-j)}](x) &= \int_{\R}dx \delta\psi(x) w_{n,j'}^{(k),^t}[\phi^{\times k};\psi^{\times(k-1)}](x).
\end{split}
\end{equation}
Substituting \eqref{eq:Pn(k)_phi_tran} and \eqref{eq:Pn(k)_psi_tran} into \eqref{eq:tlI_gd}, we arrive at the identity
\begin{equation}
\begin{split}
d\tl{I}_n[\phi,\psi](\delta\phi,\delta\psi) &= \sum_{k=1}^\infty \sum_{j=1}^k \biggl( \int_{\R}dx \delta\phi(x) w_{n,j}^{(k),t}[\phi^{\times(j-1)},\psi,\phi^{\times(k-j)};\psi^{\times(k-1)}](x) \\
&\phantom{=} \hspace{20mm} + \int_{\R}dx \delta\psi(x) w_{n,j'}^{(k),t}[\phi^{\times k};\psi^{\times(k-1)}](x) \biggr).
\end{split}
\end{equation}
Using that there are only finitely many indices $k$ yielding a nonzero contribution by \cref{lem:wn(k)_prop}, we can move the summations inside the integral to conclude that
\begin{equation}\label{symp_grad_in}
\begin{split}
d\tl{I}_n[\phi,\psi](\delta\phi,\delta\psi) &= \int_{\R}dx \delta\phi(x) \paren*{\sum_{k=1}^\infty\sum_{j=1}^k w_{n,j}^{(k),t}[\phi^{\times(j-1)},\psi,\phi^{\times(k-j)};\psi^{\times(k-1)}](x)} \\
&\phantom{=}+ \int_{\R}dx\delta\psi(x)\paren*{\sum_{k=1}^\infty \sum_{j=1}^k w_{n,j'}^{(k),t}[\phi^{\times k};\psi^{\times(k-1)}](x)},
\end{split}
\end{equation}
which yields the desired formula for the variational derivatives in light of \eqref{eq:vd_prop}.

To see the second assertion for the symplectic gradient $\grad_s I_n(\phi)$, we recall that from the fact that $I_n(\phi) = \tl{I}_n(\phi,\ol{\phi})$, \cref{rem:pb_vd}, and \eqref{eq:sgrad_F} that we have the the identity
\begin{equation*}
\grad_s I_n(\phi) = -i\ol{\grad_1\tl{I}_n(\phi,\ol{\phi})} = -i\grad_{\bar{2}}\tl{I}_n(\phi,\ol{\phi}).
\end{equation*}
Substituting the identities for $\grad_{1}\tl{I}_n(\phi,\ol{\phi}), \grad_{\bar{2}}\tl{I}_n(\phi,\ol{\phi})$ into the right-hand side of the previous equality completes the proof.
\end{proof}

\subsection{Partial trace connection of $\W_n$ to $w_n$}\label{ssec:cor_ptr}
We next connect the linear DVOs $\wt{\W}_n^{(k)}$ constructed in \cref{sec:con_GP_W} to the multilinear Schwartz-valued operators $w_n^{(k)}$ constructed in \cref{ssec:cor_multi}. We note that since the definition of the $\W_n$ is fairly straightforward given the definition of $\wt{\W}_n$, it will suffice to establish these connections for the latter operators.

It will be important to remember the following consequence of the fact that $\wt{\W}_n^{(k)}$ satisfies the good mapping property: the generalized partial trace
\begin{equation}
\Tr_{2,\ldots,k}\paren*{\wt{\W}_n^{(k)} \ket*{\bigotimes_{\ell=1}^k\phi_\ell}\bra*{\bigotimes_{\ell=1}^k \psi_\ell}},
\end{equation}
which is a priori the element of $\L(\Sc(\R),\Sc'(\R))$ given by the property
\begin{align}
&\ipp*{\Tr_{2,\ldots,k}\paren*{\wt{\W}_n^{(k)} \ket*{\bigotimes_{\ell=1}^k\phi_\ell}\bra*{\bigotimes_{\ell=1}^k \psi_\ell}}\phi,\psi}_{\Sc'(\R)-\Sc(\R)} \nonumber\\
&= \ipp*{\wt{\W}_n^{(k)}\bigotimes_{\ell=1}^k\phi_\ell, \psi\otimes\ipp*{\bigotimes_{\ell=1}^k\ol{\psi_\ell}, \phi}_{\Sc_{x_1}'(\R)-\Sc_{x_1}(\R)} }_{\Sc'(\R^k)-\Sc(\R^k)} \nonumber\\
&=\ip{\psi_1}{\phi}\ipp*{\wt{\W}_n^{(k)}\bigotimes_{\ell=1}^k\phi_\ell, \psi \otimes \bigotimes_{\ell=2}^k \ol{\psi_\ell}}_{\Sc'(\R^k)-\Sc(\R^k)},
\end{align}
for every $\phi,\psi\in\Sc(\R)$, is in fact uniquely identifiable with the element in $\Sc(\R^2)$ which we denote by
\begin{equation*}
\Phi_{\wt{\W}_n^{(k)}}(\phi_1,\ldots,\phi_k;\ol{\psi_1},\ldots,\ol{\psi_k})
\end{equation*}
via
\begin{equation}
\label{eq:Phi_Wn}
\begin{split}
&\ipp*{\Tr_{2,\ldots,k}\paren*{\wt{\W}_n^{(k)} \ket*{\bigotimes_{\ell=1}^k\phi_\ell}\bra*{\bigotimes_{\ell=1}^k \psi_\ell}} f, g}_{\Sc'(\R)-\Sc(\R)} \\
&= \int_{\R^2}dx dx' \Phi_{\wt{\W}_n^{(k)}}(\phi_1,\ldots,\phi_k;\ol{\psi_1},\ldots,\ol{\psi_k})(x;x')f(x')g(x).
\end{split}
\end{equation}
Moreover, the map
\begin{equation}
\Sc(\R)^{2k} \rightarrow \Sc(\R^2), \qquad (\phi_1,\ldots,\phi_k,\psi_1,\ldots,\psi_k) \mapsto \Phi_{\wt{\W}_n^{(k)}}(\phi_1,\ldots,\phi_k;\ol{\psi_1},\ldots,\ol{\psi_k})
\end{equation}
is continuous. The objective of the next lemma is to obtain a formula for $\Phi_{\wt{\W}_n^{(k)}}$ in terms of $w_n^{(k)}$. 

\begin{lemma}
\label{lem:Wn_wn_ptr}
Let $k,n\in\N$. Then the following properties hold:
\begin{itemize}
\item For any $\pi\in\Ss_k$ with $\pi(1) = 1$, we have that for all $(x,x')\in\R^2$,
\begin{equation}
\begin{split}
&\Phi_{\wt{\W}_{n,(\pi(1),\ldots,\pi(k))}^{(k)}}(\phi_1,\ldots,\phi_k;\ol{\psi_1},\ldots,\ol{\psi_k})(x;x') \\
&\hspace{14mm} =\ol{\psi_1(x')} w_n^{(k)}[\phi_{\pi(1)},\ldots,\phi_{\pi(k)}; \ol{\psi_{\pi(2)}},\ldots,\ol{\psi_{\pi(k)}}](x),
\end{split}
\end{equation}
and 
\begin{equation}
\begin{split}
&\Phi_{\wt{\W}_{n,(\pi(1),\ldots,\pi(k))}^{(k),*}}(\phi_1,\ldots,\phi_k;\ol{\psi_1},\ldots,\ol{\psi_k})(x;x') \\
&\hspace{14mm}= \ol{\psi_1(x') w_{n,1}^{(k),t}[\ol{\phi_{1}},\psi_{\pi(2)},\ldots,\psi_{\pi(k)};\ol{\phi_{\pi(2)}},\ldots,\ol{\phi_{\pi(k)}}](x)}.
\end{split}
\end{equation}
\item For any $\pi\in\Ss_k$ with $\pi(1) \neq 1$, we have that for all $(x,x')\in\R^2$,
\begin{equation}
\begin{split}
&\Phi_{\wt{\W}_{n,(\pi(1),\ldots,\pi(k))}^{(k)}}(\phi_1,\ldots,\phi_k;\ol{\psi_1},\ldots,\ol{\psi_k})(x;x') \\
&=\ol{\psi_1(x')} w_{n,\pi^{-1}(1)'}^{(k),t}[\phi_{\pi(1)},\ldots,\phi_{\pi(k)};\ol{\psi_{\pi(2)}},\ldots,\ol{\psi_{\pi(\pi^{-1}(1)-1)}},\ol{\psi_{\pi(1)}},\ol{\psi_{\pi(\pi^{-1}(1)+1)}},\ldots,\ol{\psi_{\pi(k)}}](x),
\end{split}
\end{equation}
and
\begin{equation}
\begin{split}
&\Phi_{\wt{\W}_{n,(\pi(1),\ldots,\pi(k))}^{(k),*}}(\phi_1,\ldots,\phi_k;\ol{\psi_1},\ldots,\ol{\psi_k})(x;x') \\
&= \ol{\psi_1(x') w_{n,\pi^{-1}(1)}^{(k),t}[\psi_{\pi(1)},\ldots,\psi_{\pi(\pi^{-1}(1)-1)},\ol{\phi_{\pi(1)}},\psi_{\pi(\pi^{-1}(1)+1)},\ldots,\psi_{\pi(k)};\ol{\phi_{\pi(2)}},\ldots,\ol{\phi_{\pi(k)}}](x)},
\end{split}
\end{equation}
\end{itemize}
\end{lemma}

\begin{proof}
We will begin by establishing the first claim for the identity permutation, that is, for each $k\in\N$ and for any $\phi_{1},\ldots,\phi_{k},\psi_{1},\ldots,\psi_{k}\in\Sc(\R)$, we have that
\begin{equation}\label{equ:whattoprove}
\begin{split}
&\Phi_{\wt{\W}_n^{(k)}}(\phi_1,\ldots,\phi_k;\ol{\psi_1},\ldots,\ol{\psi_k})(x;x')\\
&=\ol{\psi_1(x')} w_n^{(k)}[\phi_1,\ldots,\phi_k;\ol{\psi_2},\ldots,\ol{\psi_k}](x), \qquad \forall (x,x')\in\R^2.
\end{split}
\end{equation}
By Lemma \ref{lem:aa_zero}, it suffices to prove \eqref{equ:whattoprove} by induction on 
\begin{equation}
\{(k,n)\in\N^{2} : k\leq n\}.
\end{equation}
We begin with the base case, $(k,n)=(1,1)$. Since $\wt{\W}_{1}^{(1)} = Id_{1}\in\L(\Sc(\R),\Sc'(\R))$, we have the Schwartz kernel identity
\begin{equation}
\paren*{\widetilde{\W}_{1}^{(1)}\ket*{\phi_{1}}\bra*{\psi_{1}}}(x_{1};x_{1}') = \phi_{1}(x_{1})\overline{\psi_{1}(x_{1}')} =\ol{\psi_{1}(x_{1}')}w_{1}^{(1)}[\phi_{1}](x_{1}), \qquad \forall (x_{1},x_{1}')\in\R^{2},
\end{equation}
which proves \eqref{equ:whattoprove} for the base case.

For the induction step, suppose that there is some $n\in\N$ such that for all integers $j\in\N_{\leq n}$ the following assertion holds: for all integers $k\in\N_{\leq j}$ and for all functions $\phi_{1},\ldots,\phi_{k},\psi_{1},\ldots,\psi_{k}\in\Sc(\R)$, we have that
\begin{equation}
\label{eq:ptr_con_IH}
\Phi_{\wt{\W}_j^{(k)}}(\phi_1,\ldots,\phi_k; \ol{\psi_1},\ldots,\ol{\psi_k})(x;x') = \ol{\psi_1}(x') w_j^{(k)}[\phi_{1},\ldots,\phi_{k};\ol{\psi}_{2},\ldots,\ol{\psi}_k](x), \qquad \forall (x,x')\in\R^2.
\end{equation}
We now prove \eqref{eq:ptr_con_IH} with $j=n+1$. By the recursion relation \eqref{eq:Wn_recur} and the bilinearity of the generalized trace, we have that
\begin{align}
&\Tr_{2,\ldots,k}\paren*{\widetilde{\W}_{n+1}^{(k)}\ket*{\bigotimes_{\ell=1}^k \phi_{\ell}}\bra*{\bigotimes_{\ell=1}^k \psi_{\ell}}} \nonumber\\
&\phantom{=} = \Tr_{2,\ldots,k}\paren*{(-i\p_{x_{1}})\widetilde{\W}_{n}^{(k)}\ket*{\bigotimes_{r=1}^k \phi_{r}}\bra*{\bigotimes_{r=1}^k \psi_{r}}} \nonumber\\
&\phantom{=}\qquad + \kappa\sum_{m=1}^{n-1}\sum_{\ell+j=k}\Tr_{2,\ldots,k}\paren*{\delta(X_{1}-X_{\ell+1})\paren*{\widetilde{\W}_{m}^{(\ell)} \otimes \widetilde{\W}_{n-m}^{(j)}}\ket*{\bigotimes_{r=1}^k \phi_{r}}\bra*{\bigotimes_{r=1}^k \psi_{r}}} \nonumber\\
&\phantom{=} \eqqcolon \mathrm{Term}_{1,k}+\mathrm{Term}_{2,k}. \label{twoterms}
\end{align}
We first analyze $\mathrm{Term}_{1,k}$. Since $(-i\p_{x_{1}})\wt{\W}_{n}^{(k)}\in \L(\Sc(\R^{k}),\Sc'(\R^{k}))$, it follows from the definition of the generalized partial trace that
\begin{align}
\label{eq:Wn_T1_sub}
\mathrm{Term}_{1,k} &= (-i\p_{x})\Tr_{2,\ldots,k}\paren*{\wt{\W}_n^{(k)}\ket*{\bigotimes_{r=1}^k\phi_r}\bra*{\bigotimes_{r=1}^k\psi_r}}.
\end{align}
It follows from the induction hypothesis that
\begin{align}
&(-i\p_x)\Tr_{2,\ldots,k}\paren*{\wt{\W}_n^{(k)}\ket*{\bigotimes_{r=1}^k\phi_r}\bra*{\bigotimes_{r=1}^k\psi_r}}(x;x') \nonumber\\
&=\paren*{(-i\p_x)\Phi_{\wt{\W}_n^{(k)}}(\phi_1,\ldots,\phi_k;\ol{\psi_1},\ldots,\ol{\psi_k})}(x;x') \nonumber\\
&=\ol{\psi_1(x')} (-i\p_x)w_n^{(k)}[\phi_1,\ldots,\phi_k;\ol{\psi_2},\ldots,\ol{\psi_k}](x) \label{eq:Wn_T1_RHS}
\end{align}
with equality in the sense of tempered distributions on $\R^2$. Substituting \eqref{eq:Wn_T1_RHS} into \eqref{eq:Wn_T1_sub}, we obtain that
\begin{equation}
\mathrm{Term}_{1,k} = \ol{\psi_1(x')} (-i\p_x)w_n^{(k)}[\phi_1,\ldots,\phi_k;\ol{\psi_2},\ldots,\ol{\psi_k}](x).
\end{equation}

We next analyze $\mathrm{Term}_{2,k}$. By the computed action of the H\"{o}rmander product $\delta(X_{1}-X_{\ell+1})\paren*{\wt{\W}_{m}^{(\ell)}\otimes\wt{\W}_{n-m}^{(j)}}$ given by \eqref{eq:map_func} and the definition of $\Phi_{\wt{\W}_m^{(\ell)}}$ and $\Phi_{\wt{\W}_{n-m}^{(j)}}$ we have that
\begin{align}
&\Tr_{2,\ldots,k}\paren*{\delta(X_{1}-X_{\ell+1})\paren*{\widetilde{\W}_{m}^{(\ell)} \otimes \widetilde{\W}_{n-m}^{(j)}}\ket*{\bigotimes_{r=1}^k \phi_{r}}\bra*{\bigotimes_{r=1}^k \psi_{r}}}(x;x') \nonumber\\
&= \Phi_{\wt{\W}_m^{(\ell)}}(\phi_1,\ldots,\phi_\ell;\ol{\psi_1},\ldots,\ol{\psi_\ell})(x;x')\Phi_{\wt{\W}_{n-m}^{(j)}}(\phi_{\ell+1},\ldots,\phi_k;\ol{\psi_{\ell+1}},\ldots,\ol{\psi_k})(x;x) \label{eq:Wnm_sub}
\end{align}
in the sense of tempered distributions. Using the induction hypothesis for $\wt{\W}_{m}^{(\ell)}$ and $\wt{\W}_{n-m}^{(j)}$, respectively, we also have that 
\begin{equation}
\begin{split}
&\Phi_{\wt{\W}_m^{(\ell)}}(\phi_1,\ldots,\phi_\ell;\ol{\psi_1},\ldots,\ol{\psi_\ell})(x;x') \\
&=\ol{\psi_1}(x') w_m^{(\ell)}[\phi_1,\ldots,\phi_\ell;\ol{\psi_2},\ldots,\ol{\psi_\ell}](x), \qquad \forall (x,x')\in\R^2
\end{split}
\end{equation}
and
\begin{equation}
\begin{split}
&\Phi_{\wt{\W}_{n-m}^{(j)}}(\phi_{\ell+1},\ldots,\phi_k; \ol{\psi_{\ell+1}},\ldots,\ol{\psi_k}](x;x')\\
&=\ol{\psi_{\ell+1}}(x') w_{n-m}^{(j)}[\phi_{\ell+1},\ldots,\phi_k; \ol{\psi_{\ell+2}},\ldots,\ol{\psi_k}](x), \qquad \forall (x,x')\in\R^2.
\end{split}
\end{equation}
Substituting the two preceding expressions into \eqref{eq:Wnm_sub}, we find that
\begin{equation}
\eqref{eq:Wnm_sub} = \ol{\psi_1(x')}\ol{\psi_{\ell+1}(x)} w_m^{(\ell)}[\phi_1,\ldots,\phi_\ell;\ol{\psi_2},\ldots,\ol{\psi_\ell}](x) w_{n-m}^{(j)}[\phi_{\ell+1},\ldots,\phi_k; \ol{\psi_{\ell+2}},\ldots,\ol{\psi_k}](x).
\end{equation}
Hence,
\begin{align}
&\mathrm{Term}_{2,k}(x;x')  \\
&= \kappa\sum_{m=1}^{n-1}\sum_{\ell+j=k} \ol{\psi_1(x')\psi_{\ell+1}(x)}  w_m^{(\ell)}[\phi_1,\ldots,\phi_\ell;\ol{\psi_2},\ldots,\ol{\psi_\ell}](x) w_{n-m}^{(j)}[\phi_{\ell+1},\ldots,\phi_k; \ol{\psi_{\ell+2}},\ldots,\ol{\psi_k}](x).\nonumber
\end{align}
Combining our identities for $\mathrm{Term}_{1,k}$ and $\mathrm{Term}_{2,k}$, we obtain that
\begin{equation} \nonumber 
\begin{split}
&\paren*{\mathrm{Term}_{1,k}+\mathrm{Term}_{2,k}}(x;x')  \\
&= \ol{\psi_1(x')} (-i\p_x)w_n^{(k)}[\phi_1,\ldots,\phi_k;\ol{\psi_2},\ldots,\ol{\psi_k}](x) \\
&\phantom{=} \quad + \kappa\sum_{m=1}^{n-1}\sum_{\ell+j=k} \ol{\psi_1(x')\psi_{\ell+1}(x)}  w_m^{(\ell)}[\phi_1,\ldots,\phi_\ell;\ol{\psi_2},\ldots,\ol{\psi_\ell}](x) w_{n-m}^{(j)}[\phi_{\ell+1},\ldots,\phi_k; \ol{\psi_{\ell+2}},\ldots,\ol{\psi_k}](x),
\end{split}
\end{equation}
with equality in $\Sc'(\R^2)$. Now applying the recursive relation \eqref{eq:wn(k)_recur} for $w_{n+1}^{(k)}[\phi_{1},\ldots,\phi_{k};\ol{\psi_{2}},\ldots,\ol{\psi_{k}}]$, we find that
\begin{equation}
\paren*{\mathrm{Term}_{1,k} + \mathrm{Term}_{2,k}}(x;x') = \ol{\psi_1(x')} w_{n+1}^{(k)}[\phi_1,\ldots,\phi_k;\ol{\psi_2},\ldots,\ol{\psi_k}](x),
\end{equation}
which completes the proof of the induction step for showing \eqref{equ:whattoprove}.

We now use \cref{equ:whattoprove} to prove the adjoint assertion of the lemma. For $f,g\in\Sc(\R)$, we have by definition of the generalized partial trace (see \cref{prop:partial_trace}) that
\begin{equation} \label{gen_trace_exp}
\begin{split}
&\ipp*{\Tr_{2,\ldots,k}\paren*{\wt{\W}_n^{(k),*} \ket*{\bigotimes_{r=1}^k \phi_r}\bra*{\bigotimes_{r=1}^k\psi_r}}f, g}_{\Sc'(\R)-\Sc(\R)} \\
&=\ip{\psi_1}{f} \ipp*{\wt{\W}_n^{(k),*}\bigotimes_{r=1}^k \phi_r, g\otimes \bigotimes_{r=2}^k \ol{\psi_r}}_{\Sc'(\R^k)-\Sc(\R^k)}.
\end{split}
\end{equation}
By \cref{lem:dvo_adj},
\begin{equation}
\ipp*{\wt{\W}_n^{(k),*}\bigotimes_{r=1}^k \phi_r, \ol{\ol{g}\otimes \bigotimes_{r=2}^k \psi_r}}_{\Sc'(\R^k)-\Sc(\R^k)} = \ol{\ipp*{\wt{\W}_n^{(k)}( \ol{g}\otimes\bigotimes_{r=2}^k \psi_r), \bigotimes_{r=1}^k\ol{\phi_r}}}_{\Sc'(\R^k)-\Sc(\R^k)}.
\end{equation}
We can rewrite
\begin{equation}
\begin{split}
&\ip{\psi_1}{f}\ol{\ipp*{\wt{\W}_n^{(k)}( \ol{g}\otimes\bigotimes_{r=2}^k \psi_r), \bigotimes_{r=1}^k\ol{\phi_r}}}_{\Sc'(\R^k)-\Sc(\R^k)}\\
&= \ol{\ipp*{\Tr_{2,\ldots,k}\paren*{\wt{\W}_n^{(k)}\ket*{\ol{g}\otimes \bigotimes_{r=2}^k \psi_r}\bra*{f\otimes \bigotimes_{r=2}^k\phi_r}}\psi_1, \ol{\phi_1}}}_{\Sc'(\R)-\Sc(\R)}. \label{prev_ex}
\end{split}
\end{equation}
Now applying \eqref{equ:whattoprove} to this expression, we obtain that the right-hand side of \eqref{prev_ex} equals
\begin{align}
&\ol{\int_{\R^2}dxdx' \Phi_{\wt{\W}_n^{(k)}}(\ol{g},\psi_2,\ldots,\psi_k; \ol{f},\ol{\phi_2},\ldots,\ol{\phi_k})(x;x')\psi_1(x')\ol{\phi_1(x)}} \nonumber\\
&=\ol{\int_{\R^2}dxdx' \ol{f}(x') w_n^{(k)}[\ol{g},\psi_2,\ldots,\psi_k;\ol{\phi_2},\ldots,\ol{\phi_k}](x)\psi_1(x')\ol{\phi_1}(x)} \nonumber\\
&=\int_{\R^2}dxdx' f(x') \ol{w_n^{(k)}[\ol{g},\psi_2,\ldots,\psi_k;\ol{\phi_2},\ldots,\ol{\phi_k}]}(x) \ol{\psi_1}(x')\phi_1(x). \label{eq:FT_adj_app}
\end{align}
Next, using the Fubini-Tonelli theorem and applying \cref{lem:wn(k)_tran} in the $x$-integration, we find that
\begin{align}
\eqref{eq:FT_adj_app} &= \ip{\psi_1}{f} \ol{\int_{\R}dx w_{n,1}^{(k),t}[\ol{\phi_1},\psi_2,\ldots,\psi_k;\ol{\phi_2},\ldots,\ol{\phi_k}](x)\ol{g}(x)} \nonumber\\
&=\ip{\psi_1}{f}\int_{\R}dx \ol{w_{n,1}^{(k),t}[\ol{\phi_1},\psi_2,\ldots,\psi_k;\ol{\phi_2},\ldots,\ol{\phi_k}]}(x) g(x).
\end{align}
Since $f,g\in\Sc(\R)$ were arbitrary, going back to the left-hand side of \eqref{gen_trace_exp} and using the uniqueness and properties of $\Phi_{\W_n^{(k),*}}$, we conclude the pointwise in $\R^2$ identity
\begin{equation} \label{equ:whattoprove_adj}
\Phi_{\W_n^{(k),*}}(\phi_1,\ldots,\phi_k; \ol{\psi_1},\ldots,\ol{\psi_k})(x;x') = \ol{\psi_1(x')w_{n,1}^{(k),t}[\ol{\phi_1},\psi_2,\ldots,\psi_k;\ol{\phi_2},\ldots,\ol{\phi_k}](x)}.
\end{equation} 

We next need to generalize \cref{equ:whattoprove} and \cref{equ:whattoprove_adj} to arbitrary permutations $\pi\in\Ss_{k}$. By definition of the notation 
\begin{equation*}
\wt{\W}_{n,(\pi(1),\ldots,\pi(k))}^{(k)} := \pi \circ \wt{\W}_{n}^{(k)} \circ \pi^{-1},
\end{equation*}
we have that for any $\phi_1,\ldots,\phi_k\in\Sc(\R)$,
\begin{equation}
\wt{\W}_{n,(\pi(1),\ldots,\pi(k))}^{(k)}(\bigotimes_{r=1}^k \phi_r) = \pi \circ \wt{\W}_n^{(k)}((\bigotimes_{r=1}^k\phi_r)\circ\pi^{-1}),
\end{equation}
where the reader will recall from \eqref{eq:pi_vec_def} and \eqref{eq:pi_func_def} how a permutation acts on vectors and functions, respectively. Setting $f^{(k)} \coloneqq \bigotimes_{r=1}^k \phi_r$, we have by definition that
\begin{equation}
(f^{(k)}\circ\pi^{-1})(\ux_k) = f^{(k)}(x_{\pi^{-1}(1)},\ldots,x_{\pi^{-1}(k)}) = \prod_{r=1}^k \phi_r(x_{\pi^{-1}(r)}).
\end{equation}
Making the change of variable $r'=\pi^{-1}(r)$, we see that
\begin{equation}
\label{eq:tp_perm_comp}
\prod_{r=1}^k \phi_r(x_{\pi^{-1}(r)}) = \prod_{r'=1}^k \phi_{\pi(r')}(x_{r'}) = (\bigotimes_{r=1}^k \phi_{\pi(r)})(\ux_k).
\end{equation}
Therefore,
\begin{align}
\Tr_{2,\ldots,k}\paren*{\wt{\W}_{n,(\pi(1),\ldots,\pi(k))}^{(k)}\ket*{\otimes_{\ell=1}^{k}\phi_{\ell} }\bra*{\otimes_{\ell=1}^{k}\psi_{\ell}}} &= \Tr_{2,\ldots,k}\paren*{\paren*{\pi\circ\wt{\W}_n^{(k)}}\ket*{\bigotimes_{\ell=1}^k \phi_{\pi(\ell)}}\bra*{\bigotimes_{\ell=1}^k \psi_{\ell}}}
\end{align}
as elements of $\L_{gmp}(\Sc(\R), \Sc'(\R))$. Next, it follows from the characterizing property of the generalized partial trace and the fact that we define a permutation to act on tempered distribution by duality that
\begin{align}
&\ipp*{\Tr_{2,\ldots,k}\paren*{\paren*{\pi\circ\wt{\W}_n^{(k)}}\ket*{\bigotimes_{\ell=1}^k \phi_{\pi(\ell)}}\bra*{\bigotimes_{\ell=1}^k \psi_{\ell}}}f,g}_{\Sc'(\R)-\Sc(\R)} \nonumber\\
&=\ip{\psi_1}{f}\ipp*{\wt{\W}_n^{(k)}\bigotimes_{\ell=1}^k \phi_{\pi(\ell)}, (g\otimes\bigotimes_{\ell=2}^k \ol{\psi_\ell})\circ\pi^{-1}}_{\Sc'(\R^k)-\Sc(\R^k)}. \label{eq:tp_perm_sub}
\end{align}
Repeating the computation which yielded \eqref{eq:tp_perm_comp}, we find that
\begin{equation}
(g\otimes \bigotimes_{\ell=2}^k \ol{\psi_\ell})\circ \pi^{-1} = (\bigotimes_{\ell=1}^{\pi^{-1}(1)-1}\ol{\psi_{\pi(\ell)}}) \otimes g \otimes (\bigotimes_{\ell=\pi^{-1}(1)+1}^k \ol{\psi_{\pi(\ell)}}),
\end{equation}
where per our notation convention, the tensor product on the right-hand side is to be interpreted as $g \otimes \bigotimes_{\ell=2}^k\ol{\psi_{\pi(\ell)}}$ if $\pi(1)=1$. Thus,
\begin{align}
\eqref{eq:tp_perm_sub} &=\ip{\psi_1}{f}\ipp*{\wt{\W}_n^{(k)}\bigotimes_{\ell=1}^k \phi_{\pi(\ell)}, (\bigotimes_{\ell=1}^{\pi^{-1}(1)-1}\ol{\psi_{\pi(\ell)}})\otimes g\otimes (\bigotimes_{\ell=\pi^{-1}(1)+1}^k \ol{\psi_{\pi(\ell)}})}_{\Sc'(\R^k)-\Sc(\R^k)} \nonumber\\
&=\ipp*{\Tr_{2,\ldots,k}\paren*{\wt{\W}_n^{(k)}\ket*{\bigotimes_{\ell=1}^k \phi_{\pi(\ell)}}\bra*{\psi_1\otimes(\bigotimes_{\ell=2}^{\pi^{-1}(1)-1}\psi_{\pi(\ell)})\otimes\ol{g}\otimes (\bigotimes_{\ell=\pi^{-1}(1)+1}^k \psi_{\pi(\ell)}}}f,\ol{\psi_{\pi(1)}}}_{\Sc'(\R)-\Sc(\R)} \nonumber.
\end{align}
By definition of $\Phi_{\wt{\W}_n^{(k)}}$, this last expression equals
\begin{equation*}
\int_{\R^2}dxdx'\Phi_{\wt{\W}_n^{(k)}}(\phi_{\pi(1)},\ldots,\phi_{\pi(k)}; \ol{\psi_1},\ol{\psi_{\pi(2)}},\ldots,\ol{\psi_{\pi(\pi^{-1}(1)-1)}},g,\ol{\psi_{\pi(\pi^{-1}(1)+1)}},\ldots,\ol{\psi_{\pi(k)}})(x;x')f(x')\ol{\psi_{\pi(1)}(x)}.
\end{equation*}
Applying the result we have just established for the identity permutation, recorded in \eqref{equ:whattoprove}, and using the Fubini-Tonelli theorem and \cref{lem:wn(k)_tran}, we obtain
\begin{align*}
&\int_{\R^2}dxdx' \ol{\psi_1(x')} w_n^{(k)}[\phi_{\pi(1)},\ldots,\phi_{\pi(k)};\ol{\psi_{\pi(2)}},\ldots,\ol{\psi_{\pi(\pi^{-1}(1)-1)}},g,\ol{\psi_{\pi(\pi^{-1}(1)+1)}},\ldots,\ol{\psi_{\pi(k)}})(x)f(x')\ol{\psi_{\pi(1)}(x)} \\
&=\int_{\R^2}dxdx'w_{n,\pi^{-1}(1)'}^{(k),t}[\phi_{\pi(1)},\ldots,\phi_{\pi(k)};\ol{\psi_{\pi(2)}},\ldots,\ol{\psi_{\pi(\pi^{-1}(1)-1)}},\ol{\psi_{\pi(1)}},\ol{\psi_{\pi(\pi^{-1}(1)+1)}},\ldots,\ol{\psi_{\pi(k)}}](x) \\
&\phantom{=}\hspace{25mm}\times  \ol{\psi_1(x')}g(x)f(x').
\end{align*}
Since $f,g\in\Sc(\R)$ were arbitrary, we conclude that
\begin{equation}
\begin{split}
&\Phi_{\W_{n,(\pi(1),\ldots,\pi(k))}^{(k)}}(\phi_1,\ldots,\phi_k;\ol{\psi_1},\ldots,\ol{\psi_k})(x;x') \\
&= \ol{\psi_1(x')} w_{n,\pi^{-1}(1)'}^{(k),t}[\phi_{\pi(1)},\ldots,\phi_{\pi(k)};\ol{\psi_{\pi(2)}},\ldots,\ol{\psi_{\pi(\pi^{-1}(1)-1)}},\ol{\psi_{\pi(1)}},\ol{\psi_{\pi(\pi^{-1}(1)+1)}},\ldots,\ol{\psi_{\pi(k)}}](x), \qquad (x,x')\in\R^2.
\end{split}
\end{equation}

For the assertions about the adjoint, consider the expression
\begin{equation}
\label{eq:Phi_perm_adj}
\int_{\R^2}dxdx'\Phi_{\wt{\W}_n^{(k),*}}(\phi_{\pi(1)},\ldots,\phi_{\pi(k)}; \ol{\psi_1},\ol{\psi_{\pi(2)}},\ldots,\ol{\psi_{\pi(\pi^{-1}(1)-1)}},g,\ol{\psi_{\pi(\pi^{-1}(1)+1)}},\ldots,\ol{\psi_{\pi(k)}})(x;x')f(x')\ol{\psi_{\pi(1)}(x)}.
\end{equation}
By \eqref{equ:whattoprove_adj}, we have
\begin{equation} \label{char_prop}
\begin{split}
&\Phi_{\wt{\W}_n^{(k),*}}(\phi_{\pi(1)},\ldots,\phi_{\pi(k)}; \ol{\psi_1},\ol{\psi_{\pi(2)}},\ldots,\ol{\psi_{\pi(\pi^{-1}(1)-1)}},g,\ol{\psi_{\pi(\pi^{-1}(1)+1)}},\ldots,\ol{\psi_{\pi(k)}})(x;x')\\
&=\ol{\psi_1(x')}\ol{w_{n,1}^{(k),t}[\ol{\phi_{\pi(1)}},\psi_{\pi(2)},\ldots,\psi_{\pi(\pi^{-1}(1)-1)},\ol{g},\psi_{\pi(\pi^{-1}(1)+1)},\ldots,\psi_{\pi(k)};\ol{\phi_{\pi(2)}},\ldots,\ol{\phi_{\pi(k)}}](x)}.
\end{split}
\end{equation}
By the characterizing property of $w_{n,1}^{(k),t}$ from \cref{lem:wn(k)_tran}, followed by a second application of \cref{lem:wn(k)_tran},  we have that
\begin{align}
&\int_{\R}dx \ \ol{\psi_{\pi(1)}(x)w_{n,1}^{(k),t}[\ol{\phi_{\pi(1)}},\psi_{\pi(2)},\ldots,\psi_{\pi(\pi^{-1}(1)-1)},\ol{g},\psi_{\pi(\pi^{-1}(1)+1)},\ldots,\psi_{\pi(k)};\ol{\phi_{\pi(2)}},\ldots,\ol{\phi_{\pi(k)}}](x)} \nonumber\\
&=\ol{\int_{\R}dx \ \ol{\phi_{\pi(1)}(x)} w_n^{(j)}[\psi_{\pi(1)},\ldots,\psi_{\pi(\pi^{-1}(1)-1)},\ol{g},\psi_{\pi(\pi^{-1}(1)+1)},\ldots,\psi_{\pi(k)};\ol{\phi_{\pi(2)}},\ldots,\ol{\phi_{\pi(k)}}](x)} \nonumber\\
&=\ol{\int_{\R}dx \ \ol{g(x)} w_{n,\pi^{-1}(1)}^{(k),t}[\psi_{\pi(1)},\ldots,\psi_{\pi(\pi^{-1}(1)-1)},\ol{\phi_{\pi(1)}},\psi_{\pi(\pi^{-1}(1)+1)},\ldots,\psi_{\pi(k)};\ol{\phi_{\pi(2)}},\ldots,\ol{\phi_{\pi(k)}}](x)}.
\end{align}
By substituting \eqref{char_prop} into \eqref{eq:Phi_perm_adj}, then using Fubini-Tonelli theorem and the preceding identity, we conclude that
\begin{equation}
\begin{split}
&\Phi_{\W_{n,(\pi(1),\ldots,\pi(k))}^{(k),*}}(\phi_1,\ldots,\phi_k;\ol{\psi_1},\ldots,\ol{\psi_k})(x;x') \\
&=\ol{\psi_1(x') w_{n,\pi^{-1}(1)}^{(k),t}[\psi_{\pi(1)},\ldots,\psi_{\pi(\pi^{-1}(1)-1)},\ol{\phi_{\pi(1)}},\psi_{\pi(\pi^{-1}(1)+1)},\ldots,\psi_{\pi(k)};\ol{\phi_{\pi(2)}},\ldots,\ol{\phi_{\pi(k)}}](x)}
\end{split}
\end{equation}
point-wise in $\R^2$, which establishes the final claim and completes the proof.
\end{proof}

By taking the (1-particle) trace of the DVOs
\begin{equation*}
\Tr_{2,\ldots,k}\paren*{\wt{\W}_{n,(\pi(1),\ldots,\pi(k))}^{(k)} \ket*{\bigotimes_{\ell=1}^k \phi_\ell}\bra*{\bigotimes_{\ell=1}^k \psi_\ell}}, \enspace \Tr_{2,\ldots,k}\paren*{\wt{\W}_{n,(\pi(1),\ldots,\pi(k))}^{(k),*} \ket*{\bigotimes_{\ell=1}^k \phi_\ell}\bra*{\bigotimes_{\ell=1}^k \psi_\ell}}
\end{equation*}
and using the definition \eqref{eq:In(k)_def} of $I_n^{(k)}$, we obtain the following corollary of \cref{lem:Wn_wn_ptr}:

\begin{cor}\label{cor:Wn_wn_tr}
Let $k,n\in\N$. Then for any permutation $\pi\in\Ss_{k}$ and any functions $\phi_{1},\ldots,\phi_{k},\psi_{1},\ldots,\psi_{k}\in\Sc(\R)$, we have the identities
\begin{align}
&\Tr_{1,\ldots,k}\paren*{\widetilde{\W}_{n,(\pi(1),\ldots,\pi(k))}^{(k)}\ket*{\otimes_{\ell=1}^k \phi_{\ell}}\bra*{\otimes_{\ell=1}^k \psi_{\ell}}} =  I_n^{(k)}[\phi_{\pi(1)},\ldots,\phi_{\pi(k)}; \overline{\psi_{\pi(1)}},\ldots,\overline{\psi_{\pi(k)}}],\\
&\Tr_{1,\ldots,k}\paren*{\widetilde{\W}_{n,(\pi(1),\ldots,\pi(k))}^{(k),*}\ket*{\otimes_{\ell=1}^k \phi_{\ell}}\bra*{\otimes_{\ell=1}^k \psi_{\ell}}} = \ol{I_n^{(k)}[{\psi}_{\pi(1)},\ldots,{\psi}_{\pi(k)}; \overline{\phi_{\pi(1)}},\ldots,\overline{\phi_{\pi(k)}}]}.
\end{align}
\end{cor}

\section{The involution: $\H_n$ and $I_{b,n}$} \label{sec:invol}
In this section, we prove \cref{thm:GP_invol}. We recall the definition of the trace functionals
\begin{equation}\label{tr_fun}
\H_{n}(\Gamma) \coloneqq \Tr\paren*{\W_{n}\cdot\Gamma}, \qquad \forall  \Gamma \in \G_{\infty}^{*}.
\end{equation}
The statement of the theorem is then the following:
\GPinvol*

As discussed in the introduction, we prove \cref{thm:GP_invol} by showing that the Poisson commutativity of the functionals $\H_n$ on the weak Poisson manifold $(\G_\infty^*,\A_\infty,\pb{\cdot}{\cdot}_{\G_\infty^*})$ is equivalent to the Poisson commutativity of the functionals $I_{b,n}$ on the weak Poisson manifold $(\Sc(\R;\mathcal{V}), \A_{\Sc,\mathcal{V}}, \pb{\cdot}{\cdot}_{L^2,\mathcal{V}})$. See \eqref{symp_v}, \eqref{alg_v}, and \cref{prop:Schw_WP_V} for definition and properties of this manifold. Since the Poisson commutativity of the $I_{b,n}$ is established in \cref{prop:Ib_invol}, this equivalence will complete the proof of \cref{thm:GP_invol}. 

Establishing this equivalence relies on the detailed correspondence between the observable $\infty$-hierarchies $-i\W_n$ and the multilinear forms $w_n$ which we have obtained in \cref{sec:cor}, the reduction to symmetric-rank-1 tensors described in \cref{app:multi_alg}, and the demonstration of a Poisson morphism
\begin{equation*}
\iota_{\m}: (\Sc(\R;\mathcal{V}), \A_{\Sc,\mathcal{V}}, \pb{\cdot}{\cdot}_{L^2,\mathcal{V}}) \rightarrow (\G_\infty^*,\A_{\infty}, \pb{\cdot}{\cdot}_{\G_\infty^*}).
\end{equation*}
We establish the existence of this Poisson morphism in the next subsection.

\subsection{The mixed state Poisson morphism}
\label{ssec:invol_pomo}
Analogous to Theorem 2.12 from our companion paper \cite{MNPRS1_2019}, which shows that there is a Poisson morphism between $(\Sc(\R),\A_{\Sc},\pb{\cdot}{\cdot}_{L^2})$ and $(\G_{\infty}^*,\A_\infty,\pb{\cdot}{\cdot}_{\G_\infty^*})$ given by
\begin{equation}
\label{eq:iota_def}
\iota(\phi) \coloneqq (\ket*{\phi^{\otimes k}}\bra*{\phi^{\otimes k}})_{k\in\N}, \qquad\forall \phi\in\Sc(\R)
\end{equation}
\cref{thm:GP_pomo} stated below demonstrates that we have a Poisson morphism $\iota_{\m}$ between the weak Poisson manifolds $(\Sc(\R;\mathcal{V}), \A_{\Sc,\mathcal{V}},\pb{\cdot}{\cdot}_{L^2,\mathcal{V}})$ and $(\G_{\infty}^*, \A_{\infty}, \pb{\cdot}{\cdot}_{\G_{\infty}^*})$ given by
\begin{equation}
\iota_{\m}(\gamma) \coloneqq \frac{1}{2}(\ket*{\phi_1^{\otimes k}}\bra*{\phi_2^{\otimes k}}+\ket*{\phi_2^{\otimes k}}\bra*{\phi_1^{\otimes k}})_{k\in\N}, \qquad \forall \gamma = \frac{1}{2}\mathrm{odiag}(\phi_1,\ol{\phi_2},\phi_2,\ol{\phi_1}) \in \Sc(\R;\mathcal{V}).
\end{equation}

\Pomo*

Before proceeding with the proof of \cref{thm:GP_pomo}, we first record the G\^ateaux derivative of the map $\iota_{\m}$, which is used in the proof of the theorem. The computation is an easy exercise relying on multilinearity which we leave to the reader.

\begin{lemma}[Derivative of $\iota_\m$]
\label{lem:iotam_gd}
The G\^ateaux derivative of the map $\iota_\m$ is given by
\begin{equation}
\begin{split}
d\iota_{\m}[\gamma](\delta\gamma)^{(k)} &= \frac{1}{2}\sum_{\alpha=1}^k \left(\ket*{\phi_1^{\otimes (\alpha-1)} \otimes \delta\phi_1 \otimes \phi_1^{\otimes k-\alpha}}\bra*{\phi_2^{\otimes k}} +  \ket*{\phi_2^{\otimes k}}\bra*{\phi_1^{\otimes (\alpha-1)}\otimes\delta\phi_1\otimes\phi_1^{\otimes (k-\alpha)}}\right. \\
&\hspace{25mm} \left.+ \ket*{\phi_1^{\otimes k}}\bra*{\phi_2^{\otimes (\alpha-1)}\otimes\delta\phi_2\otimes\phi_2^{\otimes (k-\alpha)}} +\ket*{\phi_2^{\otimes (\alpha-1)}\otimes\delta\phi_2\otimes\phi_2^{\otimes (k-\alpha)}}\bra*{\phi_1^{\otimes k}}\right),
\end{split}
\end{equation}
for every $k\in\N$, where
\begin{equation}
\gamma = \frac{1}{2}\mathrm{odiag}(\phi_1,\ol{\phi_2},\phi_2,\ol{\phi_1}), \ \delta\gamma=\frac{1}{2}\mathrm{odiag}(\d\phi_1,\ol{\d\phi_2},\d\phi_2,\ol{\d\phi_1})\in\Sc(\R;\mathcal{V}).
\end{equation}
\end{lemma}

We now turn the proof of \cref{thm:GP_pomo}.

\begin{proof}[Proof of \cref{thm:GP_pomo}]
The proof of this result proceeds similarly to the proof of \cite[Theorem 2.12]{MNPRS1_2019}. Smoothness of $\iota_{\m}$ follows from its multilinear structure, therefore it remains to check that
\begin{enumerate}[(i)]
\item
\label{item:pomo_func}
$\iota_{\m}^* \A_{\infty} \subset \A_{\Sc,\mathcal{V}}$,
\item
\label{item:pomo_pb}
$\iota_{\m}^*\pb{\cdot}{\cdot}_{\G_{\infty}^*} = \pb{\iota_{\m}^*\cdot}{\iota_{\m}^*\cdot}_{L^2,\mathcal{V}}$.
\end{enumerate}

We first prove assertion \ref{item:pomo_func}. Let $F\in\A_{\infty}$, and set $f\coloneqq F\circ\iota_{\m}$. By the chain rule for the G\^ateaux derivative, we have that
\begin{align}
df[\gamma](\delta\gamma) &= dF[\iota_{\m}(\gamma)]\paren*{d\iota_{\m}[\gamma](\delta\gamma)} \nonumber\\
&=i\sum_{k=1}^\infty \Tr_{1,\ldots,k}\paren*{dF[\iota_{\m}(\gamma)]^{(k)} d\iota_{\m}[\gamma](\delta\gamma)^{(k)}} \nonumber\\
&=\frac{i}{2}\sum_{k=1}^\infty \Tr_{1,\ldots,k}\paren*{dF[\iota_{\m}(\gamma)]^{(k)} \ket*{\sum_{\alpha=1}^k \phi_1^{\otimes (\alpha-1)}\otimes\delta\phi_1\otimes \phi_1^{\otimes (k-\alpha)}}\bra*{\phi_2^{\otimes k}}} \nonumber\\
&\phantom{=} \hspace{20mm}+  \Tr_{1,\ldots,k}\paren*{dF[\iota_{\m}(\gamma)]^{(k)} \ket*{\phi_2^{\otimes k}}\bra*{\sum_{\alpha=1}^k \phi_1^{\otimes (\alpha-1)}\otimes \delta\phi_1\otimes \phi_1^{\otimes (k-\alpha)}}} \nonumber\\
&\phantom{=} \hspace{20mm} + \Tr_{1,\ldots,k}\paren*{dF[\iota_{\m}(\gamma)]^{(k)} \ket*{\sum_{\alpha=1}^k \phi_2^{\otimes (\alpha-1)}\otimes\delta\phi_2\otimes \phi_2^{\otimes (k-\alpha)}}\bra*{\phi_1^{\otimes k}}} \nonumber\\
&\phantom{=} \hspace{20mm}+  \Tr_{1,\ldots,k}\paren*{dF[\iota_{\m}(\gamma)]^{(k)} \ket*{\phi_1^{\otimes k}}\bra*{\sum_{\alpha=1}^k \phi_2^{\otimes (\alpha-1)}\otimes \delta\phi_2\otimes \phi_2^{\otimes (k-\alpha)}}}, \label{eq:pomo_sum_sub}
\end{align}
where the ultimate equality follows from application of \cref{lem:iotam_gd}.

Next, observe that by \cref{def:gen_trace} for the generalized trace and \cref{def:gmp} for the good mapping property, we have that
\begin{align}
&\Tr_{1,\ldots,k}\paren*{dF[\iota_{\m}(\gamma)]^{(k)} \ket*{\phi_2^{\otimes k}}\bra*{\sum_{\alpha=1}^k \phi_1^{\otimes (\alpha-1)}\otimes \delta\phi_1\otimes \phi_1^{\otimes (k-\alpha)}}}  \nonumber\\
&=\ip{\sum_{\alpha=1}^k \phi_1^{\otimes (\alpha-1)}\otimes \delta\phi_1\otimes \phi_1^{\otimes (k-\alpha)}}{dF[\iota_{\m}(\gamma)]^{(k)} \phi_2^{\otimes k}} \nonumber\\
&=\ip{\delta\phi_1}{\psi_{F,2,k}}, \label{eq:pomo_gmp_1_bar}
\end{align}
where $\psi_{F,2,k}\in\Sc(\R)$ is the necessarily unique Schwartz function coinciding with the antilinear functional
\begin{equation}
\begin{split}
\delta\phi_1&\mapsto\ipp*{\ip{\sum_{\alpha=1}^k (\cdot)\otimes_\alpha \phi_1^{\otimes (k-1)}}{dF[\iota_{\m}(\gamma)]^{(k)}\phi_2^{\otimes k}}\ , \delta\phi_1}_{\Sc'(\R)-\Sc(\R)}\\ 
&\coloneqq \ip{\sum_{\alpha=1}^k \phi_1^{\otimes (\alpha-1)}\otimes \delta\phi_1\otimes \phi_1^{\otimes(k-\alpha)}}{dF[\iota_{\m}(\gamma)]^{(k)}\phi_2^{\otimes k}}
\end{split}
\end{equation}
and where the reader will recall the definition of the notation $\otimes_\alpha$ from \eqref{eq:otimes_sub}. By the same reasoning,
\begin{align}
\label{eq:pomo_gmp_2_bar}
\Tr_{1,\ldots,k}\paren*{dF[\iota_{\m}(\gamma)]^{(k)} \ket*{\phi_1^{\otimes k}}\bra*{\sum_{\alpha=1}^k \phi_2^{\otimes (\alpha-1)}\otimes \delta\phi_2\otimes \phi_2^{\otimes (k-\alpha)}}} &=\ip{\delta\phi_2}{\psi_{F,1,k}},
\end{align}
where $\psi_{F,1,k}$ is the necessarily unique Schwartz function coinciding with the antilinear functional
\begin{equation}
\ip{\sum_{\alpha=1}^k  (\cdot)\otimes_\alpha \phi_2^{\otimes(k-1)}}{dF[\iota_{\m}(\gamma)]^{(k)}\phi_1^{\otimes k}}.
\end{equation}
Next, using that $dF[\iota_{\m}(\gamma)]^{(k)}$ is skew-adjoint,
\begin{align}
&\Tr_{1,\ldots,k}\paren*{dF[\iota_{\m}(\gamma)]^{(k)} \ket*{\sum_{\alpha=1}^k \phi_1^{\otimes (\alpha-1)}\otimes\delta\phi_1\otimes \phi_1^{\otimes (k-\alpha)}}\bra*{\phi_2^{\otimes k}}} \nonumber\\
&=-\ip{dF[\iota_{\m}(\gamma)]\phi_2^{\otimes k}}{\sum_{\alpha=1}^k \phi_1^{\otimes (\alpha-1)}\otimes\delta\phi_1\otimes\phi_1^{\otimes (k-\alpha)}} \nonumber\\
&=-\ol{\ip{\sum_{\alpha=1}^k \phi_1^{\otimes (\alpha-1)}\otimes\delta\phi_1\otimes\phi_1^{\otimes (k-\alpha)}}{dF[\iota_{\m}(\gamma)]\phi_2^{\otimes k}}} \nonumber\\
&=-\ol{\ip{\delta\phi_1}{\psi_{F,2,k}}} \nonumber\\
&=-\ip{\psi_{F,2,k}}{\delta\phi_1}. \label{eq:pomo_gmp_1}
\end{align}
By the same reasoning,
\begin{align}
\Tr_{1,\ldots,k}\paren*{dF[\iota_{\m}(\gamma)]^{(k)} \ket*{\sum_{\alpha=1}^k\phi_2^{\otimes (\alpha-1)}\otimes\delta\phi_2\otimes \phi_2^{\otimes (k-\alpha)}}\bra*{\phi_1^{\otimes k}}}&=-\ip{\psi_{F,1,k}}{\delta\phi_2}. \label{eq:pomo_gmp_2}
\end{align}

Substituting identities \eqref{eq:pomo_gmp_1_bar}, \eqref{eq:pomo_gmp_2_bar}, \eqref{eq:pomo_gmp_1}, and \eqref{eq:pomo_gmp_2} into \eqref{eq:pomo_sum_sub}, we find that
\begin{align}
df[\iota_{\m}(\gamma)](\delta\gamma) &=\frac{i}{2}\sum_{k=1}^\infty \paren*{\ip{\delta\phi_1}{\psi_{F,2,k}} + \ip{\delta\phi_2}{\psi_{F,1,k}} -\ip{\psi_{F,2,k}}{\delta\phi_1} - \ip{\psi_{F,1,k}}{\delta\phi_2}} \nonumber\\
&=\frac{i}{2}\paren*{\ip{\delta\phi_1}{\psi_{F,2}} + \ip{\delta\phi_2}{\psi_{F,1}} -\ip{\psi_{F,2}}{\delta\phi_1} -\ip{\psi_{F,1}}{\delta\phi_2}},
\end{align}
where we have defined $\psi_{F,1} \coloneqq \sum_{k=1}^\infty \psi_{F,1,k}$ and similarly for $\psi_{F,2}$. Note that these are well-defined Schwartz functions since $dF^{(k)}$ is zero for all but finitely many $k$ by assumption that $F\in\A_{\infty}$ (recall that $\A_\infty$ is generated by the set \eqref{eq:Ainf_gen}). The preceding formula can be rewritten as
\begin{equation}
\label{eq:sym_grad_rw}
df[\iota_m(\gamma)](\delta\gamma)=\frac{1}{2}\tr_{\C^2\otimes\C^2}\paren*{J \mathrm{odiag}(\psi_{F,1},\ol{\psi_{F,2}}, \psi_{F,2}, \ol{\psi_{F,1}}) \mathrm{odiag}(\delta\phi_2, \ol{\delta\phi_1},\delta\phi_1,\ol{\delta\phi_2})},
\end{equation}
where $J=\mathrm{diag}(i,-i,i,-i)$. Recalling definition \eqref{symp_v} for the symplectic form $\omega_{L^2,\V}$, we then see from \eqref{eq:sym_grad_rw} that the symplectic gradient of $f$ with respect to the form $\omega_{L^2,\V}$, which we denote by $\grad_{s,\V}f$, is given by
\begin{equation}
\label{eq:grad_form}
\grad_{s,\V} f(\gamma) = \frac{1}{2}\mathrm{odiag}(\psi_{F,1},\ol{\psi_{F,2}}, \psi_{F,2}, \ol{\psi_{F,1}}).
\end{equation}
That the map
\begin{equation}
\Sc(\R;\mathcal{V}) \rightarrow \Sc(\R;\V), \qquad \gamma \mapsto \grad_{s,\V}f(\gamma)
\end{equation}
is smooth follows from the fact that $\gamma$ depends smoothly on $(\psi_{F,1},\psi_{F,2})$, a consequence of the good mapping property. This completes our verification of assertion \ref{item:pomo_func}.

We now verify assertion \ref{item:pomo_pb} using the formula \eqref{eq:grad_form}. By definition of the Hamiltonian vector field in \ref{item:wp_P3} of \cref{def:WP} together with \cref{prop:LP}, which gives a formula for the vector field $X_G(\iota_{\m}(\gamma))$, we have that
\begin{align}
&\pb{F}{G}_{\G_{\infty}^*}(\iota_{\m}(\gamma)) \nonumber\\
&=dF[\iota_{\m}(\gamma)]\paren*{X_G(\iota_{\m}(\gamma)} \nonumber\\
&=\sum_{k=1}^\infty i\Tr_{1,\ldots,k}\paren*{dF[\iota_{\m}(\gamma)]^{(k)}\paren*{\sum_{j=1}^\infty j\Tr_{k+1,\ldots,k+j-1}\paren*{\comm{\sum_{\alpha=1}^k dG[\iota_{\m}(\gamma)]_{(\alpha,k+1,\ldots,k+j-1)}^{(j)}}{\iota_{\m}(\gamma)^{(k+j-1)}}}}}\nonumber
\end{align}
By the bosonic symmetry of $dG[\iota_{\m}(\gamma)]^{(j)}$,
\begin{align}
&\sum_{j=1}^\infty j\Tr_{k+1,\ldots,k+j-1}\paren*{\comm{\sum_{\alpha=1}^k dG[\iota_{\m}(\gamma)]_{(\alpha,k+1,\ldots,k+j-1)}^{(j)}}{\iota_{\m}(\gamma)^{(k+j-1)}}} \nonumber \\
& = \sum_{j=1}^\infty \Tr_{k+1,\ldots,k+j-1}\paren*{\comm{\sum_{\alpha=1}^k \sum_{\beta=1}^j dG[\iota_{\m}(\gamma)]_{(k+1,\ldots,k+\beta-1,\alpha,k+\beta,\ldots,\ldots,k+j-1)}^{(j)}}{\iota_{\m}(\gamma)^{(k+j-1)}}}. \label{eq:pb_c_sub}
\end{align}
It is then a short computation using the Schwartz kernel theorem and the definition of $\iota_{\m}$ that
\begin{equation}
\begin{split}
&\sum_{\beta=1}^j dG[\iota_{\m}(\gamma)]_{(k+1,\ldots,k+\beta-1,\alpha,k+\beta,\ldots,k+j-1)}^{(j)}\iota_{\m}(\gamma)^{(k+j-1)} 
\\
&= \frac{1}{2}\paren*{\ket*{\phi_1^{\otimes (k-1)} \otimes^\alpha dG[\iota_{\m}(\gamma)]^{(j)}(\phi_1^{\otimes j})}\bra*{\phi_2^{\otimes (k+j-1)}}  + \ket*{\phi_2^{\otimes (k-1)}\otimes^\alpha dG[\iota_{\m}(\gamma)]^{(j)}(\phi_2^{\otimes j})}\bra*{\phi_1^{\otimes (k+j-1)}}},
\end{split}
\end{equation}
where $\phi_1^{\otimes (k-1)}\otimes^\alpha dG[\iota_{\m}(\gamma)]^{(j)}(\phi_1^{\otimes j})$ is the element of $\Sc'(\R^{k+j-1})$ defined by
\begin{equation}
\begin{split}
&\paren*{\phi_1^{\otimes (k-1)} \otimes^\alpha dG[\iota_{\m}(\gamma)]^{(j)}(\phi_1^{\otimes j})}(\ux_{k+j-1}) \\
&\coloneqq \phi_1^{\otimes (\alpha-1)}(\ux_{\alpha-1}) \phi_1^{\otimes (k-\alpha)}(\ux_{\alpha+1;k})\paren*{\sum_{\beta=1}^j dG[\iota_{\m}(\gamma)]^{(j)}(\phi_1^{\otimes j})(\ux_{k+1;k+\beta-1},x_\alpha,\ux_{k+\beta;k+j-1})},
\end{split}
\end{equation}
and similarly for $\phi_2^{\otimes (k-1)}\otimes^\alpha dG[\iota_{\m}(\gamma)]^{(j)}(\phi_2^{\otimes j})$. Since $dG[\iota_{\m}(\gamma)]$ has the good mapping property by assumption that $G\in\A_{\infty}$, \cref{rem:gmp} and the definition of the generalized trace imply that for every $1\leq \alpha\leq k$,
\begin{equation}
\begin{split}
&\Tr_{k+1,\ldots,k+j-1}\paren*{\sum_{\beta=1}^j dG[\iota_{\m}(\gamma)]_{(k+1,\ldots,k+\beta-1,\alpha,k+\beta,\ldots,k+j-1)}^{(j)} \iota_{\m}(\gamma)^{(k+j-1)}} \\
&=\frac{1}{2}\paren*{\ket*{\phi_1^{\otimes (\alpha-1)} \otimes \psi_{G,1,j} \otimes \phi_1^{\otimes (k-\alpha)}}\bra*{\phi_2^{\otimes k}} + \ket*{\phi_2^{\otimes (\alpha-1)}\otimes\psi_{G,2,j}\otimes\phi_2^{\otimes (k-\alpha)}}\bra*{\phi_1^{\otimes k}}},
\end{split} \label{eq:psi_G_j}
\end{equation}
where $\psi_{G,1,j},\psi_{G,2,j}\in\Sc(\R)$ are the necessarily unique Schwartz functions satisfying
\begin{align}
\ip{\phi}{\psi_{G,1,j}} &= \ip{\sum_{\beta=1}^j \phi\otimes_\beta\phi_2^{\otimes (j-1)}}{dG[\iota_{\m}(\gamma)]^{(j)}\phi_1^{\otimes j}} \\
\ip{\phi}{\psi_{G,2,j}} &= \ip{\sum_{\beta=1}^j \phi\otimes_\beta\phi_1^{\otimes (j-1)}}{dG[\iota_{\m}(\gamma)]\phi_2^{\otimes j}}, \qquad \forall \phi\in\Sc(\R).
\end{align}
By repeating the same arguments and now using that the skew-adjointness of $dG[\iota_{\m}(\gamma)]^{(j)}$, we also obtain that for every $1\leq \alpha\leq k$,
\begin{equation}
\label{eq:psi_Gt_j}
\begin{split}
&\Tr_{k+1,\ldots,k+j-1}\paren*{\sum_{\beta=1}^j \iota_{\m}(\gamma)^{(k+j-1)}dG[\iota_{\m}(\gamma)]_{(\alpha,k+1,\ldots,k+j-1)}^{(j)}} \\
&=-\frac{1}{2}\paren*{\ket*{\phi_1^{\otimes k}}\bra*{\phi_2^{\otimes (\alpha-1)}\otimes \psi_{G,2,j}\otimes\phi_2^{\otimes (k-\alpha)}} + \ket*{\phi_2^{\otimes k}}\bra*{\phi_1^{\otimes (\alpha-1)}\otimes\psi_{G,1,j}\otimes\phi_1^{\otimes (k-\alpha)}}}.
\end{split}
\end{equation}

Substituting identities \eqref{eq:psi_G_j} and \eqref{eq:psi_Gt_j} into \eqref{eq:pb_c_sub} above, we find that
\begin{align}
&\pb{F}{G}_{\G_{\infty}^*}(\iota_{\m}(\gamma)) \nonumber\\
&=\frac{i}{2}\sum_{k=1}^\infty \Tr_{1,\ldots,k}\left( dF[\iota_{\m}(\gamma)]^{(k)} \left(\sum_{j=1}^\infty \ket*{\sum_{\alpha=1}^k\phi_1^{\otimes (\alpha-1)}\otimes\psi_{G,1,j}\otimes\phi_1^{\otimes(k-\alpha)}}\bra*{\phi_2^{\otimes k}} \right.\right. \nonumber\\
&\phantom{=} \hspace{70mm} \left.\left. + \ket*{\sum_{\alpha=1}^k\phi_2^{\otimes (\alpha-1)}\otimes \psi_{G,2,j}\otimes \phi_2^{\otimes (k-\alpha)}}\bra*{\phi_1^{\otimes k}}\right)\right) \nonumber\\
&\phantom{=} + \frac{i}{2}\sum_{k=1}^\infty \Tr_{1,\ldots,k}\left( dF[\iota_{\m}(\gamma)]^{(k)} \left(\sum_{j=1}^\infty \ket*{\phi_2^{\otimes k}}\bra*{\sum_{\alpha=1}^k\phi_1^{\otimes (\alpha-1)}\otimes\psi_{G,1,j}\otimes\phi_1^{\otimes(k-\alpha)}} \right.\right. \nonumber\\
&\phantom{=} \hspace{70mm} \left.\left. + \ket*{\phi_1^{\otimes k}}\bra*{\sum_{\alpha=1}^k\phi_2^{\otimes (\alpha-1)}\otimes \psi_{G,2,j}\otimes \phi_2^{\otimes (k-\alpha)}}\right)\right) \nonumber\\
&=\frac{i}{2}\sum_{j=1}^\infty\sum_{k=1}^\infty \ip{\phi_2^{\otimes k}}{dF[\iota_{\m}(\gamma)]^{(k)}\paren*{\sum_{\alpha=1}^k\phi_1^{\otimes (\alpha-1)}\otimes \psi_{G,1,j}\otimes\phi_1^{\otimes (k-\alpha)}}} \nonumber \\
&\phantom{=} \hspace{25mm} + \ip{\phi_1^{\otimes k}}{dF[\iota_{\m}(\gamma)]^{(k)}\paren*{\sum_{\alpha=1}^k\phi_2^{\otimes (\alpha-1)}\otimes \psi_{G,2,j}\otimes\phi_2^{\otimes (k-\alpha)}}} \nonumber \\
&\phantom{=} \hspace{25mm} + \ip{\sum_{\alpha=1}^k \phi_1^{\otimes(\alpha-1)}\otimes\psi_{G,1,j}\otimes\phi_1^{\otimes(k-\alpha)}}{dF[\iota_{\m}(\gamma)]^{(k)}\phi_2^{\otimes k}} \nonumber\\
&\phantom{=} \hspace{25mm} +\ip{\sum_{\alpha=1}^k \phi_2^{\otimes(\alpha-1)}\otimes\psi_{G,2,j}\otimes\phi_2^{\otimes(k-\alpha)}}{dF[\iota_{\m}(\gamma)]^{(k)}\phi_1^{\otimes k}},
\end{align}
where the ultimate equality is immediate from the definition of the generalized trace. Recalling the definitions of $\psi_{F,1,k}$ and $\psi_{F,2,k}$ in \eqref{eq:pomo_gmp_1_bar} and \eqref{eq:pomo_gmp_2_bar}, respectively, we have that
\begin{align}
\ip{\sum_{\alpha=1}^k \phi_1^{\otimes(\alpha-1)}\otimes\psi_{G,1,j}\otimes\phi_1^{\otimes(k-\alpha)}}{dF[\iota_{\m}(\gamma)]^{(k)}\phi_2^{\otimes k}} &= \ip{\psi_{G,1,j}}{\psi_{F,2,k}}, \\
\ip{\sum_{\alpha=1}^k \phi_2^{\otimes(\alpha-1)}\otimes\psi_{G,2,j}\otimes\phi_2^{\otimes(k-\alpha)}}{dF[\iota_{\m}(\gamma)]^{(k)}\phi_1^{\otimes k}} &= \ip{\psi_{G,2,j}}{\psi_{F,1,k}}.
\end{align}
Now using the skew-adjointness of $dF[\iota_{\m}(\gamma)]^{(k)}$, we find that
\begin{align}
&\ip{\phi_2^{\otimes k}}{dF[\iota_{\m}(\gamma)]^{(k)}\paren*{\sum_{\alpha=1}^k\phi_1^{\otimes (\alpha-1)}\otimes \psi_{G,1,j}\otimes\phi_1^{\otimes (k-\alpha)}}} \nonumber\\
&= -\ol{\ip{\sum_{\alpha=1}^k \phi_1^{\otimes (\alpha-1)}\otimes\psi_{G,1,j}\otimes\phi_1^{\otimes(k-\alpha)}}{dF[\iota_{\m}(\gamma)]^{(k)}\phi_2^{\otimes k}}} \nonumber\\
&=-\ip{\psi_{F,2,k}}{\psi_{G,1,j}}.
\end{align}
Similarly,
\begin{align}
&\ip{\phi_1^{\otimes k}}{dF[\iota_{\m}(\gamma)]^{(k)}\paren*{\sum_{\alpha=1}^k\phi_2^{\otimes (\alpha-1)}\otimes \psi_{G,2,j}\otimes\phi_2^{\otimes (k-\alpha)}}} =-\ip{\psi_{F,1,k}}{\psi_{G,2,j}}.
\end{align}
Hence,
\begin{align}
\pb{F}{G}_{\G_{\infty}^*}(\iota_{\m}(\gamma)) &= \frac{i}{2}\sum_{j=1}^\infty\sum_{k=1}^\infty \ip{\psi_{G,1,j}}{\psi_{F,2,k}} + \ip{\psi_{G,2,j}}{\psi_{F,1,k}} - \ip{\psi_{F,2,k}}{\psi_{G,1,j}} - \ip{\psi_{F,1,k}}{\psi_{G,2,j}} \nonumber\\
&= \frac{i}{2}\paren*{\ip{\psi_{G,1}}{\psi_{F,2}} + \ip{\psi_{G,2}}{\psi_{F,1}} -\ip{\psi_{F,2}}{\psi_{G,1}} -\ip{\psi_{F,1}}{\psi_{G,2}}} \label{eq:pb_grad_comp},
\end{align}
where we have defined $\psi_{F,\ell} \coloneqq \sum_{k=1}^\infty \psi_{F,\ell,k} $, for $\ell\in\{1,2\}$, and similarly for $\psi_{G,\ell}$. Note that these are well-defined elements of $\Sc(\R)$ since $\psi_{F,\ell,k},\psi_{G,\ell,j}$ are identically zero for all but finitely many $k,j$. By \eqref{eq:grad_form}, we know that
\begin{align}
\label{eq:grad_f_g}
\grad_{s,\V} f(\gamma) &= \frac{1}{2}\mathrm{odiag}(\psi_{F,1},\ol{\psi_{F,2}}, \psi_{F,2}, \ol{\psi_{F,1}}), \\
\grad_{s,\V} g(\gamma) &= \frac{1}{2}\mathrm{odiag}(\psi_{G,1},\ol{\psi_{G,2}}, \psi_{G,2}, \ol{\psi_{G,1}}).
\end{align}
Hence by recalling the definition \eqref{symp_v} for the symplectic form $\omega_{L^2,\V}$ and \cref{prop:Schw_WP_V}, then proceeding by direct computation, we find that
\begin{align}
&\pb{f}{g}_{L^2,\V}(\gamma) \nonumber\\
&= \omega_{L^2,\V}(\grad_{s,\V} f(\gamma), \grad_{s,\V} g(\gamma)) \nonumber\\
&=\frac{1}{2}\int_{\R}dx\tr_{\C^2\otimes\C^2}\paren*{\mathrm{diag}(i,-i,i,-i)\mathrm{odiag}(\psi_{F,1},\ol{\psi_{F,2}}, \psi_{F,2}, \ol{\psi_{F,1}}) \mathrm{odiag}(\psi_{G,2},\ol{\psi_{G,1}},\psi_{G,1},\ol{\psi_{G,2}})}(x) \nonumber\\
&=\eqref{eq:pb_grad_comp}.
\end{align}
Therefore, we have shown that
\begin{equation}
\pb{F}{G}_{\G_\infty^*}(\iota_{\m}(\gamma)) = \pb{f}{g}_{L^2,\V}(\gamma),
\end{equation}
completing the proof.
\end{proof}

\subsection{Relating the functionals $\H_n$ and $I_{b,n}$}
\label{ssec:cor_Ibn_Hn}
We now use the analysis of \cref{ssec:cor_ptr} to relate the functionals $\H_n$, defined in \eqref{Hn_trace}, on the infinite-particle phase space $\G_\infty^*$ to the functionals $I_{b,n}$, defined in \eqref{eq:Ibn_intro_def}, on the one-particle mixed-state phase space $\Sc(\R;\mathcal{V})$, defined in \eqref{eq:Schw_ms}.

\begin{prop}
\label{prop:Hn_Ibn}
For every $n\in\N$, it holds that
\begin{equation}
\H_n(\iota_{\m} (\gamma)) = I_{b,n}(\gamma), \qquad \forall \gamma\in\Sc(\R;\V).
\end{equation}
\end{prop}
\begin{proof}
Fix $n\in\N$ and let $\gamma=\frac{1}{2}\mathrm{odiag}(\phi_1,\ol{\phi_2},\phi_2,\ol{\phi_1})$, for $\phi_1,\phi_2\in\Sc(\R)$. Unpacking the definition \eqref{Hn_trace} of $\H_n$, the definition \eqref{eq:Wn_fin_def} for $\W_n$, and the bilinearity of the generalized trace, we see that
\begin{equation}
\begin{split}
\H_n(\iota_{\m}(\gamma)) &= \frac{1}{4}\sum_{k=1}^\infty  \frac{1}{k!}\sum_{\pi\in\Ss_k} \Tr_{1,\ldots,k}\paren*{\wt{\W}_{n,(\pi(1),\ldots,\pi(k))}^{(k)}\ket*{\phi_1^{\otimes k}}\bra*{\phi_2^{\otimes k}}} + \Tr_{1,\ldots,k}\paren*{\wt{\W}_{n,(\pi(1),\ldots,\pi(k))}^{(k)}\ket*{\phi_2^{\otimes k}}\bra*{\phi_1^{\otimes k}}} \\
&\phantom{=}\hspace{25mm}+ \Tr_{1,\ldots,k}\paren*{\wt{\W}_{n,(\pi(1),\ldots,\pi(k))}^{(k),*}\ket*{\phi_1^{\otimes k}}\bra*{\phi_2^{\otimes k}}} + \Tr_{1,\ldots,k}\paren*{\wt{\W}_{n,(\pi(1),\ldots,\pi(k))}^{(k),*}\ket*{\phi_2^{\otimes k}}\bra*{\phi_1^{\otimes k}}}.
\end{split}
\end{equation}
By \cref{cor:Wn_wn_tr}, we have the identities
\begin{equation}
\begin{split}
\Tr_{1,\ldots,k}\paren*{\wt{\W}_{n,(\pi(1),\ldots,\pi(k))}^{(k)}\ket*{\phi_1^{\otimes k}}\bra*{\phi_2^{\otimes k}}} &= I_n^{(k)}(\phi_1^{\times k};\ol{\phi_2}^{\times k}),  \\
\Tr_{1,\ldots,k}\paren*{\wt{\W}_{n,(\pi(1),\ldots,\pi(k))}^{(k)}\ket*{\phi_2^{\otimes k}}\bra*{\phi_1^{\otimes k}}} &= I_n^{(k)}(\phi_2^{\times k};\ol{\phi_1}^{\times k}),\\
\Tr_{1,\ldots,k}\paren*{\wt{\W}_{n,(\pi(1),\ldots,\pi(k))}^{(k),*}\ket*{\phi_1^{\otimes k}}\bra*{\phi_2^{\otimes k}}} &= \ol{I_n^{(k)}(\phi_2^{\times k};\ol{\phi_1}^{\times k})},  \\
\Tr_{1,\ldots,k}\paren*{\wt{\W}_{n,(\pi(1),\ldots,\pi(k))}^{(k),*}\ket*{\phi_2^{\otimes k}}\bra*{\phi_1^{\otimes k}}} &= \ol{I_n^{(k)}(\phi_1^{\times k};\ol{\phi_2}^{\times k})}.
\end{split}
\end{equation}
for every $k\in\N$ and $\pi\in\Ss_k$. Consequently, by \cref{rem:rec_comp},
\begin{align}
\H_n(\iota_{\m}(\gamma)) &= \frac{1}{4}\sum_{k=1}^\infty \paren*{I_n^{(k)}(\phi_1^{\times k};\ol{\phi_2}^{\times k}) + I_n^{(k)}(\phi_2^{\times k};\ol{\phi_1}^{\times k})+\ol{I_n^{(k)}(\phi_2^{\times k};\ol{\phi_1}^{\times k})}+\ol{I_n^{(k)}(\phi_1^{\times k};\ol{\phi_2}^{\times k})}} \nonumber\\
&=\frac{1}{4}\paren*{\tl{I}_n(\phi_1,\ol{\phi_2}) + \tl{I}_n(\phi_2,\ol{\phi_1}) + \ol{\tl{I}_n(\phi_1,\ol{\phi_2})} + \ol{\tl{I}_n(\phi_2,\ol{\phi_1})}}.
\end{align}
By \eqref{eq:I_invol}, we know that the $\tl{I}_n$ have the involution property
\begin{equation}
\tl{I}_n(f,\ol{g}) = \ol{\tl{I}_n(g,\ol{f})}, \qquad \forall f,g\in\Sc(\R).
\end{equation}
So, we obtain by the definition of $I_{b,n}$ in \eqref{eq:Ibn_intro_def} that
\begin{equation}
\H_n(\iota_{\m}(\gamma)) = \frac{1}{2}\paren*{\tl{I}_n(\phi_1,\ol{\phi_2}) + \tl{I}_n(\phi_2,\ol{\phi_1})} = I_{b,n}(\gamma),
\end{equation}
as required.
\end{proof}

\subsection{Proof of \protect{\cref{thm:GP_invol}} and \protect{\cref{thm:Equiv}} }
\label{ssec:invol_pb_an}

The goal of this subsection is to complete the proof of \cref{thm:GP_invol}:

\GPinvol*

As detailed in the introduction, we will establish \cref{thm:GP_invol} by proving \cref{thm:Equiv}, the statement of which we recall here.

\Equiv*

We refer to \eqref{eq:Ibn_intro_def} for the definition of $I_{b,n}$. In light of \cref{prop:Ib_invol} which establishes the validity of \eqref{ibn_inv}, \cref{thm:GP_invol} is then an immediate corollary of \cref{thm:Equiv}. Thus we focus on proving \cref{thm:Equiv}.
\begin{proof}[Proof of \protect{\cref{thm:Equiv}}]
The implication that
\begin{equation*}
\pb{\H_n}{\H_m}_{\G_\infty^*}\equiv 0 \Longrightarrow \pb{I_{b,n}}{I_{b,m}}_{L^2,\V}\equiv 0
\end{equation*}
is a consequence of \cref{thm:GP_pomo} and \cref{prop:Hn_Ibn}. Indeed, the latter states that
\[
\H_n(\iota_{\m}(\gamma)) = I_{b,n}(\gamma),
\]
and hence by \cref{thm:GP_pomo}, we have
\[
\pb{I_{b,n}}{I_{b,m}}_{L^2,\V}(\gamma)  = \pb{\H_n}{\H_m}_{\G_\infty^*}(\iota_{\m} (\gamma)) = 0.
\]

To show the reverse implication, we first claim that it suffices to show that
\begin{equation}
\label{eq:DM_form}
\pb{\H_n}{\H_m}_{\G_\infty^*}(\Gamma)=0, \quad \forall \Gamma=(\gamma^{(k)})_{k\in\N}, \ \gamma^{(k)} = \frac{1}{2}\paren*{\ket*{f_k^{\otimes k}}\bra*{g_k^{\otimes k}} + \ket*{g_k^{\otimes k}}\bra*{f_k^{\otimes k}}}, \ f_k,g_k\in\Sc(\R).
\end{equation}
Indeed, for any $k\in\N$, \cref{cor:Sch_ST_DM_d} gives that finite linear combinations of the form
\begin{equation}
\sum_{j=1}^{N_k} \frac{a_j}{2}\paren*{\ket*{f_{j}^{\otimes k}}\bra*{g_j^{\otimes k}} + \ket*{g_j^{\otimes k}}\bra*{f_j^{\otimes k}}}, \quad a_j\in\C, \ f_j,g_j\in\Sc(\R), \ N_k\in\N
\end{equation}
are dense in $\g_k^*$ (recall \cref{eq:gk*_def}). Since by definition $\G_\infty^*$ is the topological direct product of the $\g_k^*$ (recall \cref{eq:Ginf*_def}), elements $\Gamma=(\gamma^{(k)})_{k\in\N}\in\G_\infty^*$ of the form
\begin{equation}
\label{eq:DM_dense}
\gamma^{(k)} = \sum_{j=1}^\infty \frac{a_{jk}}{2}\paren*{\ket*{f_{jk}^{\otimes k}}\bra*{g_{jk}^{\otimes k}} + \ket*{g_{jk}^{\otimes k}}\bra*{f_{jk}^{\otimes k}}}, \qquad k\in\N,
\end{equation}
where $f_{jk},g_{jk}\in\Sc(\R)$ and $a_{jk}\in\C$ with $a_{jk}=0$ for all but finitely many $j\in\N$, are dense in $\G_\infty^*$. Now recalling the definition \eqref{equ:poisson_def} for the Poisson bracket $\pb{\H_n}{\H_m}_{\G_\infty^*}$ and using the bilinearity of the generalized trace, we need to show that for $\Gamma$ in the form of \eqref{eq:DM_dense},
\begin{align}
0 &= \pb{\H_n}{\H_m}_{\G_\infty^*}(\Gamma) \nonumber \\
&=\sum_{k=1}^\infty \sum_{j=1}^\infty \frac{ia_{jk}}{2}\Tr_{1,\ldots,k}\paren*{\comm{-i\W_n}{-i\W_m}_{\G_\infty}^{(k)} \paren*{\ket*{f_{jk}^{\otimes k}}\bra*{g_{jk}^{\otimes k}} + \ket*{g_{jk}^{\otimes k}}\bra*{f_{jk}^{\otimes k}}}} \nonumber\\
&=\sum_{j=1}^\infty a_{jk}\pb{\H_n}{\H_m}_{\G_\infty^*}(\Gamma_j), \label{eq:sum_pb_gamj}
\end{align}
where
\begin{equation}
\Gamma_j = (\gamma_j^{(k)})_{k\in\N}, \qquad \gamma_j^{(k)} \coloneqq \frac{1}{2}\paren*{\ket*{f_{jk}^{\otimes k}}\bra*{g_{jk}^{\otimes k}} + \ket*{g_{jk}^{\otimes k}}\bra*{f_{jk}^{\otimes k}}}.
\end{equation}
Note that because $\comm{-i\W_n}{-i\W_m}_{\G_\infty}^{(k)}$ is zero for all but finitely many $k$, and for each fixed $k\in\N$, $a_{jk}$ is zero for all but finitely many $j$, it follows that there are only finitely many nonzero terms in the double series above, and consequently, there are no issues of convergence. \eqref{eq:DM_form} will imply that each summand in \eqref{eq:sum_pb_gamj} is zero, so by continuity of $\pb{\H_n}{\H_m}_{\G_\infty^*}$ and by density of elements of the form \eqref{eq:DM_dense} in $\G_\infty^*$, we arrive at the desired implication.

Thus, we proceed to show \eqref{eq:DM_form}. Unpacking the definition of $\pb{\H_n}{\H_m}_{\G_\infty^*}(\Gamma)$, we see that
\begin{equation}
\label{eq:sum_van}
\pb{\H_n}{\H_m}_{\G_\infty^*}(\Gamma) = \frac{i}{2}\sum_{k=1}^\infty \Tr_{1,\ldots,k}\paren*{\comm{-i\W_n}{-i\W_m}_{\G_\infty}^{(k)}\paren*{\ket*{f_k^{\otimes k}}\bra*{g_k^{\otimes k}}+\ket*{g_k^{\otimes k}}\bra*{f_k^{\otimes k}}}}
\end{equation}
For each $k\in\N$ and $\lambda\in\C$, consider the element $\gamma_{k,\lambda}\in\Sc(\R;\mathcal{V})$ defined by
\begin{equation}
\gamma_{k,\lambda} \coloneqq \frac{1}{2}\mathrm{odiag}(\lambda f_k,\ol{\lambda g_k}, \lambda g_k, \ol{\lambda f_k})
\end{equation}
Then by the assumption \cref{ibn_inv} and \cref{thm:GP_pomo},
\begin{align}
0=\pb{I_{b,n}}{I_{b,m}}_{L^2,\V}(\gamma_{k,\lambda})&=\pb{\H_n}{\H_m}_{\G_\infty^*}(\iota_{\m}(\gamma_{k,\lambda})) \nonumber\\
&= \sum_{j=1}^\infty i\Tr_{1,\ldots,j}\paren*{\comm{-i\W_n}{-i\W_m}_{\G_\infty}^{(j)}\iota_{\m}(\gamma_{k,\lambda})^{(j)}} \nonumber\\
&=\frac{i}{2}\sum_{j=1}^\infty |\lambda|^{2j}\Tr_{1,\ldots,j}\paren*{\comm{-i\W_n}{-i\W_m}_{\G_\infty}^{(j)}(\ket*{f_k^{\otimes j}}\bra*{g_k^{\otimes j}}+\ket*{g_k^{\otimes j}}\bra*{f_k^{\otimes j}})} \nonumber\\
&\eqqcolon \frac{i}{2}\rho_k(\lambda).
\end{align}
$\rho_k$ is well-defined on $\C$, since there are only finitely many indices $j$ for which the summand is nonzero. Since for any $r\in\N$,
\begin{equation}
0=\paren*{(\p_{\lambda}\p_{\ol{\lambda}})^r \rho_k }(0) = r! \Tr_{1,\ldots,r}\paren*{\comm{-i\W_n}{-i\W_m}_{\G_\infty}^{(r)} (\ket*{f_k^{\otimes r}}\bra*{g_k^{\otimes r}} + \ket*{g_k^{\otimes r}}\bra*{f_k^{\otimes r}})},
\end{equation}
it follows that
\begin{equation}
\Tr_{1,\ldots,k}\paren*{\comm{-i\W_n}{-i\W_m}_{\G_\infty}^{(k)}(\ket*{f_k^{\otimes k}}\bra*{g_k^{\otimes k}}+\ket*{g_k^{\otimes k}}\bra*{f_k^{\otimes k}})}=0.
\end{equation}
Therefore, each summand in the right-hand side of \eqref{eq:sum_van} vanishes, yielding \eqref{eq:DM_form}. Thus, the proof of \cref{thm:GP_pomo} is complete.
\end{proof}

\subsection{Nontriviality}
\label{ssec:invol_ntriv}
In this subsection, we prove that the statement of \cref{thm:GP_invol} is nontrivial in the sense that the functionals $\H_n$ do not Poisson commute with every element of $\A_{\infty}$. The proof of this fact proceeds by a reduction to proving a one-particle result.

\begin{prop}
\label{prop:invol_ntriv}
For every $n\in\N$, there exists a functional $F\in\A_\infty$ and an element $\Gamma\in\G_\infty^*$ such that
\begin{equation}
\pb{F}{\H_n}_{\G_\infty^*}(\Gamma) \neq 0.
\end{equation}
\end{prop}
\begin{proof}
We proceed by contradiction and suppose that for every $F\in\A_\infty$, it holds that $\pb{F}{\H_n}_{\G_\infty^*}\equiv 0$ on $\G_\infty^*$. So by the \cref{def:WP}\ref{item:wp_P3} for the Hamiltonian vector field, we have that
\begin{align}
0=\pb{F}{\H_n}_{\G_\infty^*}(\Gamma) &= dF[\Gamma](X_{\H_n}(\Gamma)).
\end{align}
By duality, it follows that $X_{\H_n}\equiv 0$ on $\G_\infty^*$. In particular, for any pure state $\Gamma=\iota(\phi)$, where $\iota$ is as in \eqref{eq:iota_def} and $\phi\in\Sc(\R)$, we have by \cref{thm:GP_fac} (to be proved in the next section) that
\begin{equation}
X_{\H_n}(\iota(\phi))^{(1)} = \ket*{\phi}\bra*{\grad_s I_n(\phi)} + \ket*{\grad_s I_n(\phi)}\bra*{\phi} = 0\in\g_1^*.
\end{equation}
Taking the 1-particle trace of the right-hand side and using the characterization of the symplectic gradient (see \cref{def:re_grad}), we obtain that
\begin{equation}
\label{eq:In_van}
0=dI_n[\phi](\phi) = \sum_{k=1}^\infty 2k I_n^{(k)}[\phi^{\times k};\ol{\phi}^{\times k}],
\end{equation}
where the ultimate equality follows by direct computation. However, \eqref{eq:In_van} is a contradiction by \cref{lem:In_noncon}, and therefore the proof is complete.
\end{proof}

\section{The equations of motion: $n$GP and $n$NLS}\label{sec:gp_flows}
In this last section, we prove \cref{thm:GP_fac}. Before recalling the statement of this theorem, we first recall that for each $n \in \N$, the Hamiltonian functionals $\mathcal{H}_{n}$ are given by the formula
\begin{equation}\label{hamil_func}
\H_{n}(\Gamma) \coloneqq \Tr\paren*{\W_{n}\cdot\Gamma}, \qquad \forall  \Gamma \in \G_{\infty}^{*}
\end{equation}
and the Hamiltonian equation of motion defined by the functional $\mathcal{H}_{n}$ on $\G_{\infty}^{*}$, which we have called the \emph{$n$-th GP-hierarchy (nGP)}, is given by
\begin{align}
\frac{d}{dt}\Gamma &=  X_{\H_{n}}(\Gamma),
\end{align}
where $X_{\H_n}$ is the Hamiltonian vector field associated to $\H_n$.

\GPfac*

\cref{thm:GP_fac} asserts that (nGP) admits a special class of factorized solutions of the form
\begin{equation}
\Gamma=(\gamma^{(k)})_{k\in\N}, \qquad \gamma^{(k)} \coloneqq \ket*{\phi^{\otimes k}}\bra*{\phi^{\otimes k}}, \qquad \phi\in C^{\infty}(I;\mathcal{S}(\R)),
\end{equation}
where $\phi$ solves the $n$-th nonlinear Schr\"{o}dinger equation (nNLS):
\begin{equation}\label{phi_evol}
\paren*{\frac{d}{dt}\phi}(t)= \grad_{s}I_{n}(\phi(t)), \qquad \forall t\in I,
\end{equation}
and where $\grad_{s}$ is the symplectic gradient with respect to the $L^{2}$ standard symplectic structure (recall \cref{def:re_grad} and \cref{schwartz_deriv}). We note that existence and uniqueness for the (nNLS) equation in the class $C^{\infty}(I;\Sc(\R))$ follows from the inverse scattering results of \cite{BC1984, Zhou1989, Zhou1998}. 

\subsection{nGP Hamiltonian vector fields}
\label{ssec:gp_flows_vf}
We first relate the formula given by \cref{prop:LP} for the Hamiltonian vector field $X_{\H_n}$ to the nonlinear operators $w_n$. This connection underpins the proof of \cref{thm:GP_fac}. For $n\in\N$, \cref{prop:LP} gives
\begin{equation}
\label{eq:XH_n}
X_{\H_n}(\Gamma)^{(\ell)} = \sum_{j=1}^\infty j\Tr_{\ell+1,\ldots,\ell+j-1}\paren*{\comm{\sum_{\alpha=1}^\ell (-i\W_n)_{(\alpha,\ell+1,\ldots,\ell+j-1)}^{(j)}}{\gamma^{(\ell+j-1)}}}, \qquad \ell\in\N, \enspace \Gamma\in\G_\infty^*.
\end{equation}
The main lemma is a formula for
\begin{equation*}
\Tr_{\ell+1,\ldots,\ell+j-1}\paren*{\comm{\sum_{\alpha=1}^\ell (-i\W_n)_{(\alpha,\ell+1,\ldots,\ell+j-1)}^{(j)}}{\gamma^{(\ell+j-1)}}}
\end{equation*}
in the special case where $\gamma^{(\ell+j-1)}$ is a mixed state, i.e.
\begin{equation}
\label{eq:gam_form}
\gamma^{(\ell+j-1)} = \frac{1}{2}\paren*{\ket*{f^{\otimes (\ell+j-1)}}\bra*{g^{\otimes (\ell+j-1)}} + \ket*{g^{\otimes (\ell+j-1)}}\bra*{f^{\otimes (\ell+j-1)}}}, \qquad f,g\in\Sc(\R).
\end{equation}

\begin{lemma}
\label{lem:tr_Wn_form}
Let $\ell,j\in\N$. Suppose that $\gamma^{(\ell+j-1)}$ is of the form \eqref{eq:gam_form}. Then for any $\alpha\in\N_{\leq \ell}$ and $\beta \in\N_{\leq j}$, it holds that
\begin{equation} 
\label{eq:tr_Wn_form}
\begin{split}
&\Tr_{\ell+1,\ldots,\ell+j-1}\paren*{(\W_{n,sa})_{(\ell+1,\ldots,\ell+\beta-1,\alpha,\ell+\beta,\ldots,\ell+j-1)}^{(j)}\gamma^{(\ell+j-1)}}(\ux_\ell;\ux_\ell') \\
&=\frac{1}{4}f^{\otimes (\ell-1)}(\ux_{\alpha-1},\ux_{\alpha+1;\ell}) \ol{g^{\otimes \ell}(\ux_\ell')} \\
&\phantom{=}\qquad\times \paren*{w_{n,\beta'}^{(j),t}[f^{\times j};\ol{g}^{\times (j-1)}](x_\alpha) + \ol{w_{n,\beta}^{(j),t}[g^{\times(\beta-1)},\ol{f},g^{(j-\beta)};\ol{f}^{\times(j-1)}](x_\alpha)}} \\
&\phantom{=} +\frac{1}{4}g^{\otimes (\ell-1)}(\ux_{\alpha-1},\ux_{\alpha+1;\ell}) \ol{f^{\otimes \ell}(\ux_\ell')}\\
&\phantom{=} \qquad\times \paren*{w_{n,\beta'}^{(j),t}[g^{\times j};\ol{f}^{\times(j-1)}](x_\alpha) +\ol{w_{n,\beta}^{(j),t}[f^{\times (\beta-1)},\ol{g},f^{\times (j-\beta)};\ol{g}^{\times(j-1)}](x_\alpha)}}
\end{split},
\end{equation}
and
\begin{equation}
\label{eq:tr_Wn_form*}
\begin{split}
&\Tr_{\ell+1,\ldots,\ell+j-1}\paren*{\gamma^{(\ell+j-1)}(\W_{n,sa})_{(\ell+1,\ldots,\ell+\beta-1,\alpha,\ell+\beta,\ldots,\ell+j-1)}^{(j)}}(\ux_\ell;\ux_\ell') \\
&=\frac{1}{4}g^{\otimes \ell}(\ux_{\ell}) \ol{f^{\otimes (\ell-1)}(\ux_{\alpha-1}',\ux_{\alpha+1;\ell}')} \\
&\phantom{=}\qquad\times \paren*{\ol{w_{n,\beta'}^{(j),t}[f^{\times j};\ol{g}^{\times(j-1)}](x_\alpha')} + w_{n,\beta}^{(j),t}[g^{\times (\beta-1)},\ol{f},g^{\times(j-\beta)};\ol{f}^{\times(j-1)}](x_\alpha')} \\
&\phantom{=} +\frac{1}{4}f^{\otimes \ell}(\ux_{\ell}) \ol{g^{\otimes (\ell-1)}(\ux_{\alpha-1}',\ux_{\alpha+1;\ell}')}\\
&\phantom{=} \qquad\times \paren*{\ol{w_{n,\beta'}^{(j),t}[g^{\times j};\ol{f}^{\times(j-1)}](x_\alpha')} +w_{n,\beta}^{(j),t}[f^{\times (\beta-1)},\ol{g},f^{\times(j-\beta)};\ol{g}^{\times(j-1)}](x_\alpha')}
\end{split}.
\end{equation}
In all cases, equality holds in the sense of tempered distributions. 
\end{lemma}
\begin{proof}
By considerations of symmetry, it suffices to consider the case $\alpha=\ell$. Then by \cref{prop:ext_k} for the $(\ell+j-1)$-particle extension, \cref{prop:partial_trace} for the generalized partial trace, and the definition \eqref{wn_sa_def} for $\W_{n,sa}$, we find that
\begin{equation}
\begin{split}
&\Tr_{\ell+1,\ldots,\ell+j-1}\paren*{\W_{n,sa,(\ell+1,\ldots,\ell+\beta-1,\ell,\ell+\beta,\ldots,\ell+j-1)}^{(j)}\gamma^{(\ell+j-1)}}\\
&=\frac{1}{4}\Tr_{\ell+1,\ldots,\ell+j-1}\paren*{\wt{\W}_{n,(\ell+1,\ldots,\ell+\beta-1,\ell,\ell+\beta,\ldots,\ell+j-1)}^{(j)} \ket*{f^{\otimes (\ell+j-1)}}\bra*{g^{\otimes (\ell+j-1)}}} \\
&\phantom{=} +\frac{1}{4}\Tr_{\ell+1,\ldots,\ell+j-1}\paren*{\wt{\W}_{n,(\ell+1,\ldots,\ell+\beta-1,\ell,\ell+\beta,\ldots,\ell+j-1)}^{(j),*} \ket*{f^{\otimes (\ell+j-1)}}\bra*{g^{\otimes (\ell+j-1)}}} \\
&\phantom{=} +\frac{1}{4}\Tr_{\ell+1,\ldots,\ell+j-1}\paren*{\wt{\W}_{n,(\ell+1,\ldots,\ell+\beta-1,\ell,\ell+\beta,\ldots,\ell+j-1)}^{(j)} \ket*{g^{\otimes (\ell+j-1)}}\bra*{f^{\otimes (\ell+j-1)}}} \\
&\phantom{=} +\frac{1}{4}\Tr_{\ell+1,\ldots,\ell+j-1}\paren*{\wt{\W}_{n,(\ell+1,\ldots,\ell+\beta-1,\ell,\ell+\beta,\ldots,\ell+j-1)}^{(j),*} \ket*{g^{\otimes (\ell+j-1)}}\bra*{f^{\otimes (\ell+j-1)}}} \\
&= \frac{1}{4}\ket*{f^{\otimes (\ell-1)}}\bra*{g^{\otimes (\ell-1)}}\otimes\left( \Tr_{2,\ldots,j}\paren*{\wt{\W}_{n,(2,\ldots,\beta,1,\beta+1,\ldots,j)}^{(j)} \ket*{f^{\otimes j}}\bra*{g^{\otimes j}}} \right.\\
&\phantom{=}\hspace{55mm} \left.+ \Tr_{2,\ldots,j}\paren*{\wt{\W}_{n,(2,\ldots,\beta,1,\beta+1,\ldots,j)}^{(j),*} \ket*{f^{\otimes j}}\bra*{g^{\otimes j}}} \right) \\
&\phantom{=} + \frac{1}{4}\ket*{g^{\otimes (\ell-1)}}\bra*{f^{\otimes (\ell-1)}}\otimes\left(\Tr_{2,\ldots,j}\paren*{\wt{\W}_{n,(2,\ldots,\beta,1,\beta+1,\ldots,j)}^{(j)} \ket*{g^{\otimes j}}\bra*{f^{\otimes j}}} \right.\\
&\phantom{=}\hspace{55mm}\left. + \Tr_{2,\ldots,j}\paren*{\wt{\W}_{n,(2,\ldots,\beta,1,\beta+1,\ldots,j)}^{(j),*} \ket*{g^{\otimes j}}\bra*{f^{\otimes j}}} \right),
\end{split}
\end{equation}
where the ultimate equality follows from the tensor product structure. We introduce the permutation $\pi\in\Ss_j$ defined by
\begin{equation}
\pi(a) \coloneqq \begin{cases} a+1, & {1\leq a\leq \beta-1}\\ 1, &{a=\beta}\\ a, & {\beta+1\leq a\leq j} \end{cases},
\end{equation}
so that we can then write
\begin{equation}
\wt{\W}_{n,(2,\ldots,\beta,1,\beta+1,\ldots,j)}^{(j)} = \wt{\W}_{n,(\pi(1),\ldots,\pi(j))}^{(j)}
\end{equation}
and similarly for the adjoint. Using the notation $\Phi_{\wt{\W}_{n,(\pi(1),\ldots,\pi(j))}^{(j)}}$ introduced in \eqref{eq:Phi_Wn},  and similarly for the adjoint, we have that
\begin{equation}
\label{eq:Wn_f_g}
\begin{split}
&\Tr_{2,\ldots,j}\paren*{\wt{\W}_{n,(\pi(1),\ldots,\pi(j))}^{(j)} \ket*{f^{\otimes j}}\bra*{g^{\otimes j}}}(x;x') + \Tr_{2,\ldots,j}\paren*{\wt{\W}_{n,(\pi(1),\ldots,\pi(j))}^{(j),*} \ket*{f^{\otimes j}}\bra*{g^{\otimes j}}}(x;x')\\
&=\Phi_{\wt{\W}_{n,(\pi(1),\ldots,\pi(j))}^{(j)}}(f,\ldots,f;\ol{g},\ldots,\ol{g})(x;x') + \Phi_{\wt{\W}_{n,(\pi(1),\ldots,\pi(j))}^{(j),*}}(f,\ldots,f;\ol{g},\ldots,\ol{g})(x;x')
\end{split}
\end{equation}
and
\begin{equation}
\label{eq:Wn_g_f}
\begin{split}
&\Tr_{2,\ldots,j}\paren*{\wt{\W}_{n,(\pi(1),\ldots,\pi(j))}^{(j)} \ket*{g^{\otimes j}}\bra*{f^{\otimes j}}}(x;x') + \Tr_{2,\ldots,j}\paren*{\wt{\W}_{n,(\pi(1),\ldots,\pi(j))}^{(j),*} \ket*{g^{\otimes j}}\bra*{f^{\otimes j}}}(x;x') \\
&=\Phi_{\wt{\W}_{n,(\pi(1),\ldots,\pi(j))}^{(j)}}(g,\ldots,g;\ol{f},\ldots,\ol{f})(x;x') + \Phi_{\wt{\W}_{n,(\pi(1),\ldots,\pi(j))}^{(j),*}}(g,\ldots,g;\ol{f},\ldots,\ol{f})(x;x')
\end{split}
\end{equation}
in the sense of tempered distributions on $\R^2$. Next, applying \cref{lem:Wn_wn_ptr}, we obtain that for $\pi(1) = 1$ it holds that
\begin{align}
\eqref{eq:Wn_f_g} &= \ol{g(x')} \paren*{w_n^{(j)}[f^{\times j};\ol{g}^{\times(j-1)}](x) + \ol{w_{n,1}^{(j),t}[\ol{f},g^{\times(j-1)};\ol{f}^{\times(j-1)}](x)}},
\end{align}
and 
\begin{align}
\eqref{eq:Wn_g_f} &= \ol{f(x')}\paren*{w_n^{(j)}[g^{\times j};\ol{f}^{\times(j-1)}](x) + \ol{w_{n,1}^{(j),t}[\ol{g},f^{\times (j-1)};\ol{g}^{\times(j-1)}](x)}},
\end{align}
while if $\pi(1)\neq 1$, we have
\begin{align}
\eqref{eq:Wn_f_g} &= \ol{g(x')}\paren*{w_{n,\pi^{-1}(1)'}^{(j),t}[f^{\times j};\ol{g}^{\times(j-1)}](x) + \ol{w_{n,\pi^{-1}(1)}^{(j),t}[g^{\times(\pi^{-1}(1)-1)},\ol{f},g^{\times(j-\pi^{-1}(1))};\ol{f}^{\times(j-1)}](x)}},
\end{align}
and
\begin{align}
\eqref{eq:Wn_g_f} &= \ol{f(x')}\paren*{w_{n,\pi^{-1}(1)'}^{(j),t}[g^{\times j};\ol{f}^{\times(j-1)}](x)+\ol{w_{n,\pi^{-1}(1)}^{(j),t}[f^{(\pi^{-1}(1)-1)},\ol{g},f^{\times (j-\pi^{-1}(1))};\ol{g}^{\times(j-1)}](x)}}.
\end{align}
Since $\pi^{-1}(1)=\beta$ by definition of the permutation $\pi$, we obtain \eqref{eq:tr_Wn_form} after a little bookkeeping.

To obtain \eqref{eq:tr_Wn_form*} from \eqref{eq:tr_Wn_form}, observe that the self-adjointness of $\W_{n,sa}^{(j)}$ and $\gamma^{(\ell+j-1)}$ implies the Schwartz kernel identity
\begin{equation}
\begin{split}
&\ol{\Tr_{\ell+1,\ldots,\ell+j-1}\paren*{\W_{n,sa,(\ell+1,\ldots,\ell+\beta-1,\alpha,\ell+\beta,\ldots,\ell+j-1)}^{(j)}\gamma^{(\ell+j-1)}}(\ux_\ell';\ux_\ell)}\\
&=\Tr_{\ell+1,\ldots,\ell+j-1}\paren*{\gamma^{(\ell+j-1)}\W_{n,sa,(\ell+1,\ldots,\ell+\beta-1,\alpha,\ell+\beta,\ldots,\ell+j-1)}^{(j)}}(\ux_\ell;\ux_\ell').
\end{split}
\end{equation}
Substituting \eqref{eq:tr_Wn_form} into the left-hand side of the preceding identity yields the desired conclusion.
\end{proof}

We conclude this subsection by recording the required formula of the Hamiltonian vector field $X_{\H_n}$ which follows from the previous lemma and some algebraic manipulations.

\begin{lemma}
\label{lem:XHn_Sch_ker}
Suppose that $\Gamma=(\ket*{\phi^{\otimes k}}\bra*{\phi^{\otimes k}})_{k\in\N}$, for some $\phi\in\Sc(\R)$. Then for any $n\in\N$, we have the Schwartz kernel identity
\begin{equation}
\label{eq:XHn_Sch_ker}
\begin{split}
&X_{\H_n}(\Gamma)^{(\ell)}(\ux_\ell;\ux_\ell') \\
&=-\frac{i}{2}\sum_{j=1}^\infty \sum_{\alpha=1}^\ell \ket*{\phi^{\otimes (\ell-1)}}\bra*{\phi^{\otimes (\ell-1)}}(\ux_{\alpha-1},\ux_{\alpha+1;\ell};\ux_{\alpha-1}',\ux_{\alpha+1;\ell}') \\
&\phantom{=} \hspace{25mm}\times \left(\ol{\phi(x_\alpha')}\sum_{\beta=1}^j \paren*{w_{n,\beta'}^{(j),t}[\phi^{\times j};\ol{\phi}^{\times(j-1)}]+\ol{w_{n,\beta}^{(j),t}[\phi^{\times(\beta-1)},\ol{\phi},\phi^{(j-\beta)};\ol{\phi}^{\times(j-1)}]}}(x_\alpha)\right.\\
&\phantom{=} \hspace{40mm} \left.- \phi(x_\alpha)\sum_{\beta=1}^j \paren*{\ol{w_{n,\beta'}^{(j),t}[\phi^{\times j};\ol{\phi}^{\times(j-1)}]}+w_{n,\beta}^{(j),t}[\phi^{\times(\beta-1)},\ol{\phi},\phi^{(j-\beta)};\ol{\phi}^{\times(j-1)}]}(x_\alpha')\right)
\end{split}
\end{equation}
for every $\ell\in\N$.
\end{lemma}

\begin{proof}
We use the formula \eqref{eq:XH_n} and recalling definition \cref{eq:Wn_fin_def} for $\W_n$, we obtain that
\begin{equation}
\label{eq:vf_app}
\begin{split}
&X_{\H_n}(\Gamma)^{(\ell)}(\ux_\ell;\ux_\ell')\\
&=-i\sum_{j=1}^\infty \frac{1}{(j-1)!}\sum_{\pi\in\Ss_j}\Tr_{\ell+1,\ldots,\ell+-1}\paren*{\comm{\sum_{\alpha=1}^\ell \W_{n,sa,(\pi(\alpha),\pi(\ell+1),\ldots,\pi(\ell+\beta-1))}^{(j)}}{\gamma^{(\ell+j-1)}}},
\end{split}
\end{equation}
where here, $\Ss_j$ denotes the symmetric group on the set $\{\alpha,\ell+1,\ldots,\ell+j-1\}$. We can decompose $\Ss_j$ by 
\begin{equation}
\Ss_j = \bigcup_{r\in\{\alpha,\ell+1,\ldots,\ell+j-1\}} \{\pi\in\Ss_j : \pi^{-1}(\alpha) = r\} \eqqcolon \Ss_{j,r}.
\end{equation}
Note that each set in the partition has cardinality $(j-1)!$. It is a straightforward computation using the bosonic symmetry of $\gamma^{(\ell+j-1)}$ that
\begin{equation}
\begin{split}
&\Tr_{\ell+1,\ldots,\ell+j-1}\paren*{\comm{\W_{n,sa,(\pi(\alpha),\pi(\ell+1),\ldots,\pi(\ell+j-1))}^{(j)}}{\gamma^{(\ell+j-1)}}} \\
&=\begin{cases}
\Tr_{\ell+1,\ldots,\ell+j-1}\paren*{\comm{\W_{n,sa,(\alpha,\ell+1,\ldots,\ell+j-1)}^{(j)}}{\gamma^{(\ell+j-1)}}}, & r=\alpha\\
\Tr_{\ell+1,\ldots,\ell+j-1}\paren*{\comm{\W_{n,sa,(\ell+1,\ldots,r,\alpha,r+1,\ldots,\ell+j-1)}^{(j)}}{\gamma^{(\ell+j-1)}}}, & r\in\{\ell+1,\ldots,\ell+j-1\}
\end{cases}.
\end{split}
\end{equation}
Using these observations and applying \cref{lem:tr_Wn_form} to \eqref{eq:vf_app}, we obtain the Schwartz kernel identity
\begin{align}
\eqref{eq:vf_app} &= -i\sum_{j=1}^\infty \sum_{\alpha=1}^\ell \sum_{\beta=1}^j \Tr_{\ell+1,\ldots,\ell+j-1}\paren*{\comm{\W_{n,sa,(\ell+1,\ldots,\ell+\beta-1,\alpha,\ell+\beta,\ldots,\ell+j-1)}^{(j)}}{\gamma^{(\ell+j-1)}}}(\ux_\ell;\ux_\ell') \nonumber\\
&=-\frac{i}{2}\sum_{j=1}^\infty \sum_{\alpha=1}^\ell \ket*{\phi^{\otimes (\ell-1)}}\bra*{\phi^{\otimes (\ell-1)}}(\ux_{\alpha-1},\ux_{\alpha+1;\ell};\ux_{\alpha-1}',\ux_{\alpha+1;\ell}') \nonumber\\
&\phantom{=} \hspace{25mm}\times \left(\ol{\phi(x_\alpha')}\sum_{\beta=1}^j \paren*{w_{n,\beta'}^{(j),t}[\phi^{\times j};\ol{\phi}^{\times(j-1)}]+ \ol{w_{n,\beta}^{(j),t}[\phi^{\times(\beta-1)},\ol{\phi},\phi^{\times(j-\beta)};\ol{\phi}^{\times(j-1)}]}}(x_\alpha)\right. \nonumber\\
&\phantom{=} \hspace{40mm} -\left.\phi(x_\alpha)\sum_{\beta=1}^j \paren*{\ol{w_{n,\beta'}^{(j),t}[\phi^{\times j};\ol{\phi}^{\times(j-1)}]} + w_{n,\beta}^{(j),t}[\phi^{\times(\beta-1)},\ol{\phi},\phi^{\times(j-\beta)};\ol{\phi}^{\times(j-1)}]}(x_\alpha')\right).
\end{align}
This yields the desired formula.
\end{proof}

\subsection{Proof of \cref{thm:GP_fac}}\label{ssec:gp_flows_prf}
In this subsection, we prove \cref{thm:GP_fac}.

\begin{proof}[Proof of \cref{thm:GP_fac}]
Fix $n \in \N$. We would like to establish that $\Gamma=(\ket*{\phi^{\otimes k}}\bra*{\phi^{\otimes k}})_{k\in\N}$, where $\phi \in C^\infty(I; \Sc(\R))$, satisfies
\begin{align}
\label{nGP_pf}
\frac{d}{dt} \Gamma = X_{\H_n}(\Gamma),
\end{align}
i.e. $\Gamma$ is a solution to the $n$-th GP hierarchy, if
\begin{align}\label{nnls_pf}
\frac{d}{dt} \phi = \grad_s I_n(\phi),
\end{align}
i.e. $\phi$ is a solution to the $n$-th NLS. By the Leibnitz rule,
\begin{align}
\left(\frac{d}{dt} \Gamma \right)^{(\ell)} = \sum_{\alpha=1}^\ell \ket*{\phi^{\otimes (\alpha-1)}\otimes \frac{d}{dt} \phi \otimes \phi^{\otimes (\ell-\alpha)}}\bra*{\phi^{\otimes \ell}} + \ket*{\phi^{\otimes \ell}}\bra*{\phi^{\otimes (\alpha-1)}\otimes \frac{d}{dt} \phi \otimes \phi^{\otimes (\ell-\alpha)}}.
\end{align}
Substituting equation \eqref{nnls_pf} into the right-hand side of the preceding equality, we obtain that 
\begin{equation}
\label{eq:grads_I_sub}
\paren*{\frac{d}{dt}\Gamma}^{(\ell)} = \sum_{\alpha=1}^\ell \ket*{\phi^{\otimes (\alpha-1)}\otimes \grad_s I_n(\phi) \otimes \phi^{\otimes (\ell-\alpha)}}\bra*{\phi^{\otimes \ell}} + \ket*{\phi^{\otimes \ell}}\bra*{\phi^{\otimes (\alpha-1)}\otimes \grad_s I_n(\phi) \otimes \phi^{\otimes (\ell-\alpha)}}.
\end{equation}
Now the reader will recall that $\grad_s I_n$ is the symplectic gradient with respect to the form $\omega_{L^2}$ and by \eqref{symp_grad_tli} is given by the formula
\begin{equation}
\label{eq:grads_I_rec}
\grad_s I_n(\phi) = -\frac{i}{2}\sum_{j=1}^\infty \left\{\sum_{\beta=1}^j \paren*{\ol{w_{n,\beta}^{(j),t}[\phi^{\times(\beta-1)},\ol{\phi},\phi^{\times(j-\beta)};\ol{\phi}^{\times(j-1)}]} + w_{n,\beta'}^{(j),t}[\phi^{\times j};\ol{\phi}^{\times(j-1)}] }\right\}.
\end{equation}
Substituting identity \eqref{eq:grads_I_rec} into the right-hand side of \eqref{eq:grads_I_sub} and comparing the resulting expression with the formula \eqref{eq:XHn_Sch_ker} given by \cref{lem:XHn_Sch_ker} yields the desired conclusion.
\end{proof}

\subsection{An example: the fourth GP hierarchy}\label{ssec:gp_flows_ex}
We conclude this section with an example computation of one the $n$-th GP hierarchies. Specifically, we explicitly compute the equation of motion for the fourth GP hierarchy, which is the next one after the usual GP hierarchy (the third one in our terminology). In light of our \cref{thm:GP_fac}, the fourth GP hierarchy corresponds to the complex mKdV equation
\begin{equation}\label{cplx_mkdv}
\p_{t}\phi  = \p_{x}^{3}\phi -6\kappa|\phi|^{2}\p_{x}\phi, \qquad \kappa \in\{\pm 1\}.
\end{equation}

\begin{ex}[Fourth GP hierarchy]
We first recall from \cref{ex:W_4} that the
\begin{equation}\label{w4}
\W_{4}=\paren*{ (-i\p_{x_{1}})^{3}, -\frac{3\kappa i}{2}(\p_{x_{1}}+\p_{x_{2}})\delta(X_{1}-X_{2}), 0,\ldots}.
\end{equation}
Substituting \eqref{w4} into the right-hand side of the \eqref{eq:nGP}, using \cref{dual_dual} and the fact that $dH[\Gamma]^{(j)} = -i\W_{n}^{(j)}$ once again, the fourth GP equation, written in operator form, simplifies to
\begin{align}
\p_{t}\gamma^{(\ell)} &=\sum_{\alpha=1}^{\ell}\sum_{j=1}^{2}\sum_{\beta=1}^{j}\Tr_{\ell,\ldots,\ell+j-1}\paren*{(-i\W_{4}^{(j)})_{(\ell+1,\ldots,\ell+\beta-1,\alpha,\ell+\beta,\ldots,\ell+j-1)}\gamma^{(\ell+j-1)}} \nonumber\\
&\hspace{25mm}\phantom{=}-\Tr_{\ell+1,\ldots,\ell+j-1}\paren*{\gamma^{(\ell+j-1)}(-i\W_{4}^{(j)})_{(\ell+1,\ldots,\ell+\beta-1,\alpha,\ell+\beta,\ldots,\ell+j-1)}} \nonumber\\
&=-i\sum_{\alpha=1}^{\ell} \biggl( \W_{4,(\alpha)}^{(1)}\gamma^{(\ell)} + \gamma^{(\ell)}\W_{4,(\alpha)}^{(1)} + \Tr_{\ell+1}\paren*{\W_{4,(\alpha,\ell+1)}^{(2)}\gamma^{(\ell+1)} - \gamma^{(\ell+1)}\W_{4,(\alpha,\ell+1)}^{(2)}} \nonumber\\
& \hspace{60mm} \phantom{=}+\Tr_{\ell+1}\paren*{\W_{4,(\ell+1,\alpha)}^{(2)}\gamma^{(\ell+1)} - \gamma^{(\ell+1)}\W_{4,(\ell+1,\alpha)}^{(2)} } \biggr)\nonumber,
\end{align}
where we recall that the subscript notation is used to specify the variables on which the $\W_n^{(j)}$ operators act. By direct computation, this expression simplifies to yield
\begin{equation}
\p_{t}\gamma^{(\ell+1)} = \sum_{\alpha=1}^{\ell}(\p_{x_{\alpha}}^{3}+\p_{x_{\alpha}'}^{3})\gamma^{(\ell)} -6\kappa\paren*{B_{\alpha;\ell+1}^{+}(\p_{x_{\alpha}}\gamma^{(\ell+1)}) + B_{\alpha;\ell+1}^{-}(\p_{x_{\alpha}'}\gamma^{(\ell+1)})},
\end{equation}
which is the fourth GP hierarchy, and which can readily be seen to yield \eqref{cplx_mkdv} for factorized solutions.
\end{ex}

\appendix

\section{NLS Poisson commutativity}
\label{app:nls_pc}
\subsection{Transition and monodromy matrices}
In this appendix, we sketch the proof that the 1-particle functionals $I_n$ are involution with respect to the Poisson bracket $\pb{\cdot}{\cdot}_{L^2}$. We generalize the presentation to allow for the case where the two Schwartz functions $\psi,\bar{\psi}$ are independent, since this is the actual 1-particle result that we use in \cref{sec:invol}. Hence, rather than considering the scalar NLS equation \eqref{nls}, we consider the system
\begin{equation}
\label{eq:nls_sys}
\begin{cases}
(i\p_t+\Delta)\psi_1 = 2\kappa\psi_1^2\psi_2\\
(i\p_t-\Delta)\psi_2 = -2\kappa\psi_2^2\psi_1
\end{cases}, \qquad \kappa\in\{\pm 1\}.
\end{equation}
Our presentation will proceed at a high level, following the exposition in \cite[Chapter I and Chapter III]{FT07}; however, the reader may consult Chapter I, \S 7 and Chapter III, \S 4 of the aforementioned reference to fill in any omitted analytic details. We also consider the $L$-periodic case rather than entire real line. The extension to the latter case follows from truncation and periodization to fundamental domain $[-L,L]$, application of the periodic result, and then passage to the limit $L\rightarrow\infty$.

We start by fixing some notation. For $L>0$, we let $\T_L$ denote the domain $[-L,L]$ with periodic boundary conditions and $C^\infty(\T_L)$ the space of smooth functions on $\T_L$. Equivalently, $C^\infty(\T_L)$ is the space of smooth functions on the real line whose derivatives of all order are $2L$-periodic. Given a $(\C^2\otimes\C^2)$-valued functional $M_{(\psi_1,\psi_2)}$ on $C^\infty(\T_L)$, we define
\begin{equation}
M_{(\psi_1,\psi_2)}^\dagger \coloneqq \ol{M_{(\ol{\psi_2},\ol{\psi_1})}}
\end{equation}
where the complex conjugate of the matrix is taken entry-wise. Evidently, the $\dagger$ operation is involutive.

The system \eqref{eq:nls_sys} is a compatibility condition for the overdetermined system of equations
\begin{equation}
\label{eq:od_sys}
\begin{cases}
\p_x F_{(\psi_1,\psi_2)}(t,x,\lambda) = U_{(\psi_1,\psi_2)}(t,x,\lambda)F_{(\psi_1,\psi_2)}(t,x,\lambda), \\
\p_t F_{(\psi_1,\psi_2)}(t,x,\lambda) = V_{(\psi_1,\psi_2)}(t,x,\lambda)F_{(\psi_1,\psi_2)}(t,x,\lambda)
\end{cases},
\end{equation}
where $F_{(\psi_1,\psi_2)}$ is a spacetime $\C^2$-valued column vector and $U_{(\psi_1,\psi_2)}$ and $V_{(\psi_1,\psi_2)}$ are $\lambda$-dependent $2\times 2$ matrices given by
\begin{equation}
\label{eq:U_def}
U_{(\psi_1,\psi_2)}(\lambda) \coloneqq U_{0,(\psi_1,\psi_2)} + \lambda U_{1}, \qquad U_{0,(\psi_1,\psi_2)} \coloneqq \sqrt{\kappa}\begin{pmatrix} 0 & \psi_2 \\ \psi_1 & 0 \end{pmatrix}, \enspace U_1 \coloneqq \frac{1}{2i}\begin{pmatrix} 1 & 0 \\ 0 & -1 \end{pmatrix}
\end{equation}
and
\begin{equation}
\label{eq:V_def}
\begin{split}
&V_{(\psi_1,\psi_2)}(\lambda) \coloneqq V_{0,(\psi_1,\psi_2)} + \lambda V_{1,(\psi_1,\psi_2)} + \lambda^2 V_2,\\
&V_{0,(\psi_1,\psi_2)} \coloneqq i\sqrt{\kappa}\begin{pmatrix}\sqrt{\kappa}\psi_1\psi_2 & -\p_x\psi_2 \\ \p_x\psi_1 & -\sqrt{\kappa}\psi_1\psi_2\end{pmatrix}, \enspace V_{1,(\psi_1,\psi_2)} \coloneqq -U_{0,(\psi_1,\psi_2)}, \enspace V_2 \coloneqq -U_1.
\end{split}
\end{equation}
In the preceding and following material, $\lambda$ plays the role of an auxiliary spectral parameter. It will be convenient going forward to introduce notation for the $2\times 2$ Pauli matrices:
\begin{equation}
\label{eq:Pauli}
\sigma_1 \coloneqq \begin{pmatrix} 0 & 1\\ 1 & 0\end{pmatrix}, \enspace \sigma_2 \coloneqq \begin{pmatrix} 0 & -i \\ i & 0\end{pmatrix}, \enspace \sigma_3 \coloneqq \begin{pmatrix} 1 & 0 \\ 0 & -1\end{pmatrix}, \enspace \sigma_+ \coloneqq \frac{\sigma_1+i\sigma_2}{2}, \enspace \sigma_{-} \coloneqq \frac{\sigma_1-i\sigma_2}{2}.
\end{equation}
Written using $U$ and $V$, the compatibility condition for the system \eqref{eq:od_sys} is then
\begin{equation}
\label{eq:cc}
\p_t U_{(\psi_1,\psi_2)} -\p_x V_{(\psi_1,\psi_2)} + \comm{U_{(\psi_1,\psi_2)}}{V_{(\psi_1,\psi_2)}} =0
\end{equation}
point-wise in $\lambda$. In the sequel, we will omit the subscript $(\psi_1,\psi_2)$, which shows that the matrices are really matrix-valued functionals evaluated at a specific point, except when invoking the dependence is necessary. We hope that this omission will not result in any confusion on the reader's part.

There is a geometric interpretation to \eqref{eq:cc} in terms of local connection coefficients in the vector bundle $\R^2\times\C^2$. Equation \eqref{eq:cc} then says that the $(U,V)$-connection has zero curvature. For this reason, \eqref{eq:cc} is often called the \emph{zero curvature representation} in the literature. We will not emphasize this geometric aspect in the appendix, as it does not play a role for us.

Now fix a time $t_0$ and consider the \emph{auxiliary linear problem}
\begin{equation}
\label{eq:aux_lp}
\p_x F = U(t_0,x,\lambda)F.
\end{equation}
The object of interest associated to \eqref{eq:aux_lp} is the \emph{monodromy matrix}, which is the matrix of parallel transport along the contour $t=t_0$, $-L\leq x\leq L$ positively oriented:
\begin{equation}
\label{eq:mm}
T_L(\lambda,t_0) \coloneqq \overset{\curvearrowleft}\exp\paren*{\int_{-L}^L dx U(x,t_0,\lambda)},
\end{equation}
where $ \overset{\curvearrowleft}\exp$ denotes the path-ordered exponential.\footnote{For $A\in L^\infty(\T_L;\C^n\otimes\C^n)$, the \emph{path-ordered exponential} of $A$ is defined by
\begin{equation} \overset{\curvearrowleft}\exp\paren*{\int_{-L}^x dz A(z)} \coloneqq \sum_{n=0}^\infty \int_{-L}^x dx_n\int_{-L}^{x_n}dx_{n-1}\cdots\int_{-L}^{x_2}dx_1 A(x_n)\cdots A(x_1).\end{equation}} By using the superposition principle for parallel transport and the fact that parallel transport along a closed curve is trivial, one can show that the monodromy matrices are conjugate for different values of $t$. Consequently, the trace of the monodromy matrix is constant in time:
\begin{equation}
\tr_{\C^2} T_L(\lambda,t_2) = \tr_{\C^2} T_L(\lambda,t_1), \qquad \forall t_1,t_2\in\R,
\end{equation}
where $\tr_{\C^2}$ denotes the $2\times 2$ matrix trace. Furthermore, one can show that the choice of fundamental domain $[-L,L]$ in the definition \eqref{eq:mm} is immaterial to computing the trace. We conclude that
\begin{equation}
\label{eq:F_gfunc}
\tl{F}_L(\lambda) \coloneqq \tr_{\C^2} T_{L}(\lambda)
\end{equation}
is a \emph{generating function} for the conservation laws of \eqref{eq:nls_sys}.

More generally, we have the \emph{transition matrix}, which is the matrix of parallel transport from $y$ to $x$ along the $x$-axis:
\begin{equation}
\label{eq:tm}
T(x,y,\lambda) \coloneqq \overset{\curvearrowleft}\exp\paren*{\int_y^x dz U(z,\lambda)}.
\end{equation}
The monodromy matrix is then the special case of the transition matrix obtained by setting $(x,y)=(L,-L)$. From the definition \eqref{eq:tm}, it is immediate that the transition matrix satisfies the Cauchy problem
\begin{equation}
\label{eq:tm_ODE_x}
\begin{cases}
\p_x T(x,y,\lambda)  = U(x,\lambda)T(x,y,\lambda) \\
T(x,y,\lambda)|_{x=y} = I_{\C^2}
\end{cases},
\end{equation}
where $I_{\C^2}$ is the identity matrix on $\C^2$. $T(x,y,\lambda)$ is a smooth function of $(x,y)$ and is also analytic in $\lambda$ due to the analyticity of $U(x,\lambda)$ and the initial datum. By using that $\int_y^x = -\int_x^y$ in \eqref{eq:tm}, we see that $T(x,y,\lambda)$ also satisfies the ODE
\begin{equation}
\label{eq:tm_ODE_y}
\p_y T(x,y,\lambda) = -T(x,y,\lambda)U(y,\lambda).
\end{equation}
Additionally, the transition matrix has several elementary properties, which we record with the following lemma.

\begin{lemma}
\label{lem:tm_props}
The following properties hold:
\begin{enumerate}[(i)]
\item
\label{item:tm_sup}
$T(x,z,\lambda) T(z,y,\lambda) = T(x,y,\lambda)$,
\item
\label{item:tm_rev}
$T(x,y,\lambda)=T^{-1}(y,x,\lambda)$,
\item
\label{item:tm_unimod}
$\det_{\C^2} T(x,y,\lambda)=1$.
\end{enumerate}
\end{lemma}
\begin{proof}
Properties \ref{item:tm_sup} and \ref{item:tm_rev} are straightforward, and we leave them to the reader. For property \ref{item:tm_unimod}, the reader will recall Jacobi's formula that for any $n \times n$ matrix $A(t)$,
\begin{equation}
\frac{d}{dt}\det_{\C^n}(A(t)) = \tr_{\C^n}\paren*{\mathrm{adj}(A(t))\dfrac{dA(t)}{dt}},
\end{equation}
where $\mathrm{adj}(A(t))$ is the adjugate of $A(t)$ (i.e. the transpose of the cofactor matrix of $A(t)$). Fixing $y,\lambda$ and applying Jacobi's formula to $T(x,y,\lambda)$ with independent variable $x$ instead of $t$ and also using the equation \eqref{eq:tm_ODE_x}, we find that $\det_{\C^2}(T(x,y,\lambda))$ is a solution to the Cauchy problem
\begin{equation}
\begin{cases}
\p_x \det_{\C^2}(T(x,y,\lambda)) &= \tr_{\C^2}\paren*{\mathrm{adj}(T(x,y,\lambda)) U(x,\lambda)T(x,y,\lambda)},\\
\det_{\C^2}(T(x,y,\lambda))|_{x=y} &= 1
\end{cases}
\end{equation}
Since
\begin{equation}
\mathrm{adj}(T(x,y,\lambda)) = \begin{pmatrix} T^{22}(x,y,\lambda) & - T^{12}(x,y,\lambda)\\ -T^{21}(x,y,\lambda) & T^{11}(x,y,\lambda)\end{pmatrix},
\end{equation}
it follows by direct computation that
\begin{equation}
T(x,y,\lambda)\mathrm{adj}(T(x,y,\lambda)) = \det_{\C^2}(T(x,y,\lambda))Id_{\C^2}.
\end{equation}
So by the cyclicity and linearity of trace,  $\det_{\C^2}(T(x,y,\lambda))$ is the unique constant solution to the Cauchy problem
\begin{equation}
\begin{cases}
\p_x\det_{\C^2}(T(x,y,\lambda))) &= \det_{\C^2}(T(x,y,\lambda))\tr_{\C^2}\paren*{U(x,y,\lambda)I_{\C^2}} = 0 \\
\det_{\C^2}(T(x,y,\lambda))|_{x=y} &= 1
\end{cases},
\end{equation}
where we use that $U(x,y,\lambda)$ is trace-less. Thus, the proof of \ref{item:tm_unimod} is complete.
\end{proof}

It is evident from its definition \eqref{eq:U_def} that
\begin{equation}
\label{eq:U_invol}
U_{(\psi_1,\psi_2)}^\dagger(x,\lambda) = \sigma U_{(\psi_1,\psi_2)}(x,\bar{\lambda})\sigma,
\end{equation}
where
\begin{equation}
\label{eq:sigma_def}
\sigma =
\begin{cases}
\sigma_1, & \kappa=1 \\
\sigma_2, & \kappa=-1
\end{cases},
\end{equation}
where $\kappa$ is the defocusing/focusing parameter in \eqref{eq:nls_sys} and $\sigma_1,\sigma_2$ are the Pauli matrices in \eqref{eq:Pauli}. The transition matrix also satisfies an important involution relation leading to the special structure of the matrix $T(x,y,\lambda)$, which we isolate in the next lemma.

\begin{lemma}
\label{lem:tm_invol}
$T(x,y,\lambda)$ has the involution property
\begin{equation}
\label{eq:tm_invol}
\sigma T_{(\psi_1,\psi_2)}(x,y,\bar{\lambda})\sigma = T_{(\psi_1,\psi_2)}^\dagger(x,y,\lambda).
\end{equation}
Consequently, we can write the monodromy matrix $T_{L,\psi_1,\psi_2}(\lambda)$ as
\begin{equation}
T_{L,(\psi_1,\psi_2)}(\lambda) =
\begin{pmatrix}
a_{L,(\psi_1,\psi_2)}(\lambda) & \sgn(\kappa)b_{L,(\psi_1,\psi_2)}^\dagger(\bar{\lambda}) \\ b_{L,(\psi_1,\psi_2)}(\lambda) & a_{L,(\psi_1,\psi_2)}^\dagger(\bar{\lambda})
\end{pmatrix},
\end{equation}
where $a_{L,(\psi_1,\psi_2)}^\dagger(\lambda) \coloneqq \ol{a_{L,(\ol{\psi_2},\ol{\psi_1})}(\lambda)}$ and analogously for $b_{L,(\psi_1,\psi_2)}^\dagger$.
\end{lemma}
\begin{proof}
Since the Cauchy problem \eqref{eq:tm_ODE_x} has a unique solution and $\sigma^2=I_{\C^2}$, it suffices to show that the matrix
\begin{equation}
\label{eq:tl_T}
\tl{T}_{(\psi_1,\psi_2)}(x,y,\lambda) \coloneqq \sigma T_{(\psi_1,\psi_2)}^\dagger(x,y,\bar{\lambda})\sigma
\end{equation}
is a solution of \eqref{eq:tm_ODE_x}.

It is evident from $T_{(\psi_1,\psi_2)}(x,y,\lambda)|_{x=y}=I_{\C^2}$ and $\sigma^2=I_{\C^2}$ that the initial condition holds. Now using that $\p_x$ commutes with left- (and right-) multiplication by a constant matrix and complex conjugation, we find that
\begin{align}
\p_x\tl{T}_{(\psi_1,\psi_2)}(x,y,\lambda) &= \sigma\ol{\p_x T_{(\ol{\psi_2},\ol{\psi_1})}(x,y,\bar{\lambda})}\sigma \nonumber\\
&=\sigma\ol{ U_{(\ol{\psi_2},\ol{\psi_1})}(x,\ol{\lambda})T_{(\ol{\psi_2},\ol{\psi_1})}(x,y,\bar{\lambda})}\sigma \nonumber\\
&=\sigma U_{(\psi_1,\psi_2)}^\dagger(x,\ol{\lambda}) T_{(\psi_1,\psi_2)}^\dagger(x,y,\bar{\lambda})\sigma,
\end{align}
where the penultimate equality follows from application of \eqref{eq:tm_ODE_x} with $(\psi_1,\psi_2)$ replaced by $(\ol{\psi_2},\ol{\psi_1})$ and the ultimate equality follows from the definition of the dagger superscript. Since $\sigma^2=I_{\C^2}$, we can use the associativity of matrix multiplication together with the identity \eqref{eq:U_invol} to write
\begin{align}
\sigma U_{(\psi_1,\psi_2)}^\dagger(x,\ol{\lambda}) T_{(\psi_1,\psi_2)}^\dagger(x,y,\ol{\lambda})\sigma &= \paren*{\sigma U_{(\psi_1,\psi_2)}^\dagger(x,\ol{\lambda})\sigma}\paren*{\sigma T_{(\psi_1,\psi_2)}^\dagger(x,y,\ol{\lambda})\sigma} \nonumber\\
&=U_{(\psi_1,\psi_2)}(x,\lambda)\tl{T}_{(\psi_1,\psi_2)}(x,y,\lambda),
\end{align}
which is exactly what we needed to show.

We now show the second assertion concerning the structure of the monodromy matrix. We only present the details in the case $\kappa=1$ and leave the $\kappa=-1$ case as an exercise for the reader. Writing
\begin{equation}
T_{(\psi_1,\psi_2)}(x,y,\lambda) =
\begin{pmatrix}
T_{(\psi_1,\psi_2)}^{11}(x,y,\lambda) & T_{(\psi_1,\psi_2)}^{12}(x,y,\lambda) \\
T_{(\psi_1,\psi_2)}^{21}(x,y,\lambda) & T_{(\psi_1,\psi_2)}^{22}(x,y,\lambda)
\end{pmatrix},
\end{equation}
we see from direct computation that
\begin{align}
\sigma T_{(\psi_1,\psi_2)}(x,y,\bar{\lambda})\sigma &=
\begin{pmatrix}
0 & 1\\
1 & 0
\end{pmatrix}
\begin{pmatrix}
T_{(\psi_1,\psi_2)}^{12}(x,y,\bar{\lambda}) & T_{(\psi_1,\psi_2)}^{11}(x,y,\bar{\lambda}) \\
T_{(\psi_1,\psi_2)}^{22}(x,y,\bar{\lambda}) & T_{(\psi_1,\psi_2)}^{21}(x,y,\bar{\lambda})
\end{pmatrix} \nonumber\\
&=\begin{pmatrix} T_{(\psi_1,\psi_2)}^{22}(x,y,\bar{\lambda}) & T_{(\psi_1,\psi_2)}^{21}(x,y,\bar{\lambda}) \\ T_{(\psi_1,\psi_2)}^{12}(x,y,\bar{\lambda}) & T_{(\psi_1,\psi_2)}^{11}(x,y,\bar{\lambda}) \end{pmatrix}.
\end{align}
Now by the involution property \eqref{eq:tm_invol} and the definition of $T_{(\psi_1,\psi_2)}^\dagger$, we see that
\begin{align}
\begin{pmatrix}
\ol{T_{(\ol{\psi_2},\ol{\psi_1})}^{11}(x,y,\lambda)} & \ol{T_{(\ol{\psi_2},\ol{\psi_1})}^{12}(x,y,\lambda)} \\
\ol{T_{(\ol{\psi_2},\ol{\psi_1})}^{21}(x,y,\lambda)} & \ol{T_{(\ol{\psi_2},\ol{\psi_1})}^{22}(x,y,\lambda)}
\end{pmatrix}
&= T_{(\psi_1,\psi_2)}^\dagger(x,y,\lambda) \nonumber\\
&= 
\begin{pmatrix} T_{(\psi_1,\psi_2)}^{22}(x,y,\bar{\lambda}) & T_{(\psi_1,\psi_2)}^{21}(x,y,\bar{\lambda}) \\ T_{(\psi_1,\psi_2)}^{12}(x,y,\bar{\lambda}) & T_{(\psi_1,\psi_2)}^{11}(x,y,\bar{\lambda}) \end{pmatrix}.
\end{align}
Evaluating this identity at $(x,y) = (L,-L)$ and defining
\begin{equation}
a_{L,(\psi_1,\psi_2)}(\lambda) \coloneqq T_{L,(\psi_1,\psi_2)}^{11}(\lambda) , \quad b_{L,(\psi_1,\psi_2)}(\lambda) \coloneqq T_{L,(\psi_1,\psi_2)}^{21}(\lambda),
\end{equation}
we obtain the desired conclusion.
\end{proof}

\begin{remark}
Since the transition matrix is an entire function of $\lambda$, it follows that the functions $a_{L,(\psi_1,\psi_2)}, a_{L,(\psi_1,\psi_2)}^\dagger, b_{L,(\psi_1,\psi_2)}, b_{L,(\psi_1,\psi_2)}^\dagger$ are entire functions as well. In fact, they are of exponential type $L$. Moreover, the unimodularity property \ref{item:tm_unimod} for the transition matrix implies the normalization condition
\begin{equation}
a_{L,(\psi_1,\psi_2)}(\lambda)a_{L,(\psi_1,\psi_2)}^\dagger(\lambda) -\sgn(\kappa) b_{L,(\lambda_1,\lambda_2)}(\lambda)b_{L,(\psi_1,\psi_2)}^\dagger(\lambda)=1, \qquad \lambda\in\R.
\end{equation}
\end{remark}

We close this subsection with an alternative way to see that the trace of the monodromy matrix, which we called $\tl{F}_L(\lambda)$ in \eqref{eq:F_gfunc}, is conserved in time. By differentiating both sides of equation \eqref{eq:tm_ODE_x} with respect to time and performing some algebraic manipulation, one finds that
\begin{equation}
\p_t T(t,x,y,\lambda) = V(t,x,\lambda)T(t,x,y,\lambda) - T(t,x,y,\lambda) V(t,y,\lambda)
\end{equation}
Since $V$ is $2L$-periodic and therefore $V(t,L,\lambda) = V(t,-L,\lambda)$, it follows that the monodromy matrix satisfies the von Neumann equation
\begin{equation}
\p_t T_L(t,\lambda) = \comm{V(t,L,\lambda)}{T_L(t,\lambda)}.
\end{equation}
Since differentiation commutes with the trace and the trace of a commutator is zero, it follows that
\begin{equation}
\p_t \tr_{\C^2}(T_L(t,\lambda)) = 0.
\end{equation}

\subsection{Integrals of motion}
\label{ssec:nls_pc_im}
We now use an asymptotic expansion for the generating functional $\tl{F}_L(\lambda)$ (recall \eqref{eq:F_gfunc}) to identify conserved quantities for the system \eqref{eq:nls_sys}. We start by finding a gauge transformation that reduces the transition matrix to diagonal form $\exp Z(x,y,\lambda)$:
\begin{equation}
\label{eq:tm_gt}
T(x,y,\lambda) = \paren*{I_{\C^2}+W(x,\lambda)} \exp\paren*{Z(x,y,\lambda)}\paren*{I_{\C^2}+W(y,\lambda)}^{-1},
\end{equation}
where $W$ and $Z$ are off-diagonal and diagonal matrices, respectively. We will see that $W$ and $Z$ have the large real $\lambda$ asymptotic expansions
\begin{equation}
\label{eq:W_asy}
W(x,\lambda) \sim \sum_{n=1}^\infty \frac{W_n(x)}{\lambda^n}, \enspace Z(x,y,\lambda) \sim \frac{(x-y)\lambda\sigma_3}{2i}+ \sum_{n=1}^\infty \frac{Z_n(x,y,\lambda)}{\lambda^n},
\end{equation}
where the reader will recall the Pauli matrix $\sigma_3$ from \eqref{eq:Pauli}. Here and throughout the appendix, the asymptotic should be interpreted as follows: for any $k\in\N$,
\begin{equation}
\begin{split}
o(|\lambda|^{-k}) &= \sup_{-L\leq x\leq L} \|W(x,\lambda)-\sum_{n=1}^k \frac{W_n(x)}{\lambda^n}\| \\
&\phantom{=} + \sup_{-L\leq x,y,\leq L} \|Z(x,y,\lambda)-\frac{(x-y)\lambda\sigma_3}{2i} -\sum_{n=1}^k \frac{Z_n(x,y,\lambda)}{\lambda^n}\|
\end{split}
\end{equation}
as $|\lambda| \rightarrow \infty$ on the real line, where $\|\cdot\|$ denotes any matrix norm.

Proceeding formally to identify the relevant equations, we substitute \eqref{eq:tm_gt} into the transition matrix differential equation \eqref{eq:tm_ODE_x} and use the Leibnitz rule to obtain that
\begin{equation}
\begin{split}
&U(x,\lambda)\paren*{I_{\C^2}+W(x,\lambda)}\exp\paren*{Z(x,y,\lambda)}\paren*{I_{\C^2}+W(y,\lambda)}^{-1}\\
&=\p_x W(x,\lambda) \exp\paren*{Z(x,y,\lambda)}\paren*{I_{\C^2}+W(y,\lambda)}^{-1}\\
&\phantom{=} + \paren*{I_{\C^2}+W(x,\lambda)}\p_x Z(x,y,\lambda) \exp\paren*{Z(x,y,\lambda)}\paren*{I_{\C^2}+W(y,\lambda)}^{-1},
\end{split}
\end{equation}
which can be manipulated to yield
\begin{equation}
\label{eq:di_odi_split}
U(x,\lambda)\paren*{I_{\C^2}+W(x,\lambda)} = \p_x W(x,\lambda)+\paren*{I_{\C^2}+W(x,\lambda)}\p_x Z(x,y,\lambda).
\end{equation}
Recalling from \eqref{eq:U_def} that $U(x,\lambda)=U_0(x) + \lambda U_1$, where $U_0$ is off-diagonal and $U_1$ is diagonal, and decomposing both sides of \eqref{eq:di_odi_split} into off-diagonal and diagonal parts, we find that $W$ and $Z$ satisfy the coupled system of equations
\begin{equation}
\begin{cases}
\p_x W + W\p_x Z = U_0 + \lambda U_1 W \\
\p_x Z= U_0 W+\lambda U_1
\end{cases}.
\end{equation}
Substituting the second equation into the first one and using that $U_1$ anticommutes with $W$, we find that $W$ satisfies the matrix Riccati equation
\begin{equation}
\label{eq:W_ODE}
\p_x W+i\lambda\sigma_3 W+ W U_0 W-U_0 = 0.
\end{equation}
One can rewrite \eqref{eq:W_ODE} as an integral equation and use the fixed-point method to show that \eqref{eq:W_ODE} has a smooth solution on $\T_L$ for sufficiently large $\lambda$ depending on the data $(\|\phi\|_{L^1(\T_L)}, \|\phi\|_{L^\infty(\T_L)}, L)$, with the asymptotic expansion \eqref{eq:W_asy}. We can then solve for $Z$ subject to the initial condition $Z(x,y,\lambda)|_{x=y}=0_{\C^2}$ by
\begin{equation}
\label{eq:Z_sol}
Z(x,y,\lambda) = \frac{\lambda(x-y)}{2i}\sigma_3 + \int_y^x dz \ U_0(z)W(z,\lambda).
\end{equation}
In particular, the asymptotic expansion of $Z$ is then determined by the asymptotic expansion for $W$. $W$ and $Z$ satisfy \eqref{eq:tm_gt} since both the left-hand side and right-hand side of the equation \eqref{eq:tm_gt} are solutions to the same Cauchy problem, which has a unique solution.

Next, substituting the expansion $\sum_{n=1}^\infty \frac{W_n(x)}{\lambda^n}$ into equation \eqref{eq:W_ODE}, we find that the coefficients $W_n(x)$ satisfy the recursion relation
\begin{equation}
\label{eq:Wn_m_recur}
\begin{split}
W_1(x) &= -i\sigma_3 U_0(x) = i\sqrt{\kappa}\begin{pmatrix} 0 & -\psi_2(x) \\ \psi_1(x) & 0\end{pmatrix}, \\
W_{n+1}(x) &= i\sigma_3\paren*{\p_x W_n(x) + \sum_{k=1}^{n-1} W_k(x) U_0(x) W_{n-k}(x)}.
\end{split}
\end{equation}
Evidently, the matrices $W_n(x)$ are $2L$-periodic and are polynomials of the derivatives of $U_0(x)$. By equation \eqref{eq:W_ODE} for $W$ and the continuity method together with the equation \eqref{eq:Z_sol} for $Z$, one can show that the asymptotic \eqref{eq:W_asy} holds. In the next lemma, we record an important involution property of the $W_n$. As before with $U$, we include the subscripts $(\psi_1,\psi_2)$ in the sequel to denote the underlying dependence.

\begin{lemma}
\label{lem:W_invol}
For every $n\in\N$, it holds that $W_n$ is off-diagonal and 
\begin{equation}
\label{eq:W_invol}
W_{n,(\psi_1,\psi_2)}^\dagger(x) = \sigma W_{n,(\psi_1,\psi_2)}(x)\sigma,
\end{equation}
where $\sigma$ is as in \eqref{eq:sigma_def}. Additionally, $W_{n,(\psi_1,\psi_2)}(x)$ has the form
\begin{equation}
i\sqrt{\kappa}
\begin{pmatrix}
0 & -w_{n,(\psi_1,\psi_2)}^\dagger(x)\\
w_{n,(\psi_1,\psi_2)}(x) & 0
\end{pmatrix},
\end{equation}
where the functions $w_{n,(\psi_1,\psi_2)}(x)$ satisfy the recursion relation
\begin{equation}
\label{eq:wn_b_rec}
\begin{split}
w_{1,(\psi_1,\psi_2)}(x) &= \psi_1(x), \\
w_{n+1,(\psi_1,\psi_2)}(x) &= -i\p_x w_n(x) + \kappa\psi_2(x)\sum_{k=1}^{n-1}w_{k,(\psi_1,\psi_2)}(x)w_{n-k,(\psi_1,\psi_2)}(x).
\end{split}
\end{equation}
\end{lemma}
\begin{proof}
We prove the lemma by strong induction on $n$ using the recursion formula \eqref{eq:Wn_m_recur}. The base case $n=1$ follows from
\begin{equation}
\label{eq:U0_invol}
U_{0,(\psi_1,\psi_2)}^\dagger(x) = \sigma U_{0,(\psi_1,\psi_2)}(x)\sigma
\end{equation}
and the fact that $\sigma$ anti-commutes with $\sigma_3$.

For the induction step, suppose that for some $n\in\N$, the involution relation holds for all $k\in\N_{\leq n-1}$. Multiplying \eqref{eq:Wn_m_recur} by $\sigma$ on the left and right and using that $\sigma^2=I_{\C^2}$, we find that
\begin{align}
&\sigma W_{n+1,(\psi_1,\psi_2)}(x)\sigma \nonumber\\
&= i\sigma \sigma_3\paren*{\p_x W_{n,(\psi_1,\psi_2)}(x) + \sum_{k=1}^{n-1} W_{k,(\psi_1,\psi_2)}(x) U_{0,(\psi_1,\psi_2)}(x) W_{n-k,(\psi_1,\psi_2)}(x)}\sigma \nonumber\\
&=-i\sigma_3\paren*{\p_x (\sigma W_{n,(\psi_1,\psi_2)}(x)\sigma) + \sum_{k=1}^{n-1} (\sigma W_{k,(\psi_1,\psi_2)}(x)\sigma)(\sigma U_{0,(\psi_1,\psi_2)}(x)\sigma) (\sigma W_{n-k,(\psi_1,\psi_2)}(x)\sigma)} \nonumber\\
&=-i\sigma_3\paren*{\p_x W_{n,(\psi_1,\psi_2)}^\dagger(x) + \sum_{k=1}^n W_{k,(\psi_1,\psi_2)}^\dagger(x) U_{0,(\psi_1,\psi_2)}^\dagger(x) W_{n-k,(\psi_1,\psi_2)}^\dagger(x)},
\end{align}
where we again use \eqref{eq:U0_invol} and the anti-commutativity of $\sigma$ and $\sigma_3$ to obtain the penultimate equality and the induction hypothesis to obtain the ultimate equality. Since $(i\sigma_3)^\dagger=-i\sigma_3$ and the $\dagger$ operation is a homomorphism of algebras which commutes with differentiation, \eqref{eq:W_invol} is proved. Since $W_{1,(\psi_1,\psi_2)},\ldots,W_{n,(\psi_1,\psi_2)}$ are off-diagonal, it it follows from some basic algebra and the diagonality and off-diagonality of $\sigma_3$ and $U_0$, respectively, that $W_{n+1, (\psi_1,\psi_2)}$ is off-diagonal. Thus, the proof of the induction step is complete.

Now since $W_{n,(\psi_1,\psi_2)}$ is off-diagonal, it takes the form
\begin{equation}
W_{n,(\psi_1,\psi_2)} =
\begin{pmatrix}
0 & w_{n,(\psi_1,\psi_2)}^{12} \\
w_{n,(\psi_1,\psi_2)}^{21} & 0
\end{pmatrix}, \qquad w_{n,(\psi_1,\psi_2)}^{12}, w_{n,(\psi_1,\psi_2)}^{21}\in C^\infty(\T_L),
\end{equation}
which by direct computation implies that
\begin{equation}
\sigma W_{n,(\psi_1,\psi_2)} \sigma = \begin{pmatrix} 0 & \sgn(\kappa)w_{n,(\psi_1,\psi_2)}^{21} \\ \sgn(\kappa)w_{n,(\psi_1,\psi_2)}^{12} & 0 \end{pmatrix},
\end{equation}
Now the involution relation \eqref{eq:W_invol} implies the equality
\begin{equation}
\begin{pmatrix} 0 & \sgn(\kappa)w_{n,(\psi_1,\psi_2)}^{21} \\ \sgn(\kappa)w_{n,(\psi_1,\psi_2)}^{12} & 0 \end{pmatrix} = W_{n,(\psi_1,\psi_2)}^{\dagger} = \begin{pmatrix} 0 & w_{n,(\psi_1,\psi_2)}^{12,\dagger} \\ w_{n,(\psi_1,\psi_2)}^{21,\dagger} & 0\end{pmatrix}.
\end{equation}
Therefore, defining $w_{n,(\psi_1,\psi_2)} \coloneqq w_{n,(\psi_1,\psi_2)}^{21}/(i\sqrt{\kappa})$, we can write $W_{n,(\psi_1,\psi_2)}$ in the form
\begin{equation}
W_{n,(\psi_1,\psi_2)} = 
i\sqrt{\kappa}
\begin{pmatrix}
0 & -w_{n,(\psi_1,\psi_2)}^\dagger(x)\\
w_{n,(\psi_1,\psi_2)}(x) & 0
\end{pmatrix},
\end{equation}
where by \eqref{eq:Wn_m_recur}, the functions $w_{n,(\psi_1,\psi_2)}(x)$ satisfy the recursion relation
\begin{equation}
\label{eq:wn(,)_def}
\begin{split}
w_{1,(\psi_1,\psi_2)}(x) = \psi_1(x), \\
w_{n+1,(\psi_1,\psi_2)}(x) = -i\p_x w_{n,(\psi_1,\psi_2)}(x) + \kappa\psi_2(x)\sum_{k=1}^{n-1}w_{k,(\psi_1,\psi_2)}(x)w_{n-k,(\psi_1,\psi_2)}(x).
\end{split}
\end{equation}
Thus, the proof of the lemma is complete.
\end{proof}

By using the equation \eqref{eq:W_ODE}, one can also show that $W_{(\psi_1,\psi_2)}(x,\lambda)$ satisfies the same involutive property as $W_n$. So we can write
\begin{equation}
\label{eq:W_w}
W_{(\psi_1,\psi_2)}(x,\lambda) = i\sqrt{\kappa}\paren*{w_{(\psi_1,\psi_2)}(x,\lambda)\sigma_{-} - w_{(\psi_1,\psi_2)}^\dagger(x,\bar{\lambda})\sigma_+},
\end{equation}
where $\sigma_{\pm}$ are defined in \eqref{eq:Pauli} and where $w_{(\psi_1,\psi_2)}$ has the large real lambda asymptotic expansion
\begin{equation}
\label{eq:w_asy}
w_{(\psi_1,\psi_2)}(x,\lambda)\sim \sum_{n=1}^\infty \frac{w_{n,(\psi_1,\psi_2)}(x)}{\lambda^n}.
\end{equation}
Using equation \eqref{eq:Z_sol} for $Z_{(\psi_1,\psi_2)}(x,y,\lambda)$ and evaluating $(x,y)=(L,-L)$, we find that
\begin{align}
Z_{L,(\psi_1,\psi_2)}(\lambda) &\coloneqq Z_{(\psi_1,\psi_2)}(L,-L,\lambda) \nonumber\\
&= \frac{\lambda L}{i}\sigma_3 + \int_{-L}^L dz U_{(\psi_1,\psi_2)}(z)W_{(\psi_1,\psi_2)}(z,\lambda) \nonumber\\
&=\begin{pmatrix} -i\lambda L & 0 \\ 0 & i\lambda L\end{pmatrix} + \int_{-L}^L dz \begin{pmatrix}0 &\sqrt{\kappa}\psi_2(z)\\ \sqrt{\kappa}\psi_1(z) & 0\end{pmatrix}\begin{pmatrix} 0 & -i\sqrt{\kappa}w_{(\psi_1,\psi_2)}^\dagger(z,\lambda) \\ i\sqrt{\kappa}w_{(\psi_1,\psi_2)}(z,\lambda) & 0\end{pmatrix} \nonumber\\
&=\begin{pmatrix} -i\lambda L + i\kappa\int_{-L}^L dz \psi_2(z)w_{(\psi_1,\psi_2)}(z,\lambda) & 0 \\ 0 & i\lambda L -i\kappa\int_{-L}^L dz \psi_1(z)w_{(\psi_1,\psi_2)}^\dagger(z,\lambda) \end{pmatrix} \label{eq:ZL_form}
\end{align}
Evaluating both sides of equation \eqref{eq:tm_gt} at $(x,y)=(L,-L)$, we find that the monodromy matrix $T_L(\lambda)$ has the representation
\begin{equation}
\label{eq:TL_form}
\begin{split}
T_{L,(\psi_1,\ol{\psi_2})}(\lambda) &= \paren*{I_{\C^2}+W_{(\psi_1,\ol{\psi_2})}(L,\lambda)}\exp(Z_{L,(\psi_1,\ol{\psi_2})}(\lambda))\paren*{I_{\C^2}+W_{(\psi_1,\ol{\psi_2})}(-L,\lambda)}^{-1}.
\end{split}
\end{equation}

We now turn to finding a formula for the generating function $\tl{F}_L(\lambda)$ (recall \eqref{eq:F_gfunc}) in terms of the functions $w$ and $w^\dagger$. We first have an important involution property for the entries of $Z_L(\lambda)$.

\begin{lemma}
\label{lem:I_invol}
For every $(\psi_1,\ol{\psi_2})\in\Sc(\R)^2$ and $\lambda\in\R$ sufficiently large so that $w_{(\psi_1,\ol{\psi_2})}(\cdot,\lambda)$ exists, it holds that
\begin{equation}
\int_{-L}^L dx \ol{\psi_2}(x)w_{(\psi_1,\ol{\psi_2})}(x,\lambda) = \int_{-L}^L dx\psi_1(x) w_{(\psi_1,\ol{\psi_2})}^\dagger(x,\lambda) = \ol{\int_{-L}^L dx \ol{\psi_1}(x)w_{(\psi_2,\ol{\psi_1})}(x,\lambda)}.
\end{equation}
In particular, if for every $n\in\N$, we define
\begin{equation}
\label{eq:tlIn_def}
\tl{I}_n(\psi_1,\ol{\psi_2}) \coloneqq \int_{-L}^L dx \ol{\psi_2}(x)w_{n,(\psi_1,\ol{\psi_2})}(x), \qquad \forall (\psi_1,\ol{\psi_2})\in\Sc(\R)^2,
\end{equation}
then
\begin{equation}
\label{eq:I_invol}
\tl{I}_n(\psi_1,\ol{\psi_2}) = \ol{\tl{I}_n(\psi_2,\ol{\psi_1})}.
\end{equation}
\end{lemma}
\begin{proof}
Since $\det_{\C^2}(T_{L,(\psi_1,\ol{\psi_2})}(\lambda))=1$ by the unimodularity property \cref{lem:tm_props}\ref{item:tm_unimod} and
\begin{equation}
\paren*{I_{\C^2}+W_{(\psi_1,\ol{\psi_2})}(L,\lambda)}^{-1} = I_{\C^2}+W_{(\psi_1,\ol{\psi_2})}(-L,\lambda)
\end{equation}
by the $2L$-periodicity of $W(\cdot,\lambda)$, it follows from the multiplicative property of determinant that
\begin{equation}
1=\det_{\C^2}(T_{L,(\psi_1,\ol{\psi_2})}(\lambda)) = \det_{\C^2}\paren*{\exp Z_{L,(\psi_1,\ol{\psi_2})}(\lambda)}.
\end{equation}
Now for any matrix $A\in\C^n\otimes \C^n$, Jacobi's formula implies the trace identity
\begin{equation}
\det_{\C^n}(e^A)=\exp(\tr_{\C^n} A).
\end{equation}
Hence,
\begin{equation}
1= \exp\paren*{\tr_{\C^2} Z_{L,(\psi_1,\ol{\psi_2})}(\lambda)} = 1 \Longrightarrow \tr_{\C^2} Z_{L,(\psi_1,\ol{\psi_2})}(\lambda)=0.
\end{equation}
So by identity \eqref{eq:ZL_form}, we obtain that
\begin{align}
\int_{-L}^L dx \ol{\psi_2}(x)w_{(\psi_1,\ol{\psi_2})}(x,\lambda) = \int_{-L}^L dx\psi_1(x) w_{(\psi_1,\ol{\psi_2})}^\dagger(x,\lambda) = \ol{\int_{-L}^L dx \ol{\psi_1}(x)w_{(\psi_2,\ol{\psi_1})}(x,\lambda)}.
\end{align}
where the ultimate equality follows by definition of the $\dagger$ superscript. Substituting the asymptotic expansions \eqref{eq:w_asy} for $w_{(\psi_1,\ol{\psi_2})}(x,\lambda)$ and $w_{(\psi_2,\ol{\psi_1})}(x,\lambda)$ into the left-hand and right-hand sides of the preceding equation, respectively, and using the definition \cref{eq:tlIn_def} for $\tl{I}_n(\psi_1,\ol{\psi_2})$ and $\tl{I}_n(\psi_2,\ol{\psi_1})$, the second assertion follows as well.
\end{proof}

\begin{lemma}
\label{lem:FL_form}
For every $(\psi_1,\ol{\psi_2})\in\Sc(\R)^2$ and $\lambda\in\R$ sufficiently large as in \cref{lem:I_invol}, it holds that
\begin{equation}
\label{eq:FL_form}
\begin{split}
\tl{F}_{L}(\psi_1,\ol{\psi_2};\lambda) &= 2\cos\paren*{-\lambda L + \kappa\int_{-L}^L dx \ol{\psi_2}(x)w_{(\psi_1,\ol{\psi_2})}(x,\lambda)},
\end{split}
\end{equation}
where $\tl{F}_L$ is defined in \cref{eq:F_gfunc}.
\end{lemma}
\begin{proof}
Since the trace is invariant under unitary transformation and $W_{(\psi_1,\ol{\psi_2})}$ is $2L$-periodic, we have that
\begin{equation}
\tl{F}_L(\psi_1,\ol{\psi_2};\lambda) = \tr_{\C^2} T_{L,(\psi_1,\ol{\psi_2})}(\lambda) = \tr_{\C^2}\exp\paren*{Z_{L,(\psi_1,\ol{\psi_2})}(\lambda)},
\end{equation}
so we have reduced to considering the right-hand side expression.

Using that $Z_{L,(\psi_1,\ol{\psi_2})}(\lambda)$ is diagonal and applying formula \eqref{eq:ZL_form} and \cref{lem:I_invol}, we find that
\begin{align}
\label{eq:Z_form}
Z_{L,(\psi_1,\ol{\psi_2})}(\lambda) &=\begin{pmatrix}
-i\lambda L + i\kappa\int_{-L}^L dx \ol{\psi_2}(x)w_{(\psi_1,\ol{\psi_2})}(x,\lambda) & 0\\
0 & i\lambda L - i\kappa \int_{-L}^L dx \ol{\psi_2}(x) w_{(\psi_1,\ol{\psi_2})}(x,\lambda))
\end{pmatrix},
\end{align}
it follows that the exponential of $Z_{L,(\psi_1,\ol{\psi_2})}(\lambda)$ is the diagonal matrix with the entries given by the exponential of the entries of $Z_L(\lambda)$. Using the elementary trigonometric identity
\begin{equation}
e^{iz} + e^{-iz} = 2\cos(z), \qquad z\in\C,
\end{equation}
we then obtain that
\begin{equation}
\begin{split}
\tr_{\C^2} \exp\paren*{Z_{L,(\psi_1,\ol{\psi_2})}(\lambda)} &= 2\cos\paren*{-\lambda L + \kappa\int_{-L}^L dx \ol{\psi_2}(x)w_{(\psi_1,\ol{\psi_2})}(x,\lambda)},
\end{split}
\end{equation}
which completes the proof of the lemma.
\end{proof}

\begin{remark}
By \cref{lem:tm_invol}, we have the involution relation
\begin{equation}
\tr_{\C^2} T_{L,(\psi_1,\ol{\psi_2})}(\lambda) = \tr_{\C^2}\paren*{\sigma T_{L,(\psi_1,\ol{\psi_2})}^\dagger(\bar{\lambda}) \sigma} =  \tr_{\C^2} T_{L,(\psi_1,\ol{\psi_2})}^\dagger(\bar{\lambda}) = \ol{\tr_{\C^2}\paren*{T_{L,(\psi_2,\ol{\psi_1})}(\bar{\lambda})}},
\end{equation}
where we use the cyclicity of trace and $\sigma^2=I_{\C^2}$ to obtain the penultimate equality. Consequently, we have that
\begin{equation}
\tl{F}_L(\psi_1,\ol{\psi_2};\lambda) = \ol{\tl{F}_L(\psi_2,\ol{\psi_1};\bar{\lambda})}.
\end{equation}
Consequently, if we take twice the real part of $\tl{F}_L(\psi_1,\ol{\psi_2};\lambda)$,
\begin{equation}
\label{eq:FL_Re}
F_{L,\Re}(\psi_1,\ol{\psi_2};\lambda) \coloneqq 2\Re{\tl{F}_L(\psi_1,\ol{\psi_2};\lambda)}, \qquad \forall (\psi_1,\ol{\psi_2},\lambda)\in C^\infty(\T_L)^2\times \C,
\end{equation}
then we obtain from \eqref{eq:FL_form} that
\begin{equation}
\label{eq:FL_Re_form}
\begin{split}
F_{L,\Re}(\psi_1,\ol{\psi_2};\lambda) &= 2\cos\paren*{-\lambda L+\kappa\int_{-L}^L dx \ol{\psi_2}(x)w_{(\psi_1,\ol{\psi_2})}(x,\lambda)} \\
&\phantom{=}+ 2\cos\paren*{-\ol{\lambda} L+\kappa\int_{-L}^L dx\ol{\psi_1}(x)w_{(\psi_2,\ol{\psi_1})}(x,\bar{\lambda})}.
\end{split}
\end{equation}
Similarly, if we take twice the imaginary part of $\tl{F}_{L}(\psi_1,\ol{\psi_2};\lambda)$,
\begin{equation}
\label{eq:FL_Im}
F_{L,\Im}(\psi_1,\ol{\psi_2};\lambda) \coloneqq 2\Im{\tl{F}_L(\psi_1,\ol{\psi_2})},
\end{equation}
then we have that
\begin{equation}
\label{eq:FL_Im_form}
\begin{split}
F_{L,\Im}(\psi_1,\ol{\psi_2};\lambda) &= -i \left(2\cos\paren*{-\lambda L+\kappa\int_{-L}^L dx\ol{\psi_2}(x)w_{(\psi_1,\ol{\psi_2})}(x,\lambda)} \right.\\
&\phantom{=}\hspace{15mm} \left.-2\cos\paren*{-\ol{\lambda} L+\kappa\int_{-L}^L dx\ol{\psi_1}(x)w_{(\psi_2,\ol{\psi_1})}(x,\ol{\lambda})}\right).
\end{split}
\end{equation}
Moreover, we can regard $F_{L,\Re}(\cdot,\cdot;\lambda)$ and $F_{L,\Im}(\cdot,\cdot;\lambda)$, respectively, as restrictions of the complex functionals of four variables to the subspace $\psi_{\bar{1}} =\ol{\psi_1}, \psi_{\bar{2}}=\ol{\psi_2}$. More precisely, for fixed $\lambda\in\C$, define complex-valued functionals on $C^\infty(\T_L)^4$ by
\begin{equation}
\begin{split}
\tl{F}_{L,\Re}(\psi_1, \psi_{\bar{2}}, \psi_{2},\psi_{\bar{1}};\lambda) &\coloneqq \tl{F}_{L}(\psi_1,\psi_{\bar{2}};\lambda) + \tl{F}_L(\psi_2,\psi_{\bar{1}};\ol{\lambda}), \\
\tl{F}_{L,\Im}(\psi_1, \psi_{\bar{2}}, \psi_{2}, \psi_{\bar{1}};\lambda) &\coloneqq -i\paren*{\tl{F}_L(\psi_1,\psi_{\bar{2}};\lambda) - \tl{F}_L(\psi_2,\psi_{\bar{1}};\ol{\lambda})},
\end{split}
\end{equation}
so that
\begin{equation}
\label{eq:FL_restric}
\begin{split}
F_{L,\Re}(\psi_1,\psi_{\bar{2}};\lambda) &= \tl{F}_{L,\Re}(\psi_1,\ol{\psi_2},\psi_2,\ol{\psi_1};\lambda)\\
F_{L,\Im}(\psi_1,\psi_{\bar{2}};\lambda) &= \tl{F}_{L,\Im}(\psi_1,\ol{\psi_2},\psi_2,\ol{\psi_1};\lambda).
\end{split}
\end{equation}
Consequently, $F_{L,\Re}(\lambda)$ and $F_{L,\Im}(\lambda)$ extend with an abuse of notation to well-defined smooth functionals on the space $C^\infty(\T_L;\mathcal{V})$ (recall the space of matrices $\V$ in \eqref{equ:mixed}) given by
\begin{equation}
\begin{cases} F_{L,\Re}(\gamma;\lambda) \coloneqq F_{L,\Re}(\phi_1,\ol{\phi_2};\lambda), \\ F_{L,\Im}(\gamma;\lambda) \coloneqq F_{L,\Im}(\phi_1,\ol{\phi_2};\lambda)\end{cases}, \qquad \forall \gamma=\frac{1}{2}\mathrm{odiag}(\phi_1,\ol{\phi_2},\phi_2,\ol{\phi_1}),
\end{equation}
which belong to the admissible algebra $\A_{\Sc,\V}$, provided that $\tl{F}_{L}\in\A_{\Sc,\C}$, a result we postpone until the next subsection. By the same reasoning, the functionals
\begin{equation}
\label{eq:Ibn_def}
\begin{split}
I_{b,n}(\gamma) &\coloneqq \frac{1}{2}\paren*{\tl{I}_n(\phi_1,\ol{\phi_2}) + \tl{I}_n(\phi_2,\ol{\phi_1})} \\
&=\frac{1}{2}\int_{-L}^L dx\paren*{\ol{\phi_2}(x)w_{n,(\phi_1,\ol{\phi_2})}(x)+\ol{\phi_1}(x)w_{n,(\phi_2,\ol{\phi_1})}(x)}, \quad \forall \gamma=\frac{1}{2}\mathrm{odiag}(\phi_1,\ol{\phi_2},\phi_2,\ol{\phi_1}),
\end{split}
\end{equation}
where the subscript $b$ is to denote the dependence on two inputs, extend to smooth functionals on $C^\infty(\T_L;\V)$ which belong to $\A_{\Sc,\V}$. This latter admissibility can be verified using the results of \cref{ssec:cor_symgrad}. Note that by \cref{lem:I_invol}, the functionals $I_{b,n}$ are real-valued.
\end{remark}

\subsection{Poisson commutativity}
\label{ssec:nls_pc}
In this last subsection of the appendix, we show that the functionals $I_{b,n}$ defined in \eqref{eq:Ibn_def} are in involution with respect to the Poisson bracket $\pb{\cdot}{\cdot}_{L^2,\mathcal{V}}$ defined in \cref{prop:Schw_WP_V}. We obtain this result by first showing that the generating functionals $\tl{F}_{L}(\lambda),\tl{F}_L(\mu)$, for $\lambda,\mu\in\C$, are in involution with respect to the Poisson bracket $\pb{\cdot}{\cdot}_{L^2,\C}$. The reader will recall that the $\tl{F}_L$ was defined in \eqref{eq:F_gfunc} above.

Given two complex-valued functionals $F,G$ on $C^\infty(\T_L)^2$ satisfying the conditions of \cref{rem:pb_vd_C}, we recall their Poisson bracket is defined by
\begin{equation}
\label{eq:pb_L2C_rev}
\pb{F}{G}_{L^2,\C}(\psi_1,\psi_2) = -i\int_{-L}^L dx\paren*{\grad_1 F(\psi_1,\psi_2)\grad_{\bar{2}} G(\psi_1,\psi_2) - \grad_{\bar{2}} F(\psi_1,\psi_2)\grad_1 G(\psi_1,\psi_2)}(x),
\end{equation}
where $\grad_1$ and $\grad_{\bar{2}}$ denote the variational derivatives defined in \eqref{eq:vd_prop}. Now let $A$ and $B$ be two complex-matrix-valued functionals on $C^\infty(\T_L)^2$. We introduce the notation
\begin{equation}
\{A\tens B\}_{L^2,\C}(\psi_1,\psi_2) \coloneqq -i\int_{-L}^Ldx \paren*{\grad_1 A(\psi_1,\psi_2)\otimes \grad_{\bar{2}} B(\psi_1,\psi_2) - \grad_{\bar{2}} A(\psi_1,\psi_2)\otimes\grad_1 B(\psi_1,\psi_2)}(x),
\end{equation}
where our identification of the tensor product is the $4\times 4$ matrix
\begin{equation}
(A\otimes B)_{jk,mn} = A_{jm}B_{kn}, \qquad j,m,k,n\in\{1,2\},
\end{equation}
so that
\begin{equation}
{\{A\tens B\}_{L^2,\C}}_{jk,mn} = \pb{A_{jm}}{B_{kn}}_{L^2,\C}.
\end{equation}

\begin{remark}
An observation important for our identities in the sequel is that the notation $\{\tens\}$ admits an obvious extension to general $n\times n$ matrices.
\end{remark}

The reader may check that the above tensor Poisson bracket notation has the following properties:
\begin{description}
\item[Skew-symmetry]
\begin{equation}
\{A\tens B\}_{L^2,\C} = -P\{B\tens A\}_{L^2,\C}P,
\end{equation}
where $P$ is the permutation matrix in $\C^2\otimes\C^2$ defined by $P(\xi\otimes \eta) = \eta\otimes\xi$, for $\xi,\eta\in\C^2$.
\item[Leibnitz rule]
\begin{equation}
\{A\tens BC\}_{L^2,\C} = \{A\tens B\}_{L^2,\C}(I_{\C^2}\otimes C) + (I_{\C^2}\otimes B)\{A\tens C\}_{L^2,\C},
\end{equation}
\item[Jacobi identity]
\begin{equation}
\begin{split}
0&=\{A\tens\{B\tens C\}_{L^2,\C}\}_{L^2,\C} + P_{13}P_{23}\{C\tens\{A\tens B\}_{L^2,\C}\}_{L^2,\C}P_{23}P_{13} \\
&\phantom{=} + P_{13}P_{12}\{B\tens\{C\tens A\}_{L^2,\C}\}_{L^2,\C}P_{12}P_{13},
\end{split}
\end{equation}
where $P_{ij}$ is the permutation matrix in $(\C^2)^{\otimes 3}$ which swaps the $i^{th}$ and $j^{th}$ element of a tensor $\xi_1\otimes\xi_2\otimes\xi_3$, for $i,j\in\{1,2,3\}$.
\end{description}

\begin{remark}
The reader can also check that $P$ is idempotent (i.e. $P^2=I_{\C^2}$) and $P(A\otimes B)=(B\otimes A)P$, for any $2\times 2$ matrices $A,B$.
\end{remark}

With the above notation in hand, we proceed to compute Poisson brackets. Let us consider $U_{(\psi_1,\psi_2)}(z,\lambda)$ from \eqref{eq:U_def} as a functional of $(\psi_1,\psi_2)$, for fixed $(z,\lambda)$. For the reader's benefit, we recall that
\begin{equation}
\label{eq:U_recall}
U_{(\psi_1,\psi_2)}(x,\lambda) = \frac{\lambda}{2i}\sigma_3 + U_0(x)=\frac{\lambda}{2i}\sigma_3 + \sqrt{\kappa}\paren*{\psi_2(x)\sigma_{+}+\psi_1(x)\sigma_{-}},
\end{equation}
where $U_0(x)$ is defined in \eqref{eq:U_def}. The first objective is to prove the following lemma which gives the so-called \emph{fundamental Poisson brackets}.

\begin{lemma}[Fundamental Poisson brackets]
\label{lem:fund_PB}
For any $(\lambda,\mu)\in\C^2$, we have the distributional (on $\T_L^2$) identity
\begin{equation}
\{U(x,\lambda)\tens U(y,\mu)\}_{L^2,\C}= -\comm{r(\lambda-\mu)}{U(x,\lambda)\otimes I_{\C^2} + I_{\C^2}\otimes U(y,\mu)}\delta(x-y),
\end{equation}
where $r(\lambda-\mu)\coloneqq -\frac{\kappa}{(\lambda-\mu)}P$.\footnote{This matrix $r$ is called an \emph{$r$-matrix} in the integrable systems literature and is a central object in the study of such systems.}
\end{lemma}
\begin{proof}
We recall the (classical) canonical commutation relations
\begin{equation}
\pb{\psi_1(x)}{\psi_1(y)}_{L^2,\C}=\pb{\psi_2(x)}{\psi_2(y)}_{L^2,\C}=0, \quad \pb{\psi_1(x)}{\psi_2(y)}_{L^2,\C}=-i\delta(x-y),
\end{equation}
which should be interpreted in the sense of tempered distributions on $\T_L^2$. It then follows from \eqref{eq:U_recall} that
\begin{equation}
(\grad_1 U(x,\lambda))(\psi_1,\psi_2) = \sqrt{\kappa}\sigma_{-}\delta_x, \quad (\grad_{\bar{2}} U(x,\lambda))(\psi_1,\psi_2) = \sqrt{\kappa}\sigma_{+}\delta_x,
\end{equation}
where $\delta_x$ is the Dirac mass centered at the point $x$. Hence,
\begin{align}
&\{U(x,\lambda)\tens U(y,\mu)\}_{L^2,\C}(\psi_1,\psi_2) \nonumber\\
&= -i\int_{-L}^L dz \paren*{(\grad_1 U(x,\lambda))(\psi_1,\psi_2) (\grad_{\bar{2}} U(y,\mu))(\psi_1,\psi_2)-(\grad_{\bar{2}} U(x,\lambda))(\psi_1,\psi_2)(\grad_1 U(y,\mu))(\psi_1,\psi_2)}(z) \nonumber\\
&=-i\kappa\int_{-L}^L dz \delta(z-x)\delta(z-y)\paren*{\sigma_{-}\otimes\sigma_{+} - \sigma_{+}\otimes\sigma_{-}} \nonumber\\
&=-i\kappa\delta(x-y)\paren*{\sigma_{-}\otimes\sigma_{+} - \sigma_{+}\otimes\sigma_{-}} \nonumber .
\end{align}
One can check from the commutation relations for the Pauli matrices defined in \eqref{eq:Pauli} that
\begin{equation}
\sigma_{-}\otimes\sigma_{+}-\sigma_{+}\otimes\sigma_{-} = \frac{1}{2}\comm{P}{\sigma_3\otimes I_{\C^2}}=-\frac{1}{2}\comm{P}{I_{\C^2}\otimes\sigma_3}.
\end{equation}
Therefore,
\begin{align}
i\kappa\paren*{\sigma_{-}\otimes\sigma_{+} - \sigma_{+}\otimes\sigma_{-}} &= \frac{i\kappa\lambda}{\lambda-\mu}\paren*{\sigma_{-}\otimes\sigma_{+} - \sigma_{+}\otimes\sigma_{-}}  -\frac{i\kappa\mu}{\lambda-\mu}\paren*{\sigma_{-}\otimes\sigma_{+} - \sigma_{+}\otimes\sigma_{-}} \nonumber\\
&=-\frac{\kappa}{\lambda-\mu}\paren*{\frac{\lambda}{2i}\comm{P}{\sigma_3\otimes I_{\C^2}} +\frac{\mu}{2i}\comm{P}{I_{\C^2}\otimes\sigma_3}}.
\end{align}
Now recalling the definition of $U(x,\lambda)$ in \eqref{eq:U_recall} and that $P$ commutes with the tensor $U_0(x)\otimes I_{\C^2} + I_{\C^2}\otimes U_0(x)$ by the symmetry of the latter, we obtain the desired conclusion.
\end{proof}

The importance of the fundamental Poisson brackets is that they yield a formula for the Poisson brackets between the entries of the transition matrices $T(x,y,\lambda)$ and $T(x,y,\mu)$, regarded as matrix-valued functionals, as the next lemma shows.

\begin{lemma}
\label{lem:tm_pb}
For fixed $-L<y<x<L$ and $(\lambda,\mu)\in\C^2$, regard $T(x,y,\lambda)$ as the $\C^2\otimes\C^2$-matrix valued functional $C^\infty(\T_L)^2$ defined by $(\psi_1,\psi_2) \mapsto T_{(\psi_1,\psi_2)}(x,y,\lambda)$ and similarly for $T(x,y,\mu)$. Then it holds that
\begin{equation}
\{T(x,y,\lambda)\tens T(x,y,\mu)\}_{L^2,\C} = -\comm{r(\lambda-\mu)}{T(x,y,\lambda)\otimes T(x,y,\mu)}.
\end{equation}
\end{lemma}
\begin{proof}
We use the differential equations \eqref{eq:tm_ODE_x} and \eqref{eq:tm_ODE_y} for the transition matrix in order to prove the lemma. Since the $(a,b)$ entry of the matrix-valued functional $T(x,y,\lambda)$ depends on $(\psi_1,\psi_2)$ through the entries of the matrix-valued functional $U(z,\lambda)$ it follows from the definition of the Poisson bracket $\pb{\cdot}{\cdot}_{L^2,\C}$ reviewed in \eqref{eq:pb_L2C_rev} and the chain rule that
\begin{equation}
\label{eq:pb_chain}
\begin{split}
&\pb{T^{ab}(x,y,\lambda)}{T^{cd}(x,y,\mu)}_{L^2,\C}(\psi_1,\psi_2)\\
&=\int_{y}^{x}\int_{y}^x dzdz' (\grad_{U^{jk}(\lambda)}T^{ab}(x,y,\lambda)(\psi_1,\psi_2))(z)\pb{U^{jk}(z,\lambda)}{U^{\ell m}(z',\mu)}_{L^2,\C}(\psi_1,\psi_2) \\
&\phantom{=} \hspace{35mm}\times (\grad_{U^{\ell m}(\mu)} T^{cd}(x,y,\mu)(\psi_1,\psi_2))(z'),
\end{split}
\end{equation}
where $\grad_{U^{jk}(\lambda)}T^{ab}(x,y,\lambda)$ and $\grad_{U^{\ell m}(\mu)}T^{cd}(x,y,\mu)$ are the variational derivatives uniquely defined by (a priori in the sense of distributions)
\begin{equation}
\label{eq:Tab_vd}
\begin{split}
dT^{ab}(x,y,\lambda)[\psi_1,\psi_2](\delta U^{jk}(\lambda)) &= \int_{-L}^Ldz (\grad_{U^{jk}(\lambda)}T^{ab}(x,y,\lambda)(\psi_1,\psi_2))(z)\delta U^{jk}(z,\lambda), \\
dT^{cd}(x,y,\mu)[\psi_1,\psi_2](\delta U^{\ell m}(\mu)) &= \int_{-L}^L dz (\grad_{U^{\ell m}(\mu)}T^{cd}(x,y,\mu)(\psi_1,\psi_2))(z')\delta U^{\ell m}(z',\mu).
\end{split}
\end{equation}
In \eqref{eq:pb_chain}, we use the convention of Einstein summation, so the summation over repeated indices is implicit.

We now seek a formula for $\grad_{U^{jk}(\lambda)}T^{ab}(x,y,\lambda)$ and $\grad_{U^{\ell m}(\mu)}T^{cd}(x,y,\mu)$. To find such a formula, we take the G\^ateaux derivative of both sides of \eqref{eq:tm_ODE_x} at the point $U(\cdot,\lambda)$ in the direction $\delta U(\cdot,\lambda)$ to obtain the equation
\begin{equation}
\label{eq:T_tl_ODE}
\begin{cases}
\p_x dT(x,y,\lambda)[U(\cdot,\lambda)](\delta U(\cdot,\lambda))=U(x,\lambda)dT(x,y,\lambda)[U(\cdot,\lambda)](\delta U(\cdot,\lambda)) + \delta U(x,\lambda) T(x,y,\lambda), \\
dT(x,y,\lambda)[U(\cdot,\lambda)](\delta U(\cdot,\lambda))|_{x=y}=I_{\C^2}.
\end{cases}
\end{equation}
The reader can check by direct computation that the solution to this equation is given by
\begin{equation}
\label{eq:dT_U}
dT(x,y,\lambda)[U(\cdot,\lambda)](\delta U(\cdot,\lambda)) = \int_y^x dz T(x,y,\lambda)\delta U(z,\lambda) T(z,y,\lambda).
\end{equation}
Examining identity \eqref{eq:dT_U} entry-wise, we have that
\begin{equation}
\begin{split}
dT^{ab}(x,y,\lambda)[U(\cdot,\lambda)](\delta U(\cdot,\lambda)) &= \int_y^x dz T^{a j}(x,y,\lambda) \delta U^{j k}(z,\lambda) T^{k b}(z,y,\lambda),\\
dT^{cd}(x,y,\mu)[U(\cdot,\lambda)](\delta U(\cdot,\lambda)) &= \int_y^x dz T^{c \ell}(x,y,\mu) \delta U^{\ell m}(z',\mu) T^{md}(z',y,\mu),
\end{split}
\end{equation}
which upon comparison with \eqref{eq:Tab_vd} yields the identity
\begin{equation}
\label{eq:grad_U_T}
\begin{split}
(\grad_{U^{jk}(\lambda)} T^{ab}(x,y,\lambda)(\psi_1,\psi_2))(z) = \begin{cases} T_{(\psi_1,\psi_2)}^{aj}(x,y,\lambda) T_{(\psi_1,\psi_2)}^{kb}(z,y,\lambda), & -L<y<z<x<L\\ 0, & {\text{otherwise}}\end{cases},\\
(\grad_{U^{\ell m}(\lambda)} T^{cd}(x,y,\mu)(\psi_1,\psi_2))(z') = \begin{cases} T_{(\psi_1,\psi_2)}^{c\ell}(x,y,\mu) T_{(\psi_1,\psi_2)}^{md}(z',y,\mu), & -L<y<z'<x<L\\ 0, & {\text{otherwise}}\end{cases}
\end{split}.
\end{equation}
Substituting the identity \eqref{eq:grad_U_T} into \eqref{eq:pb_chain}, we find that
\begin{equation}
\begin{split}
&\{T(x,y,\lambda)\tens T(x,y,\mu)\}_{L^2,\C}(\psi_1,\psi_2) \\
&=\int_y^x\int_y^x dzdz' \big(T_{(\psi_1,\psi_2)}(x,z,\lambda)\otimes T_{(\psi_1,\psi_2)}(x,z',\mu)\big)\{U(z,\lambda)\tens U(z',\mu)\}_{L^2,\C}(\psi_1,\psi_2) \\
&\phantom{=}\hspace{35mm}\times \big(T_{(\psi_1,\psi_2)}(z,y,\lambda)\tens T_{(\psi_1,\psi_2)}(z',y,\mu)\big).
\end{split}
\end{equation}
Using the formula given by \cref{lem:fund_PB}, we obtain that the right-hand equals
\begin{equation}
\begin{split}
&-\int_y^x dz \big(T_{(\psi_1,\psi_2)}(x,z,\lambda)\otimes T_{(\psi_1,\psi_2)}(x,z,\mu)\big) \comm{r(\lambda-\mu)}{U(z,\lambda)\otimes I_{\C^2} + I_{\C^2}\otimes U(z,\mu)} \\
&\phantom{=}\hspace{35mm}\times \big(T_{(\psi_1,\psi_2)}(z,y,\lambda)\otimes T_{(\psi_1,\psi_2)}(z,y,\mu)\big).
\end{split}
\end{equation}
We now claim that the integrand is the partial derivative with respect to $z$ of
\begin{equation}
\begin{split}
\paren*{T_{(\psi_1,\psi_2)}(x,z,\lambda)\otimes T_{(\psi_1,\psi_2)}(x,z,\mu)}r(\lambda-\mu)\paren*{T_{(\psi_1,\psi_2)}(z,y,\lambda)\otimes T_{(\psi_1,\psi_2)}(z,y,\mu)},
\end{split}
\end{equation}
which then completes the proof. Indeed, the reader may verify this is the case by direct computation using the Leibnitz rule and the equations \eqref{eq:tm_ODE_x} and \eqref{eq:tm_ODE_y} for the transition matrix. So upon application of the fundamental theorem of calculus and using the initial condition $T(x,y,\lambda)|_{x=y}=I_{\C^2}$, we obtain the desired conclusion.
\end{proof}

We next check that the functional $\tl{F}_L(\lambda)$ defined in \eqref{eq:F_gfunc}, is admissible (i.e. it belongs to $\A_{\Sc,\C}$ defined in \eqref{equ:Asc}). This admissibility will then imply that $F_{L,\Re}(\lambda)$ and $F_{L,\Im}(\lambda)$ defined in \eqref{eq:FL_Re} and \eqref{eq:FL_Im}, respectively, belong to $\A_{\Sc,\V}$ defined in \eqref{alg_v}. First, observe that by taking the direction
\begin{equation}
\delta U(z,\lambda) = \sqrt{\kappa}\paren*{\delta\psi_2(z)\sigma_{+}+\delta\psi_1(z)\sigma_{-}}
\end{equation}
in \eqref{eq:dT_U}, we find that
\begin{equation}
\begin{split}
(\grad_{1}T(x,y,\lambda)(\psi_1,\psi_2))(z) &= \sqrt{\kappa} T_{(\psi_1,\psi_2)}(x,z,\lambda)\sigma_{-} T_{(\psi_1,\psi_2)}(z,y,\lambda), \\
(\grad_{\bar{2}}T(x,y,\lambda)(\psi_1,\psi_2))(z) &= \sqrt{\kappa}T_{(\psi_1,\psi_2)}(x,z,\lambda)\sigma_{+}T_{(\psi_1,\psi_2)}(z,y,\lambda),
\end{split}
\end{equation}
for $z\in [y,x]$, and zero for $z\in (-L,L)\setminus (y,x)$. Letting $x\rightarrow L^+$ and $y\rightarrow L^{-}$, we find that
\begin{equation}
\begin{split}
(\grad_1 T_L(\lambda)(\psi_1,\psi_2))(z) &= \sqrt{\kappa} T_{(\psi_1,\psi_2)}(L,z,\lambda)\sigma_{-}T_{(\psi_1,\psi_2)}(z,-L,\lambda),\\
(\grad_{\bar{2}}T_L(\lambda)(\psi_1,\psi_2))(z) &= \sqrt{\kappa} T_{(\psi_1,\psi_2)}(L,z,\lambda)\sigma_{+}T_{(\psi_1,\psi_2)}(z,-L,\lambda).
\end{split}
\end{equation}
Note that $\grad_1 T_L(\lambda)(\psi_1,\psi_2),\grad_{\bar{2}}T_L(\lambda)(\psi_1,\psi_2)$ are smooth in $(-L,L)$ but discontinuous at the boundary, and consequently do no belong to $C^\infty(\T_L)$ (i.e. $T_L(\lambda)$ is not an admissible functional). However, if we take the $2\times 2$ matrix trace of both sides of the preceding identities and use that the variational derivative commutes with the trace together with the cyclicity of trace, we obtain that the resulting expressions extend smoothly periodically to the entire real line. We summarize the preceding discussion with the following lemma.

\begin{lemma}
\label{lem:FL_adm}
For any $\lambda\in\C$, $\tl{F}_L\in \A_{\Sc,\C}$. Consequently, $F_{L,\Re}(\lambda), F_{L,\Im}(\lambda)\in \A_{\Sc,\mathcal{V}}$.
\end{lemma}

We now show that traces $\tl{F}_L(\lambda), \tl{F}_L(\mu)$, for fixed $\mu,\lambda\in\C$, are in involution with respect to the Poisson bracket $\pb{\cdot}{\cdot}_{L^2,\C}$. They key ingredient of this result is the identity of \cref{lem:tm_pb} for the Poisson brackets between the entries of the transition matrices.

\begin{lemma}
\label{lem:tl_F_invol}
For any $\lambda,\mu\in\C$, we have that
\begin{equation}
\pb*{\tl{F}_L(\lambda)}{\tl{F}_L(\mu)}_{L^2,\C}\equiv 0.
\end{equation}
\end{lemma}
\begin{proof}
Applying \cref{lem:tm_pb}, we have that
\begin{equation}
\begin{split}
&\comm{r(\lambda-\mu)}{T_{(\psi_1,\psi_2)}(x,y,\lambda)\otimes T_{(\psi_1,\psi_2)}(x,y,\mu)}\\
&=\int_{-L}^L dz \paren*{\grad_1 T(x,y,\lambda)\otimes \grad_{\bar{2}} T(x,y,\mu) - \grad_{\bar{2}} T(x,y,\lambda)\otimes\grad_{1}T(x,y,\mu)}(\phi_1,\phi_2)(z).
\end{split}
\end{equation}
Taking the $4\times 4$ matrix trace $\tr_{\C^2\otimes\C^2}$ of both sides and using that the trace of a commutator is zero together with the algebraic identity
\begin{equation}
\tr_{\C^2\otimes \C^2}(A\otimes B) = \tr_{\C^2}(A)\tr_{\C^2}(B),
\end{equation}
for any $2\times 2$ matrices $A,B$, we obtain that
\begin{equation}
\begin{split}
0 &= -\int_{-L}^L dz \left(\grad_1 \paren*{\tr_{\C^2}(T(x,y,\lambda))\grad_{\bar{2}}\tr_{\C^2}(T(x,y,\mu))}(\phi_1,\phi_2)(z)  \right.\\
&\phantom{=}\hspace{30mm} \left.- \paren*{\grad_{\bar{2}}\tr_{\C^2}(T(x,y,\lambda))\grad_{1}\tr_{\C^2}(T(x,y,\mu))}(\phi_1,\phi_2)(z)\right),
\end{split}
\end{equation}
where we also use that the trace commutes with the variational derivative. Now using the continuity in $(x,y)$ of the integrand, we can let $x\rightarrow L^{-}$ and $y\rightarrow -L^{+}$ and use that $\tr_{\C^2}(T_L(\lambda))=\tl{F}_L(\lambda)$ by definition \eqref{eq:F_gfunc} and $\tr_{\C^2}(T_L(\mu))=\tl{F}_L(\mu)$ to obtain the desired conclusion.
\end{proof}

Now we show that the functionals $I_{b,n}$ defined in \cref{eq:Ibn_def} are mutually involutive with respect to the Poisson structure on $C^\infty(\T_L;\mathcal{V})$. We begin by defining the generating functional
\begin{equation}
\label{eq:tlp_def}
\tl{p}_L(\phi_1,\ol{\phi_2};\lambda) \coloneqq \arccos(\frac{1}{2}\tl{F}_L(\phi_1,\ol{\phi_2};\lambda)), \qquad \forall (\phi_1,\ol{\phi_2},\lambda)\in C^\infty(\T_L)^2\times\C,
\end{equation}
where we take the principal branch of the function $\arccos$. We first want to show that
\begin{equation}
\label{eq:p_invol}
\pb{\tl{p}_L(\lambda)}{\tl{p}_L(\mu)}_{L^2,\C}(\phi_1,\ol{\phi_2}) = 0, \qquad \forall (\phi_1,\ol{\phi_2})\in C^\infty(\T_L)^2,
\end{equation}
for $\lambda,\mu\in\R$ with sufficiently large modulus, which requires us to compute the variational derivatives of $\tl{p}_L(\lambda),\tl{p}_L(\mu)$.

Recall from \eqref{eq:FL_form} that
\begin{equation}
\frac{1}{2}\tl{F}_L(\phi_1,\ol{\phi_2};\lambda) = \cos\paren*{-\lambda L + \kappa\int_{-L}^L dx\ol{\phi_2}(x)w_{(\phi_1,\ol{\phi_2})}(x,\lambda)}.
\end{equation}
We want to show that we can choose $\lambda$ so that the $\cos$ in the right-hand side of the preceding equation is at positive distance from $\pm 1$ for all $(\phi_1,\ol{\phi_2})$ in a closed ball of $C^\infty(\T_L)$. To this end, we know from \cref{ssec:nls_pc_im} that given $(\phi_1,\ol{\phi_2})\in C^\infty(\T_L)^2$, we can choose
\begin{equation*}
\lambda=\lambda(\|\phi_1\|_{L^1(\T_L)}, \|\phi_1\|_{L^\infty(\T_L)}, \|\phi_2\|_{L^1(\T_L)}, \|\phi_2\|_{L^\infty(\T_L)}, L)\in \R
\end{equation*}
with sufficiently large modulus so that there exists $w_{(\phi_1,\ol{\phi_2})}(\lambda)$ in \eqref{eq:W_w} with the asymptotic expansion \eqref{eq:w_asy}. Consequently, for any $k\in\N$, we have that
\begin{align}
\|w_{(\phi_1,\ol{\phi_2})}(\lambda)\|_{L^\infty(\T_L)} &\leq \left\|w_{(\phi_1,\ol{\phi_2})}(\lambda) - \sum_{n=1}^k \frac{w_{k,(\phi_1,\ol{\phi_2})}}{\lambda^n}\right\|_{L^\infty(\T_L)} + \sum_{n=1}^k \frac{\|w_{k,(\phi_1,\ol{\phi_2})}\|_{L^\infty(\T_L)}}{\lambda^n} \nonumber\\
&= o(|\lambda|^k) + \sum_{n=1}^k \frac{\|w_{k,(\phi_1,\ol{\phi_2})}\|_{L^\infty(\T_L)}}{\lambda^n},
\end{align}
where the implicit constant in $o(|\lambda|^k)$ depends only the data $\|\p_x^{n-1}\phi_j\|_{L^\infty(\T_L)}$ for $n\in\N_{\leq k+1}$ and $j\in\{1,2\}$. By the analysis of \cref{ssec:cor_multi},
\begin{equation}
\|w_{k,(\phi_1,\ol{\phi_2})}\|_{L^\infty(\T_L)} \lesssim_k \sum_{n=0}^k \paren*{\|\p_x^n\phi_1\|_{L^\infty(\T_L)} + \|\p_x^n\phi_2\|_{L^\infty(\T_L)}}.
\end{equation}
Hence,
\begin{align}
\left|\int_{-L}^L dx \ol{\phi_2}(x)w_{(\phi_1,\ol{\phi_2})}(x,\lambda)\right| &\leq 2L \|\phi_2\|_{L^\infty(\T_L)} \|w_{(\phi_1,\ol{\phi_2})}(\lambda)\|_{L^\infty(\T_L)} \nonumber\\
&\lesssim \frac{2L}{\lambda}\sum_{n=0}^1 \paren*{\|\p_x^n\phi_1\|_{L^\infty(\T_L)} + \|\p_x^n\phi_2\|_{L^\infty(\T_L)}}.
\end{align}
Thus, given $\varepsilon>0$, we can choose $\lambda\in\R$ with sufficiently large modulus depending the data
\begin{equation*}
(\varepsilon,L,\|\p_x^n\phi_j\|_{L^\infty(\T_L)}), \qquad \forall (n,j)\in\{0,1\}\times\{1,2\},
\end{equation*}
so that
\begin{equation}
\left|\int_{-L}^L dx \ol{\phi_2}(x)w_{(\phi_1,\ol{\phi_2})}(x,\lambda)\right| < \varepsilon.
\end{equation}
Also choosing $\lambda$ so that $\min_{k\in\Z}\{|\lambda L - k\pi|\}>2\varepsilon$, we conclude that given $R>0$,
\begin{equation}
\min_{k\in\Z}\left\{\left|k\pi-\lambda L + \kappa\int_{-L}^L dx \ol{\phi_2}(x)w_{(\phi_1,\ol{\phi_2})}(x,\lambda)\right|\right\} \geq \delta>0
\end{equation}
for all $\phi_1,\ol{\phi_2}\in C^\infty(\T_L)$ with $\|\p_x^n\phi_1\|_{L^\infty(\T_L)}, \|\p_x^n\phi_2\|_{L^\infty(\T_L)} \leq R$, for $n\in\{0,1\}$. For such choice of $\lambda$, we have that
\begin{equation}
\label{eq:tl_pL_form}
\tl{p}_L(\phi_1,\ol{\phi_2};\lambda) = -\lambda L + \kappa\int_{-L}^L dx \ol{\phi_2}(x)w_{(\phi_1,\ol{\phi_2})}(x,\lambda), \qquad \phi_1,\ol{\phi_2}\in C^\infty(\T_L),
\end{equation}
for all $\phi_1,\ol{\phi_2}\in C^\infty(\T_L)$ with $\max\{\|\p_x^n\phi_1\|_{L^\infty(\T_L)}, \|\phi_2\|_{L^\infty(\T_L)}\}\leq R, \enspace n\in\{0,1\}$.  Moreover, for such $\phi_1,\ol{\phi_2}$, we can use the chain rule without concern over the singularity of $\arccos(z)$ at $z=\pm 1$ to compute the variational derivatives $\tl{p}_L$, finding
\begin{equation}
\begin{split}
(\grad_1 \tl{p}_L(\lambda))(\phi_1,\ol{\phi_2}) &= \frac{1}{2}\paren*{1-\paren*{\frac{\tl{F}_L(\phi_1,\ol{\phi_2};\lambda)}{2}}^2}^{-1/2} (\grad_1 \tl{F}(\lambda))(\phi_1,\ol{\phi_2}), \\
(\grad_{\bar{2}}\tl{p}_L(\lambda))(\phi_1,\ol{\phi_2}) &= \frac{1}{2}\paren*{1-\paren*{\frac{\tl{F}_L(\phi_1,\ol{\phi_2};\lambda)}{2}}^2}^{-1/2} (\grad_{\bar{2}} \tl{F}(\lambda))(\phi_1,\ol{\phi_2}),
\end{split}
\end{equation}
where by \cref{lem:FL_adm}, the variational derivatives of $\tl{F}_L(\lambda)$ are elements of $C^\infty(C^\infty(\T_L)^2;C^\infty(\T_L))$. Recalling the definition \eqref{eq:pb_L2_C} for the Poisson bracket $\pb{\cdot}{\cdot}_{L^2,\C}$, we then find that for appropriate $\lambda,\mu\in\R$,
\begin{align}
&\pb{\tl{p}_L(\lambda)}{\tl{p}_L(\mu)}_{L^2,\C}(\phi_1,\ol{\phi_2}) \nonumber\\
&= -\frac{i}{4} \paren*{1-\paren*{\frac{\tl{F}_L(\phi_1,\ol{\phi_2};\lambda)}{2}}^2}^{-1/2}\paren*{1-\paren*{\frac{\tl{F}_L(\phi_1,\ol{\phi_2};\mu)}{2}}^2}^{-1/2} \nonumber\\
&\phantom{=}\hspace{5mm}\times \int_{-L}^L dx \paren*{(\grad_1 \tl{F}_L(\lambda))(\phi_1,\ol{\phi_2})(\grad_{\bar{2}}\tl{F}_L(\mu))(\phi_1,\ol{\phi_2}) - (\grad_{\bar{2}}\tl{F}_L(\lambda))(\phi_1,\ol{\phi_2})(\grad_1\tl{F}_L(\mu))(\phi_1,\ol{\phi_2})}(x) \nonumber\\
&=\frac{1}{4} \paren*{1-\paren*{\frac{\tl{F}_L(\phi_1,\ol{\phi_2};\lambda)}{2}}^2}^{-1/2}\paren*{1-\paren*{\frac{\tl{F}_L(\phi_1,\ol{\phi_2};\mu)}{2}}^2}^{-1/2}\pb{\tl{F}_L(\lambda)}{\tl{F}_L(\mu)}_{L^2,\C}(\phi_1,\ol{\phi_2}) \nonumber\\
&=0, \nonumber
\end{align}
where the ultimate equality follows from an application of \cref{lem:tl_F_invol}.

\medskip
We now use \eqref{eq:p_invol} to prove the mutual involution of the functionals $I_{b,n}$.
\begin{prop}
\label{prop:Ib_invol}
For any $n,m\in\N$, it holds that
\begin{equation}
\pb{I_{b,n}}{I_{b,m}}_{L^2,\mathcal{V}}\equiv 0 \text{ on $C^\infty(\T_L;\mathcal{V})$}.
\end{equation}
\end{prop}
\begin{proof}
Fix $n,m\in\N$, and let $\gamma=\frac{1}{2}\mathrm{odiag}(\phi_1,\ol{\phi_2},\phi_2,\ol{\phi_1})\in C^\infty(\T_L;\mathcal{V})$. Let us first introduce some notation that will simplify the computations in the sequel. Define
and
\begin{equation}
\label{eq:pL_def}
p_L(\gamma;\lambda) \coloneqq \tl{p}_L(\phi_1,\ol{\phi_2};\lambda) + \tl{p}_L(\phi_2,\ol{\phi_1};\lambda), \qquad \forall (\gamma,\lambda)\in C^\infty(\T_L;\mathcal{V})\times\C,
\end{equation}
where we recall that $\tl{p}_L$ is defined in \eqref{eq:tlp_def}. Note that it is tautological that $p_L$ is the restriction of a complex-valued functional on $C^\infty(\T_L)^4$, which by an abuse of notation we write as
\begin{equation}
\label{eq:pL_ext}
p_L(\phi_1,\phi_{\bar{2}},\phi_2,\phi_{\bar{1}};\lambda) = \tl{p}_L(\phi_1,\phi_{\bar{2}};\lambda) + \tl{p}_L(\phi_2,\phi_{\bar{1}};\lambda), \qquad \phi_1,\phi_{\bar{1}},\phi_2,\phi_{\bar{2}}\in C^\infty(\T_L).
\end{equation}

Now for $\gamma\in C^\infty(\T_L;\V)$, we have by the variational derivative formulation of the Poisson bracket $\pb{p_L(\lambda)}{p_L(\mu)}_{L^2,\V}$ (recall \eqref{eq:L2V_pb_vd}) and \eqref{eq:pL_ext} that
\begin{align}
&\pb{p_L(\lambda)}{p_L(\mu)}_{L^2,\V}(\gamma) \nonumber\\
&= -i\int_{-L}^L dz \paren*{(\grad_1 p_L(\lambda))(\grad_{\bar{2}}p_L(\mu)) - (\grad_{\bar{2}}p_L(\lambda))(\grad_{1} p_L(\mu))}(\phi_1,\ol{\phi_2},\phi_2,\ol{\phi_1})(z) \nonumber\\
&\phantom{=} -i\int_{-L}^L dz \paren*{(\grad_2 p_L(\lambda))(\grad_{\bar{1}}p_L(\mu)) - (\grad_{\bar{1}}p_L(\lambda))(\grad_{2} p_L(\mu))}(\phi_1,\ol{\phi_2},\phi_2,\ol{\phi_1})(z) \nonumber\\
&=-i\int_{-L}^L dz \paren*{(\grad_1\tl{p}_L(\lambda)(\grad_{\bar{2}}\tl{p}_L(\mu)) - (\grad_{\bar{2}} \tl{p}_L(\lambda))(\grad_1 \tl{p}_L(\mu))}(\phi_1,\ol{\phi_2})(z) \nonumber\\
&\phantom{=} -i\int_{-L}^L dz \paren*{(\grad_1\tl{p}_L(\lambda)(\grad_{\bar{2}}\tl{p}_L(\mu)) - (\grad_{\bar{2}} \tl{p}_L(\lambda))(\grad_1 \tl{p}_L(\mu))}(\phi_2,\ol{\phi_1})(z). \label{eq:pL_vd_sub}
\end{align}
Recalling \cref{rem:pb_vd_C} for the variational derivative formulation of the Poisson bracket $\pb{\cdot}{\cdot}_{L^2,\C}$, we can rewrite the right-hand side of the preceding equality to obtain that
\begin{equation}
\pb{p_L(\lambda)}{p_L(\mu)}_{L^2,\V}(\gamma) = \pb{\tl{p}_L(\lambda)}{\tl{p}_L(\mu)}_{L^2,\C}(\phi_1,\ol{\phi_2}) + \pb{\tl{p}_L(\lambda)}{\tl{p}_L(\mu)}_{L^2,\C}(\phi_2,\ol{\phi_1}).
\end{equation}
Given $R>0$, for all $\gamma \in C^\infty(\T_L;\V)$ with $\|\p_x^n\gamma\|_{L^\infty(\T_L)}\leq R$, for $n\in\{0,1\}$, we can choose $\lambda,\mu\in\R$ arbitrarily large to apply \eqref{eq:p_invol}, yielding that both terms in the right-hand side of the preceding equality are zero. Hence,
\begin{equation}
\pb{p_L(\lambda)}{p_L(\mu)}_{L^2,\V}(\gamma) = 0.
\end{equation}

Now by the formula \eqref{eq:tl_pL_form} for $\tl{p}_L(\lambda)$ and the large real $\lambda$ asymptotic expansion \eqref{eq:w_asy} for $w_{(\phi_1,\ol{\phi_2})}(\lambda)$, we see that
\begin{equation}
\tl{p}_L(\phi_1,\ol{\phi_2};\lambda) \sim -\lambda L + \kappa\sum_{k=1}^\infty \frac{\int_{-L}^Ldx \ol{\phi_2}(x) w_{k,(\phi_1,\ol{\phi_2})}(x)}{\lambda^k} = -\lambda L + \kappa\sum_{k=1}^\infty \frac{\tl{I}_k(\phi_1,\ol{\phi_2})}{\lambda^k},
\end{equation}
where the ultimate equality follows from the definition \eqref{eq:tlIn_def} for $\tl{I}_k$. Taking the variational derivatives of both sides of the preceding identity, we find that
\begin{equation}
\label{eq:tlpL_vd}
\grad_1 \tl{p}_L (\phi_1,\ol{\phi_2};\lambda) \sim \kappa\sum_{k=1}^\infty \frac{\grad_1 \tl{I}_k(\phi_1,\ol{\phi_2})}{\lambda^k}, \qquad \grad_{\bar{2}} \tl{p}_L(\phi_1,\ol{\phi_2};\lambda) \sim \kappa\sum_{k=1}^\infty \frac{\grad_{\bar{2}}\tl{I}_k(\phi_1,\ol{\phi_2})}{\lambda^k}.
\end{equation}
Substituting the asymptotic expansions \eqref{eq:tlpL_vd} into \eqref{eq:pL_vd_sub}, we see that
\begin{align}
0 &= \pb{p_L(\lambda)}{p_L(\mu)}_{L^2,\mathcal{V}}(\gamma) \nonumber\\
&\sim -i\kappa^2\sum_{k,j=1}^\infty \frac{1}{\lambda^k \mu^j} \int_{-L}^L dz \left(\grad_1 \tl{I}_k(\phi_1,\ol{\phi_2})\grad_{\bar{2}}\tl{I}_j(\phi_1,\ol{\phi_2}) - \grad_{\bar{2}}\tl{I}_k(\phi_1,\ol{\phi_2})\grad_1\tl{I}_j(\phi_1,\ol{\phi_2}) \right)(z) \nonumber\\
&\phantom{=} - i\kappa^2\sum_{k,j=1}^\infty \frac{1}{\lambda^k \mu^j} \int_{-L}^L dz \left(\grad_1 \tl{I}_k(\phi_2,\ol{\phi_1})\grad_{\bar{2}}\tl{I}_j(\phi_2,\ol{\phi_1}) - \grad_{\bar{2}}\tl{I}_k(\phi_2,\ol{\phi_1})\grad_1\tl{I}_j(\phi_2,\ol{\phi_1}) \right)(z) \nonumber\\
&=\sum_{k,j=1}^\infty \frac{4\pb{I_{b,k}}{I_{b,j}}_{L^2,\V}(\gamma)}{\lambda^k\mu^j},
\end{align}
where the ultimate equality follows from \cref{rem:pb_V_vd} and the definition \eqref{eq:Ibn_def} of the functionals $I_{b,n}$. By the uniqueness of coefficients of asymptotic expansions, we conclude that $\pb{I_{b,k}}{I_{b,j}}_{L^2,\V}\equiv 0$ on $C^\infty(\T_L;\V)$, completing the proof of the proposition.
\end{proof}

\section{Multilinear algebra}
\label{app:multi_alg}
In this appendix, we review some useful facts from multilinear algebra about symmetric tensors, which we make use of to prove \cref{thm:GP_invol}. Throughout this appendix, $V$ denotes a finite-dimensional complex vector space unless specified otherwise. For concreteness, the reader can just take $V=\C^d$, where $d$ is the dimension of $V$. For more details and the omitted proofs, we refer the reader to \cite{Greub1978} and \cite{CGLM2008}, in particular the latter for a concise, pedestrian exposition.

Let $n\in\N$, and let $V^{\times n}\rightarrow V^{\otimes n}$ be an algebraic $n$-fold tensor product\footnote{The reader will recall that the tensor product is only defined up to unique isomorphism.} for $V$. Now given any $n$-linear map $T:V^{\times n}\rightarrow W$, where $W$ is another complex finite-dimensional vector space, the universal property of the tensor product asserts that there exists a unique linear map $\bar{T}: V^{\otimes n}\rightarrow W$, such that the following diagram commutes
\begin{equation}
\begin{tikzcd}
V^{\times n} \arrow[rd, "T"] \arrow[r] & V^{\otimes n} \arrow[d, dashrightarrow, "\bar{T}"] \\
& W
\end{tikzcd}.
\end{equation}
In particular, given any permutation $\pi\in \Ss_n$, there is a unique map $\bar{\pi}: V^{\otimes n} \rightarrow V^{\otimes n}$ with the property that
\begin{equation}
\bar{\pi}(v_1\otimes \cdots \otimes v_n) = v_{\pi(1)}\otimes \cdots \otimes v_{\pi(n)}, \qquad \forall v_1,\ldots,v_n\in V.
\end{equation}
Using these maps $\bar{\pi}$, we can define the symmetrization operator $\Sym_n$ on $V^{\otimes n}$ by
\begin{equation}
\Sym_n(u) \coloneqq \frac{1}{n!}\sum_{\pi\in\Ss_n} \bar{\pi}(u), \qquad \forall u\in\ V^{\otimes n}
\end{equation}
and define what it means for a tensor to be symmetric.

\begin{mydef}[Symmetric tensor]
We say that $u\in V^{\otimes n}$ is \emph{symmetric} if $\Sym_n(u) = u$. Equivalently, $\bar{\pi}(u) = u$ for every $\pi\in\Ss_n$. We let $\Sym_n(V^{\otimes n})$, alternatively $\bigotimes_s^n V$ or $V^{\otimes_s^n}$, denote the subspace of $V^{\otimes n}$ consisting of symmetric tensors.
\end{mydef}

\begin{remark}
If $\{e_1,\ldots,e_d\}$ is a basis for $V$, then $\{e_{j_1}\otimes\cdots\otimes e_{j_n}\}_{j_1,\ldots,j_n=1}^d$ is a basis for $V^{\otimes n}$, so that $\dim(V^{\otimes n}) = d^n$. Similarly, $\{\Sym_n(e_{j_1}\otimes \cdots \otimes e_{j_n})\}_{1\leq j_1\leq \cdots \leq j_n\leq d}$ is a basis for $V^{\otimes_s^n}$, so that $\dim(V^{\otimes_s^n})={d+n-1\choose n}$.
\end{remark}

We now claim that any element of $V^{\otimes_s^n}$ is uniquely identifiable with an element of $\C[x_1,\ldots,x_d]_n$, the space of homogeneous polynomials of degree $n$ in $d$ variables. Indeed, fix a basis $\{e_1,\ldots,e_d\}$ for $V$, so that $\{\Sym_n(e_{j_1}\otimes\cdots\otimes e_{j_n})\}_{1\leq j_1\leq\cdots\leq j_n\leq d}$ is a basis for $V^{\otimes_s^n}$. By mapping
\begin{equation}
\Sym_n(e_{j_1}\otimes \cdots\otimes e_{j_n}) \mapsto x_1^{\alpha_1}\cdots x_d^{\alpha_d} \eqqcolon \ux_d^{\ul{\alpha}_d},
\end{equation}
where $\ul{\alpha}_d$ is the multi-index defined by
\begin{equation}
\label{eq:alpha_func}
\alpha_j \coloneqq \sum_{i=1}^n \delta_j(j_i), \qquad \forall j\in \N_{\leq d},
\end{equation}
where $\delta_j$ is the discrete Dirac mass centered at $j$, one obtains a linear map from $V^{\otimes_s^n} \rightarrow \C[x_1,\ldots,x_d]_n$. One can show this map is, in fact, an isomorphism. Consequently, if
\begin{equation}
u = \sum_{1\leq j_1\leq \cdots \leq j_n \leq d} u_{j_1\cdots j_n} \Sym_n\paren*{e_{j_1}\otimes\cdots\otimes e_{j_n}}
\end{equation}
is an element of $V^{\otimes_s^n}$, then $u$ is uniquely identifiable with the element $F\in \C[x_1,\ldots,x_d]_n$ given by
\begin{equation}
\label{eq:sym_pol_iso}
F(\ux_d) = \sum_{1\leq j_1\leq \cdots\leq j_n\leq d}  u_{j_1\cdots j_n}\ux_d^{\ul{\alpha}_d(\ul{j}_n)},
\end{equation}
where we write $\ul{\alpha}_d(\ul{j}_n)$ to emphasize that $\ul{\alpha}_d$ is intended as a function of $\ul{j}_n$ according to the rule \eqref{eq:alpha_func}.

There is a useful bilinear form on $\C[x_1,\ldots,x_d]_n$ defined as follows: if $F,G\in \C[x_1,\ldots,x_d]_n$ are respectively given by
\begin{equation}
F(\ux_d) = \sum_{|\ul{\alpha}_d|=n} {n\choose {\alpha_1,\ldots,\alpha_d}} a_{\ul{\alpha}_d} \ul{x}_d^{\ul{\alpha}_d}, \qquad G(\ux_d) = \sum_{|\ul{\alpha}_d|=n} {n\choose {\alpha_1,\ldots,\alpha_d}} b_{\ul{\alpha}_d} \ul{x}_d^{\ul{\alpha}_d},
\end{equation}
then we define
\begin{equation}
\ipp{F,G} \coloneqq \sum_{|\ul{\alpha}_d|=n} {n\choose {\alpha_1,\ldots,\alpha_d}} a_{\ul{\alpha}_d}b_{\ul{\alpha}_d}.
\end{equation}

The form $\ipp{\cdot,\cdot}$, which is evidently symmetric, has the important property of nondegeneracy, as the next lemma shows.

\begin{lemma}[Nondegeneracy]
\label{lem:form_nd}
The symmetric bilinear form $\ipp{\cdot,\cdot}: \C[x_1,\ldots,x_d]_n \times \C[x_1,\ldots,x_d]_n \rightarrow\C$ is nondegenerate: if $\ipp{F,G}=0$ for all $G\in\C[x_1,\ldots,x_d]_n$, then $F\equiv 0$.
\end{lemma}

When $G$ is of the form $G(\ux_d) = (\beta_1 x_1+\cdots+\beta_d x_d)^n$ (i.e. an $n^{th}$ power of a linear form), then the next lemma provides an explicit formula for $\ipp{F,G}$.

\begin{lemma}
\label{lem:form_eval}
If $G(\ux_d) = (\beta_1 x_1+\cdots+\beta_d x_d)^n$, where $\ul{\beta}_d\in\C^d$, then for every $F\in \C[x_1,\ldots,x_d]_n$, we have that
\begin{equation}
\ipp{F,G} = F(\ul{\beta}_d).
\end{equation}
\end{lemma}

We now use \cref{lem:form_eval} to prove the existence of a special decomposition for elements of $V^{\otimes_s^n}$. We have included a proof as it is a nice argument.

\begin{lemma}[Symmetric rank-1 decomposition]
\label{lem:sr1_dcomp}
For any $u\in V^{\otimes_s^n}$, there exists an integer $N\in\N$, coefficients $\{a_j\}_{j=1}^N\subset\C$, and elements $\{v_j\}_{j=1}^N \subset V$, such that
\begin{equation}
\label{eq:sr1_dcomp}
u = \sum_{j=1}^N a_j v_j^{\otimes n}.
\end{equation}
\end{lemma}
\begin{proof}
Let $W\subset V^{\otimes_s^n}$ denote the set of elements which admit a decomposition of the form \eqref{eq:sr1_dcomp}. Evidently, $W$ is a subspace of $V^{\otimes_s^n}$. Fix a basis $\{e_1,\ldots,e_d\}$ for $V$. If $v=\beta_1 e_1+\cdots+ \beta_d e_d$, then one can check that under the isomorphism given by \eqref{eq:sym_pol_iso}, $v^{\otimes n}$ is uniquely identifiable with the polynomial
\begin{equation*}
(\beta_1 x_1+\cdots+\beta_d x_d)^n,
\end{equation*}
i.e. an $n^{th}$ power of a linear form. Consequently, $W$ is isomorphic to the span of $n^{th}$ powers of linear forms in $\C[x_1,\ldots,x_d]_n$.

Assume for the sake of a contradiction that $W$ is a proper subspace, so that the orthogonal complement $W^{\perp}$ with respect to the form $\ipp{\cdot,\cdot}$ is nontrivial. Then it follows from the embedding of $W\subset \C[x_1,\ldots,x_d]_n$ that there exists a nonzero polynomial $F\in \C[x_1,\ldots,x_d]_n$, such that $\ipp{F,G} = 0$ for every $G\in W$. \cref{lem:form_eval} then implies that $F(\ul{\beta}_d) = 0$ for every $\ul{\beta}_d\in \C^d$, which contradicts that $F$ is a nonzero polynomial.
\end{proof}

\begin{remark}
Since \cref{lem:sr1_dcomp} asserts that a decomposition of the form \eqref{eq:sr1_dcomp} always exists, one can define the \emph{symmetric rank} of an element $u\in V^{\otimes_s^n}$ by the minimal integer $N$. Evidently, a tensor of the form $v^{\otimes n}$ has symmetric rank 1. Although we will not need the notion of symmetric rank in this work, we will refer to the decomposition \eqref{eq:sr1_dcomp} as a symmetric-rank-1 decomposition. 
\end{remark}

As an application of the symmetric-rank-1 tensor decomposition, we now show an approximation result for bosonic Schwartz functions (i.e. elements of $\Sc_s(\R^d)$).

\begin{lemma}
\label{lem:Sch_ST_d}
Let $f\in\Sc_s(\R^d)$. Then given $\varepsilon>0$ and a Schwartz seminorm $\mathcal{N}$, there exist $N\in\N$, elements $\{f_i\}_{i=1}^N\subset \Sc(\R)$, and coefficients $\{a_{i}\}_{i=1}^N\subset\C$, such that
\begin{equation}
\mathcal{N}\paren*{f-\sum_{i=1}^N a_i f_i^{\otimes d}}\leq \varepsilon.
\end{equation}
In other words, finite linear combinations of symmetric-rank-1 tensor products are dense in $\Sc_s(\R^d)$.
\end{lemma}
\begin{proof}
Fix $f\in\Sc_s(\R^d)$, $\varepsilon>0$, and seminorm $\mathcal{N}$. Since $\Sc_s(\R^d) \cong \hat{\bigotimes}_s^d \Sc(\R)$, there exists an integer $M\in\N$, elements $\{g_{ij}\}_{{1\leq i\leq d}\atop {1\leq j\leq M}}\subset \Sc(\R)$, and coefficients $\{a_j\}_{1\leq j\leq M}\subset \C$, such that
\begin{equation}
\label{eq:approx_wtp}
\mathcal{N}\paren*{f-\sum_{j=1}^M a_j \Sym_d\paren*{\bigotimes_{i=1}^d g_{ij} } }\leq \varepsilon.
\end{equation}
Define the complex vector space
\begin{equation}
V\coloneqq \mathrm{span}_\C\{g_{ij} : 1\leq i\leq d, \ 1\leq j\leq M\},
\end{equation}
which is evidently finite-dimensional. For each $j\in\N_{\leq M}$, consider the symmetric tensor
\begin{equation}
\Sym_d\paren*{\bigotimes_{i=1}^d g_{ij}}\in V^{\otimes_s^d}.
\end{equation}
By \cref{lem:sr1_dcomp}, there exists an integer $N_j\in\N$, elements $\{f_{j\ell}\}_{\ell=1}^{N_j}\subset V$, coefficients $\{a_{j\ell}\}_{\ell=1}^{N_j}\subset\C$, such that
\begin{equation}
\Sym_d\paren*{\bigotimes_{i=1}^d g_{ij}} = \sum_{\ell=1}^{N_j} a_{j\ell} f_{j\ell}^{\otimes_d}.
\end{equation}
Consequently,
\begin{equation}
\sum_{j=1}^M a_{j} \Sym_d\paren*{\bigotimes_{i=1}^d g_{ij}} = \sum_{j=1}^M\sum_{\ell=1}^{N_j} a_j a_{j\ell} f_{j\ell}^{\otimes_d},
\end{equation}
so upon substitution of this identity into \eqref{eq:approx_wtp}, we obtain the desired conclusion.
\end{proof}

As a corollary of \cref{lem:Sch_ST_d}, we obtain the following decomposition for elements in $\L(\Sc_s'(\R^d),\Sc_s(\R^d))$.

\begin{cor}
\label{cor:Sch_ST_DM_d}
Let $\gamma^{(d)}\in \L(\Sc_s'(\R^d),\Sc_s(\R^d))$. Then given $\varepsilon>0$ and a Schwartz seminorm $\mathcal{N}$, there exists $N\in\N$, elements $\{f_i,g_i\}_{i=1}^N\subset\Sc(\R)$, and coefficients $\{a_i\}_{i=1}^N \subset\C$, such that
\begin{equation}
\mathcal{N}\paren*{\gamma^{(d)} - \sum_{i=1}^N a_i f_i^{\otimes d} \otimes g_i^{\otimes d}} \leq \varepsilon.
\end{equation}
\end{cor}
\begin{proof}
Fix $\gamma^{(d)}\in\L(\Sc_s'(\R^d),\Sc_s(\R^d))$, $\varepsilon>0$, and seminorm $\mathcal{N}$. Since
\begin{equation*}
\L(\Sc_s'(\R^d),\Sc_s(\R^d)) \cong \Sc_s(\R^d)\hat{\otimes}\Sc_s(\R^d),
\end{equation*}
there exists an integer $N$, elements $\{\tl{f}_i,\tl{g}_i\}_{i=1}^N \subset\Sc_s(\R^d)$, and coefficients $\{a_i\}_{i=1}^N\subset\C$, such that
\begin{equation}
\label{eq:approx_dm_wtp}
\mathcal{N}\paren*{\gamma^{(d)} - \sum_{i=1}^N a_i \tl{f}_i \otimes \tl{g}_i } \leq \varepsilon.
\end{equation}
For each $i\in\N_{\leq N}$, \cref{lem:Sch_ST_d} implies that there exist integers $N_{i,f}, N_{i,g}\in\N$, elements $\{f_{ij}\}_{j=1}^{N_{i,f}}, \{g_{ij}\}_{j=1}^{N_{i,g}}\subset\Sc(\R)$, and coefficeints $\{a_{ij,f}\}_{j=1}^{N_{i,f}}, \{a_{ij,g}\}_{j=1}^{N_{i,g}}\subset\C$, such that
\begin{equation}
\tl{f}_i = \sum_{j=1}^{N_{i,f}} a_{ij,f} f_{ij}^{\otimes d}, \qquad \tl{g}_i = \sum_{j=1}^{N_{i,g}} a_{ij,g} g_{ij}^{\otimes d}.
\end{equation}
By setting coefficients equal to zero, we may assume without loss of generality that $N_{i,f}=N_{i,g}=M\in\N$, for every $i\in\N_{\leq N}$. So by the bilinearity of tensor product, we obtain the decomposition
\begin{equation}
\sum_{i=1}^N a_i \tl{f}_i\otimes \tl{g}_i = \sum_{i=1}^N \sum_{j,j'=1}^M a_i a_{ij,f}a_{ij',g} f_{ij}^{\otimes d} \otimes g_{ij'}^{\otimes d}.
\end{equation}
Substitution of this identity into \eqref{eq:approx_dm_wtp} and relabeling/re-indexing of the summation yields the desired conclusion.
\end{proof}

\section{Distribution-valued operators}\label{app:DVO}
Following Appendix B of our companion paper \cite{MNPRS1_2019}, we review and develop some properties of distribution-valued operators (DVOs) (i.e. elements of $\L(\Sc(\R^{k}),\Sc'(\R^{k}))$),  which are used extensively in this work. Most of these properties are a special case of a more general theory involving topological tensor products of locally convex spaces for which we refer the reader to \cite{Schwartz1966, Horvath1966, Treves1967} for further reading.

\subsection{Adjoint}
In this subsection, we record some properties of the adjoint of a DVO as well as some properties of the map taking a DVO to its adjoint.

\begin{lemma}[Adjoint map]\label{lem:dvo_adj}
Let $k\in\N$, and let $A^{(k)}\in \L(\Sc(\R^{k}),\Sc'(\R^{k}))$. Then there is a unique map $(A^{(k)})^{*}\in\L(\Sc(\R^{k}),\Sc'(\R^{k}))$ such that
\begin{equation}\label{eq:adj_prop}
\ipp*{(A^{(k)})^{*}g^{(k)},  \ol{f^{(k)}}}_{\Sc'(\R^{k})-\Sc(\R^{k})} = \ol{\ipp*{A^{(k)}f^{(k)}, \ol{g^{(k)}}}}_{\Sc'(\R^{k})-\Sc(\R^{k})}, \qquad \forall f^{(k)},g^{(k)}\in\Sc(\R^{k}).
\end{equation}
Furthermore, the adjoint map
\begin{equation}
*:\L(\Sc(\R^{k}),\Sc'(\R^{k})) \rightarrow \L(\Sc(\R^{k}),\Sc'(\R^{k})), \qquad A^{(k)} \mapsto (A^{(k)})^{*}
\end{equation}
is a continuous involution. 

Additionally, for $B^{(k)}\in\L(\Sc'(\R^{k}),\Sc'(\R^{k}))$, there exists a unique linear map in $(B^{(k)})^{*}\in\L(\Sc(\R^{k}),\Sc(\R^{k}))$ such that
\begin{equation}
\ipp*{u^{(k)},\ol{(B^{(k)})^{*}g^{(k)}}}_{\Sc'(\R^{k})-\Sc(\R^{k})} = \ipp*{B^{(k)}u^{(k)}, \ol{g^{(k)}}}_{\Sc'(\R^{k})-\Sc(\R^{k})}, \quad \forall (g^{(k)},u^{(k)})\in\Sc(\R^{k})\times\Sc'(\R^{k}).
\end{equation}
Moreover, the adjoint map
\begin{equation}
*:\L(\Sc'(\R^{k}),\Sc'(\R^{k})) \rightarrow \L(\Sc(\R^{k}),\Sc(\R^{k}))
\end{equation}
is a continuous involution.
\end{lemma}

The next lemma is useful for computing the adjoint of the composition of maps.

\begin{lemma}
\label{lem:adj_comp}
Let $A^{(k)}\in\L(\Sc(\R^{k}),\Sc'(\R^{k}))$ and $B^{(k)}\in\L(\Sc'(\R^{k}),\Sc'(\R^{k}))$. Then
\begin{equation}
\paren*{B^{(k)}A^{(k)}}^{*} = (A^{(k)})^{*}(B^{(k)})^{*}.
\end{equation}
\end{lemma}

\begin{mydef}[Self- and skew-adjoint]\label{def:dvo_sa}
Given $k\in\N$, we say that an operator $A^{(k)} \in \L(\Sc(\R^{k}),\Sc'(\R^{k}))$ is self-adjoint if $(A^{(k)})^{*}=A^{(k)}$. Similarly, we say that $A^{(k)}\in \L(\Sc(\R^{k}),\Sc'(\R^{k}))$ is skew-adjoint if $(A^{(k)})^{*}=-A^{(k)}$.
\end{mydef}

\begin{remark}
Note that if $A^{(k)}\in \L(\Sc(\R^{k}),\Sc'(\R^{k}))$ is an operator mapping $\Sc(\R^{k}) \rightarrow L^{2}(\R^{k})$, then our definition of self-adjoint does \emph{not} coincide with the usual Hilbert space definition for densely defined operators, but instead with the definition of a symmetric operator.
\end{remark}

\subsection{Trace and partial trace}\label{sec:trace}
In this subsection, we generalize the trace of an operator on a separable Hilbert space to the DVO setting. Viewing the trace as a \emph{bilinear} map and using the canonical isomorphisms
\begin{equation}
\L(\Sc(\R^{N}),\Sc'(\R^{N})) \cong \Sc'(\R^{2N}) \enspace \text{and} \enspace \L(\Sc'(\R^{N}),\Sc(\R^{N})) \cong \Sc(\R^{2N})
\end{equation}
given by the Schwartz kernel theorem, we can define the generalized trace of the right-composition of an operator in $\L(\Sc(\R^{N}),\Sc'(\R^{N}))$ with an operator in $\L(\Sc'(\R^{N}),\Sc(\R^{N}))$ through the pairing of their Schwartz kernels. More precisely,
\begin{equation}
\label{eq:gtr_dpair}
\Tr_{1,\ldots,N}(A^{(N)}\gamma^{(N)}) = \ipp{A^{(N)}, (\gamma^{(N)})^t}_{\Sc'(\R^{2N})-\Sc(\R^{2N})}
\end{equation}
is, with an abuse of notation, the distributional pairing of the Schwartz kernel of $A^{(N)}$, which belongs to $\Sc'(\R^{2N})$, with the Schwartz kernel of the transpose of $\gamma^{(N)}$,\footnote{$(\gamma^{(N)})^t$ is the operator $f\mapsto \int_{\R^N}d\ux_N'\gamma(\ux_N';\ux_N)f(\ux_N')$.}, which belongs to $\Sc(\R^{2N})$.

\begin{mydef}[Generalized trace]\label{def:gen_trace}
We define
\begin{equation}
\begin{split}
&\Tr_{1,\ldots,N}:\L(\Sc(\R^{N}),\Sc'(\R^{N}))\times\L(\Sc'(\R^{N}),\Sc(\R^{N}))\rightarrow\C \\
&\Tr_{1,\ldots,N}\paren*{A^{(N)}\gamma^{(N)}} \coloneqq \ipp{A^{(N)}, (\gamma^{(N)})^t}_{\Sc'(\R^{2N})-\Sc(\R^{2N})}.
\end{split}
\end{equation}
\end{mydef}

\begin{remark}
The Schwartz kernel theorem implies that for $A^{(N)}\in \L(\Sc(\R^N),\Sc'(\R^N))$,
\begin{equation} \label{equ:tr}
\Tr_{1,\ldots,N}\paren*{A^{(N)}(f\otimes g)}=\ipp{A^{(N)}f,g}_{\Sc'(\R^{N})-\Sc(\R^{N})}, \qquad \forall f,g\in\Sc(\R^N).
\end{equation}
\end{remark}

\begin{remark}
The reader can check that if $A^{(N)}\in \L(\Sc(\R^N),\Sc'(\R^N))$ and $\gamma^{(N)}\in \L(\Sc'(\R^N),\Sc(\R^N))$ are such that $A^{(N)}\gamma^{(N)}$ is a trace-class operator $\rho^{(N)}$, then our definition of the generalized trace of $A^{(N)}\gamma^{(N)}$ coincides with the usual definition of the trace of $\rho^{(N)}$ as an operator on the Hilbert space $L^2(\R^N)$.
\end{remark}

We now record some properties of the generalized trace which are reminiscent of properties of the usual trace encountered in functional analysis.

\begin{prop}[Properties of generalized trace]
\label{prop:gtr_prop}
Let $A^{(N)}\in \L(\Sc(\R^{N}),\Sc'(\R^{N}))$, and let $\gamma^{(N)}\in\L(\Sc'(\R^{N}),\Sc(\R^{N}))$. The following properties hold:
\begin{enumerate}[(i)]
\item\label{item:gtr_sc}
$\Tr_{1,\ldots,N}$ is separately continuous.
\item\label{item:gtr_adj} We have the following identity:
\begin{equation}
\Tr_{1,\ldots,N}\paren*{(A^{(N)})^{*}\gamma^{(N)}} = \ol{\Tr_{1,\ldots,N}\paren*{A^{(N)}(\gamma^{(N)})^{*}}}.
\end{equation}
\item\label{item:gtr_cyc}
If $B^{(N)}\in \L(\Sc'(\R^{N}),\Sc'(\R^{N}))$, then $\Tr_{1,\ldots,N}$ satisfies the cyclicity property
\begin{equation}
\Tr_{1,\ldots,N}\paren*{\paren*{B^{(N)}A^{(N)}}\gamma^{(N)}} = \Tr_{1,\ldots,N}\paren*{A^{(N)}\paren*{\gamma^{(N)}B^{(N)}}}.
\end{equation}
\end{enumerate}
\end{prop}

We now extend the partial trace map to our setting using our bilinear perspective.

\begin{prop}[Generalized partial trace]\label{prop:partial_trace}
Let $N\in\N$ and let $k\in\{0,\ldots,N-1\}$. Then there exists a unique bilinear, separately continuous map 
\begin{equation}
\Tr_{k+1,\ldots,N}: \L(\Sc(\R^{N}),\Sc'(\R^{N})) \times \L(\Sc'(\R^{N}),\Sc(\R^{N})) \rightarrow \L(\Sc(\R^{k}),\Sc'(\R^{k})),
\end{equation}
which satisfies
\begin{equation}\label{eq:gpt_up}
\Tr_{k+1,\ldots,N}\paren*{A^{(N)}(f^{(N)}\otimes g^{(N)})} = \int_{\R^{N-k}}d\ux_{k+1;N}(A^{(N)}f^{(N)})(\ux_{k},\ux_{k+1;N})g^{(N)}(\ux_{k}',\ux_{k+1;N}).
\end{equation}
for all $A^{(N)}\in \L(\Sc(\R^{N}),\Sc'(\R^{N}))$, and $f^{(N)}, g^{(N)} \in \Sc(\R^{N})$. That is,
\begin{equation}
\begin{split}
&\ipp*{\Tr_{k+1,\ldots,N}\paren*{A^{(N)}(f^{(N)}\otimes g^{(N)})}\phi^{(k)},\psi^{(k)}}_{\Sc'(\R^k)-\Sc(\R^k)} \\
&= \ipp*{A^{(N)}f^{(N)},\psi^{(k)}\otimes \ipp{g^{(N)},\phi^{(k)}}_{\Sc_{\ux_k}'(\R^k)-\Sc_{\ux_k}(\R^k)}}_{\Sc'(\R^N)-\Sc(\R^N)},
\end{split}
\end{equation}
for all $\phi^{(k)},\psi^{(k)}\in\Sc(\R^k)$.
\end{prop}

\begin{remark}
Our notation $\Tr_{k+1, \ldots, N}$ implies a partial trace over the variables with indices belonging to the index set $\{ i \,:\, k+1 \leq i \leq N\}$. To alleviate some notational complications, we will use the convention that if the index set of the partial trace is empty, we do not take a partial trace.
\end{remark}

\subsection{Contractions and the ``good mapping property''}\label{ssec:GMP}
Given $A^{(i)}\in \L(\Sc(\R^{i}),\Sc'(\R^{i}))$, an integer $k\geq i$, and a cardinality-$i$ subset $\{\ell_{1},\ldots,\ell_{i}\} \subset \N_{\leq k}$, we want to define to an operator acting only on the variables associated to $\{\ell_{1},\ldots,\ell_{i}\}$. We have the following result.

\begin{prop}[$k$-particle extensions]
\label{prop:ext_k}
There exists a unique $A_{(\ell_{1},\ldots,\ell_{i})}^{(i)}\in \L(\Sc(\R^{k}),\Sc'(\R^{k}))$, which satisfies
\begin{equation}\label{eq:op_index_notation}
A_{(\ell_{1},\ldots,\ell_{i})}^{(i)}(f_{1}\otimes\cdots\otimes f_{k})(\ux_{k}) = A^{(i)}(f_{\ell_{1}}\otimes\cdots\otimes f_{\ell_{i}})(x_{\ell_{1}},\ldots,x_{\ell_{i}}) \cdot \biggl(\prod_{\ell\in\N_{\leq k}\setminus\{\ell_{1},\ldots,\ell_{i}\}} f_{\ell}(x_{\ell})\biggr)
\end{equation}
in the sense of tempered distributions.
\end{prop}

An important property of the above $k$-particle extension is that it preserves self- and skew-adjointness.

\begin{lemma}\label{lem:ext_sa}
Let $i\in\N$, let $k\in\N_{\geq i}$, and let $A^{(i)}\in\L(\Sc(\R^{k}),\Sc'(\R^{i}))$ be self-adjoint (resp skew-adjoint). Then for any cardinality-$i$ subset $\{\ell_{1},\ldots,\ell_{i}\}\subset\N_{\leq k}$, we have that $A_{(\ell_{1},\ldots,\ell_{i})}^{(i)}$ is self-adjoint (resp. skew-adjoint).
\end{lemma}

Now let $i,j\in\N$, let $k\coloneqq i+j-1$, and let $(\alpha,\beta)\in\N_{\leq i}\times\N_{\leq j}$. The proof of \cref{prop:G_inf_br} in \cite{MNPRS1_2019} requires us to give meaning to the composition
\begin{equation}
\label{eq:comp_wd}
A_{(1,\ldots,i)}^{(i)}B_{(i+1,\ldots,i+\beta-1,\alpha,i+\beta,\ldots,k)}^{(j)}
\end{equation}
as an operator in $\L(\Sc(\R^{k}),\Sc'(\R^{k}))$, when $A^{(i)}\in\L(\Sc(\R^{i}),\Sc'(\R^{i}))$ and $B^{(j)}\in \L(\Sc(\R^{j}),\Sc'(\R^{j}))$. 

\begin{remark}\label{not_well_def}
Without further conditions on $A^{(i)}$ or $B^{(j)}$, the composition \cref{eq:comp_wd} may not be well-defined. Indeed, consider the operator $A\in \L(\Sc(\R^{2}),\Sc'(\R^{2}))$ defined by
\begin{equation}
Af \coloneqq \delta_{0} f, \qquad \forall f\in\Sc(\R^{2}),
\end{equation}
where $\delta_{0}$ denotes the Dirac mass about the origin in $\R^{2}$. Then for $f,g\in\Sc(\R)$,
\begin{equation}
\int_{\R}dx_{2}(Af^{\otimes 2})(x_{1},x_{2})g^{\otimes 2}(x_{1}',x_{2}) = f(0)g(0) f(x_{1})g(x_{1}')\delta_{0}(x_{1}) \in\Sc'(\R)\otimes\Sc(\R).
\end{equation}
It is easy to show that $f\delta_{0}\in\Sc'(\R)$ does not coincide with a Schwartz function.
\end{remark}

This issue leads us to a property we call the \emph{good mapping property}. The intuition for the good mapping property is the basic fact from distribution theory that the convolution of a distribution of compact support with a Schwartz function is again a Schwartz function. We recall the definition of the good mapping property here.

\gmp*

\begin{remark}\label{rem:gmp}
By tensoring with identity, we see that if $A^{(i)}$ has the good mapping property, then $A_{(\ell_{1},\ldots,\ell_{i})}^{(i)}$ has the good mapping property, where $i$ is replaced by $k$ and $\alpha\in\N_{\leq k}$.
\end{remark}

\subsection{The subspace $\L_{gmp}(\Sc(\R^{k}),\Sc'(\R^{k}))$}
Lastly, we expand more on $\L_{gmp}(\Sc(\R^{k}),\Sc'(\R^{k}))$ as a topological vector subspace of $\L(\Sc(\R^{k}),\Sc'(\R^{k}))$ with the following lemma.

\begin{lemma}\label{gmp_dense}
For $k\in\N$, it holds that
\begin{enumerate}[(i)]
\item
$\L_{gmp}(\Sc(\R^{k}),\Sc'(\R^{k}))$ is a dense subspace of $\L(\Sc(\R^{k}),\Sc'(\R^{k}))$;
\item
The topological dual $\L_{gmp}(\Sc(\R^{k}),\Sc'(\R^{k}))^{*}$ endowed with the strong dual topology is isomorphic to $\L(\Sc'(\R^{k}),\Sc(\R^{k}))$.
\end{enumerate}
\end{lemma}

\section{Products of distributions and the wave front set}\label{app:WF}
In this appendix, we review some basic facts from microlocal analysis about the wave front set of a distribution and its application to proving the well-definedness of the product of two distributions, as used in \cref{ssec:con_GP_W_S1}. We mostly follow the exposition in Chapter VIII of \cite{HormPDOI}, but refer the reader to Chapter IX, \S{10} of \cite{RSII} for a more pedestrian treatment.

\begin{mydef}[Singular support]
Let $u\in \mathcal{D}'(\R^{k})$. We say that $x\in\R^{k}$ is a \emph{regular point} of $u$ if and only if there exists an open neighborhood $U\ni x$ and a function $f:U\rightarrow\C$ which is $C^{\infty}$ on $U$, such that
\begin{equation}
\ipp{u,\phi}_{\D'(\R^{k})-\D(\R^{k})} =\int_{\R^{k}}f(x)\phi(x)dx, \qquad \forall \phi \in C_{c}^{\infty}(\R^{k}) \enspace \text{with $\supp(\phi)\subset U$}.
\end{equation}
We call the set
\begin{equation}
\R^{k}\setminus\{x\in\R^{k} : \enspace \text{$x$ is a regular point for $u$}\}
\end{equation}
the \emph{singular support} of $u$, denoted by $\ssupp(u)$.
\end{mydef}

\begin{remark}
It is evident that $\ssupp(u)\subset\supp(u)$. Since the set of regular points is open (any other point in the neighborhood $U$ above also belongs to the singular support), it follows that $\ssupp(u)$ is a closed subset of $\supp(u)$.
\end{remark}

The singular support is useful for establishing the well-definedness of a product of distributions $uv$ via localization, as the next proposition shows.

\begin{prop}\label{prop:sing_prop}
Let $u,v\in\mathcal{D}'(\R^{k})$, and suppose that $\ssupp(u)\cap \ssupp(v)=\emptyset$. Then there is a unique $w\in \mathcal{D}'(\R^{k})$ such that the following holds:
\begin{enumerate}[(i)]
\item
If $x\notin \ssupp(v)$ and $v=f$ in a neighborhood of $x$, where $f\in C^{\infty}(\R^{k})$, then $w=fu$ in a neighborhood of $x$.
\item
If $x\notin \ssupp(u)$ and $u=g$ in a neighborhood of $x$, where $g\in C^{\infty}(\R^{k})$, then $w=gv$ in a neighborhood of $x$.
\end{enumerate}
\end{prop}
\begin{proof}
See Theorem IX.42 in \cite{RSII}.
\end{proof}

Next, we introduce the wave front set of a distribution. While the singular support captures the location of the singularities of a distribution, the wave front set also contains information about the directions of the high frequencies that cause these singularities.

\begin{mydef}[Wave front set]\label{def:WF}
Let $u\in\D'(\R^{k})$. We say that a point $(\ux_{k},\uxi_{k})\in\R^{k}\times(\R^{k}\setminus\{0\})$ is a \emph{regular directed point} for $u$ if and only if there exist radii $\varepsilon_{x}, \varepsilon_{\xi}>0$ and a function $\phi\in C_c^\infty(\R^{k})$ which is identically one on the open ball $B(\ux_{k}, \varepsilon_{x})$, such that
\begin{equation}
\left|\wh{\phi u}(\lambda \ueta_{k})\right| \lesssim_{N} \paren*{1+|\lambda|}^{-N}, \qquad \forall (\ueta_{k},\lambda) \in B(\uxi_{k}, \varepsilon_{\xi}) \times [0,\infty), \enspace \forall N\in\N_{0}.
\end{equation}
We define the \emph{wave front set} of $u$ to be the complement in $\R^k\times (\R^k\setminus\{0\})$ of the set of regular directed points:
\begin{equation}
\WF(u) \coloneqq \paren*{\R^{k}\times(\R^{k}\setminus\{0\})} \setminus \{(\ux_{k},\uxi_{k})\in\R^{k}\times (\R^{k}\setminus\{0\}) : (\ux_{k},\uxi_{k}) \enspace \text{is a regular directed point for $u$}\}.
\end{equation}
\end{mydef}

\begin{remark}
In \cite{HormPDOI}, H\"{o}rmander uses a definition of the wave front set of a distribution $u$, which is seemingly different from our \cref{def:WF}. More precisely, for any $\ux_{k}\in\R^{k}$ and $\phi\in C_{c}^{\infty}(\R^{k})$, such that $\phi(\ux_{k})\neq 0$, he defines the set $\Sigma(\phi u)$ consisting of all $\uxi_{k}\in\R^{k}\setminus\{0\}$ having no conic neighborhood $U$ such that
\begin{equation}
|\wh{\phi u}(\uxi_{k})| \lesssim_{N} \paren*{1+|\uxi_{k}|}^{-N}, \qquad \forall \uxi_{k}\in U, \enspace \forall N\in\N.
\end{equation}
He then defines the set $\Sigma_{x}(u)$ by
\begin{equation}
\Sigma_{\ux_k}(u) \coloneqq \bigcap_{\phi} \Sigma(\phi u), \enspace \phi\in C_{c}^{\infty}(\R^{k}) \text{ s.t. } \phi(\ux_{k}) \neq 0.
\end{equation}
H\"{o}rmander's definition of the wave front set of $u$, which we denote by $\wt{\WF}(u)$, is then given by
\begin{equation}
\wt{\WF}(u) \coloneqq \{(\ux_{k},\uxi_{k}) \in \R^{k} \times (\R^{k}\setminus\{0\}) : \uxi_{k}\in\Sigma_{\ux_k}(u)\}.
\end{equation}
It follows from \cref{lem:WF_loc_incl} below that $\wt{\WF}(u)=\WF(u)$ (i.e. the two definitions are equivalent).
\end{remark}

We record some properties of the wave front set.

\begin{lemma}\label{lem:WF_loc_incl}
If $u\in\D'(\R^{k})$ and $g\in C_c^\infty(\R^{k})$, then $\WF(gu)\subset \WF(u)$. Similarly, if $u\in\Sc'(\R^{k})$ and $g\in\Sc(\R^{k})$, then $\WF(gu) \subset \WF(u)$.
\end{lemma}

\begin{prop}\label{prop:wf_prop}
Let $u\in\D'(\R^{k})$.
\begin{enumerate}[(a)]
\item\label{item:wf_prop_closed}
$\WF(u)$ is a closed subset of $\R^{k}\times (\R^{k}\setminus\{0\})$.
\item\label{item:wf_prop_cone}
For each $\ux_{k}\in\R^{k}$, the set
\begin{equation}
\WF_{\ux_{k}}(u) \coloneqq \{\uxi_{k} \in\R^{k}\setminus\{0\} : (\ux_{k},\uxi_{k}) \in \WF(u)\}
\end{equation}
is a cone.
\item\label{item:wf_prop_sum}
If $v\in\D'(\R^{k})$, then
\begin{equation}
\WF(u+v) \subset \WF(u) \cup \WF(v).
\end{equation}
\item\label{item:wf_prop_ssupp}
$\ssupp(u) = \{\ux_{k}\in\R^{k} : \WF_{\ux_{k}}(u) \neq \emptyset\}$.
\item\label{item:wf_prop_tp}
If $v\in \D'(\R^{j})$, then
\begin{equation}
\WF(u\otimes v) \subset \paren*{\WF(u) \times \WF(v)} \cup \paren*{\paren*{\supp(u) \times \{0\}} \times \WF(v)} \cup \paren*{\WF(u) \times \paren*{\supp(v)\times\{0\}}}.
\end{equation}
\item\label{inclusion}
If $u \in \mathcal{S}'(\R^i), v \in \mathcal{S}'(\R^j)$ and $w \in \mathcal{S}(\R^{i+j})$ then
\[
\WF((u\otimes v) w ) \subset \WF(u \otimes v).
\]
\end{enumerate}
\end{prop}
\begin{proof}
Properties \ref{item:wf_prop_closed} - \ref{item:wf_prop_sum} are quick consequences of the definition of the wave front set. For \ref{item:wf_prop_ssupp}, see Theorem IX.44 in \cite{RSII}. For property \ref{item:wf_prop_tp}, see Theorem 8.2.9 in \cite{HormPDOI}. Property \ref{inclusion} follows from \cref{lem:WF_loc_incl}.
\end{proof}

In our proof of \cref{lem:Wn_wd}, we will need the following result.
\begin{lemma}[Wave front set of $\delta(x_{i}-x_{j})$]\label{lem:wf_del}
Let $k\in\N$, and let $i<j\in\N_{\leq k}$. Then
\begin{equation*}
\WF(\delta(x_{i}-x_{j})) = \{(\ux_{k},\uxi_{k})\in\R^{k}\times (\R^{k}\setminus\{0\}) : x_{i}=x_{j}, \enspace \xi_{i}+\xi_{j}=0, \enspace \text{and} \enspace \xi_{\ell}=0 \enspace \forall l\in\N_{\leq k}\setminus\{i,j\}\}.
\end{equation*}
\end{lemma}
\begin{proof}
By symmetry, it suffices to consider the case $(i,j)=(1,2)$. Since $\delta(x_{1}-x_{2})$ has singular support in the hyperplane $\{x_{1}=x_{2}\}\subset \R^{k}$, it follows from \cref{prop:wf_prop}\ref{item:wf_prop_ssupp} that $(\ux_{k},\uxi_{k})\in\WF(\delta(x_{1}-x_{2}))$ implies that $x_{1}=x_{2}$.

Now suppose that $(\ux_{k},\uxi_{k})\in \R^{k}\times (\R^{k}\setminus\{0\})$ and $\xi_{1}+\xi_{2}\neq 0$. We claim that such a point is a regular directed point for $\delta(x_{1}-x_{2})$ (i.e. it does not belong to the wave front set). Indeed, let $\varphi\in C_{c}^{\infty}(\R^{k})$ be such that $\varphi(\ux_{k})\neq 0$. Then
\begin{equation}
\mathcal{F}\paren*{\delta(x_{1}-x_{2})\varphi}(\uxi_{k}') = \int_{\R^{k-1}}d\ul{y}_{2;k} \varphi(y_{2},\ul{y}_{2;k})e^{-i(\xi_{1}'+\xi_{2}')y_{2} + \uxi_{3;k}'\cdot \ul{y}_{3;k}}, \qquad \forall \uxi_{k}'\in \R^{k}.
\end{equation}
Since $\varphi$ is Schwartz class, repeated integration by parts in $\ul{y}_{2;k}$ yields
\begin{equation}
\left|\mathcal{F}\paren*{\delta(x_{1}-x_{2})\varphi}(\uxi_{k}')\right| \lesssim_{N} \paren*{1+|\xi_{1}'+\xi_{2}'| + |\uxi_{3;k}'|}^{-N}, \qquad \forall N\in\N_{0}.
\end{equation}
We consider two cases based on the values of $\xi_{1}$ and $\xi_{2}$.
\begin{enumerate}[I.]
\item
If $\sgn(\xi_{2})=\sgn(\xi_{1})$, then
\begin{equation}
|\xi_{1}+\xi_{2}| \geq \max\{|\xi_{1}|,|\xi_{2}|\},
\end{equation}
which implies that
\begin{equation}
\paren*{1+|\xi_{1}+\xi_{2}| + |\uxi_{3;k}|}^{-N} \lesssim_N \paren*{1+|\uxi_{k}|}^{-N}.
\end{equation}
Hence, if $\varepsilon>0$ is sufficiently small so that $\sgn(\xi_{1}')=\sgn(\xi_{2}')$ for all $\uxi_{k}'\in B(\uxi_{k},\varepsilon)$, then
\begin{equation}
\left|\mathcal{F}\paren*{\delta(x_{1}-x_{2})\varphi}(\lambda\uxi_{k}')\right| \lesssim_{N} \paren*{1+\lambda|\uxi_{k}|}^{-N}, \qquad \forall \uxi_{k}'\in B(\uxi_{k},\varepsilon), \enspace \lambda \in [0,\infty).
\end{equation}
\item
If $\sgn(\xi_{2})=-\sgn(\xi_{1})$, then without loss of generality suppose that $|\xi_{1}|>|\xi_{2}|$. Then for $\varepsilon>0$ sufficiently small, we have that there exists $\theta\in (0,1)$ such that
\begin{equation}
\frac{|\xi_{2}'|}{|\xi_{1}'|} \geq \theta, \qquad \forall \uxi_{k}'\in B(\uxi_{k},\varepsilon).
\end{equation}
So by the reverse triangle inequality,
\begin{equation}
\paren*{1+\lambda |\xi_{1}'+\xi_{2}'| + \lambda |\uxi_{3;k}'|}^{-N} \lesssim_{\theta,N} \paren*{1+\lambda|\uxi_{k}|}^{-N}, \qquad \forall \uxi_{k}'\in B(\uxi_{k},\varepsilon), \enspace \lambda \in[0,\infty).
\end{equation}
\end{enumerate}

Now suppose that $(\ux_{k},\uxi_{k}) \in\R^{k}\times (\R^{k}\setminus\{0\})$, $\xi_{1}+\xi_{2}= 0$, and $\uxi_{3;k}\neq 0\in\R^{k-2}$. We claim that such a point is a regular directed point. We consider two cases based on the magnitude of $|\xi_{2}|$ relative to $|\uxi_{3;k}|$.
\begin{enumerate}[I.]
\item
If $|\xi_{1}| \leq |\uxi_{3;k}|$, then for $\varepsilon>0$ sufficiently small,
\begin{equation}
\paren*{1+\lambda|\xi_{1}'+\xi_{2}'| + \lambda|\uxi_{3;k}'|}^{-N} \lesssim_{N} \paren*{1+\lambda|\uxi_{k}|}^{-N}, \qquad \forall \uxi_{k}'\in B(\uxi_{k},\varepsilon), \enspace \lambda \in [0,\infty).
\end{equation}
\item
If $|\xi_{1}|>|\uxi_{3;k}|$, then for $\varepsilon>0$ sufficiently small, there exists $\theta\in (0,1)$ such that
\begin{equation}
\frac{|\uxi_{3;k}'|}{|\xi_{1}'|} \geq \theta, \qquad \forall \uxi_{k}'\in B(\uxi_{k},\varepsilon).
\end{equation}
Hence,
\begin{equation}
|\uxi_{3;k}'| \geq \frac{|\uxi_{3;k}'|}{2} + \frac{\theta}{4}\paren*{|\uxi_{1}'|+|\uxi_{2}'|},
\end{equation}
which implies that
\begin{equation}
\paren*{1+\lambda |\uxi_{3;k}'|}^{-N} \lesssim_{\theta,N} \paren*{1+\lambda|\uxi_{k}|}^{-N}, \qquad \forall \uxi_{k}'\in B(\uxi_{k},\varepsilon), \enspace \lambda\in [0,\infty).
\end{equation}
\end{enumerate}

Thus, we have shown that
\begin{equation}
\WF(\delta(x_{1}-x_{2})) \subset \{(\ux_{k},\uxi_{k})\in\R^{k}\times (\R^{k}\setminus\{0\}) : x_{1}=x_{2}, \enspace \xi_{1}+\xi_{2}=0, \enspace \text{and} \enspace \uxi_{3;k}=0\}.
\end{equation}
For the reverse inclusion, we claim that $(\ux_{k},(-\xi_{2},\xi_{2},\ul{0}_{3;k}))\in \R^{k}\times (\R^{k}\setminus\{0\})$ is not a regular directed point for $\delta(x_{1}-x_{2})$. Indeed, this claim follows from observing that for a bump function $\varphi\in C_{c}^{\infty}(\R^{k})$ about $\ux_{k}$, we have that for all $\lambda\in [0,\infty)$,
\begin{equation}
\left|\mathcal{F}\paren*{\delta(x_{1}-x_{2})\varphi}(-\lambda\xi_{2},\lambda\xi_{2},\ul{0}_{3;k})\right| = \int_{\R^{k-1}}d\ux_{2;k} \varphi(x_{2},\ux_{2;k}).
\end{equation}
\end{proof}

We now seek to systematically give meaning to the product of distributions and, in particular, preserve the property that the Fourier transform maps products to convolution. We accomplish this task with a useful criterion due to H\"{o}rmander--one which we make heavy use of in \cref{sec:con_GP_W}--for how to ``canonically'' define the product of two distributions. Before stating H\"{o}rmander's result, we need a few technical preliminaries.

For a closed cone $\Gamma\subset \R^{k}\times (\R^{k}\setminus\{0\})$, define the set
\begin{equation}
\D_{\Gamma}'(\R^{k}) \coloneqq \{u\in\D'(\R^{k}) : \WF(u) \subset \Gamma\}.
\end{equation}

\begin{lemma}
$u\in \D'(\R^{k})$ belongs to $\D_{\Gamma}'(\R^{k})$ if and only if for every $\phi\in C_{c}^{\infty}(\R^{k})$ and every closed cone $V\subset \R^{k}$ satisfying
\begin{equation}\label{eq:cone_con}
\Gamma \cap (\supp(\phi)\times V) = \emptyset,
\end{equation}
we have that
\begin{equation}
\sup_{\uxi_{k} \in V} |\uxi_{k}|^{N} |\wh{(\phi u)}(\uxi_{k})| < \infty, \qquad \forall N\in\N.
\end{equation}
\end{lemma}
\begin{proof}
See Lemma 8.2.1 in \cite{HormPDOI}.
\end{proof}

It is clear that $\D_{\Gamma}'(\R^{k})$ is a subspace of $\D'(\R^{k})$. We say that a sequence $\{u_{j}\}_{j=1}^{\infty}$ in $\D_{\Gamma}'(\R^{k})$ and $u\in\D_{\Gamma}'(\R^{k})$, we say that $u_{j}\rightarrow u$ in $\D_{\Gamma}'(\R^{k})$ as $j\rightarrow\infty$ if $u_{j}\rightarrow u$ in the weak-* topology on $\D'(\R^{k})$ and for every $N\in\N$,
\begin{equation}
\sup_{\uxi_{k}\in V} |\uxi_{k}|^{N} |\wh{(\phi u)}(\uxi_{k}) - \wh{(\phi u_{j})}(\uxi_{k})| \rightarrow 0,
\end{equation}
as $j\rightarrow\infty$, for every $\phi\in C_{c}^{\infty}(\R^{k})$ and closed cone $V\subset \R^{k}$ such that \eqref{eq:cone_con} holds.

\medskip
The next lemma shows that $C_{c}^{\infty}(\R^{k})$ is sequentially dense in the space $\D_{\Gamma}'(\R^{k})$.

\begin{lemma}
For every $u\in \D_{\Gamma}'(\R^{k})$, there exists a sequence $u_{j} \in C_{c}^{\infty}(\R^{k})$ such that $u_{j}\rightarrow u$ in $\D_{\Gamma}'(\R^{k})$.
\end{lemma}
\begin{proof}
See Theorem 8.2.3 in \cite{HormPDOI}.
\end{proof}

\begin{lemma}
Let $m,n\in\N$ and let $f:\R^{m}\rightarrow\R^{n}$ be a $C^{\infty}$ map. Define the set of normals of the map $f$ by
\begin{equation}
N_{f} \coloneqq \{(f(\ux_{m}),\ueta_{n}) \in \R^{n}\times \R^{n} : f'(\ux_{m})^{T}\ueta_{n}=0\},
\end{equation}
where $f'(\ux_{m})^{T}$ denotes the transpose of the matrix $f'(\ux_{m})$. Then the pullback distribution $f^{*}u$ can be defined in one and only one way for all $u\in\D'(\R^{n})$ with
\begin{equation}\label{eq:wf_norm}
N_{f} \cap \WF(u) = \emptyset
\end{equation}
so that $f^{*}u=u\circ f$, when $u\in C^{\infty}(\R^{n})$ and for any closed conic subset $\Gamma\subset \R^{n}\times (\R^{n}\setminus\{0\})$ satisfying $\Gamma\cap N_{f}=\emptyset$, we have a continuous map $f^{*}:\D'_{\Gamma}(\R^{n})\rightarrow \D'_{f^{*}\Gamma}(\R^{m})$, where
\begin{equation}
f^{*}\Gamma \coloneqq \{(\ux_{m},f'(\ux_{m})^{T}\ueta_{n}) : (f(\ux_{m}),\ueta_{n})\in\Gamma\}.
\end{equation}
In particular, for every $u\in\D'(\R^{n})$ satisfying \eqref{eq:wf_norm}, we have that
\begin{equation}
\WF(f^{*}u)\subset f^{*}\WF(u).
\end{equation}
\end{lemma}
\begin{proof}
See Theorem 8.2.4 in \cite{HormPDOI}.
\end{proof}

We are now prepared to state H\"{o}rmander's criterion for the existence of the product of two distributions.

\begin{prop}[H\"{o}rmander's criterion]\label{prop:H_crit}
Let $u_{1},u_{2}\in\D'(\R^{k})$, and suppose that
\begin{equation}
\WF(u_{1}) \oplus \WF(u_{2}) \coloneqq \{(\ux_{k}, \uxi_{k}) \in \R^{k} \times (\R^{k}\setminus\{0\}) : \uxi_{k} = \uxi_{1,k} + \uxi_{2,k}, \enspace (\ux_{k},\uxi_{j,k}) \in \WF(u_{j}) \enspace \text{for $j=1,2$}\}
\end{equation}
does not contain an element of the form $(\ux_{k},0)$. Then the product $u_{1}u_{2}$ can be defined as the pullback of the tensor product $u_{1}\otimes u_{2}$ by the diagonal map $d:\R^{k}\rightarrow \R^{2k}$. Moreover,
\begin{equation}
\WF(u_{1}u_{2}) \subset \WF(u_{1}) \cup \WF(u_{2}) \cup \paren*{\WF(u_{1}) \oplus \WF(u_{2})}.
\end{equation}
We refer to this definition of the product $u_1 u_2$ as the \emph{H\"ormander product}.
\end{prop}
\begin{proof}
See Theorem 8.2.10 in \cite{HormPDOI}.
\end{proof}

Sometimes it is easy to make an ansatz for an explicit formula for the product of two distributions, for example $\delta(x_1-x_2)\delta(x_2-x_3)$. The next lemma is useful for verifying that the ansatz indeed coincides with the product distribution defined by \cref{prop:H_crit}.

\begin{lemma}\label{lem:p_uniq}
Let $u,v\in\D'(\R^{k})$. Then there exists at most one distribution $w\in \D'(\R^{k})$ such that for every $\ux_{k}\in\R^{k}$, there exists $\phi\in C_{c}^{\infty}(\R^{k})$ which is $\equiv 1$ on $B(\ux_{k},\varepsilon)$, for some $\varepsilon>0$, and such that for every $\uxi_{k}\in\R^{k}$,
\begin{equation}
\mathcal{F}(\phi u) \cdot \mathcal{F}(\phi v)(\uxi_{k}-\cdot) \in L^{1}(\R^{k}),
\end{equation}
the map
\begin{equation}
\R^{k} \rightarrow \C, \qquad \uxi_{k} \mapsto \paren*{\F(\phi u) \ast \F(\phi v)}(\uxi_{k})
\end{equation}
is polynomially bounded, and
\begin{equation}\label{eq:pdis_con}
\mathcal{F}(\phi^{2} w)(\uxi_{k}) = (2\pi)^{-k/2}\int_{\R^{k}}d\ueta_{k}\mathcal{F}(\phi u)(\ueta_{k})\mathcal{F}(\phi v)(\uxi_{k}-\ueta_{k}).
\end{equation}
\end{lemma}
\begin{proof}
We first claim that for any $\psi\in C_{c}^{\infty}(\R^{k})$,
\begin{equation}\label{eq:pre_res}
\F(\psi \phi^{2}w)(\uxi_{k}) = (2\pi)^{-k/2}\paren*{\F(\psi\phi u_{1}) \ast \F(\phi u_{2})}(\uxi_{k}) = (2\pi)^{-k/2} \paren*{\F(\phi u_{1})\ast \F(\psi\phi u_{2})}(\uxi_{k}), 
\end{equation}
for all $\uxi_{k}\in\R^{k}$ where the integrals defining the convolutions converge absolutely for $\uxi_{k}$ fixed. Indeed, since $\wh{\psi}$ is Schwartz and $\F(\phi^{2}w)$ is analytic,
\begin{align}
\F(\psi \phi^{2}w)(\uxi_{k}) &= (2\pi)^{-k/2}\int_{\R^{k}}d\ueta_{k} \F(\psi)(\uxi_{k}-\ueta_{k}) \F(\phi^{2}w)(\ueta_{k}) \nonumber\\
&= (2\pi)^{-k/2}\int_{\R^{k}}d\ueta_{k} \F(\psi)(\uxi_{k}-\ueta_{k}) \paren*{\int_{\R^{k}}d\ueta_{k}'\F(\phi u_{1})(\ueta_{k}-\ueta_{k}') \F(\phi u_{2})(\ueta_{k}')},
\end{align}
where the integrals are absolutely convergent. Hence, by the Fubini-Tonelli theorem,
\begin{equation}
\begin{split}
&\int_{\R^{k}}d\ueta_{k} \F(\psi)(\uxi_{k}-\ueta_{k}) \paren*{\int_{\R^{k}}d\ueta_{k}'\F(\phi u_{1})(\ueta_{k}-\ueta_{k}') \F(\phi u_{2})(\ueta_{k}')} \\
&\phantom{=} = \int_{\R^{k}}d\ueta_{k}'\F(\phi u_{2})(\ueta_{k}') \paren*{\int_{\R^{k}}d\eta_{k}\F(\psi)(\uxi_{k}-\ueta_{k})\F(\phi u_{1})(\ueta_{k}-\ueta_{k}')}.
\end{split}
\end{equation}
By the translation invariance of the Lebesgue measure,
\begin{align}
\int_{\R^{k}} d\ueta_{k}\F(\psi)(\uxi_{k}-\ueta_{k})\F(\phi u_{1})(\ueta_{k}-\ueta_{k}') &= \int_{\R^{k}}d\ueta_{k}\F(\psi)(\uxi_{k}-\ueta_{k}'-\ueta_{k})\F(\phi u_{1})(\ueta_{k}) \nonumber\\
&= \paren*{\F(\psi) \ast \F(\phi u_{1})}(\uxi_{k}-\ueta_{k}') \nonumber\\
&= (2\pi)^{k/2} \F(\psi\phi u_{1})(\uxi_{k}-\ueta_{k}'),
\end{align}
where the ultimate equality follows from Fourier inversion. Therefore,
\begin{equation}
(2\pi)^{-k/2}\int_{\R^{k}}d\ueta_{k}'\F(\phi u_{2})(\ueta_{k}') \paren*{\int_{\R^{k}}d\eta_{k}\F(\psi)(\uxi_{k}-\ueta_{k})\F(\phi u_{1})(\ueta_{k}-\ueta_{k}')} = \paren*{\F(\psi\phi u_{1}) \ast \F(\phi u_{2})}(\uxi_{k}).
\end{equation}
By symmetry, we have also shown that
\begin{equation}
\F(\psi \phi^{2}w)(\uxi_{k}) = \paren*{\F(\phi u_{1}) \ast \F(\psi\phi u_{2})}(\uxi_{k}).
\end{equation}

Now suppose that $w_{1},w_{2}\in\D'(\R^{k})$ are two distributions such that there exist $\phi_{1},\phi_{2}\in C_{c}^{\infty}(\R^{k})$ so that
\begin{align}
\F(\phi_{1}^{2}w_{1}) &= \paren*{\F(\phi_{1}u_{1})\ast\F(\phi_{1}u_{2})}\\
\F(\phi_{2}^{2}w_{2}) &= \paren*{\F(\phi_{2}u_{1}) \ast \F(\phi_{2}u_{2})},
\end{align}
where the integrals defining the convolutions are absolutely convergent for fixed $\uxi_{k}$ and there exists $N_{1},N_{2}\in \N_{0}$ so that
\begin{align}
&\sup_{\uxi_{k}\in\R^{k}}\jp{\uxi_{k}}^{-N_{1}}\int_{\R^{k}}d\ueta_{k}\left|\F(\phi_{1}u_{1})(\uxi_{k}-\ueta_{k})\F(\phi_{1}u_{2})(\ueta_{k})\right| <\infty  \label{eq:poly_bd}\\
&\sup_{\uxi_{k}\in\R^{k}} \jp{\uxi_{k}}^{-N_{2}}\int_{\R^{k}}d\ueta_{k} \left|\F(\phi_{2}u_{1})(\uxi_{k}-\ueta_{k})\F(\phi_{2}u_{2})(\ueta_{k})\right| <\infty.
\end{align}
Then by \eqref{eq:pre_res},
\begin{align}
\F(\phi_{1}^{2}\phi_{2}^{2}w_{1}) = (2\pi)^{-k/2} \F(\phi_{2}) \ast \F(\phi_{2}\phi_{1}^{2}w_{1}) &= (2\pi)^{-k/2}\F(\phi_{2}) \ast \paren*{\F(\phi_{1}u_{1}) \ast \F(\phi_{1}\phi_{2}u_{2})} \nonumber\\
&= (2\pi)^{-k/2} \F(\phi_{2}\phi_{1}u_{1})\ast\F(\phi_{2}\phi_{1}u_{2}),
\end{align}
where the ultimate equality is justified since $\F(\phi_{2})$ is a Schwartz function and the fact that there exists some $N\in\N$ so that
\begin{equation}
\sup_{\uxi_{k}\in\R^{k}} \jp{\uxi_{k}}^{-N}\int_{\R^{k}}d\ueta_{k} \left|\F(\phi_{1}u_{1})(\uxi_{k}-\ueta_{k}) \F(\phi_{1}\phi_{2}u_{2})(\ueta_{k})\right| < \infty,
\end{equation}
which is a consequence of \eqref{eq:poly_bd}. Similarly,
\begin{equation}
\F(\phi_{1}^{2}\phi_{2}w_{2}) = (2\pi)^{-k/2}\F(\phi_{1}\phi_{2}u_{1}) \ast \F(\phi_{1}\phi_{2}u_{2}),
\end{equation}
which shows that $\F(\phi_{1}^{2}\phi_{2}^{2}w_{1})=\F(\phi_{1}^{2}\phi_{2}^{2}w_{2})$. By a localization argument (see, for instance, Theorem 2.2.1 in \cite{HormPDOI}), it follows that $w_{1}=w_{2}$ in $\D'(\R^{k})$, completing the proof of the lemma.
\end{proof}

Lastly, we record some basic properties of the product of two distributions, when it exists.

\begin{prop}[Properties of product]
The following properties hold:
\begin{enumerate}[(a)]
\item
If $f\in\D(\R^{k})$ and $u\in\D'(\R^{k})$, then the usual definition of the $fu$ coincides with \cref{prop:H_crit}.
\item
If $u,v,w\in\D'(\R^{k})$ and the products $uv$, $(uv)w$, $vw$, and $u(vw)$ all exists, then $u(vw) = (uv)w$. Furthermore, if $uv$ exists, then $vu$ also exists and $uv=vu$.
\item
If $u,v\in \D'(\R^{k})$ have disjoint singular supports, then $uv$ exists and is given by the product distribution guaranteed by \cref{prop:sing_prop}.
\item
If $u,v\in\D'(\R^{k})$ and $uv$ exists, then $\supp(uv)\subset \supp(u)\cap \supp(v)$.
\end{enumerate}
\end{prop}
\begin{proof}
See Theorem IX.43 in \cite{RSII}.
\end{proof}

\pagebreak
\begin{table}[h]
   \caption{Notation} 
   \label{tab:notation}
   \small 
   \centering 
   \begin{tabular}{|p{4cm}|p{12cm}|} 
      \hline
   \textbf{Symbol} & \textbf{Definition} \\ 
      \hline
   $\ul{x}_{k}$ or $\ul{x}_{i;i+k}$ & $(x_1, \ldots, x_k)$ or $(x_i, \ldots, x_{i+k})$ \\
      $d\ul{x}_{k}$ or $d\ul{x}_{i;i+k}$ & $ dx_1 \cdots dx_k$ or $dx_i \cdots dx_{i+k}$ \\
         $\N$ or $\N_0$ & natural numbers or natural numbers inclusive of zero\\
      $\mathbb{N}_{\leq i}$ or $\mathbb{N}_{\geq i}$  & $\{ n \in \mathbb{N} \,:\, n \leq i \}$ or $\{ n \in \mathbb{N} \,:\, n \geq i \}$ \\
      $\Ss_{k}$ & symmetric group on $k$ elements \\
      $C_c^\infty(\R^k)$ or $\D(\R^k)$ & smooth, compactly supported functions on $\R^k$\\
      $\Sc(\R^{k})$ or $\Sc_s(\R^{k})$ & Schwartz space or bosonic Schwartz space on $\R^k$: \cref{sym_schwartz} \\
      $\Sc(\R^k;\V)$ & Schwartz functions on $\R^k$ with values in $\V$:  \eqref{eq:Schw_ms}, \eqref{equ:mixed} \\
      $\Sc'(\R^{k})$ or $\Sc_s'(\R^k)$ & tempered distributions or bosonic tempered distributions on $\R^k$ \\
      $\mathcal{D}'(\R^{k})$ & distributions on $\R^k$ \\
      $\L(E, F)$ & continuous linear maps between locally convex spaces $E$ and $F$\\
      $dF$ & the G\^ateaux derivative of $F$: \cref{gateaux_deriv}\\
      $\grad$ or $\grad_{s}$, $\grad_{s,\V}$, $\grad_{s,\C}$ & the real or symplectic $L^2$ gradients: \cref{def:re_grad} and \cref{schwartz_deriv}, \cref{prop:Schw_WP_V}, \cref{schwartz2_wpoiss}\\
      $\grad_{1},\grad_{\bar{1}}, \grad_{2},\grad_{\bar{2}}$ & variational derivatives: \eqref{eq:vd_prop}, \eqref{eq:vd_prop_V} \\
       $A_{(\pi(1),\ldots, \pi(k))}^{(k)}$ & conjugation of an operator by a permutation: see \eqref{eq:op_coord} \\
       $\Sym_k(f)$ & symmetrization operator for functions: \cref{def:sym_f} \\
       $\Sym_k(A^{(k)})$, $\Sym(A)$ & symmetrization operator for operators: \cref{def:sym_A}\\
       $B_{i;j}^{\pm}, B_{i;j}$ & contraction operators: \cref{eq:B_con} \cref{eq:B_pm}\\
       $\phi^{\otimes k}$ or  $\phi^{\times k}$ & $k$-fold tensor or cartesian product of $\phi$ with itself: \eqref{tensor_def} or  \eqref{prod_coords}\\
       $\omega_{L^2}$, $\omega_{L^2,\V}$, $\omega_{L^2,\V}$ & $L^2$ symplectic forms: \eqref{l2_symp}, \eqref{symp_v}, \eqref{eq:om_L2_sq}  \\
       $\A_{\Sc}$, $\A_{\Sc,\V}$, $\A_{\Sc,\C}$ & see \eqref{equ:Asc}, \eqref{alg_v}, \eqref{eq:Asc}\\
       $\pb{\cdot}{\cdot}_{L^2}, \pb{\cdot}{\cdot}_{L^2,\V}, \pb{\cdot}{\cdot}_{L^2,\C}$  & $L^2$ Poisson brackets: \eqref{l2_bracket}, \eqref{eq:pb_L2_V}, \eqref{eq:pb_L2_C} \\
        $(\G_\infty, [\cdot, \cdot]_{\G_\infty})$ & Lie algebra of observable $\infty$-hierarchies: see discussion around \cref{prop:G_inf_br}\\
        $(\G_\infty^*, \A_\infty, \{\cdot, \cdot\}_{\G_\infty^*})$ & {Lie-Poisson manifold of density matrix $\infty$-hierarchies: \eqref{equ:poisson_def} and discussion around \cref{prop:LP}}\\
        $w_n$, $w_{n,(\psi_1,\psi_2)}$ & recursive functions: \eqref{eq:w_rec}, \eqref{eq:wn_bv_rec} \\
        $w_{n}^{(k)}$; $w_{n,j}^{(k),t}$, $w_{n,j'}^{(k),t}$ & $k$-particle component of $w_n$: \eqref{eq:wn(k)_recur}; partial transposes of $w_{n}^{(k)}$: \cref{lem:wn(k)_tran}  \\
        $I_n$, $\tl{I}_n$, $I_{b,n}$ & involutive functionals: \eqref{eq:In_def}, \eqref{tilde_in_intro}, \eqref{eq:Ibn_intro_def} \\
        $\widetilde{\W}_{n}$ & the unsymmetrized operators: \eqref{eq:Wn_rec}\\
	$\W_{n,sa}$ & the self-adjoint operators: \eqref{wn_sa_def}\\
	$\W_{n}$ & the bosonic, self-adjoint operators:  \eqref{eq:Wn_fin_def}\\
        $\H_n$ & the $n$-th Hamiltonian functional: \eqref{Hn_trace}\\
        $\Tr_{1,\ldots,N}$ & generalized trace: \cref{def:gen_trace}\\
        $\Tr_{k+1,\ldots,N}$ & generalized partial trace: \cref{prop:partial_trace}\\
        $WF(u)$ & wave front set of a distribution $u$: \cref{def:WF}\\
        \hline
   \end{tabular}
\end{table}

\bibliographystyle{siam}
\bibliography{GPHam}

\begin{thebibliography}{10}

\bibitem{AKNS74}
{\sc M.~J. Ablowitz, D.~J. Kaup, A.~C. Newell, and H.~Segur}, {\em The inverse
  scattering transform-{Fourier} analysis for nonlinear problems}, Studies in
  Applied Mathematics, 53 (1974), pp.~249--315.

\bibitem{ABGT2004}
{\sc R.~Adami, C.~Bardos, F.~Golse, and A.~Teta}, {\em Towards a rigorous
  derivation of the cubic {NLSE} in dimension one}, Asymptotic Analysis, 40
  (2004), pp.~93--108.

\bibitem{AmBre2012}
{\sc Z.~Ammari and S.~Breteaux}, {\em Propagation of chaos for many-boson
  systems in one dimension with a point pair-interaction}, Asymptot. Anal., 76
  (2012), pp.~123--170.

\bibitem{BC1984}
{\sc R.~Beals and R.~R. Coifman}, {\em Scattering and inverse scattering for
  first order systems}, Communications on Pure and Applied Mathematics, 37
  (1984), pp.~39--90.

\bibitem{BC1985}
\leavevmode\vrule height 2pt depth -1.6pt width 23pt, {\em Inverse scattering
  and evolution equations}, Communications on pure and applied mathematics, 38
  (1985), pp.~29--42.

\bibitem{BC1987}
{\sc R.~Beals and R.~R. Coifman}, {\em Scattering and inverse scattering for
  first-order systems. {II}}, Inverse Problems, 3 (1987), pp.~577--593.

\bibitem{BC1989}
{\sc R.~Beals and R.~R. Coifman}, {\em Linear spectral problems, non-linear
  equations and the $\delta$-method}, Inverse problems, 5 (1989), p.~87.

\bibitem{Bethe1931}
{\sc H.~Bethe}, {\em Zur theorie der metalle}, Zeitschrift f{\"u}r Physik, 71
  (1931), pp.~205--226.

\bibitem{Calogero1991}
{\sc F.~Calogero}, {\em Why are certain nonlinear {PDE}s both widely applicable
  and integrable?}, in What is integrability?, Springer Ser. Nonlinear Dynam.,
  Springer, Berlin, 1991, pp.~1--62.

\bibitem{CGLM2008}
{\sc P.~Comon, G.~Golub, L.-H. Lim, and B.~Mourrain}, {\em Symmetric tensors
  and symmetric tensor rank}, SIAM J. Matrix Anal. Appl., 30 (2008),
  pp.~1254--1279.

\bibitem{Davies1990}
{\sc B.~Davies}, {\em Higher conservation laws for the quantum nonlinear
  {S}chr\"{o}dinger equation}, Phys. A, 167 (1990), pp.~433--456.

\bibitem{deift_survey}
{\sc P.~Deift}, {\em {Fifty Years of KdV: An Integrable System }},
  arXiv:1902.10267.

\bibitem{DZ2003}
{\sc P.~Deift and X.~Zhou}, {\em Long-time asymptotics for solutions of the
  {NLS} equation with initial data in a weighted {S}obolev space}, vol.~56,
  2003, pp.~1029--1077.
\newblock Dedicated to the memory of J\"{u}rgen K. Moser.

\bibitem{FT07}
{\sc L.~D. Faddeev and L.~A. Takhtajan}, {\em {Hamiltonian methods in the
  theory of solitons}}, Classics in Mathematics, Springer, Berlin, english~ed.,
  2007.

\bibitem{FKP}
{\sc J.~Fr{\"{o}}hlich, A.~Knowles, and A.~Pizzo}, {\em {Atomism and
  quantization}}, J. Phys. A, 40 (2007), pp.~3033--3045.

\bibitem{FTY2000}
{\sc J.~Fr\"{o}hlich, T.-P. Tsai, and H.-T. Yau}, {\em On a classical limit of
  quantum theory and the non-linear {H}artree equation}, no.~Special Volume,
  Part I, 2000, pp.~57--78.
\newblock GAFA 2000 (Tel Aviv, 1999).

\bibitem{GGKM}
{\sc C.~S. Gardner, J.~M. Greene, M.~D. Kruskal, and R.~M. Miura}, {\em
  Korteweg-de{V}ries equation and generalization. {VI}. {M}ethods for exact
  solution}, Comm. Pure Appl. Math., 27 (1974), pp.~97--133.

\bibitem{Gaudin2014}
{\sc M.~Gaudin}, {\em The {B}ethe wavefunction}, Cambridge University Press,
  New York, 2014.
\newblock Translated from the 1983 French original by Jean-S\'{e}bastien Caux.

\bibitem{Greub1978}
{\sc W.~Greub}, {\em Multilinear algebra}, Springer-Verlag, New
  York-Heidelberg, second~ed., 1978.
\newblock Universitext.

\bibitem{HormPDOI}
{\sc L.~H{\"{o}}rmander}, {\em The Analysis of Linear Partial Differential
  Operators: Distribution Theory and Fourier Analysis}, Springer-Verlag, 1983.

\bibitem{Horvath1966}
{\sc J.~Horv{\'a}th}, {\em Topological vector spaces and distributions}, no.~v.
  1 in Addison-Wesley series in mathematics, Addison-Wesley Pub. Co., 1966.

\bibitem{HM1976}
{\sc R.~L. Hudson and G.~R. Moody}, {\em Locally normal symmetric states and an
  analogue of {de Finetti's} theorem}, Zeitschrift f{\"{u}}r
  Wahrscheinlichkeitstheorie und Verwandte Gebiete, 33 (1976), pp.~343--351.

\bibitem{KBI1993}
{\sc V.~E. Korepin, N.~M. Bogoliubov, and A.~G. Izergin}, {\em Quantum inverse
  scattering method and correlation functions}, Cambridge Monographs on
  Mathematical Physics, Cambridge University Press, Cambridge, 1993.

\bibitem{Lax68}
{\sc P.~D. Lax}, {\em Integrals of nonlinear equations of evolution and
  solitary waves}, Comm. Pure Appl. Math., 21 (1968), pp.~467--490.

\bibitem{LNR2015}
{\sc M.~Lewin, P.~T. Nam, and N.~Rougerie}, {\em {Derivation of nonlinear
  {G}ibbs measures from many-body quantum mechanics}}, J. {\'{E}}c. polytech.
  Math., 2 (2015), pp.~65--115.

\bibitem{LL1963_I}
{\sc E.~H. Lieb and W.~Liniger}, {\em Exact analysis of an interacting {B}ose
  gas. {I}. {T}he general solution and the ground state}, Phys. Rev. (2), 130
  (1963), pp.~1605--1616.

\bibitem{MW1983}
{\sc J.~Marsden and A.~Weinstein}, {\em Coadjoint orbits, vortices, and
  {C}lebsch variables for incompressible fluids}, vol.~7, 1983, pp.~305--323.
\newblock Order in chaos (Los Alamos, N.M., 1982).

\bibitem{MMW1984}
{\sc J.~E. Marsden, P.~J. Morrison, and A.~Weinstein}, {\em The {H}amiltonian
  structure of the {BBGKY} hierarchy equations}, in Fluids and plasmas:
  geometry and dynamics ({B}oulder, {C}olo., 1983), vol.~28 of Contemp. Math.,
  Amer. Math. Soc., Providence, RI, 1984, pp.~115--124.

\bibitem{MR2013}
{\sc J.~E. Marsden and T.~S. Ratiu}, {\em Introduction to mechanics and
  symmetry: a basic exposition of classical mechanical systems}, vol.~17,
  Springer Science \& Business Media, 2013.

\bibitem{MNPRS1_2019}
{\sc D.~Mendelson, A.~R. Nahmod, N.~Pavlovi{\'c}, M.~Rosenzweig, and
  G.~Staffilani}, {\em A rigorous derivation of the {Hamiltonian} structure for
  the nonlinear {Schr\"odinger} equation}, arXiv preprint arXiv:1908.03847,
  (2019).

\bibitem{MNPS2016}
{\sc D.~Mendelson, A.~R. Nahmod, N.~Pavlovi\'{c}, and G.~Staffilani}, {\em An
  infinite sequence of conserved quantities for the cubic {G}ross-{P}itaevskii
  hierarchy on {$\Bbb{R}$}}, Trans. Amer. Math. Soc., 371 (2019),
  pp.~5179--5202.

\bibitem{Milnor1984}
{\sc J.~Milnor}, {\em {Remarks on infinite-dimensional Lie groups}}, in
  Relativ. groups Topol. 2, 1984.

\bibitem{NST2014}
{\sc K.-H. Neeb, H.~Sahlmann, and T.~Thiemann}, {\em Weak poisson structures on
  infinite dimensional manifolds and hamiltonian actions}, in Lie Theory and
  Its Applications in Physics, V.~Dobrev, ed., Tokyo, 2014, Springer Japan,
  pp.~105--135.

\bibitem{Omori1979}
{\sc H.~Omori}, {\em Infinite-dimensional {L}ie groups}, vol.~158 of
  Translations of Mathematical Monographs, American Mathematical Society,
  Providence, RI, 1997.
\newblock Translated from the 1979 Japanese original and revised by the author.

\bibitem{Palais1997}
{\sc R.~Palais}, {\em The symmetries of solitons}, Bulletin of the American
  Mathematical Society, 34 (1997), pp.~339--403.

\bibitem{RSII}
{\sc M.~Reed and B.~Simon}, {\em Methods of modern mathematical physics. {II}.
  {F}ourier analysis, self-adjointness}, Academic Press [Harcourt Brace
  Jovanovich, Publishers], New York-London, 1975.

\bibitem{Schwartz1966}
{\sc L.~Schwartz}, {\em Th\'{e}orie des distributions}, Publications de
  l'Institut de Math\'{e}matique de l'Universit\'{e} de Strasbourg, No. IX-X.
  Nouvelle \'{e}dition, enti\'{e}rement corrig\'{e}e, refondue et
  augment\'{e}e, Hermann, Paris, 1966.

\bibitem{Stormer1969}
{\sc E.~St{\"o}rmer}, {\em {Symmetric states of infinite tensor products of
  {$C^*$} Algebras}}, J. Funct. Anal., 3 (1969), pp.~48--68.

\bibitem{Terng1997}
{\sc C.-L. Terng}, {\em Soliton equations and differential geometry}, J.
  Differential Geom, 45 (1997), pp.~407--445.

\bibitem{TU1998}
{\sc C.-L. Terng and K.~Uhlenbeck}, {\em Poisson actions and scattering theory
  for integrable systems}, Surveys in Differential Geometry, 4 (1998),
  pp.~315--402.

\bibitem{Thacker1978_pcl}
{\sc H.~B. Thacker}, {\em Polynomial conservation laws in (1 + 1)-dimensional
  classical and quantum field theory}, Phys. Rev. D, 17 (1978), pp.~1031--1040.

\bibitem{Treves1967}
{\sc F.~Tr{\`e}ves}, {\em Topological vector spaces, distributions and
  kernels}, Academic Press, New York-London, 1967.

\bibitem{ZS79}
{\sc V.~E. Zaharov and A.~B. {\v{S}}abat}, {\em {Integration of the nonlinear
  equations of mathematical physics by the method of the inverse scattering
  problem. {II}}}, Funktsional. Anal. i Prilozhen., 13 (1979), pp.~13--22.

\bibitem{Zhakharov91}
{\sc V.~E. Zakharov}, ed., {\em What is integrability?}, Springer Series in
  Nonlinear Dynamics, Springer-Verlag, Berlin, 1991.

\bibitem{ZS72}
{\sc V.~E. Zakharov and A.~B. Shabat}, {\em {Exact theory of two-dimensional
  self-focusing and one-dimensional self-modulation of waves in nonlinear
  media}}, {\v{Z}}. {\`{E}}ksper. Teoret. Fiz., 61 (1971), pp.~118--134.

\bibitem{Zhou1989}
{\sc X.~Zhou}, {\em Direct and inverse scattering transforms with arbitrary
  spectral singularities}, Comm. Pure Appl. Math., 42 (1989), pp.~895--938.

\bibitem{Zhou1998}
\leavevmode\vrule height 2pt depth -1.6pt width 23pt, {\em {$L^2$}-{S}obolev
  space bijectivity of the scattering and inverse scattering transforms}, Comm.
  Pure Appl. Math., 51 (1998), pp.~697--731.

\end{thebibliography}
\end{document}